%% file: SecrecyArxiv.tex
\documentclass[10pt,twoside,a4paper]{article}
\usepackage[ngerman,USenglish]{babel} % Different languages:
\usepackage{amsmath}                % AMS Fontsthenth
\usepackage{amsfonts}               % AMS Fonts 
\usepackage{amssymb}                % Mathemat. symbols of AMS
\usepackage{amsopn}                 % Further math. operators like
\allowdisplaybreaks[3] % Page breaks within eqnarray and IEEEeqnarray
\usepackage{amsthm}                % Other choices for Theorem. 
\usepackage{bbm}                    % $\mathbbm{1}$ yields a special ``1''
\usepackage{mathrsfs}               % new mathfonts: \mathscr{ABC}
\usepackage{array}                  % Improved array and
\usepackage{calc}                   % Calculations within the tex-file
\usepackage{graphicx}        % Further graphic-commands (use
\usepackage{psfrag}                 % Include PS-graphics with LaTeX-
\usepackage{subfigure}              % Put figures into groups
\usepackage{wrapfig}                % Put pictures to the border of
\usepackage{color}           % Coloring of text
\usepackage{fancyhdr}               % Special page layouts
\usepackage{verbatim}               % citations directly included
\usepackage{moreverb}               % even better citations...
\usepackage{exscale}                % Improving on a weakness of
\usepackage{enumerate}
\usepackage{IEEEtrantools}

\usepackage{ifthen}
\usepackage{pgf,tikz}
\usetikzlibrary{positioning,arrows,shapes,decorations}
\usepackage{float}
\usepackage{wasysym}
\newcommand{\overbar}[1]{\mkern 1.5mu\overline{\mkern-1.5mu#1\mkern-1.5mu}\mkern 1.5mu}

\input{math-macrosRD.tex}

\newtheorem{theorem}{Theorem}
\newtheorem{lemma}[theorem]{Lemma}
\newtheorem{corollary}[theorem]{Corollary}

\newtheorem{defin}{Definition}
\newtheorem{remark}[theorem]{Remark}

\newtheorem{fact}[theorem]{Fact}
\newtheorem{example}{Example}

\fussy
\large \normalsize
\begin{document}
%
%%%%%%%%%%%%%%%%%%%%%%%%%%%%%%%%%%%%%%%%%%%%%%%%%%%%%%%%%%%%%%%%%%
%% LANGUAGE OF DOCUMENT
%%%%%%%%%%%%%%%%%%%%%%%%%%%%%%%%%%%%%%%%%%%%%%%%%%%%%%%%%%%%%%%%%%
%
\selectlanguage{USenglish}
\pagenumbering{arabic}
%
%%%%%%%%%%%%%%%%%%%%%%%%%%%%%%%%%%%%%%%%%%%%%%%%%%%%%%%%%%%%%%%%%%
\title{Guessing Attacks on Distributed-Storage Systems}
  
\author{Annina Bracher, Eran Hof, and Amos Lapidoth}

\maketitle

\huge
\begin{abstract}
  \setcounter{page}{1}
%  \addcontentsline{toc}{chapter}{\numberline{}Abstract}
  \normalsize
  \vspace{0.5cm}

  \let\thefootnote\relax\footnotetext{This paper was presented in part at the 2014 IEEE Information Theory Workshop and in part at the 2015 IEEE International Symposium on Information Theory.}
  \let\thefootnote\relax\footnotetext{A.\ Bracher and A.\ Lapidoth are with the Department of Information Technology and Electrical Engineering, ETH Zurich, Switzerland (e-mail: bracher@isi.ee.ethz.ch; lapidoth@isi.ee.ethz.ch), and E.\ Hof is with an R\&D center, Ramat-Gan, Israel (email: eran.hof@gmail.com).}

The secrecy of a distributed-storage system for passwords is studied. The encoder, Alice, observes a length-$n$ password and describes it using two hints, which she stores in different locations. The legitimate receiver, Bob, observes both hints. In one scenario the requirement is that the expected number of guesses it takes Bob to guess the password approach one as $n$ tends to infinity, and in the other that the expected size of the shortest list that Bob must form to guarantee that it contain the password approach one. The eavesdropper, Eve, sees only one of the hints. Assuming that Alice cannot control which hints Eve observes, the largest normalized (by $n$) exponent that can be guaranteed for the expected number of guesses it takes Eve to guess the password is characterized for each scenario. Key to the proof are new results on Arikan's guessing and Bunte and Lapidoth's task-encoding problem; in particular, the paper establishes a close relation between the two problems. A rate-distortion version of the model is also discussed, as is a generalization that allows for Alice to produce $\delta$ (not necessarily two) hints, for Bob to observe $\nu$ (not necessarily two) of the hints, and for Eve to observe $\eta$ (not necessarily one) of the hints. The generalized model is robust against $\delta - \nu$ disk failures.
\end{abstract}
\normalsize

\section{Introduction} \label{sec:intro}

Suppose that some sensitive information $X$ (e.g.\ password) is drawn from a finite set $\setX$ according to some probability mass function (PMF) $P_X$. A (stochastic) encoder, Alice, maps (possibly using randomization) $X$ to two hints $M_1$ and $M_2$ and stores them on different disks in different locations. The hints are intended for a legitimate receiver, Bob, who knows where they are stored and sees both. An eavesdropper, Eve, sees one of the hints but not both; we do not know which. Which hint is revealed to Eve is a subtle question. We adopt a conservative approach and assume that, after observing $X$, an adversarial ``genie'' reveals to Eve the hint that minimizes her ambiguity. Not allowing the genie to observe $X$ would lead to a weaker form of secrecy (Example~\ref{ex:notionOfSecrecy}). Given some notion of ambiguity, we would ideally like Bob's ambiguity about $X$ to be small and Eve's large.

There are several ways to define ambiguity. One approach would be to require that Bob be able to reconstruct $X$ whenever $X$ is ``typical'' and that the conditional entropy of $X$ given Eve's observation be large. For some scenarios, such an approach might be unsuitable. First, it may not properly address Bob's needs when $X$ is not typical. For example, if Bob must guess $X$, this approach does not guarantee that the expected number of guesses be small: it only guarantees that the probability of success after one guess be large. It does not indicate the number of guesses that Bob might need when $X$ is atypical. Second, conditional entropy need not be an adequate measure of Eve's ambiguity. For example, if $X$ is some password that Eve wishes to uncover, then we may care more about the number of guesses that Eve needs than about the conditional entropy \cite{arikanmerhav99}.

In this paper, we assume that Eve wants to guess $X$ with the minimal number of guesses of the form ``Is X = x?". We quantify Eve's ambiguity about $X$ by the expected number of guesses that she needs to uncover $X$. In this sense, Eve faces an instance of the Massey-Arikan guessing problem \cite{massey94,arikan96}: When faced with the problem of guessing $X$ after observing that $Z = z$, where $Z$ denotes Eve's observation, Eve must come up with a guessing order for the elements of $\setX$. Such an order can be specified using a bijective function $\guess {} { \cdot | z }$ from $\setX$ onto the set $\bigl\{ 1, \ldots, | \setX | \bigr\}$---a guessing function with the understanding that if Eve observes $z$, then the question ``Is $X = x$?'' will be her $\guess {}{x|z}$-th question. Eve's expected number of guesses is $\bigEx {}{\guess {}{ X | Z }}$. This expectation is minimized if for each $z \in \setZ$ the guessing function $\guess {} { \cdot | z }$ orders the elements of $\setX$ in decreasing order of their posterior probabilities given $Z = z$.

As to Bob, we will consider two different criteria: In the ``guessing version'' the criterion is the expected number of guesses it takes Bob to guess $X$, and in the ``list version'' the criterion is the expected size of the list that Bob must form to guarantee that it contain $X$. %We shall see that the two criteria lead to similar results.

The former criterion is natural when Bob can check whether a guess is correct: if $X$ is some password, then Bob can stop guessing as soon as he has gained access to the account that is secured by $X$. The latter criterion is appropriate if Bob does not know whether a guess is correct. For example, if $X$ is a task that Bob must perform, then the only way for Bob to make sure that he performs $X$ is to perform all the tasks in the list $\setL_{M_1,M_2}$ comprising the tasks that have positive posterior probabilities given his observation. In this scenario, a good measure for Bob's ambiguity about $X$ is the expected number of tasks that he must perform, i.e., $\bigEx {}{| \setL_{M_1,M_2} |}$, and this will be small whenever Alice is a good task-encoder for Bob \cite{buntelapidoth14}.
%To describe the list of tasks that Bob must perform more explicitly, let us denote by $$\distof{ M_1 = m_1, M_2 = m_2 | X = x }, \quad ( m_1, m_2 ) \in \setM_1 \times \setM_2$$ the probability that Alice produces the pair of hints $( M_1, M_2 ) = ( m_1,m_2 )$ upon observing that $X = x$. It is $\{ 0,1 \}$-valued if Alice does not use randomization. Upon observing that $ ( M_1, M_2 ) = ( m_1, m_2 )$, Bob produces the list $\setL_{m_1,m_2}$ of all the tasks $x \in \setX$ whose posterior probability $\distof{ X = x | M_1 = m_1, M_2 = m_2 }$ is positive. Our notion of Bob's ambiguity about $X$ is $\bigEx {}{| \setL_{M_1,M_2} |}$.

Alternatively, the list-size criterion can also be viewed as a worst-case version of the guessing criterion: Even if Bob is incognizant of the PMF of $X$, the number of guesses it takes him to guess $X$ can be guaranteed not to exceed the size of the smallest list that is guaranteed to contain $X$.

The guessing and the list-size criterion for Bob lead to similar results in the following sense: Clearly, every guessing function $\guess {}{ \cdot | M_1,M_2 }$ for $X$ that guesses the elements of $\setX$ of zero posterior probability only after those of positive posterior probabilities satisfies $\bigEx {}{\guess {}{X | M_1,M_2}} \leq \bigEx {}{| \setL_{M_1,M_2} |}$. Conversely, one can prove that every pair of ambiguities for Bob and Eve that is achievable in the guessing version is---up to polylogarithmic factors of $| \setX |$---also achievable in the list version (Remark~\ref{remark:guessListClose}). These polylogarithmic factors wash out in the asymptotic regime where the sensitive information is an $n$-tuple and $n$ tends to infinity.

Things are different for Eve: applying the list-size criterion for Eve would lead to results that markedly differ from those that apply under the guessing criterion; see Theorem~\ref{th:fairOpponent} and the subsequent discussion.

To derive our results, we establish new results on guessing and task-encoding: we relate task-encoders to guessing functions (Theorem~\ref{th:relGuessEnc}), and we quantify how additional side information can help guessing (Lemma~\ref{le:ImproveGuess}). These results may be of interest in their own right. For example, the former result leads to alternative proofs of Bunte and Lapidoth's asymptotic task-encoding results \cite[Theorems~I.2 and VI.2]{buntelapidoth14} as well as the direct part of \cite[Theorem~I.1]{buntelapidoth14fb}, which states that, in the presence of feedback, the listsize capacity of a discrete-memoryless channel (DMC) with positive zero-error capacity equals the cutoff rate with feedback (which is in fact equal to that without feedback \cite[Corollary~I.4]{buntelapidoth14fb}). The latter result on how additional side information can help guessing is related to \cite{lapidothpfister15}: To quantify how additional side information can help guessing, we establish how an encoder must describe $X$ to minimize the expected number of guesses that a decoder needs to guess $X$. The list-size analog is Lapidoth and Pfister's optimal task-encoder \cite{lapidothpfister15}, which describes $X$ to minimize the expected size of the decoder's list. Despite the close relation between task-encoding and guessing, an optimal encoder for a guessing decoder is typically quite different from an optimal task-encoder.

We also generalize our problem in two different directions. The first, along the lines of \cite{arikanmerhav98, buntelapidoth14}, is a rate-distortion version of the model where Bob and Eve are content with reconstructing the sensitive information to within some given allowed distortion. The second considers the case where Alice produces $\delta$ $s$-bit hints, Bob sees $\nu \leq \delta$ hints, and Eve sees $\eta < \nu$ hints (not necessarily a subset of those that Bob sees). This may model a scenario where the hints are stored on different disks and we want to guarantee robustness against the failure of $\delta - \nu$ disks and the compromise of $\eta$ disks. We adopt again a conservative approach and assume that, after observing $X$, an adversarial genie reveals to Bob the $\nu$ hints that maximize his ambiguity and to Eve the $\eta$ hints that minimize her ambiguity. This guarantees that---no matter which disks fail---the model be robust against the failure of $\delta - \nu$ disks and the compromise of $\eta$ disks. The generalized model is a distributed-storage system, which is static in the sense that failed disks are not replaced.

The case where $X$ is drawn uniformly, Bob must reconstruct $X$, and Eve's observation must satisfy some information-theoretic security criterion (e.g., the mutual information between Eve's observation and $X$ must be null) corresponds to the erasure-erasure wiretap channel studied in \cite{subramanianmclaughlin09} and is a special case of the wiretap networks in \cite{caiyeung11,elrouayhebsoljaninsprintson12}. In the literature, this setting is also known as ``secret sharing.'' In traditional secret sharing, each set of hints either reveals $X$ or reveals no information about $X$ \cite{blakley79,shamir79}. More general are ramp schemes, where any $\nu$ hints reveal $X$ and the amount of information that fewer-than-$\nu$ hints reveal is controlled (see e.g.\ \cite{blakleychaum85}). Our setting is different in that we assume $X \sim P_X$ and in that, using some notion of ambiguity, we quantify how difficult it is for Bob and Eve to reconstruct $X$.

To better bring out the role of R\'{e}nyi entropy, we generalize the models and replace expectations with $\rho$-th moments. (The generalization comes with no extra effort.) For an arbitrary $\rho > 0$, we thus study the $\rho$-th (instead of the first) moment of the list-size and of the number of guesses. Moreover, we shall allow some side information $Y$ that is available to all parties.

The connection between R\'{e}nyi entropy and the $\rho$-th moment of the minimum number of guesses has been studied extensively in the literature \cite{arikan96,pfistersullivan04,sundaresan07,hanawalsundaresan11}. The connection with encoding tasks was studied in \cite{buntelapidoth14}.

The idea to quantify Eve's ambiguity by the $\rho$-th moment of the number of guesses she needs to uncover $X$ is due to Arikan and Merhav, who studied the Shannon cipher system with a guessing wiretapper \cite{arikanmerhav99}. Their approach was later adopted in \cite{hayashiyamamoto08,hanawalsundaresan11b}. The current setting differs from the ones in \cite{arikanmerhav99,hayashiyamamoto08,hanawalsundaresan11b} in the following sense: Instead of mapping $X$ to a public message using a secret key, which is available to Bob but not to Eve, here Alice produces two hints and stores them so that Bob sees both but Eve sees only one. Moreover, unlike \cite{arikanmerhav99,hayashiyamamoto08,hanawalsundaresan11b} we do not measure Bob's ambiguity in terms of the probability that $X$ is not his first guess.\\

The rest of this paper is structured as follows. Section~\ref{sec:notation} briefly describes our notation and summarizes some notions and results pertaining to the guessing problem and the problem of encoding tasks. In Section~\ref{sec:listsAndGuesses}, we quantify how additional side information can help guessing and relate task-encoders to guessing functions, thereby establishing the prerequisites for the proofs of our main results. Section~\ref{sec:problemStatement} contains the problem statement and the main results (both finite-blocklength and asymptotic). The results are discussed in Section~\ref{sec:discussion} and proved in Section~\ref{sec:distStorProofs}. Section~\ref{sec:extDiskFailures} generalizes the model to allow for a limited number of disk failures. Section~\ref{sec:rateDist} considers the rate-distortion version of the problem stated in Section~\ref{sec:problemStatement} and extends the results on guessing and task-encoding of Section~\ref{sec:listsAndGuesses} accordingly. Section~\ref{sec:conclusion} concludes the paper.

\section{Notation and Preliminaries}\label{sec:notation}
In this paper $( X, Y )$ is a pair of chance variables that is drawn from the finite set $\setX \times \setY$ according to the PMF $P_{X,Y}$, and $\rho > 0$ is fixed. We denote by $P_X$ the marginal PMF of $X$ and by $P_Y$ the marginal PMF of $Y$, e.g., $$P_X ( x ) = \sum_{y \in \setY} P_{X,Y} ( x,y ), \,\, \forall \, x \in \setX.$$ For every positive integer $n \in \naturals$ we denote by $P_{X,Y}^n$ the $n$-fold product of $P_{X,Y}$, i.e., $$P_{X,Y}^n ( \vecx, \vecy ) = \prod^n_{i = 1} P_{X,Y} ( x_i, y_i ), \,\, \forall \, (\vecx,\vecy) \in \setX^n \times \setY^n.$$ A generic probability measure on a measurable space $(\Omega, \setF)$ is denoted $\dist$, i.e., whenever we introduce a set of chance variables (e.g., $X$ and $Y$), we denote by $\dist$ the probability measure associated with the probability space $\left( \Omega, \setF, \dist \right)$ on which the chance variables live.

For some positive integer $k$, we denote by $\oplus_k$ addition modulo $k$, so $\alpha \oplus_k \beta$ is for any pair of integers $( \alpha, \beta)$ the unique element $\gamma \in \{ 0, \ldots, k-1 \}$ satisfying $$\gamma \equiv \alpha + \beta \mod k.$$ We denote by $\mathbb F_{q}$ the Galois field with $q$ elements.

By default $\log (\cdot)$ denotes base-2 logarithm, and $\ln (\cdot)$ denotes natural logarithm. We denote by $\alpha \vee \beta$ the maximum of two real numbers $\alpha$ and $\beta$ and by $\alpha \wedge \beta$ their minimum. For some real number $\alpha$, we denote by $[ \alpha ]^+$ the maximum of $\alpha$ and zero $$[ \alpha ]^+ = \alpha \vee 0,$$ by $\left\lceil \alpha \right\rceil$ the smallest integer that is at least as large as $\alpha$, and by $\left\lfloor \alpha \right\rfloor$ the largest integer that is at most as large as $\alpha$. We sometimes use the identity
\begin{equation}
\left\lceil \xi \right\rceil^\rho < 1 + 2^\rho \xi^\rho, \quad \xi \in \reals^+_0, \label{eq:ceilApprox}
\end{equation}
which is easily checked by considering separately the cases $0 \leq \xi \leq 1$ and $\xi > 1$ \cite{buntelapidoth14}.

\subsection{The Conditional R\'enyi Entropy}

To describe our results, we shall need the conditional version of R\'enyi entropy (originally proposed by Arimoto \cite{arimoto77} and also studied in \cite{buntelapidoth14,berensfehr14})
\be \label{eq:conditoinalRenaEntrDef}
\renent {\alpha}{ X | Y } = \frac{\alpha}{1 - \alpha} \log \sum_{y \in \setY} \Biggl( \sum_{x \in \setX} P_{X,Y} ( x,y )^{\alpha} \Biggr)^{1/\alpha},
\ee
where $\alpha \in [ 0,\infty ]$ is the order and where the cases where $\alpha$ is $0$, $1$, or $\infty$ are treated by a limiting argument. Let $\{(X_i,Y_i)\}_{i \in \naturals}$ be a discrete-time stochastic process with finite alphabet $\setX \times \setY$. Whenever the limit as $n$ tends to infinity of $\renent {\alpha}{ X^n | Y^n } / n$ exists, we denote it by $\renent {\alpha}{ \rndvecX | \rndvecY }$ and call it conditional R\'enyi entropy-rate. In this paper $\alpha$ will equal $1 / (1+\rho)$, and thus, since $\rho > 0$, will take values in the set $( 0,1 )$. To simplify notation, we henceforth write $\tirho$ for $1 / (1+\rho)$
\begin{equation}
\tirho \triangleq \frac{1}{1 + \rho}.
\end{equation}

The conditional R\'enyi entropy satisfies the following properties (see, e.g.~\cite[Theorem~2]{berensfehr14}): 

\begin{lemma}\label{le:uncertaintyIncreasesRenEnt}
Let $( X,Y,Z )$ be a triple of chance variables taking values in the finite set $\setX \times \setY \times \setZ$ according to the joint PMF $P_{X,Y,Z}$. For every $\alpha \in [ 0,\infty ]$
\be 
\renent \alpha {X | Y } \leq \renent \alpha {X,Z|Y}. \label{eq:uncertaintyIncreasesRenEnt}
\ee
\end{lemma}

%\begin{proof}
%For the sake of completeness, a proof is provided in Appendix~\ref{app:pfLeUncertInc}.
%\end{proof}

\begin{lemma}\cite[Theorem~3]{berensfehr14}\label{le:chainRuleRenEnt}
Let $( X,Y,Z )$ be a triple of chance variables taking values in the finite set $\setX \times \setY \times \setZ$ according to the joint PMF $P_{X,Y,Z}$. For every $\alpha \in [ 0,\infty ]$
\be 
\renent \alpha {X | Y,Z } \geq \renent \alpha {X,Z | Y } - \log | \setZ |. \label{eq:renEntChainRule}
\ee
\end{lemma}

%\begin{proof}
%For the sake of completeness, a proof is provided in Appendix~\ref{app:pfLeChRule}.
%\end{proof}

\subsection{Optimal Guessing Functions and Task-Encoders} \label{sec:optGuessTaskEnc}

Suppose we want to guess $X$ with guesses of the form ``Is $X = x$?'' Following the notation of \cite{arikan96}, we call a bijection $G \colon \setX \rightarrow \bigl\{ 1, \ldots, | \setX | \bigr\}$ a \emph{guessing function} for $X$. The guessing function determines the guessing order: If we use $G (\cdot)$ to guess $X$, then the question ``Is $X = x$?'' will be our $G (x)$-th question. With a slight abuse of the term ``function,'' we call $\guess {}{\cdot | Y }$ a guessing function for $X$ given $Y$ if the mapping $\guess {} { \cdot | y } \colon \setX \rightarrow \bigl\{ 1, \ldots, | \setX | \bigr\}$ is for every $y \in \setY$ a guessing function for $X$. If we use $\guess {}{\cdot | Y }$ to guess $X$ from the observation $Y$ and observe that $Y = y$, then the question ``Is $X = x$?'' will be our $\guess {} { x | y }$-th question.

In the following we shall consider guessing functions for $X$ given $Y$. Since every guessing function for $X$ can be viewed as a guessing function for $X$ given $Y$ for the case where $Y$ is null, the results also apply to guessing functions for $X$.

The performance of a guessing function is studied in terms of the $\rho$-th moment of the number of guesses that we need to guess $X$ when we use that function. That is, the expectation $\bigEx {}{\guess {} {X|Y}^\rho}$ is the performance of $\guess {}{\cdot|Y}$. We say that a guessing function $\guess {}{\cdot | Y}$ is optimal if its performance is optimal, i.e., $\guess {}{\cdot|Y}$ is optimal if, and only if, (iff) it minimizes $\bigEx {}{\guess {} {X|Y}^\rho}$ among all the guessing functions for $X$ given $Y$. It is easy to see that a guessing function $\guess {}{ \cdot | Y }$ is optimal iff for every $y \in \setY$, the function $\guess {}{ \cdot | y }$ orders the possible realizations of $X$ in decreasing order of their posterior probabilities given $Y = y$. We can use Arikan's results on guessing \cite{arikan96} to bound the performance of optimal guessing functions:

\begin{theorem}[On the Performance of Optimal Guessing Functions]\cite[Theorem~1 and Proposition~4]{arikan96}\label{th:optGuessFun}
There exists some guessing function $\guess {}{\cdot | Y}$ for which
\be 
\bigEx {}{\guess {} {X |Y }^\rho} \leq 2^{\rho \renent {\tirho}{X|Y}}.
\ee
Conversely, for every guessing function $\guess {}{\cdot | Y}$
\be 
\bigEx {}{\guess {} {X |Y }^\rho} \geq \bigl( 1 + \ln | \setX | \bigr)^{-\rho} 2^{\rho \renent {\tirho}{X|Y}} \vee 1. \label{eq:lbOptGuessFun}
\ee
\end{theorem}

For task-encoders we adopt the terminology of \cite{buntelapidoth14}. Given some finite set of descriptions $\setZ$, we call a mapping $f \colon \setX \rightarrow \setZ$ a \emph{task-encoder} for $X$. We associate every task-encoder with a decoder of the form
\begin{IEEEeqnarray}{rCl}
f^{-1} \colon \setZ & \rightarrow & 2^{\setX} \nonumber \\*[-0.5\normalbaselineskip]
% start row with number
  \label{eq:taskDecForm}
% end row with number
\\*[-0.5\normalbaselineskip]
z & \mapsto & \Bigl\{x \in \setX \colon \bigl\{ P_X ( x ) > 0 \bigr\} \cap \bigl\{ \enc x = z \bigr\} \Bigr\}. \nonumber
\end{IEEEeqnarray}
If the encoder describes $X$ by $Z \triangleq \enc X$, then the list $\setL_Z \triangleq \dec Z$ produced by the decoder is the list containing all the realizations of $X$ of positive a priori probability that the encoder could have described by $Z$. (This is the shortest list that is almost-surely guaranteed to contain $X$ given its description $Z$.)

Consider now the scenario where some side information $Y$ is revealed to the encoder and decoder~\cite[Section~VI]{buntelapidoth14}. In this scenario we call $\enc { \cdot | Y }$ a task-encoder for $X$ given $Y$ if the mapping $\enc { \cdot | y } \colon \setX \rightarrow \setZ$ is for every $y \in \setY$ a task-encoder for $X$. We associate every task-encoder with a decoder $\dec { \cdot | Y }$ satisfying for every $y \in \setY$ that $\dec { \cdot | y }$ is of the form \eqref{eq:taskDecForm}, i.e., that
\begin{IEEEeqnarray}{rCl}
f^{-1} ( \cdot | y ) \colon \setZ & \rightarrow & 2^{\setX} \nonumber \\*[-0.5\normalbaselineskip]
% start row with number
  \label{eq:deterministicListDef}
% end row with number
\\*[-0.5\normalbaselineskip]
z & \mapsto & \Bigl\{x \in \setX \colon \bigl\{ P_{X|Y} ( x | y ) > 0 \bigr\} \cap \bigl\{ \enc {x|y} = z \bigr\} \Bigr\}. \nonumber
\end{IEEEeqnarray}
If, upon observing $Y$, the encoder describes $X$ by $Z \triangleq \enc { X | Y }$, then the list $\setL^Y_Z \triangleq \dec {Z|Y}$ produced by the decoder is the list containing all the realizations of $X$ that---given the side information $Y$---have a positive posterior probability under $P_{X|Y}$ and that the encoder could have described by $Z$.

In the following we shall consider task-encoders for $X$ given $Y$. Since every task-encoder for $X$ can be viewed as a task-encoder for $X$ given $Y$ for the case where $Y$ is null, the results also apply to task-encoders for $X$.

We shall also need the notion of a \emph{stochastic task-encoder}. Such an encoder associates with every possible realization $(x,y) \in \setX \times \setY$ of the pair $(X,Y)$ a PMF on $\setZ$ and, upon observing the side information $y$, describes $x$ by drawing $Z$ from $\setZ$ according to the PMF associated with $(x,y)$. The conditional probability that $Z = z$ given $( X,Y ) = (x,y)$ is thus determined by the stochastic encoder, and we denote it by
\begin{equation} \label{eq:stochasticEncoderDef}
\distof { Z = z | X = x, Y = y }, \quad \, (x,y,z) \in \setX \times \setY \times \setZ.
\end{equation}
Based on $(Y,Z)$ the decoder associated with the encoder \eqref{eq:stochasticEncoderDef} produces the smallest list $\setL^Y_Z$ that is guaranteed to contain $X$, i.e., if $(Y,Z) = (y,z)$, then the decoder produces the list
\begin{equation} \label{eq:stochasticListDef} 
\setL^y_z = \bigl\{ x \in \setX \colon \distof { X = x | Y = y, Z = z } > 0 \bigr\}, \quad (y,z) \in \setY \times \setZ
\end{equation}
of all the possible realizations $x \in \setX$ of $X$ of positive posterior probability
\begin{equation}
\distof { X = x | Y = y, Z = z } = \frac{ P_{X,Y} ( x,y ) \, \distof { Z = z | X = x, Y = y} }{ \sum_{\tilde x \in \setX} P_{X,Y} ( \tilde x, y ) \, \distof { Z = z | X = \tilde x, Y = y} }.
\end{equation}

We assess the performance of a task-encoder in terms of the $\rho$-th moment $\bigEx {}{| \setL^Y_Z |^\rho}$ of the size of the list that the associated decoder must form. As we argue shortly, deterministic task-encoders are optimal in the sense that for every stochastic task-encoder there exists a deterministic task-encoder that performs at least as well. Therefore, we can use Bunte and Lapidoth's results on deterministic task-encoders \cite{buntelapidoth14} to bound the performance of optimal stochastic task-encoders:

\begin{theorem}[On the Performance of the Optimal Task-Encoders]\cite[Theorem~VI.1]{buntelapidoth14}\label{th:optTaskEnc}
Let $\setZ$ be a finite set. If $\card \setZ > \log \card \setX + 2$, then there exists a deterministic task-encoder $\enc {\cdot | Y}$ for which
\begin{equation}
\BigEx {}{ \bigl| \setL^Y_Z \bigr|^\rho} = \BigEx {}{ \bigdec{ \enc{X|Y} \bigl| Y }^\rho } < 1 + 2^{\rho ( \renent {\tirho}{X|Y} - \log ( | \setZ | - \log | \setX | - 2 ) + 2 )}.
\end{equation}
Conversely, given any stochastic task-encoder \eqref{eq:stochasticEncoderDef}, the associated decoding lists $\{ \setL^y_z \}$ \eqref{eq:stochasticListDef} satisfy
\be 
\BigEx {}{ \bigl| \setL^Y_Z \bigr|^\rho } \geq 2^{\rho ( \renent {\tirho}{X|Y} - \log | \setZ | )} \vee 1.
\ee
\end{theorem}

We conclude this section by showing that for every stochastic task-encoder there exists a deterministic task-encoder that performs at least as well. Given a stochastic task-encoder \eqref{eq:stochasticEncoderDef} with associated decoding lists \eqref{eq:stochasticListDef}, we can construct a deterministic task-encoder $\enc {\cdot | Y }$ as follows. If $(x,y) \in \setX \times \setY$ satisfies $P_{X|Y} ( x | y ) > 0$, then we choose $\enc { x | y }$ as one that---among all elements of $\{ z \in \setZ \colon x \in \setL^y_z \}$---minimizes $| \setL^y_z |$, so
\begin{equation}\label{eq:constructedDetTaskEnc}
\enc { x | y } \in \argmin_{z \in \setZ \colon x \in \setL^y_z} | \setL^y_z |.
\end{equation}
Otherwise, we choose $\enc { x | y }$ to be an arbitrary element of $\setZ$. It then follows from \eqref{eq:deterministicListDef} that the deterministic task-encoder performs at least as well as the stochastic task-encoder:
\begin{IEEEeqnarray}{l}
\BigEx {}{\bigl| \setL^Y_Z \bigr|^\rho} \nonumber \\
\quad = \sum_{(x,y) \in \setX \times \setY} \sum_{z \in \setZ} P_{X,Y} (x,y) \, \distof {Z = z | X = x, Y = y} | \setL^y_z |^\rho \\
\quad \geq \sum_{\substack{ (x,y) \in \setX \times \setY \colon \\ P_{X,Y} (x,y) > 0 }} \sum_{z \in \setZ} P_{X,Y} (x,y) \, \distof {Z = z | X = x, Y = y} \min_{z^\prime \in \setZ \colon x \in \setL^y_{z^\prime}} | \setL^y_{z^\prime} |^\rho \\
\quad = \sum_{\substack{ (x,y) \in \setX \times \setY \colon \\ P_{X,Y} (x,y) > 0 }} P_{X,Y} ( x,y ) \min_{z^\prime \in \setZ \colon x \in \setL^y_{z^\prime}} | \setL^y_{z^\prime} |^\rho \\
\quad = \sum_{(x,y) \in \setX \times \setY} P_{X,Y} ( x,y ) \, \bigl| \setL^y_{\enc { x | y }} \bigr|^\rho \\
\quad \stackrel{(a)}\geq \sum_{(x,y) \in \setX \times \setY} P_{X,Y} ( x,y ) \, \bigl| \bigdec { \enc { x | y } \bigl| y } \bigr|^\rho \\
\quad = \BigEx {}{\bigl| \bigdec { \enc { X | Y } \bigl| Y } \bigr|^\rho}, \label{eq:stochEncToDetEnc}
\end{IEEEeqnarray}
where $(a)$ holds because \eqref{eq:deterministicListDef} and \eqref{eq:constructedDetTaskEnc} imply that $\bigdec { \enc { x | y } | y } \subseteq \setL^y_{ \enc { x | y } }$.

\section{Lists and Guesses}\label{sec:listsAndGuesses}

In this section we relate task-encoders to guessing functions and explain why the performance guarantees for optimal guessing functions (Theorem~\ref{th:optGuessFun}) and task-encoders (Theorem~\ref{th:optTaskEnc}) are remarkably similar. Moreover, we quantify how additional side information can help guessing. We shall need these results to characterize the secrecy of the distributed-storage systems we study in the present paper, but they may also be of independent interest.

We start by quantifying how some additional information $Z$ (e.g., some description produced by an encoder) can help guessing. As the following lemma shows, $Z$ can reduce the $\rho$-th moment of the number of guesses by at most a factor of $\card \setZ^{-\rho}$:

\begin{lemma}\label{le:ImproveGuess}
Given a finite set $\setZ$, draw $Z$ from $\setZ$ according to some conditional PMF $P_{Z|X,Y}$, so $(X,Y,Z) \sim P_{X,Y} \times P_{Z|X,Y}$. For optimal guessing functions $\guessast {}{\cdot|Y,Z}$ and $\guessast {}{\cdot|Y}$ (which minimize $\bigEx {}{\guess {}{ X | Y, Z }^\rho}$ and $\bigEx {}{\guess {}{ X | Y }^\rho}$, respectively)
\begin{equation}
\bigEx {}{\guessast {}{X|Y,Z}^\rho} \geq \BigEx {}{\bigl\lceil \guessast {}{X | Y} / \card \setZ \bigr\rceil^\rho}. \label{eq:impGuess}
\end{equation}
Equality holds whenever $Z = f ( X,Y )$ for some mapping $f \colon \setX \times \setY \rightarrow \setZ$ for which $f ( x,y ) = f ( \tilde x, y )$ implies either $\bigl\lceil \guessast {}{ x | y } / | \setZ | \bigr\rceil \neq \bigl\lceil \guessast {}{ \tilde x | y } / | \setZ | \bigr\rceil$ or $x = \tilde x$. Such a mapping always exists, because for all $l \in \naturals$ at most $| \setZ |$ different $x \in \setX$ satisfy $\bigl\lceil \guessast {}{ x | y } / | \setZ | \bigr\rceil = l$.
\end{lemma}

\begin{proof}
To prove \eqref{eq:impGuess} we first show that $$\bigEx {}{\guessast {}{X|Y,Z}^\rho}$$ is minimum if $Z$ is deterministic given $(X,Y)$. Indeed, define the function $g \colon \setX \times \setY \rightarrow \setZ$ so that $g ( x, y ) \in \argmin_{z \in \setZ} \guessast {}{ x | y, z }$ holds for all $( x,y ) \in \setX \times \setY$. This implies that
\begin{IEEEeqnarray}{rCl}
\bigguessast {}{X|Y,Z} \geq \bigguessast {}{X|Y,g ( X,Y )}
\end{IEEEeqnarray}
and consequently that
\begin{IEEEeqnarray}{rCl}
\bigEx {}{\guessast {}{X|Y,Z}^\rho} & \geq & \min_{\guess {}{ \cdot | Y }} \BigEx {}{\bigguess {}{X|Y,g ( X,Y )}^\rho}. 
\end{IEEEeqnarray}
It thus suffices to prove \eqref{eq:impGuess} for the case where $Z$ is deterministic given $(X,Y)$, and we thus assume w.l.g.\ that $Z = g ( X,Y )$ for some function $g \colon \setX \times \setY \rightarrow \setZ$. For every guessing function $\bigguess {}{\cdot | Y, g ( X,Y )}$ we have
\begin{IEEEeqnarray}{l}
\BigEx {}{\bigguess {}{X|Y,g(X,Y)}^\rho} = \sum_{(x,y) \in \setX \times \setY} P_{X,Y} ( x,y ) \, \bigguess {}{x|y, g ( x, y )}^\rho. \label{eq:expToMinImpGuess}
\end{IEEEeqnarray}
Moreover, for every distinct $x, \, \tilde x \in \setX$ and every $y \in \setY$ the equality 
\be \nonumber
\bigguess {}{x | y, g ( x,y )} = \bigguess {}{\tilde x | y, g ( \tilde x,y )}
\ee
implies that $g ( x,y ) \neq g ( \tilde x,y )$, because $\guess {}{\cdot | y, z} \colon \setX \rightarrow \bigl\{ 1, \ldots, |\setX| \bigr\}$ is for every $z \in \setZ$ one-to-one. Consequently, for every $\ell \in \naturals$ there are at most $| \setZ |$ different $x \in \setX$ for which $\bigguess {}{x|y,g ( x,y )} = \ell$. For every $y \in \setY$ order the possible realizations of $X$ in decreasing order of $P_{X,Y} ( x,y )$ or, equivalently, in decreasing order of their posterior probabilities given $Y = y$, and let $x_j^y$ denote the $j$-th element. Clearly, \eqref{eq:expToMinImpGuess} is minimum over $g ( \cdot, \cdot )$ and $\bigguess {}{\cdot | Y, g ( X,Y )}$ if for every $\ell \in \naturals$ and every $y \in \setY$ we have $\bigguess {}{x|y, g ( x, y )} = \ell$ whenever $x = x_j^y$ for some $j$ satisfying $(\ell - 1) | \setZ | + 1 \leq j \leq \ell | \setZ |$ or, equivalently, $\bigl\lceil j / | \setZ | \bigr\rceil = \ell$. Since $\guessast {}{\cdot | Y}$ minimizes $\bigEx {}{\guess {}{X | Y}^\rho}$, it orders the elements of $\setX$ in decreasing order of their posterior probabilities given $Y$, and consequently we can choose $x^y_j$ to be the unique $x \in \setX$ for which $\guessast {}{x | y} = j$. Hence, \eqref{eq:expToMinImpGuess} is minimized if $f ( \cdot, \cdot )$ satisfies the specifications in the lemma, $g ( \cdot, \cdot ) = f ( \cdot, \cdot )$, and $\bigguess {}{x|y, f ( x, y )} = \bigl\lceil \guessast {}{x|y} / | \setZ | \bigr\rceil$ (see Figure~\ref{fig:optimalEncGuessing}). Moreover, the minimum equals the RHS of \eqref{eq:impGuess}.
\end{proof}

One can infer from Lemma~\ref{le:ImproveGuess} how to construct an optimal encoder $f \colon \setX \times \setY \rightarrow \setZ$ for a guessing decoder, i.e., an encoder $Z = f ( X,Y )$ that minimizes $\min_{\guess {}{\cdot |Y,Z}} \bigEx {}{\guess {}{ X | Y, Z }^\rho}$ among all the possible descriptions $Z$ that are drawn from $\setZ$ according to some conditional PMF $P_{Z|X,Y}$. To that end recall that a guessing function $\guess {}{\cdot | Y}$ is optimal, i.e., minimizes $\bigEx {}{\guess {}{ X | Y }^\rho}$, iff for every $y \in \setY$ $\guess {}{ \cdot | y }$ orders the possible realizations of $X$ in decreasing order of their posterior probabilities given $Y = y$. An optimal encoder $f \colon \setX \times \setY \rightarrow \setZ$ for a guessing decoder can be constructed as follows: For every $y \in \setY$ we first order the possible realizations of $X$ in decreasing order of $P_{X,Y} ( x,y )$ or, equivalently, in decreasing order of their posterior probabilities given $Y = y$, and we let $x_j^y$ denote the $j$-th element. (Ties are resolved at will.) We then choose some mapping $f \colon \setX \times \setY \rightarrow \setZ$ for which $f ( x_j^y, y ) = f ( x_{j^\prime}^y, y )$ implies either $\bigl\lceil j / | \setZ | \bigr\rceil \neq \bigl\lceil j^\prime / | \setZ | \bigr\rceil$ or $j = j^\prime$, e.g., by indexing the elements of $\setZ$ by the elements of $\bigl\{ 0, \ldots, |\setZ| - 1 \bigr\}$ and choosing $f (  x_j^y, y )$ as the element of $\setZ$ indexed by the remainder of the Euclidean division of $j - 1$ by $| \setZ |$ (see Figure~\ref{fig:optimalEncGuessing}).

\begin{figure}[ht]
\vspace{-2mm}

\begin{center}
\def\pgfsysdriver{pgfsys-dvipdfm.def}
\begin{tikzpicture}

\path (0.4,3.6) node (prob) [] {\rotatebox{60}{$P(\cdot|y)$}};
\path (1,3.64) node (guess) [] {\rotatebox{60}{$\guessast {}{\cdot|y}$}};
\path (1.65,3.9) node (zed) [] {\rotatebox{60}{$z = f (\cdot,y)$}};
\path (2.15,3.8) node (zed) [] {\rotatebox{60}{$\guessast {}{\cdot|y,z}$}};

  \foreach \y in {0,0.5,...,2.5} {
      \foreach \x in {0,0.5} {
          \pgfmathparse{\y/2.5*(\x==0.5)+(1-\y/2.5)*(\x==0)}
          \definecolor{MyColor}{gray}{\pgfmathresult}
          \pgfmathparse{2*\x}
          \ifthenelse{1 = \pgfmathresult}{\path[fill=MyColor, opacity = 0.5] (\x,\y) rectangle node [opacity = 1] {\pgfmathparse{int(6-2*\y)}\pgfmathresult} ++(0.5,0.5)}{\path[fill=MyColor, opacity = 0.5] (\x,\y) rectangle ++(0.5,0.5)};
      }
  }
  
  \foreach \y in {0,0.5,...,2.5} {
      \foreach \x in {1} {
          \definecolor{MyColor}{gray}{1}
          \pgfmathparse{2*\y}
          \ifthenelse{5 = \pgfmathresult \OR 2 = \pgfmathresult}{\path[fill=MyColor, opacity = 0.5] (\x,\y) rectangle node [opacity = 1] {$\star$} ++(0.5,0.5)}{};
          \ifthenelse{4 = \pgfmathresult \OR 1 = \pgfmathresult}{\path[fill=MyColor, opacity = 0.5] (\x,\y) rectangle node [opacity = 1] {$\bullet$} ++(0.5,0.5)}{};
          \ifthenelse{3 = \pgfmathresult \OR 0 = \pgfmathresult}{\path[fill=MyColor, opacity = 0.5] (\x,\y) rectangle node [opacity = 1] {$\diamond$} ++(0.5,0.5)}{};
      }
  }
  
  \foreach \y in {0,0.5,...,2.5} {
      \foreach \x in {1.5} {
          \definecolor{MyColor}{gray}{1}
          \definecolor{MyOtherColor}{gray}{0.8}
          \pgfmathparse{3-int(1+\y/1.5)}
          \ifthenelse{1 = \pgfmathresult}{\path[fill=MyColor, opacity = 0.5] (\x,\y) rectangle node [opacity = 1] {\pgfmathparse{int(3-int(1+\y/1.5))}\pgfmathresult} ++(0.5,0.5)}{\path[fill=MyOtherColor, opacity = 0.5] (\x,\y) rectangle node [opacity = 1] {\pgfmathparse{int(3-int(1+\y/1.5))}\pgfmathresult} ++(0.5,0.5)};
      }
  }

  \draw[step=0.5,thick] (0,0) grid (2,3);
  
  \draw[step=0.5,thick] (0,0) grid (1,3); 
  \path (-0.3,3) -- node [] {\rotatebox{90}{$x \in \setX$}} (-0.3,0);
\end{tikzpicture}

\caption[Constructing an optimal encoder for a guessing decoder]{How to construct an optimal encoder $f \colon \setX \times \setY \rightarrow \setZ$ for a guessing decoder when $\setZ = \{ \star, \bullet, \diamond \}$. Light background tones indicate small values of $P (\cdot | y)$ or $\guessast {}{\cdot | y}$. }
\label{fig:optimalEncGuessing}
\end{center}
\vspace{-2mm}

\end{figure}

Lemma~\ref{le:ImproveGuess} and \eqref{eq:ceilApprox} imply the following corollary:

\begin{corollary}\label{co:impGuess}
Given a finite set $\setZ$, there exists some mapping $f \colon \setX \times \setY \rightarrow \setZ$ such that
\begin{equation}
\min_{\guess {}{\cdot |Y,Z}} \bigEx {}{\guess {}{X|Y,Z}^\rho} < 1 + 2^{\rho} | \setZ |^{-\rho} \min_{\guess {}{\cdot | Y}} \bigEx {}{\guess {}{X|Y}^\rho}, \label{eq:lbGuessSI}
\end{equation}
where $Z$ denotes $f ( X,Y )$. Conversely, for every chance variable $Z$ that takes values in $\setZ$
\begin{equation}
\min_{\guess {}{\cdot |Y,Z}} \bigEx {}{\guess {}{X|Y,Z}^\rho} \geq | \setZ |^{-\rho} \min_{\guess {}{\cdot | Y}} \bigEx {}{\guess {}{X|Y}^\rho} \vee 1. \label{eq:ubGuessSI}
\end{equation}
\end{corollary}

From Corollary~\ref{co:impGuess} and Theorem~\ref{th:optGuessFun}, which characterizes the performance of optimal guessing functions $\guess {}{\cdot | Y}$, we obtain the following upper and lower bounds on the smallest ambiguity $\min_{\guess {}{\cdot | Y, Z}} \bigEx {}{\guess {}{X|Y,Z}^\rho}$ that is achievable for a given $|\setZ|$. The bounds are tight up to polylogarithmic factors of $|\setX|$.

\begin{corollary}\label{co:equivBunteResultGuessing}
Given a finite set $\setZ$, there exists some mapping $f \colon \setX \times \setY \rightarrow \setZ$ for which
\begin{equation}
\min_{\guess {}{\cdot | Y, Z}} \bigEx {}{\guess {}{X|Y,Z}^\rho} < 1 + 2^{\rho ( \renent {\tirho}{X|Y} - \log | \setZ | + 1 )}, \label{eq:lbGuessSIExp}
\end{equation}
where $Z$ denotes $f ( X,Y )$. Conversely, for every chance variable $Z$ that takes values in $\setZ$
\begin{equation}
\min_{\guess {}{\cdot | Y, Z}} \bigEx {}{\guess {}{X|Y,Z}^\rho} \geq \bigl( 1 + \ln \card \setX \bigr)^{-\rho} 2^{\rho ( \renent {\tirho}{X|Y} - \log \card \setZ )} \vee 1. \label{eq:ubGuessSIExp}
\end{equation}
\end{corollary}

Note that \eqref{eq:ubGuessSIExp} also follows from~\eqref{eq:lbOptGuessFun} in Theorem~\ref{th:optGuessFun} and the properties of conditional R\'enyi entropy in Lemmas~\ref{le:uncertaintyIncreasesRenEnt} and \ref{le:chainRuleRenEnt}.\\

The performance guarantees for optimal guessing functions (Theorem~\ref{th:optGuessFun} and Corollary~\ref{co:equivBunteResultGuessing}) and task-encoders (Theorem~\ref{th:optTaskEnc}) are remarkably similar. To provide some intuition on this, we relate task-encoders to guessing functions. As the following theorem shows, a ``good'' guessing function ``induces'' a ``good'' task-encoder and vice versa:\footnote{We call a guessing function or task-encoder ``good'' if its performance is nearly optimal, and ``induce'' means here that---without knowing the PMF $P_{X,Y}$---we can construct from a guessing function a task-encoder and vice versa.}

\begin{theorem}\label{th:relGuessEnc}
Let $\setZ$ be a finite set.
\begin{enumerate}
\item Given any stochastic task-encoder \eqref{eq:stochasticEncoderDef}, the associated decoding lists $\{ \setL^y_z \}$ \eqref{eq:stochasticListDef} induce a guessing function $\guess {}{\cdot | Y }$ that satisfies
\begin{equation}
\bigEx {}{\guess {} { X | Y }^\rho} \leq | \setZ |^\rho \BigEx {}{ \bigl| \setL^Y_Z \bigr|^\rho}. \label{eq:listToGuess}
\end{equation}
\item Every guessing function $\guess {}{ \cdot | Y }$ and every positive integer $\omega \leq \card \setX$ satisfying
\begin{equation}
| \setZ | \geq \omega \biggl( 1 + \Bigl\lfloor \log \bigl\lceil | \setX | / \omega \bigr\rceil \Bigr\rfloor \biggr) \label{eq:relCardMandV}
\end{equation}
induce a deterministic task-encoder, i.e., a stochastic task-encoder whose conditional PMF \eqref{eq:stochasticEncoderDef} is $\{ 0,1 \}$-valued, whose associated decoding lists $\{ \setL^y_z \}$~\eqref{eq:stochasticListDef} satisfy
\begin{equation}
\BigEx {}{\bigl| \setL^Y_Z \bigr|^\rho} \leq \BigEx {}{\bigl\lceil \guess {} { X | Y } / \omega \bigr\rceil^\rho}. \label{eq:guessToList}
\end{equation}
\end{enumerate}
\end{theorem}

To prove Theorem~\ref{th:relGuessEnc}, we need the following fact:

\begin{fact}\label{fa:1}
For every $k \in \naturals$
\begin{equation}
\bigl| \bigr\{ \tilde k \in \naturals \colon \lfloor \log \tilde k \rfloor  = \lfloor \log k \rfloor  \bigr\} \bigl| \leq k.
\end{equation}
\end{fact}

\begin{proof}[Proof of Fact~\ref{fa:1}]
If $k, \, \tilde k \in \naturals$ are such that $\lfloor \log \tilde k \rfloor  = \lfloor \log k \rfloor$, then
\begin{equation}
2^{\lfloor \log k \rfloor} \leq \tilde k < 2^{\lfloor \log k \rfloor + 1}.
\end{equation}
Hence,
\begin{equation}
\bigl| \bigl\{ \tilde k \in \naturals \colon \lfloor \log \tilde k \rfloor  = \lfloor \log k \rfloor \bigr\} \bigr| \leq 2^{\lfloor \log k \rfloor} \leq k.
\end{equation}
\end{proof}

\begin{proof}[Proof of Theorem~\ref{th:relGuessEnc}] As to the first part, suppose we are given a stochastic task-encoder \eqref{eq:stochasticEncoderDef} with associated decoding-lists $\{ \setL^y_z \}$ \eqref{eq:stochasticListDef}. For every $y \in \setY$ order the lists $\set{\setL^y_z}_{z \in \setZ}$ in increasing order of their cardinalities, and order the elements in each list in some arbitrary way. Now consider the guessing order where we first guess the elements of the first (and smallest) list in their respective order followed by those elements in the second list that have not yet been guessed (i.e., that are not contained in the first list), and where we continue until concluding by guessing those elements of the last (and longest) list that have not been previously guessed. Let $\guess {}{ \cdot | Y }$ be the corresponding guessing function, and observe that
\begin{IEEEeqnarray}{rCl}
\bigEx {}{\guess {}{X|Y}^\rho} & = & \sum_{x,y} P_{X,Y} ( x,y ) \, \bigl| \bigl\{ \tilde x \colon \guess {}{\tilde x|y} \leq \guess {}{x|y} \bigr\} \bigr|^\rho \\
& \stackrel{(a)}\leq & \sum_{x,y} P_{X,Y} ( x,y ) \, | \setZ |^\rho \min_{z \colon x \in \setL^y_z} \card {\setL^y_z}^\rho \\
& \stackrel{(b)}\leq & | \setZ |^\rho \, \BigEx {}{\bigl| \setL^Y_Z \bigr|^\rho},
\end{IEEEeqnarray}
where $(a)$ holds because for every $x, \, \tilde x \in \setX$ and $y \in \setY$ a necessary condition for $\guess {}{\tilde x|y} \leq \guess {}{x|y}$ is that $\tilde x \in \setL^y_{\tilde z}$ for some $\tilde z \in \setZ$ satisfying $$| \setL^y_{\tilde z} | \leq \min_{z \colon x \in \setL^y_z} | \setL^y_z |,$$ and because the number of lists whose size does not exceed $\min_{z \colon x \in \setL^y_z} | \setL^y_z |$ is at most $| \setZ |$; and $(b)$ holds because the list $\setL^Y_Z$ contains $X$ \eqref{eq:stochasticListDef}.\\

As to the second part, suppose we are given a guessing function $\guess {}{\cdot | Y }$ and a positive integer $\omega \leq \card \setX$ satisfying \eqref{eq:relCardMandV}. Let $\setO = \{ 0, \ldots, \omega - 1 \}$ and $$\setS = \biggl\{ 0, \ldots, \Bigl\lfloor \log \bigl\lceil | \setX | / \omega \bigr\rceil \Bigr\rfloor \biggr\}.$$ From \eqref{eq:relCardMandV} it follows that $| \setZ | \geq | \setO | \, | \setS |$. It thus suffices to prove the existence of a task-encoder that uses only $| \setO | \, | \setS |$ possible descriptions, and we thus assume w.l.g.\ that $\setZ = \setO \times \setS$. That is, using the side-information $y$ the task-encoder (deterministically) describes $x$ by $z = (o,s)$. The encoding involves two steps:

\textbf{Step~1:} In Step~1 the encoder first computes $O \in \setO$ as the remainder of the Euclidean division of $\guess {}{X|Y} - 1$ by $\card \setO$. This guarantees that if $(Y,O) = (y,o)$, then $X$ be in the set $$\setX_{y,o} \triangleq \Bigl\{ x \in \setX \colon \bigl( \guess {} { x | y } - 1 \bigr) \equiv o \mod | \setO | \Bigr\}.$$ It then constructs from $\guess {}{ \cdot | Y }$ a guessing function $\guess {}{ \cdot | Y,O }$ as follows. The encoder constructs the guessing function $\guess {}{ \cdot | y,o }$ so that---in the corresponding guessing order---we first guess the elements of $\setX_{y,o}$ in increasing order of $\guess {}{ x | y }$. Our first $| \setX_{y,o} |$ guesses are thus the elements of $\setX_{y,o}$ with $x \in \setX_{y,o}$ being guessed before $\tilde x \in \setX_{y,o}$ whenever $\guess {}{ x | y } < \guess {}{ \tilde x | y }$. Once we have guessed all the elements of $\setX_{y,o}$, we guess the remaining elements of $\setX$ in some arbitrary order. This order is immaterial, because $X$ is guaranteed to be in the set $\setX_{y,o}$. As we argue next, the guessing function $\guess {}{\cdot | Y,O}$ for $X$ satisfies
\begin{equation} \label{eq:guessToListO}
\guess {} { X | Y, O } = \bigl\lceil \guess {} { X | Y } / | \setO | \bigr\rceil.
\end{equation}
Indeed, observe that for every $(y,o) \in \setY \times \setO$ and $l \in \bigl\{ 1, \ldots, | \setX_{y,o} | \bigr\}$ our $l$-th guess $x_l$ is the element of $\setX_{y,o}$ for which $\guess {}{ x_l | y } = o + 1 + ( l - 1 ) | \setO |$. Since $o + 1 \in \bigl\{ 1, \ldots, | \setO | \bigr\}$, we find that $\guess {}{ x | y,o } = \bigl\lceil \guess {}{ x | y } / | \setO | \bigr\rceil$ whenever $x \in \setX_{y,o}$. But $X$ is guaranteed to be in the set $\setX_{y,o}$. This proves that the guessing function $\guess {}{\cdot | Y,O }$ for $X$ satisfies \eqref{eq:guessToListO}. By \eqref{eq:guessToListO} and because $|\setO| = \omega$,
\begin{equation} \label{eq:guessToListOV}
\guess {} { X | Y, O } = \bigl\lceil \guess {} { X | Y } / \omega \bigr\rceil.
\end{equation}

\textbf{Step~2:} In Step~2 the encoder first computes $S = \bigl\lfloor \log \guess {}{X | Y, O} \bigr\rfloor$ and then describes $X$ by $Z \triangleq ( O,S )$. By \eqref{eq:guessToListOV} $$1 \leq \guess {}{ X | Y, O } \leq \bigl\lceil |\setX| / \omega \bigr\rceil$$ and consequently $S \in \setS$. Since $O$ and $S$ are deterministic given $( X,Y )$, the conditional PMF \eqref{eq:stochasticEncoderDef} corresponding to the description $Z = (O,S)$ is $\{ 0,1 \}$-valued. It remains to show that the decoding lists $\{ \setL^y_z \}$~\eqref{eq:stochasticListDef} satisfy \eqref{eq:guessToList}. To this end note that if $(Y,O,S) = (y,o,s)$, then $X$ is in the set $$\setX_{y,o,s} \triangleq \Bigl\{ x \in \setX \colon \bigl\lfloor \log \guess {}{ x | y, o } \bigr\rfloor = s \Bigr\}.$$ Because every pair $x, \, \tilde x \in \setX_{y,o,s}$ satisfies $\bigl\lfloor \log \guess {}{ x | y, o } \bigr\rfloor = \bigl\lfloor \log \guess {}{ \tilde x | y, o } \bigr\rfloor$, Fact~\ref{fa:1} and the fact that the guessing function $\guess {}{\cdot | y,o }$ is a bijection imply that
\begin{equation} \label{eq:sizeSetXGivYOS}
| \setX_{y,o,s} | \leq \guess {}{ x | y, o }, \,\, \forall \, x \in \setX_{y,o,s}.
\end{equation}
Recalling that
\begin{equation} \label{eq:xInSetXYOS}
\Bigl( (Y,O,S) = (y,o,s) \Bigr) \implies X \in \setX_{y,o,s},
\end{equation}
we obtain from \eqref{eq:sizeSetXGivYOS} that
\begin{equation} \label{eq:sizeSetXGivYOSGXYOS}
| \setX_{Y,O,S} | \leq \guess {}{ X | Y, O }.
\end{equation}
By \eqref{eq:xInSetXYOS} and because $Z = (O,S)$, the list $\setL^Y_Z$ \eqref{eq:stochasticListDef} is contained in the set $\setX_{Y,O,S}$ and consequently satisfies $\bigl| \setL^Y_Z \bigr| \leq | \setX_{Y,O,S} |$. Hence, \eqref{eq:sizeSetXGivYOSGXYOS} implies that
\begin{equation} \label{eq:guessToListS}
\bigl| \setL^Y_Z \bigr| \leq \guess {} { X | Y,O }.
\end{equation}
From \eqref{eq:guessToListOV} and \eqref{eq:guessToListS} we conclude that
\begin{IEEEeqnarray}{l}
\BigEx {}{\bigl| \setL^Y_Z \bigr|^\rho } \leq \bigEx {}{\guess {} { X | Y,O }^\rho} = \BigEx {}{\bigl\lceil \guess {} {X | Y } / \omega \bigr\rceil^\rho}.
\end{IEEEeqnarray}
\end{proof}

To better understand the second part of Theorem~\ref{th:relGuessEnc}, we briefly discuss the construction of a deterministic task-encoder from an optimal guessing function $\guessast {}{\cdot|Y}$ (which minimizes $\bigEx {}{\guess {}{ X | Y }^\rho}$). If $\guessast {}{\cdot|Y}$ is an optimal guessing function, then the two-step construction in the proof of Theorem~\ref{th:relGuessEnc} can be alternatively described as follows. We construct a task-encoder that describes $X$ by $$Z = (O,S),$$ where $O$ takes values in some set $\setO$ of size $\omega$, where $$1 \leq \omega \leq |\setX|,$$ and $S$ takes values in some set $\setS$ of size $$1 + \Bigl\lfloor \log \bigl\lceil | \setX | / \omega \bigr\rceil \Bigr\rfloor \leq 1 + \log |\setX|.$$ (Note that the description $Z$ assumes at most $|\setO| \, |\setS|$ different values, and by \eqref{eq:relCardMandV} $|\setO| \, |\setS| \leq |\setZ|$.) In the first step of the construction, we choose the first part of the description, $O$. We choose $O$ as one that---among all $O$'s that are drawn from $\setO$ according to some conditional PMF $P_{O|X,Y}$---minimizes $\min_{\guess {}{\cdot | Y, O}} \bigEx {}{\guess {}{ X | Y, O }^\rho}$. From Lemma~\ref{le:ImproveGuess} (and the subsequent paragraph) we already know how to construct $O$. Indeed, from Lemma~\ref{le:ImproveGuess} it follows that $$\min_{\guess {}{\cdot | Y, O}} \bigEx {}{\guess {}{ X | Y, O }^\rho} \geq \BigEx {}{\bigl\lceil \guessast {}{X | Y} / \card \setO \bigr\rceil^\rho},$$ where equality is achieved by choosing $O = f_1 (X,Y)$ for some mapping $f_1 \colon \setX \times \setY \rightarrow \setO$ for which $f_1 ( x,y ) = f_1 ( \tilde x, y )$ implies either $\bigl\lceil \guessast {}{ x | y } / | \setO | \bigr\rceil \neq \bigl\lceil \guessast {}{ \tilde x | y } / | \setO | \bigr\rceil$ or $x = \tilde x$. For example, in the case where $\setO = \{ 0, \ldots, \omega - 1 \}$ we can choose $O$ as the remainder of the Euclidean division of $\guess {}{X|Y} - 1$ by $|\setO|$. Based on the optimal guessing function $\guessast {}{\cdot|Y}$ and the first part of the description,  $O$, we can construct an optimal guessing function $\guessast {}{\cdot|Y,O}$ (which minimizes $\bigEx {}{\guess {}{ X | Y, O }^\rho}$) by choosing some $\guessast {}{\cdot|Y,O}$ for which $$\bigguessast {}{x|y,f_1(x,y)} = \bigl\lceil \guessast {}{x|y} / |\setO| \bigr\rceil, \,\, \forall \, (x,y) \in \setX \times \setY.$$

In the second step of the construction we choose the second part of the description, $S$. We choose $S = f_2 (x,y)$, where $$f_2 (x,y) = \Bigl\lfloor \log \bigguessast {}{ x | y, f_1 (x,y) } \Bigr\rfloor, \,\, \forall \, (x,y) \in \setX \times \setY.$$ This will guarantee that the decoding lists satisfy $$\BigEx {}{\bigl| \setL^Y_Z \bigr|^\rho } \leq \bigEx {}{\guessast {} { X | Y,O }^\rho} = \BigEx {}{\bigl\lceil \guessast {}{X | Y} / |\setO| \bigr\rceil^\rho},$$ where $$Z = (O,S) = \bigl( f_1 (X,Y), f_2 (X,Y) \bigr).$$ Note that the size of the support $\setS$ of $S$ is only logarithmic in $|\setX|$ and thus negligible in asymptotic settings, i.e., in asymptotic settings $|\setZ| \approx |\setO|$.\\

The following corollary results from Theorem~\ref{th:relGuessEnc} and \eqref{eq:ceilApprox} by setting $$\omega = \biggl\lfloor | \setZ | / \Bigl( 1 + \bigl\lfloor \log | \setX | \bigr\rfloor \Bigr) \biggr\rfloor$$ in Theorem~\ref{th:relGuessEnc}.

\begin{corollary}\label{co:guessToBestList}
Given a set $\setZ$ of cardinality $|\setZ| \geq 1 + \bigl\lfloor \log | \setX | \bigr\rfloor$, any guessing function $\guess {}{ \cdot | Y }$ induces a deterministic task-encoder, i.e., a stochastic task-encoder whose conditional PMF \eqref{eq:stochasticEncoderDef} is $\{ 0,1 \}$-valued, whose associated decoding lists $\{ \setL^y_z \}$~\eqref{eq:stochasticListDef} satisfy
\be 
\BigEx {}{\bigl| \setL^Y_Z \bigr|^\rho} \leq 1 + 2^{\rho} \bigEx {}{\guess {}{ X | Y }^\rho} \biggl( \frac{| \setZ |}{1 + \log | \setX |} - 1 \biggr)^{-\rho}. \label{eq:guessToBestList}
\ee
\end{corollary}

Combined with Theorem~\ref{th:optGuessFun}, which bounds the performance of an optimal guessing function, Equations~\eqref{eq:listToGuess} and \eqref{eq:guessToBestList} provide an upper and a lower bound on the smallest $\Ex {}{| \setL^Y_Z |^\rho}$ that is achievable for a given $\card \setZ$. These bounds are weaker than \cite[Theorem~I.1 and Theorem~VI.1]{buntelapidoth14} (see Theorem~\ref{th:optTaskEnc}) in the finite blocklength regime but tight enough to prove the asymptotic results \cite[Theorem~I.2 and Theorem~VI.2]{buntelapidoth14}.

Another interesting corollary to Theorem~\ref{th:relGuessEnc} results from the choice $\omega = 1$ in Theorem~\ref{th:relGuessEnc}:

\begin{corollary}\label{co:guessToList}
Given a set $\setZ$ of cardinality $|\setZ| = 1 + \bigl\lfloor \log | \setX | \bigr\rfloor$, any guessing function $\guess {}{ \cdot | Y }$ induces a deterministic task-encoder, i.e., a stochastic task-encoder whose conditional PMF \eqref{eq:stochasticEncoderDef} is $\{ 0,1 \}$-valued, whose associated decoding lists $\{ \setL^y_z \}$~\eqref{eq:stochasticListDef} satisfy
\begin{equation} \label{eq:guessToListV1}
\BigEx {}{\bigl| \setL^Y_Z \bigr|^\rho} \leq \bigEx {}{\guess {} { X | Y }^\rho}.
\end{equation}
E.g., if $$\setZ = \Bigl\{ 0, \ldots, \bigl\lfloor \log | \setX | \bigr\rfloor \Bigr\},$$ then the task-encoder $\enc { \cdot | Y }$ defined by
\begin{subequations} \label{bl:guessToListV1Det}
\begin{equation}
f ( \cdot | y ) = \bigl\lfloor \log \guess {}{ \cdot | y } \bigr\rfloor, \,\, \forall \, y \in \setY
\end{equation}
satisfies \eqref{eq:guessToListV1} or, equivalently,
\begin{equation}
\BigEx {}{\bigdec { \enc{ X | Y } \bigl| Y }^\rho} \leq \bigEx {}{\guess {} { X | Y }^\rho}.
\end{equation}
\end{subequations}
\end{corollary}

An implication of Corollary~\ref{co:guessToList} for the problems studied in this paper is discussed in Remark~\ref{remark:guessListClose}. Another example where Corollary~\ref{co:guessToList} is useful is in determining the feedback listsize capacity of a DMC $W (y|x)$ with positive zero-error capacity. Corollary~\ref{co:guessToList} can be used to give an elegant proof of the direct part of \cite[Theorem~I.1]{buntelapidoth14fb}, which states that in the presence of perfect feedback the listsize capacity of $W ( y | x )$ equals the cutoff rate $R_{\textnormal{cutoff}} (\rho)$ with feedback (which is in fact equal to the cutoff rate without feedback \cite[Corollary~I.4]{buntelapidoth14fb}). To see this, suppose that we are given a sequence of (feedback) codes of rate $R$ for which the $\rho$-th moment of the number of guesses $\guessast {}{ M | Y^n }$ a decoder needs to guess the transmitted message $M$ based on the channel-outputs $Y^n$ approaches one as the blocklength~$n$ tends to infinity. (Recall that $R_{\textnormal{cutoff}} (\rho)$ is the supremum of all rates for which such a sequence exists.) Suppose now that the transmission does not stop after $n$ channel uses. Instead, the encoder computes $$Z \triangleq \bigl\lfloor \log \guessast {}{ M | Y^n } \bigr\rfloor \in \bigl\{ 0, \ldots, \lfloor n R \rfloor \bigr\}$$ from the feedback $Y^n$ and uses another $n^\prime$ channel uses to transmit $Z$ at a positive rate while guaranteeing that the receiver can decode it with probability one. Since a positive zero-error (feedback) capacity cannot be smaller than one \cite{shannon56}, it is enough to take $n^\prime \leq \lceil \log ( n R ) \rceil$. Hence, $( n + n^\prime ) / n$ converges to one as $n$ tends to infinity, and the rate of the code thus converges to $R$. At the same time, when we substitute $(M,Y^n,Z)$ for $(X,Y,Z)$ in Corollary~\ref{co:guessToList}, Corollary~\ref{co:guessToList} implies that the size of the smallest decoding-list $\setL^{Y^{n+n^\prime}}$ that is guaranteed to contain $M$ satisfies $\bigl| \setL^{Y^{n+n^\prime}} \bigr| = \bigl| \setL^{Y^n}_Z \bigr| \leq \guessast {}{ M | Y^n }$, and consequently that the $\rho$-th moment of $\bigl| \setL^{Y^{n+n^\prime}} \bigr|$ converges to one as~$n$ tends to infinity. This proves that in the presence of perfect feedback the listsize capacity of $W ( y | x )$ is lower-bounded by $R_{\textnormal{cutoff}} (\rho)$.\\

\section{Problem Statement and Main Results}\label{sec:problemStatement}

We consider two problems: the ``guessing version'' and the ``list
version.'' The two differ in the definition of Bob's ambiguity. In
both versions a pair $( X,Y )$ is drawn from the finite set
$\setX \times \setY$ according to the PMF $P_{X,Y}$, and $\rho > 0$ is
fixed. Upon observing $( X,Y ) = ( x,y )$, Alice draws the hints $M_1$
and $M_2$ from some finite set $\setM_1 \times \setM_2$ according to
some conditional PMF
\begin{equation}
\distof{M_1 = m_1, M_2 = m_2 | X = x, Y = y }. \label{eq:aliceEncPMF}
\end{equation}
Bob sees both hints and the side information $Y$. In the guessing version Bob's ambiguity about $X$ is
\begin{equation}
\mathscr A^{(\textnormal g)}_{\textnormal B} ( P_{X,Y} ) = \min_{\guess {}{\cdot | M_1, M_2}} \bigEx {}{\guess {}{X | Y,M_1,M_2 }^\rho}. \label{eq:bobAmbiguityGuessing}
\end{equation}
In the list version Bob's ambiguity about $X$ is
\begin{equation} 
\mathscr A^{(\textnormal l)}_{\textnormal B} ( P_{X,Y} ) = \BigEx {}{\bigl| \setL^Y_{M_1,M_2} \bigr|^\rho}, \label{eq:bobAmbiguityList}
\end{equation}
where for all $y \in \setY$ and $( m_1,m_2 ) \in \setM_1 \times \setM_2$
\begin{equation}
\setL^y_{m_1,m_2} =  \bigl\{ x \colon \distof{ X = x | Y = y, M_1 = m_1, M_2 = m_2 } > 0 \bigr\}
\end{equation}
is the list of all the realizations of $X$ of positive posterior probability
\begin{IEEEeqnarray}{l}
\distof{ X = x | Y = y, M_1 = m_1, M_2 = m_2 } \nonumber \\
\quad = \frac{P_{X,Y} ( x,y ) \, \distof{M_1 = m_1, M_2 = m_2 | X = x, Y = y }}{\sum_{\tilde x} P_{X,Y} ( \tilde x,y ) \, \distof{M_1 = m_1, M_2 = m_2 | X = \tilde x, Y = y }}.
\end{IEEEeqnarray}
Eve sees one of the hints and guesses $X$ based on this hint and the
side information $Y$. Which of the hints is revealed to her is
determined by an accomplice of hers to minimize her guessing
efforts. In both versions Eve's ambiguity about $X$ is
\begin{equation} 
\mathscr A_{\textnormal E} ( P_{X,Y} ) = \min_{\guess 1 {\cdot | Y,M_1}, \, \guess 2 {\cdot | Y, M_2}} \bigEx {}{\guess 1 { X | Y, M_1 }^\rho \wedge \guess 2 { X | Y, M_2 }^\rho}. \label{eq:distEncSecrecyMeasure}
\end{equation}

Optimizing over Alice's mapping, i.e., the choice of the conditional PMF in \eqref{eq:aliceEncPMF}, we wish to characterize the largest ambiguity that we can guarantee that Eve will have subject to a given upper bound on the ambiguity that Bob may have.

Note that by quantifying Eve's ambiguity
    using~\eqref{eq:distEncSecrecyMeasure}, we are implicitly assuming that Eve's accomplice observes $X$ and $Y$ before determining the hint that minimizes Eve's guessing efforts. Less conservative is the
    ambiguity
\begin{IEEEeqnarray}{l}
\tilde {\mathscr A}_{\textnormal E} ( P_{X,Y} ) = \min_{k \in \{1,2 \}} \min_{\guess k {\cdot | Y, M_k}} \bigEx {}{ \guess k { X | Y, M_k }^\rho }, \label{eq:distEncSecrecyMeasureAlternative}
\end{IEEEeqnarray}
which applies if the accomplice does not observe $(X,Y)$ and reveals to
Eve the hint that in expectation over $(X,Y)$ minimizes her guessing
efforts. Definition~\eqref{eq:distEncSecrecyMeasureAlternative} is less conservative than~\eqref{eq:distEncSecrecyMeasureAlternative} in the sense that
\begin{IEEEeqnarray}{rCl}
\mathscr A_{\text{E}} ( P_{X,Y} ) \leq \tilde {\mathscr A}_{\textnormal E} ( P_{X,Y} ). \label{eq:eveAltAmb}
\end{IEEEeqnarray}
Why we prefer \eqref{eq:distEncSecrecyMeasure} over \eqref{eq:distEncSecrecyMeasureAlternative} is explained in Section~\ref{sec:discussion}.\\

Of special interest to us is the asymptotic regime where $( X,Y )$ is an $n$-tuple (not necessarily drawn IID), and where $$\setM_1 = \bigl\{ 1, \ldots, 2^{n R_1} \bigr\}, \quad \setM_2 = \bigl\{ 1, \ldots, 2^{n R_2} \bigr\},$$ where $( R_1,R_2 )$ is a nonnegative pair corresponding to the rate.\footnote{When we say that a positive integer $k \in \naturals$ assumes the value $2^{n R}$, where $R > 0$ corresponds to a rate, we mean that $k = \lfloor 2^{n R} \rfloor$.} For both versions of the problem, we shall characterize the largest exponential growth that we can guarantee for Eve's ambiguity subject to the constraint that Bob's ambiguity tend to one.\footnote{Note that in the guessing version $\guess {}{X | Y,M_1,M_2}^\rho$ is one iff Bob's first guess is $X^n$, and in the list version $\bigl| \setL^Y_{M_1,M_2} \bigr|^\rho$ is one iff Bob forms the ``perfect'' list comprising only $X^n$.} This asymptote turns out not to depend on the version of the problem, and in the asymptotic analysis $\mathscr A_{\textnormal B}$ can stand for either $\mathscr A_{\textnormal B}^{(\textnormal g)}$ or $\mathscr A_{\textnormal B}^{(\textnormal l)}$.

The following definition phrases mathematically what we mean by the ``largest exponential growth that we can guarantee for Eve's ambiguity:''

\begin{defin}[Privacy-Exponent] \label{de:EvesAmbig}
{\em Let $\bigl\{ (X_i, Y_i) \bigr\}_{i \in \naturals}$ be a stochastic process over the finite alphabet $\setX \times \setY$, and denote by $P_{X^n,Y^n}$ the PMF of $( X^n, Y^n )$. Given a nonnegative rate-pair $( R_1, R_2 )$, we call $E_{\textnormal E}$ an \emph{achievable ambiguity-exponent} if there exists a sequence of stochastic encoders such that Bob's ambiguity (which is always at least one) satisfies
\begin{equation} 
\lim_{n \rightarrow \infty} \mathscr A_{\textnormal B} ( P_{X^n,Y^n} ) = 1, \label{eq:bobAmbTo1}
\end{equation}
and such that Eve's ambiguity satisfies
\begin{equation} 
\liminf_{n \rightarrow \infty} \frac{\log \bigl( \mathscr A_{\textnormal E} ( P_{X^n,Y^n} ) \bigr)}{n} \geq E_{\textnormal E}. \label{eq:eveGuessBobAmbTo1}
\end{equation}
The \emph{privacy-exponent} $\overbar{E_{\textnormal E}}$ is the supremum of all achievable ambiguity-exponents. If \eqref{eq:bobAmbTo1} cannot be satisfied, then the set of achievable ambiguity-exponents is empty, and we define the privacy-exponent as negative infinity.}
\end{defin}

We also consider a scenario where we impose only a modest requirement on Bob's ambiguity and allow it to grow exponentially with a given normalized (by $n$) exponent $E_{\textnormal B}$. For this scenario the following definition introduces the mathematical quantity by which we characterize the largest exponential growth that we can guarantee for Eve's ambiguity:

\begin{defin}[Modest Privacy-Exponent] \label{de:EvesAmbigModest}
{\em
Let $E_{\textnormal B} \geq 0$. We call $E^{\textnormal m}_{\textnormal E} ( E_{\textnormal B} )$ an \emph{achievable modest-ambiguity-exponent} if there is a sequence of stochastic encoders such that Bob's ambiguity satisfies
\begin{equation}
\limsup_{n \rightarrow \infty} \frac{\log \bigl( \mathscr A_{\textnormal B} ( P_{X^n,Y^n} ) \bigr)}{n} \leq E_{\textnormal B}, \label{eq:bobAmbToEB}
\end{equation}
and such that Eve's ambiguity satisfies
\begin{equation}
\liminf_{n \rightarrow \infty} \frac{\log \bigl( \mathscr A_{\textnormal E} ( P_{X^n,Y^n} ) \bigr)}{n} \geq E^{\textnormal m}_{\textnormal E} ( E_{\textnormal B} ). \label{eq:eveGuessBobAmbToEB}
\end{equation}
For every $E_{\textnormal B} \geq 0$, the \emph{modest privacy-exponent} $\overbar{E^{\textnormal m}_{\textnormal E} ( E_{\textnormal B} )}$ is the supremum of all achievable modest-ambiguity-exponents. If \eqref{eq:bobAmbToEB} cannot be satisfied, then the set of achievable modest-ambiguity-exponents is empty, and we define the modest privacy-exponent as negative infinity.}
\end{defin}

We next present our results to the stated problems in the finite-blocklength regime (Section~\ref{sec:finiteBlocklengthResults}) and in the asymptotic regime (Section~\ref{sec:asymptoticResults}).

\subsection{Finite-Blocklength Results}\label{sec:finiteBlocklengthResults}

In the next two theorems $c_{\textnormal s}$ is related to how much information can be gleaned about the secret $X$ from the pair of hints $( M_1,M_2 )$ but not from one hint alone; $c_1$ is related to how much can be gleaned from $M_1$; and $c_2$ is related to how much can be gleaned from $M_2$. More precisely, in the proof of the two theorems (see Section~\ref{sec:distStorProofs} ahead) we shall see that Alice first maps $( X,Y )$ to the triple $( V_{\textnormal s}, V_1, V_2 )$, which takes value in a set $\setV_{\textnormal s} \times \setV_1 \times \setV_2$, whose marginal cardinalities satisfy $| \setV_\nu | = c_\nu, \,\, \nu \in \{ \textnormal s,1,2 \}$. Independently of $( X,Y )$ she then draws a (one-time-pad like) random variable $U$ uniformly over $\setV_{\textnormal s}$ and maps $( U, V_{\textnormal s} )$ to a variable $\widetilde V_{\textnormal s}$ choosing the (XOR like) mapping so that $V_{\textnormal s}$ can be recovered from $( \widetilde V_{\textnormal s}, U )$ while $\widetilde V_{\textnormal s}$ alone is independent of $( X,Y )$. The hints are $M_1 = ( \widetilde V_{\textnormal s}, V_1 )$ and $M_2 = ( U, V_2 )$. Since the tuple $( \widetilde V_{\textnormal s}, V_1 )$ takes value in the set $\setV_{\textnormal s} \times \setV_1$ of size $c_{\textnormal s} c_1$, we must have that $c_{\textnormal s} c_1 \leq | \setM_1 |$. Likewise, we must have that $c_{\textnormal s} c_2 \leq | \setM_2 |$. Because $c_{\textnormal s}$, $c_1$, and $c_2$ are positive integers, they thus satisfy \eqref{bl:condCsC1C2Guess} ahead. Alice does not use randomization if $c_{\textnormal s} = 1$.

\begin{theorem}[Finite-Blocklength Guessing-Version]\label{th:distStorGuess}
For every triple $( c_{\textnormal s}, c_1, c_2 ) \in \naturals^3$ satisfying
\begin{subequations}\label{bl:condCsC1C2Guess}
\begin{IEEEeqnarray}{l}
c_{\textnormal s} \leq \card {\setM_1} \wedge \card {\setM_2},\quad c_1 \leq \bigl\lfloor \card {\setM_1} / c_s \bigr\rfloor, \quad c_2 \leq \bigl\lfloor \card {\setM_2} / c_s \bigr\rfloor, 
\end{IEEEeqnarray}
\end{subequations}
there is a choice of the conditional PMF in \eqref{eq:aliceEncPMF} for which Bob's ambiguity about $X$ is upper-bounded by
\begin{IEEEeqnarray}{l}
\mathscr A_{\textnormal B}^{(\textnormal g)} ( P_{X,Y} ) < 1 + 2^{\rho ( \renent {\tirho} {X | Y } - \log ( c_{\textnormal s} c_1 c_2 ) + 1 )}, \label{eq:BobMomDistStorDirGuess}
\end{IEEEeqnarray}
and Eve's ambiguity about $X$ is lower-bounded by
\begin{IEEEeqnarray}{l}
\mathscr A_{\textnormal E} ( P_{X,Y} ) \geq \bigl( 1 + \ln \card \setX \bigr)^{-\rho} 2^{\rho ( \renent {\tirho} { X | Y } - \log ( c_1 + c_2 ) )}. \label{eq:EveMomDistStorDirGuess}
\end{IEEEeqnarray}
Conversely, for every conditional PMF, Bob's ambiguity is lower-bounded by
\begin{IEEEeqnarray}{l}
\mathscr A_{\textnormal B}^{(\textnormal g)} ( P_{X,Y} ) \geq \bigl( 1 + \ln \card \setX \bigr)^{- \rho} 2^{\rho ( \renent {\tirho} { X | Y } - \log ( \card {\setM_1} \, \card {\setM_2} ) )} \vee 1, \label{eq:BobMomDistStorConvGuess}
\end{IEEEeqnarray}
and Eve's ambiguity is upper-bounded by
\begin{IEEEeqnarray}{l}
\mathscr A_{\textnormal E} ( P_{X,Y} ) \leq \bigl( \card {\setM_1} \wedge \card {\setM_2} \bigr)^\rho \mathscr A^{(\textnormal g)}_{\textnormal B} ( P_{X,Y} ) \wedge 2^{\rho \renent {\tirho} {X | Y }}, \label{eq:EveMomDistStorConvGuess}
\end{IEEEeqnarray}
where \eqref{eq:EveMomDistStorConvGuess} holds even if we replace \eqref{eq:distEncSecrecyMeasure} by \eqref{eq:distEncSecrecyMeasureAlternative}, i.e.,
\begin{IEEEeqnarray}{l}
\tilde {\mathscr A}_{\textnormal E} ( P_{X,Y} ) \leq \bigl( \card {\setM_1} \wedge \card {\setM_2} \bigr)^\rho \mathscr A^{(\textnormal g)}_{\textnormal B} ( P_{X,Y} ) \wedge 2^{\rho \renent {\tirho} {X | Y }}, \label{eq:EveMomDistStorConvGuessAlternative}
\end{IEEEeqnarray}
\end{theorem}

\begin{proof}
See Section~\ref{sec:distStorProofsThmsFiniteBlockl}.
\end{proof}

\begin{theorem}[Finite-Blocklength List-Version]\label{th:distStor} \hspace{-2mm} If $| \setM_1 | \, | \setM_2 | > \log \card \setX + 2$, then for every triple $( c_{\textnormal s}, c_1, c_2 ) \in \naturals^3$ satisfying
\begin{subequations}\label{bl:condCsC1C2}
\begin{IEEEeqnarray}{l}
c_{\textnormal s} \leq \card {\setM_1} \wedge \card {\setM_2},\quad c_1 \leq \bigl\lfloor \card {\setM_1} / c_{\textnormal s} \bigr\rfloor, \quad c_2 \leq \bigl\lfloor \card {\setM_2} / c_{\textnormal s} \bigr\rfloor, \\
c_{\textnormal s} c_1 c_2 > \log \card \setX + 2, \label{eq:condCsC1C2List}
\end{IEEEeqnarray}
\end{subequations}
there is a choice of the conditional PMF in \eqref{eq:aliceEncPMF} for which Bob's ambiguity about $X$ is upper-bounded by
\begin{IEEEeqnarray}{l}
\mathscr A_{\textnormal B}^{(\textnormal l)} ( P_{X,Y} ) < 1 + 2^{\rho ( \renent {\tirho} {X | Y } - \log ( c_{\textnormal s} c_1 c_2 - \log \card \setX - 2 ) + 2 )}, \label{eq:BobMomDistStorDir}
\end{IEEEeqnarray}
and Eve's ambiguity about $X$ is lower-bounded by
\begin{IEEEeqnarray}{l}
\mathscr A_{\textnormal E} ( P_{X,Y} ) \geq \bigl( 1 + \ln \card \setX \bigr)^{-\rho} 2^{\rho ( \renent {\tirho} { X | Y } - \log ( c_1 + c_2 ) )}. \label{eq:EveMomDistStorDir}
\end{IEEEeqnarray}
Conversely, for every conditional PMF, Bob's ambiguity is lower-bounded by
\begin{IEEEeqnarray}{l}
\mathscr A_{\textnormal B}^{(\textnormal l)} ( P_{X,Y} ) \geq 2^{\rho ( \renent {\tirho} {X | Y } - \log ( \card {\setM_1} \, \card {\setM_2} ))} \vee 1, \label{eq:BobMomDistStorConv}
\end{IEEEeqnarray}
and Eve's ambiguity is upper-bounded by
\begin{IEEEeqnarray}{l}
\mathscr A_{\textnormal E} ( P_{X,Y} ) \leq \bigl( \card {\setM_1} \wedge \card {\setM_2} \bigr)^\rho \mathscr A^{(\textnormal l)}_{\textnormal B} ( P_{X,Y} ) \wedge 2^{\rho \renent {\tirho} {X | Y }}, \label{eq:EveMomDistStorConv}
\end{IEEEeqnarray}
where \eqref{eq:EveMomDistStorConv} holds even if we replace \eqref{eq:distEncSecrecyMeasure} by \eqref{eq:distEncSecrecyMeasureAlternative}, i.e.,
\begin{IEEEeqnarray}{l}
\tilde {\mathscr A}_{\textnormal E} ( P_{X,Y} ) \leq \bigl( \card {\setM_1} \wedge \card {\setM_2} \bigr)^\rho \mathscr A^{(\textnormal l)}_{\textnormal B} ( P_{X,Y} ) \wedge 2^{\rho \renent {\tirho} {X | Y }}, \label{eq:EveMomDistStorConvAlternative}
\end{IEEEeqnarray}
\end{theorem}

\begin{proof}
See Section~\ref{sec:distStorProofsThmsFiniteBlockl}.
\end{proof}

We next present the finite-blocklength results (Theorems~\ref{th:distStorGuess} and \ref{th:distStor}) in a simplified and more accessible form:

\begin{corollary}[Simplified Finite-Blocklength Guessing-Version]\label{prop:intGuessing}
For any constant $\mathscr U_{\textnormal B}$ satisfying
\begin{IEEEeqnarray}{l}
\mathscr U_{\textnormal B} \geq 1 + 2^\rho \bigl( \card {\setM_1} \, \card {\setM_2} \bigr)^{-\rho} 2^{\rho \renent {\tirho}{X|Y} }, \label{eq:intUBLB}
\end{IEEEeqnarray}
there is a choice of the conditional PMF in \eqref{eq:aliceEncPMF} for which Bob's ambiguity about $X$ is upper-bounded by
\begin{IEEEeqnarray}{l}
\mathscr A^{(\textnormal g)}_{\textnormal B} ( P_{X,Y} ) < \mathscr U_{\textnormal B}, \label{eq:intBobAmbUB}
\end{IEEEeqnarray}
and Eve's ambiguity about $X$ is lower-bounded by
\begin{IEEEeqnarray}{l} 
\mathscr A_{\textnormal E} ( P_{X,Y} ) \geq 2^{-\rho} \bigl( 1 + \ln \card \setX \bigr)^{-\rho} \Bigl[ 2^{-4 \rho} \bigl( \card {\setM_1} \wedge \card {\setM_2} \bigr)^\rho ( \mathscr U_{\textnormal B} - 1 ) \wedge 2^{\rho \renent {\tirho}{X|Y}} \Bigr]. \label{eq:intEveAmbLB}
\end{IEEEeqnarray}
Conversely, \eqref{eq:intBobAmbUB} cannot hold for
\begin{IEEEeqnarray}{l} 
\mathscr U_{\textnormal B} < \bigl( 1 + \ln \card \setX \bigr)^{-\rho} \bigl( \card {\setM_1} \, \card {\setM_2} \bigr)^{-\rho} 2^{\rho \renent {\tirho}{X|Y} } \vee 1, \label{eq:intUBLBConv}
\end{IEEEeqnarray}
and if Bob's ambiguity satisfies \eqref{eq:intBobAmbUB} for some $\mathscr U_{\text B}$, then Eve's ambigutiy about $X$ is upper-bounded by
\begin{IEEEeqnarray}{l} 
\mathscr A_{\textnormal E} ( P_{X,Y} ) \leq \bigl( \card {\setM_1} \wedge \card {\setM_2} \bigr)^\rho \mathscr U_{\textnormal B} \wedge 2^{\rho \renent {\tirho}{X|Y}}. \label{eq:intEveAmbUB}
\end{IEEEeqnarray}
\end{corollary}

\begin{proof}
The result is a corollary to Theorem~\ref{th:distStorGuess} (see Appendix~\ref{app:pfIntGuessing} for a proof).
\end{proof}

\begin{corollary}[Simplified Finite-Blocklength List-Version]\label{prop:intList}
For $|\setM_1| \, |\setM_2| > \log | \setX | + 2$ and any constant $\mathscr U_{\textnormal B}$ satisfying
\begin{IEEEeqnarray}{l} 
\mathscr U_{\textnormal B} \geq 1 + 2^{\rho ( \renent {\tirho} { X | Y }  - \log ( | \setM_1 | \, | \setM_2 | - \log | \setX | - 2 ) + 2 ) }, \label{eq:intUBLBList}
\end{IEEEeqnarray} 
there is a choice of the conditional PMF in \eqref{eq:aliceEncPMF} for which Bob's ambiguity about $X$ is upper-bounded by
\begin{IEEEeqnarray}{l} 
\mathscr A^{(\textnormal l)}_{\textnormal B} ( P_{X,Y} ) < \mathscr U_{\textnormal B}, \label{eq:intBobAmbUBList}
\end{IEEEeqnarray}
and Eve's ambiguity about $X$ is lower-bounded by
\begin{IEEEeqnarray}{rcl} 
\mathscr A_{\textnormal E} ( P_{X,Y} ) \geq 2^{-\rho} \bigl( 1 + \ln \card \setX \bigr)^{-\rho} &\Bigl[ &2^{-6 \rho} \bigl( \card {\setM_1} \wedge \card {\setM_2} \bigr)^\rho ( \mathscr U_{\textnormal B} - 1 ) \nonumber \\
&&\wedge \, 2^{-4 \rho} \bigl( 2 + \log \card \setX \bigr)^{-\rho} \bigl( \card {\setM_1} \wedge \card {\setM_2} \bigr)^\rho 2^{\rho \renent {\tirho}{X|Y}} \nonumber \\
&&\wedge \, 2^{\rho \renent {\tirho}{X|Y}} \Bigr]. \label{eq:intEveAmbLBList}
\end{IEEEeqnarray}
Conversely, \eqref{eq:intBobAmbUBList} cannot hold for
\begin{IEEEeqnarray}{l} 
\mathscr U_{\textnormal B} < \bigl( \card {\setM_1} \, \card {\setM_2} \bigr)^{-\rho} 2^{\rho \renent {\tirho}{X|Y} } \vee 1, \label{eq:intUBLBConvList}
\end{IEEEeqnarray}
and if Bob's ambiguity satisfies \eqref{eq:intBobAmbUBList} for some $\mathscr U_{\textnormal B}$, then Eve's ambigutiy about $X$ is upper-bounded by
\begin{IEEEeqnarray}{l} 
\mathscr A_{\textnormal E} ( P_{X,Y} ) \leq \bigl( \card {\setM_1} \wedge \card {\setM_2} \bigr)^\rho \mathscr U_{\textnormal B} \wedge 2^{\rho \renent {\tirho} {X \left| Y \right.}}. \label{eq:intEveAmbUBList}
\end{IEEEeqnarray}
\end{corollary}

\begin{proof}
The result is a corollary to Theorem~\ref{th:distStor} (see Appendix~\ref{app:pfIntList} for a proof).
\end{proof}

Note that the simplified achievability results (namely \eqref{eq:intUBLB}--\eqref{eq:intEveAmbLB} in the guessing version and \eqref{eq:intUBLBList}--\eqref{eq:intEveAmbLBList} in the list version) match the corresponding converse results (namely \eqref{eq:intUBLBConv}--\eqref{eq:intEveAmbUB} in the guessing version and \eqref{eq:intUBLBConvList}--\eqref{eq:intEveAmbUBList} in the list version) up to polylogarithmic factors of $\card \setX$.

\subsection{Asymptotic Results} \label{sec:asymptoticResults}

Suppose now that $( X,Y )$ is an $n$-tuple. We study the asymptotic regime where $n$ tends to infinity. Recall that in this regime we refer to both $\mathscr A^{(\textnormal g)}_{\textnormal B}$ and $\mathscr A^{(\textnormal l)}_{\textnormal B}$ by $\mathscr A_{\textnormal B}$, because the results are the same for both versions of the problem. Theorems~\ref{th:distStorGuess} and \ref{th:distStor} imply the following asymptotic result:

\begin{theorem}[Privacy-Exponent] \label{th:asympDistStor}
Let $\bigl\{ (X_i,Y_i) \bigr\}_{i \in \naturals}$ be a discrete-time stochastic process with finite alphabet $\setX \times \setY$, and suppose its conditional R\'enyi entropy-rate $\renent {\tirho}{ \rndvecX | \rndvecY }$ is well-defined. Given any positive rate-pair $( R_1, R_2 )$, the privacy-exponent is
\begin{IEEEeqnarray}{l} 
\overbar{E_{\textnormal E}} = \begin{cases} \rho \bigl( R_1 \wedge R_2 \wedge \renent {\tirho}{ \rndvecX | \rndvecY } \bigr) & R_1 + R_2 > \renent {\tirho}{ \rndvecX | \rndvecY }, \\ - \infty, & R_1 + R_2 < \renent {\tirho}{ \rndvecX | \rndvecY }. \end{cases} \label{eq:secStrongConst}
\end{IEEEeqnarray}
\end{theorem}

\begin{proof}
See Section~\ref{sec:distStorProofsThmAsymp}.
\end{proof}

Suppose now that Bob's ambiguity need not tend to one but can grow exponentially with a given normalized (by $n$) exponent $E_{\textnormal B}$. For this case Theorems~\ref{th:distStorGuess} and \ref{th:distStor} imply the following asymptotic result:

\begin{theorem}[Modest Privacy-Exponent] \label{th:asympDistStorEB}
Let $\bigl\{ (X_i,Y_i) \bigr\}_{i \in \naturals}$ be a discrete-time stochastic process with finite alphabet $\setX \times \setY$, and suppose its conditional R\'enyi entropy-rate $\renent {\tirho}{ \rndvecX | \rndvecY }$ is well-defined. Given any positive rate-pair $( R_1, R_2 )$, the modest privacy-exponent for $E_{\textnormal B} \geq 0$ is
\begin{IEEEeqnarray}{l} 
\!\!\!\!\! \overbar{E^{\textnormal m}_{\textnormal E} ( E_{\textnormal B} )} = \begin{cases} \bigl( \rho ( R_1 \wedge R_2 ) + E_{\textnormal B} \bigr) \wedge \rho \renent {\tirho}{ \rndvecX | \rndvecY } & R_1 + R_2 \geq \renent {\tirho}{ \rndvecX | \rndvecY } - \rho^{-1} E_{\textnormal B}, \\ - \infty & R_1 + R_2 < \renent {\tirho}{ \rndvecX | \rndvecY } - \rho^{-1} E_{\textnormal B}. \end{cases} \label{eq:secModConst}
\end{IEEEeqnarray}
\end{theorem}

\begin{proof}
See Section~\ref{sec:distStorProofsThmAsympEB}.
\end{proof}

\section{Discussion}\label{sec:discussion}

This section provides some intuition and discusses some of the models
and their underlying  assumptions.  We begin with some intuition as to why
the guessing and list-size criteria for Bob lead to similar
results. Then, we explain why we quantify Eve's ambiguity by
\eqref{eq:distEncSecrecyMeasure}. We show that if---rather than guessing---Eve
were required to form a list, then perfect secrecy would come almost
for free. Finally, we explain how our results change in the following
two scenarios: 1) Alice knows which hint Eve observes; or 2) Alice
describes $X$ using only one hint, but Alice and Bob see a secret key,
which is not revealed to Eve.

The following remark explains why the results for the guessing and the
list version differ only by polylogarithmic factors of $| \setX |$
(and are consequently the same in the asymptotic regime):

\begin{remark}[\textbf{Why Do the Two Criteria for Bob Lead to Similar Results?}] \label{remark:guessListClose}
{\em Consider any choice of the conditional PMF in \eqref{eq:aliceEncPMF}. In the guessing version Bob uses an optimal guessing function $\guessast {}{\cdot | Y,M_1,M_2 }$ (which minimizes $\bigEx {}{\guess {}{X | Y,M_1,M_2 }^\rho}$) to guess $X$ based on the side information $Y$ and the hints $M_1$ and $M_2$, and his ambiguity is $\bigEx {}{\guessast {}{X | Y,M_1,M_2 }^\rho}$. By Corollary~\ref{co:guessToList} we can construct from $\guessast {}{\cdot | Y,M_1,M_2 }$ an additional hint $M$ that takes values in a set of size at most $1 + \bigl\lfloor \log \card \setX \bigr\rfloor$ such that
\begin{equation}
\BigEx {}{ \bigl| \setL^Y_{M_1,M_2,M} \bigr|^\rho} \leq \bigEx {}{\guessast {}{ X | Y, M_1, M_2 }^\rho},
\end{equation}
where $\setL^Y_{M_1,M_2,M}$ is the smallest list that is guaranteed to contain $X$ given $( Y,M_1,M_2,M )$. Suppose now that Alice maps $X$ to the hints $M_1^\prime \triangleq (M_1,M)$ and $M_2^\prime \triangleq M_2$. This implies that Bob's ambiguity in the list version is $$\BigEx {}{\bigl| \setL^Y_{M_1^\prime,M_2^\prime} \bigr|^\rho} = \BigEx {}{ \bigl| \setL^Y_{M_1,M_2,M} \bigr|^\rho}$$ and consequently no larger than $\bigEx {}{\guessast {}{X | Y,M_1,M_2 }^\rho}$. Moreover, because $M$ takes values in a set of size at most $1 + \bigl\lfloor \log \card \setX \bigr\rfloor$, we can use Lemma~\ref{le:ImproveGuess} to show that---compared to the case where the hints are $M_1$ and $M_2$---Eve's ambiguity decreases by at most a polylogarithmic factor of $| \setX |$.}
\end{remark}

We next explain why we choose to quantify Eve's ambiguity by \eqref{eq:distEncSecrecyMeasure} and not by \eqref{eq:distEncSecrecyMeasureAlternative}. As we have seen, \eqref{eq:distEncSecrecyMeasure} is more conservative than \eqref{eq:distEncSecrecyMeasureAlternative} in the sense that \eqref{eq:eveAltAmb} holds. Consequently, it follows from \eqref{eq:EveMomDistStorConvGuessAlternative} and \eqref{eq:EveMomDistStorConvAlternative} that the results of Theorems~\ref{th:distStorGuess} and \ref{th:distStor} hold irrespective of whether we quantify Eve's ambiguity by \eqref{eq:distEncSecrecyMeasure} or by \eqref{eq:distEncSecrecyMeasureAlternative}. We prefer to quantify Eve's ambiguity by
\eqref{eq:distEncSecrecyMeasure}, because---as the following example shows---\eqref{eq:distEncSecrecyMeasureAlternative} leads to a weaker notion
of secrecy than \eqref{eq:distEncSecrecyMeasure}:

\begin{example} \label{ex:notionOfSecrecy} {\em Suppose that $Y$ is
    null, $X$ is uniform over $\setX$, and Alice produces the hints at
    random: they are equally likely to be $( M_1 = X, M_2 = \ast )$ or
    $( M_1 = \ast, M_2 = X )$, where the symbol~$\ast$ is not in
    $\setX$. Since Bob can recover $X$ from $( M_1, M_2 )$ (by
    producing the hint that is not $\ast$),
  $$\min_{\guess {}{\cdot | M_1, M_2}} \bigEx {}{\guess {}{X |
      M_1,M_2}^\rho} = \bigEx {}{\card {\setL_{M_1,M_2}}^\rho} = 1.$$
  The system is clearly insecure, because one of the hints always
  reveals $X$, and ${\mathscr A}_{\textnormal E} ( P_{X,Y} ) = 1$.
  However, as we next argue, this weakness is not captured by
  $\tilde {\mathscr A}_{\textnormal E} ( P_{X,Y} )$.  The probability
  of $M_1$ being $\ast$ is $1/2$, so the $\rho$-th moment of
  $\guess 1 {X | M_1}$ is at least
  $\min_{\guess {}{\cdot}} \bigEx {}{\guess {}{X}^\rho}/2$. Likewise,
  by symmetry, for $\guess 2 {X | M_2}$. Thus
  $\tilde {\mathscr A}_{\textnormal E} ( P_{X,Y} )$ differs from
  $\min_{\guess {}{\cdot}} \bigEx {}{\guess {}{X}^\rho}$ by a factor
  of at most $1/2$.}
\end{example}

So far, we have explained why we prefer
\eqref{eq:distEncSecrecyMeasure} over \eqref{eq:distEncSecrecyMeasureAlternative}. But
why do we allow Eve to guess even in the list version of our
problem? That is, why do we prefer \eqref{eq:distEncSecrecyMeasure}
over
\begin{equation} \label{eq:fairOpponentEveAmbList}
\mathscr A_{\textnormal E}^{(\textnormal l)} = \BigEx {}{\bigl|\setL_{M_1}^Y\bigr|^\rho \wedge \bigl|\setL_{M_2}^Y \bigr|^\rho}
 \end{equation}
%i.e., why we allow Eve to guess $X$ and not measure her ambiguity
%using lists. Indeed, in the list version it would seem fair to require
%that---like Bob---also Eve form a list. 
even when Bob must form a list? 

We prefer \eqref{eq:distEncSecrecyMeasure} over
\eqref{eq:fairOpponentEveAmbList} because, as
Theorem~\ref{th:fairOpponent} ahead will show, forcing Eve to produce
a short list would severely handicap her and make it trivial to defeat
her: when Eve must form a list, perfect secrecy is almost free. 
% To see why, suppose we were to quantify Bob's ambiguity by
% \eqref{eq:bobAmbiguityList} and Eve's by
% \eqref{eq:fairOpponentEveAmbList}. The following theorem characterizes
% the largest ambiguity that can be guaranteed for Eve subject to a
% given upper bound on Bob's ambiguity:
\begin{theorem}[Eve Must Form a List] \label{th:fairOpponent}
If 
\begin{equation}
\label{eq:amosCOND}
| \setM_1 | \wedge | \setM_2 | \geq 1 + \bigl\lfloor \log |\setX|
\bigr\rfloor,
\end{equation}
then there exists a conditional PMF as in
\eqref{eq:aliceEncPMF} for which Bob's ambiguity about $X$ is
upper-bounded by
\begin{IEEEeqnarray}{l}
\mathscr A_{\textnormal B}^{(\textnormal l)} ( P_{X,Y} ) \leq 1 + 2^{\rho ( \renent {\tirho} {X | Y } - \log ( |\setM_1| \, |\setM_2| ) + 2 \log (1 + \lfloor \log |\setX \rfloor) + 3 )}, \label{eq:BobMomEveListsDir}
\end{IEEEeqnarray}
and Eve's ambiguity about $X$ is
\begin{IEEEeqnarray}{l}
\mathscr A_{\textnormal E}^{(\textnormal l)} ( P_{X,Y} ) = \bigEx {}{|\setL_Y|^\rho}, \label{eq:EveMomEveListsDir}
\end{IEEEeqnarray}
where
\begin{IEEEeqnarray}{l}
\bigEx {}{|\setL_Y|^\rho} = \sum_y P_Y (y) \, \bigl| \bigl\{ x \in \setX \colon P_{X|Y} (x|y) > 0 \bigr\} \bigr|^\rho.
\end{IEEEeqnarray}
Conversely, for every conditional PMF, Bob's ambiguity is lower-bounded by
\begin{IEEEeqnarray}{l}
\mathscr A_{\textnormal B}^{(\textnormal l)} ( P_{X,Y} ) \geq 2^{\rho ( \renent {\tirho} {X | Y } - \log ( \card {\setM_1} \, \card {\setM_2} ) )} \vee 1, \label{eq:BobMomEveListsConv}
\end{IEEEeqnarray}
and Eve's ambiguity is upper-bounded by
\begin{IEEEeqnarray}{l}
\mathscr A_{\textnormal E}^{(\textnormal l)} ( P_{X,Y} ) \leq \bigEx {}{|\setL_Y|^\rho}. \label{eq:EveMomEveListsConv}
\end{IEEEeqnarray}
\end{theorem}

\begin{proof}
See Appendix~\ref{app:pfThFairOpponent}.
\end{proof}

To see why perfect secrecy is almost free when Eve is required to form a list,
note that the RHS of~\eqref{eq:EveMomEveListsDir} would also be Eve's
list size if she only saw $Y$ and did not get to see \emph{any} hint,
so in this sense achieving~\eqref{eq:EveMomEveListsDir} is tantamount
to achieving perfect secrecy. And the cost is very small:
Condition~\eqref{eq:amosCOND} is satisfied in the large-blocklength
regime whenever the rates of the two hints are positive; and the RHS
of~\eqref{eq:BobMomEveListsDir} will tend to one in this regime
whenever the sum of the rates exceeds the conditional R\'enyi entropy
rate---a condition that is necessary even in the absence of an
adversay (Theorem~\ref{th:optTaskEnc}).

% As a result of Theorem~\ref{th:fairOpponent}, we see that, if Eve were to form a list that is guaranteed to contain $X$, then perfect secrecy would come almost for free. Indeed, \eqref{eq:BobMomEveListsDir} and \eqref{eq:EveMomEveListsDir} are achievable whenever the size of each of the sets $\setM_1$ and $\setM_2$ is at least $1 + \bigl\lfloor \log |\setX| \bigr\rfloor$; and in asymptotic settings this holds as soon as each hint has a positive rate, no matter how small it be. By Theorem~\ref{th:optTaskEnc} achieving \eqref{eq:BobMomEveListsDir} is tantamount to guaranteeing that Bob's ambiguity is---up to polylogarithmic factors of $|\setX|$---as small as it can be when he sees a hint whose support has size $|\setM_1| \, |\setM_2|$. Moreover, achieving \eqref{eq:EveMomEveListsDir} is tantamount to guaranteeing perfect secrecy, because \eqref{eq:EveMomEveListsDir} means that Eve's ambiguity is as large as it can possibly be. This proves that perfect secrecy would come almost for free and consequently explains why a list-decoder is not a fair opponent for a list-decoder.

That perfect secrecy is (almost) free when we quantify Eve's
ambiguity by \eqref{eq:fairOpponentEveAmbList} is highly intuitive: By
forcing Eve to form a list that is guaranteed to contain $X$, we force
her to include in her list all the realizations of $X$ that have a
positive posterior probability, no matter how small. This implies
that, if Eve were to form a list, then perfect secrecy could be
attained by hiding very little information from Eve. The situation is
different in case Eve guesses $X$, because allowing Eve to guess $X$,
i.e., quantifying Eve's ambiguity by \eqref{eq:distEncSecrecyMeasure},
is tantamount to first indexing the elements of the list in
\eqref{eq:fairOpponentEveAmbList}---which she would otherwise have to
form---in decreasing order of their posterior probability, and to then
downweigh the large indices of the realizations at the bottom of the
list by their small posterior probabilities.

To conclude the discussion of how to quantify Eve's ambiguity, we relate Eve's ambiguity \eqref{eq:distEncSecrecyMeasure} to the concept of \emph{equivocation}. In the classical Shannon cipher system \cite{shannon49}, a popular way to measure imperfect secrecy is in terms of equivocation, i.e., in terms of the conditional entropy $H (X|Z)$, where $X$ denotes some sensitive information and $Z$ Eve's observation. In the settings where Bob is a list-decoder or a guessing decoder, R\'enyi entropy plays the role of Shannon entropy in the sense that the minimum required rate to encode an $n$-tuple $X = X^n$ is the R\'enyi entropy rate $\renent {\tirho}{\rndvecX}$ rather than the Shannon entropy rate $H (\rndvecX) = \renent {1} {\rndvecX}$ (this follows from Theorems~\ref{th:optTaskEnc} and Corollary~\ref{co:equivBunteResultGuessing}). Consequently, in these settings the conditional R\'enyi entropy $\renent {\tirho}{X|Z}$ qualifies as a ``natural'' equivalent for equivocation. But $\renent {\tirho}{X|Z}$ has a nice operational characterization: $2^{\rho \renent {\tirho}{X|Z}}$ is (up to polylogarithmic factors of $| X |$) the $\rho$-th moment of the number of guesses that Eve needs to guess $X$ from her observation $Z$ (see Theorem~\ref{th:optGuessFun}). This is another reason why it makes sense to quantify Eve's ambiguity in terms of the $\rho$-th moment of the number of guesses that she needs to guess $X$.\\

In the remainder of this section we briefly discuss how the results of
Theorems~\ref{th:distStorGuess} and \ref{th:distStor} change in the
following two scenarios: 1) Alice knows which hint Eve observes; or 2)
Alice describes $X$ using only one hint, but Alice and Bob share a
secret key, which is unknown to Eve. We begin with Scenario~1. In
this scenario Alice draws the public hint $M_{\textnormal p}$ and the
secret hint $M_{\textnormal s}$ from some finite set
$\setM_{\textnormal p} \times \setM_{\textnormal s}$ according to some
conditional PMF
\begin{equation}
\distof{M_{\textnormal p} = m_{\textnormal p}, M_{\textnormal s} = m_{\textnormal s} | X = x, Y = y }. \label{eq:aliceEncPMFSecMsg}
\end{equation}
Bob sees both hints. In the guessing version his ambiguity about $X$ is
\begin{IEEEeqnarray}{l}
\mathscr A^{(\textnormal g)}_{\textnormal B} ( P_{X,Y} ) = \min_{\guess {}{\cdot | Y, M_{\textnormal p}, M_{\textnormal s}}} \bigEx {}{\guess {}{X | Y,M_{\textnormal p},M_{\textnormal s}}^\rho} \label{eq:bobAmbiguityGuessingSecMsg}
\end{IEEEeqnarray}
and in the list version
\begin{IEEEeqnarray}{l} 
\mathscr A^{(\textnormal l)}_{\textnormal B} ( P_{X,Y} ) = \BigEx {}{\bigl| \setL^Y_{M_{\textnormal p},M_{\textnormal s}} \bigr|^\rho}. \label{eq:bobAmbiguityListSecMsg}
\end{IEEEeqnarray}
Eve sees only the public hint. In both versions her ambiguity about $X$ is
\begin{IEEEeqnarray}{l}
\mathscr A_{\textnormal E} ( P_{X,Y} ) = \min_{\guess {}{\cdot | Y, M_{\textnormal p}}} \bigEx {}{\guess {}{ X | Y, M_{\textnormal p} }^\rho }. \label{eq:distEncSecrecyMeasureSecMsg}
\end{IEEEeqnarray}

The next two theorems characterize the largest ambiguity that we can guarantee that Eve will have subject to a given upper bound on the ambiguity that Bob may have (see Appendix~\ref{app:pfThsSecMess} for a proof). As in the case where the hints are not secret and public, the guessing and the list version lead to similar results (cf.\ Remark~\ref{remark:guessListClose}). In the next two theorems $c$ is related to how much can be gleaned about $X$ from $M_{\textnormal p}$.

\begin{theorem}[Secret Hint Guessing-Version]\label{th:secMessGuess}
For every $c \in \naturals$ satisfying
\begin{equation}
c \leq | \setM_{\textnormal p} |, \label{eq:condCGuess}
\end{equation}
there is a $\{ 0,1 \}$-valued choice of the conditional PMF in \eqref{eq:aliceEncPMFSecMsg} for which Bob's ambiguity about $X$ is upper-bounded by
\begin{IEEEeqnarray}{l}
\mathscr A_{\textnormal B}^{(\textnormal g)} ( P_{X,Y} ) < 1 + 2^{\rho ( \renent {\tirho} {X | Y } - \log ( c \, | \setM_{\textnormal s} | ) + 1 )}, \label{eq:BobMomSecMsgDirGuess}
\end{IEEEeqnarray}
and Eve's ambiguity about $X$ is lower-bounded by
\begin{IEEEeqnarray}{l}
\mathscr A_{\textnormal E} ( P_{X,Y} ) \geq \bigl( 1 + \ln | \setX | \bigr)^{-\rho} 2^{\rho ( \renent {\tirho} { X | Y } - \log c )}. \label{eq:EveMomSecMsgDirGuess}
\end{IEEEeqnarray}
Conversely, for every conditional PMF, Bob's ambiguity is lower-bounded by
\begin{IEEEeqnarray}{l}
\mathscr A_{\textnormal B}^{(\textnormal g)} ( P_{X,Y} ) \geq \bigl( 1 + \ln | \setX | \bigr)^{-\rho} 2^{\rho ( \renent {\tirho} {X | Y } - \log ( | \setM_{\textnormal p} | \, | \setM_{\textnormal s} | ) )} \vee 1, \label{eq:BobMomSecMsgConvGuess}
\end{IEEEeqnarray}
and Eve's ambiguity is upper-bounded by
\begin{IEEEeqnarray}{l}
\mathscr A_{\textnormal E} ( P_{X,Y} ) \leq | \setM_{\textnormal s}|^\rho \mathscr A^{(\textnormal g)}_{\textnormal B} ( P_{X,Y} ) \wedge 2^{\rho \renent {\tirho} {X | Y }}. \label{eq:EveMomSecMsgConvGuess}
\end{IEEEeqnarray}
\end{theorem}

\begin{theorem}[Secret Hint List-Version]\label{th:secMess} If $| \setM_{\textnormal p} | \, | \setM_{\textnormal s} | > \log | \setX | + 2$, then for every $c \in \naturals$ satisfying
\begin{equation}
c \leq | \setM_{\textnormal p} |, \quad c \, | \setM_{\textnormal s} | > \log | \setX | + 2, \label{eq:condCList}
\end{equation}
there is a $\{ 0,1 \}$-valued choice of the conditional PMF in \eqref{eq:aliceEncPMFSecMsg} for which Bob's ambiguity about $X$ is upper-bounded by
\begin{IEEEeqnarray}{l}
\mathscr A_{\textnormal B}^{(\textnormal l)} ( P_{X,Y} ) < 1 + 2^{\rho ( \renent {\tirho} {X | Y } - \log ( c \, | \setM_{\textnormal s} | - \log | \setX | - 2 ) + 2 )}, \label{eq:BobMomSecMsgDir}
\end{IEEEeqnarray}
and Eve's ambiguity about $X$ is lower-bounded by
\begin{IEEEeqnarray}{l}
\mathscr A_{\textnormal E} ( P_{X,Y} ) \geq \bigl( 1 + \ln | \setX | \bigr)^{-\rho} 2^{\rho ( \renent {\tirho} { X | Y } - \log c )}. \label{eq:EveMomSecMsgDir}
\end{IEEEeqnarray}
Conversely, for every conditional PMF, Bob's ambiguity is lower-bounded by
\begin{IEEEeqnarray}{l}
\mathscr A_{\textnormal B}^{(\textnormal l)} ( P_{X,Y} ) \geq 2^{\rho ( \renent {\tirho} {X | Y } - \log ( | \setM_{\textnormal p} | \, | \setM_{\textnormal s} | )} \vee 1, \label{eq:BobMomSecMsgConv}
\end{IEEEeqnarray}
and Eve's ambiguity is upper-bounded by
\begin{IEEEeqnarray}{l}
\mathscr A_{\textnormal E} ( P_{X,Y} ) \leq | \setM_{\textnormal s} |^\rho \mathscr A^{(\textnormal l)}_{\textnormal B}  ( P_{X,Y} ) \wedge 2^{\rho \renent {\tirho} {X | Y }}. \label{eq:EveMomSecMsgConv}
\end{IEEEeqnarray}
\end{theorem}

We next contrast Theorems~\ref{th:secMessGuess} and \ref{th:secMess} to their counterparts in the previous scenario, i.e., to Theorems~\ref{th:distStorGuess} and \ref{th:distStor}. By comparing the respective upper and lower bounds on Eve's ambiguity, we see that $c$ and $|\setM_{\textnormal s}|$ in the current scenario, which relate to how much information can be gleaned about $X$ from $M_{\textnormal p}$ and $M_{\textnormal s}$, play the roles of $c_1 + c_2 \approx c_1 \vee c_2$ and $|\setM_1| \wedge |\setM_2|$ in the previous scenario, which relate to how much information can be gleaned about $X$ from the hint that---among $M_1$ and $M_2$---reveals more information about $X$ and the one that---among $M_1$ and $M_2$---reveals less information about $X$. This reflects the fact that in the current scenario Eve always sees $M_{\textnormal p}$, whereas in the previous scenario she sees the hint that reveals more information about $X$ and hence minimizes her ambiguity.

Unlike Theorems~\ref{th:distStorGuess} and \ref{th:distStor}, Theorems~\ref{th:secMessGuess} and \ref{th:secMess} imply that in the current scenario Alice can describe $X$ deterministically by choosing a $\{ 0,1 \}$-valued conditional PMF \eqref{eq:aliceEncPMFSecMsg}. To see why, recall that in the current scenario Eve sees only the public hint $M_{\textnormal p}$, and hence there is no need to encrypt information that can be gleaned from the secret hint $M_{\textnormal s}$. Consequently, Alice need not draw a one-time-pad like random variable and ensure that some information can be gleaned about $X$ from $(M_{\textnormal p}, M_{\textnormal s})$ but not from one hint alone. Instead, she can store that information on $M_{\textnormal s}$ without prior encryption.\\

We now proceed to Scenario~2, where Alice describes $X$ using only one
hint, but Alice and Bob share a secret key, which is unknown to
Eve. The secret key $K$ is drawn independently of the pair $(X,Y)$ and
uniformly over some finite set $\setK$. Upon observing $(X,Y) = (x,y)$
and $K = k$, Alice draws the hint $M$ from some finite set $\setM$
according to some conditional PMF
\begin{equation}
\distof{M = m | X = x, Y = y, K = k }. \label{eq:aliceEncPMFKey}
\end{equation}
Throughout, we assume that $| \setK | \leq | \setM |$. Bob sees the secret key and the hint. In the guessing version his ambiguity about $X$ is
\begin{IEEEeqnarray}{l}
\mathscr A^{(\textnormal g)}_{\textnormal B} ( P_{X,Y} ) = \min_{\guess {}{\cdot | Y, K, M}} \bigEx {}{\guess {}{ X | Y,K,M }^\rho} \label{eq:bobAmbiguityGuessingKey}
\end{IEEEeqnarray}
and in the list version
\begin{IEEEeqnarray}{l}
\mathscr A^{(\textnormal l)}_{\textnormal B} ( P_{X,Y} ) = \BigEx {}{\bigl| \setL^{Y,K}_M \bigr|^\rho}. \label{eq:bobAmbiguityListKey}
\end{IEEEeqnarray}
Eve sees sees only the hint. In both versions her ambiguity about $X$ is
\begin{IEEEeqnarray}{l}
\mathscr A_{\textnormal E} ( P_{X,Y} ) = \min_{\guess {}{\cdot | Y, M}} \bigEx {}{\guess {}{ X | Y, M }^\rho }. \label{eq:distEncSecrecyMeasureKey}
\end{IEEEeqnarray}

The next two theorems characterize the largest ambiguity that we can guarantee that Eve will have subject to a given upper bound on the ambiguity that Bob may have (see Appendix~\ref{app:pfThsKey} for a proof). Again, the guessing and the list version lead to similar results. Here $|\setK|$ is related to how much information can be gleaned about $X$ from $(K,M)$ but not from $M$ alone, i.e., to the ``encrypted'' information stored on $M$, and $c$ is related to how much information can be gleaned about $X$ from $M$, i.e., to the ``unencrypted'' information stored on $M$.

\begin{theorem}[Secret Key Guessing-Version]\label{th:keyGuess}
For every $c \in \naturals$ satisfying
\begin{IEEEeqnarray}{l}
c \, | \setK | \leq | \setM |, \label{eq:condCGuessKey}
\end{IEEEeqnarray}
there is a $\{ 0,1 \}$-valued choice of the conditional PMF in \eqref{eq:aliceEncPMFKey} for which Bob's ambiguity about $X$ is upper-bounded by
\begin{IEEEeqnarray}{l}
\mathscr A_{\textnormal B}^{(\textnormal g)} ( P_{X,Y} ) < 1 + 2^{\rho ( \renent {\tirho} { X | Y } - \log ( c \, | \setK | ) + 1 )}, \label{eq:BobMomKeyDirGuess}
\end{IEEEeqnarray}
and Eve's ambiguity about $X$ is lower-bounded by
\begin{IEEEeqnarray}{l}
\mathscr A_{\textnormal E} ( P_{X,Y} ) \geq \bigl( 1 + \ln | \setX | \bigr)^{-\rho} 2^{\rho ( \renent {\tirho} { X | Y } - \log c )}. \label{eq:EveMomKeyDirGuess}
\end{IEEEeqnarray} 
Conversely, for every conditional PMF, Bob's ambiguity is lower-bounded by
\begin{IEEEeqnarray}{l}
\mathscr A_{\textnormal B}^{(\textnormal g)} ( P_{X,Y} ) \geq \bigl( 1 + \ln | \setX | \bigr)^{-\rho} 2^{\rho ( \renent {\tirho} {X | Y } - \log | \setM | )} \vee 1, \label{eq:BobMomKeyConvGuess}
\end{IEEEeqnarray}
and Eve's ambiguity is upper-bounded by
\begin{IEEEeqnarray}{l}
\mathscr A_{\textnormal E} ( P_{X,Y} ) \leq | \setK |^\rho \mathscr A^{(\textnormal g)}_{\textnormal B} ( P_{X,Y} ) \wedge 2^{\rho \renent {\tirho} { X | Y }}. \label{eq:EveMomKeyConvGuess}
\end{IEEEeqnarray}
\end{theorem}

\begin{theorem}[Secret Key List-Version]\label{th:key} If $\bigl\lfloor | \setM | / | \setK | \bigr\rfloor | \setK | > \log | \setX | + 2$, then for every $c \in \naturals$ satisfying
\begin{IEEEeqnarray}{l}
c \, | \setK | \leq | \setM |, \quad c \, | \setK | > \log | \setX | + 2, \label{eq:condCListKey}
\end{IEEEeqnarray}
there is a $\{ 0,1 \}$-valued choice of the conditional PMF in \eqref{eq:aliceEncPMFKey} for which Bob's ambiguity about $X$ is upper-bounded by
\begin{IEEEeqnarray}{l}
\mathscr A_{\textnormal B}^{(\textnormal l)} ( P_{X,Y} ) < 1 + 2^{\rho ( \renent {\tirho} { X | Y } - \log ( c \, | \setK | - \log | \setX | - 2 ) + 2 )}, \label{eq:BobMomKeyDir}
\end{IEEEeqnarray}
and Eve's ambiguity about $X$ is lower-bounded by
\begin{IEEEeqnarray}{l}
\mathscr A_{\textnormal E} ( P_{X,Y} ) \geq \bigl( 1 + \ln | \setX | \bigr)^{-\rho} 2^{\rho ( \renent {\tirho} { X | Y } - \log c )}. \label{eq:EveMomKeyDir}
\end{IEEEeqnarray}
Conversely, for every conditional PMF, Bob's ambiguity is lower-bounded by
\begin{IEEEeqnarray}{l}
\mathscr A_{\textnormal B}^{(\textnormal l)} ( P_{X,Y} ) \geq 2^{\rho ( \renent {\tirho} {X | Y } - \log | \setM | )} \vee 1, \label{eq:BobMomKeyConv}
\end{IEEEeqnarray}
and Eve's ambiguity is upper-bounded by
\begin{IEEEeqnarray}{l}
\mathscr A_{\textnormal E} ( P_{X,Y} ) \leq | \setK |^\rho \mathscr A^{(\textnormal l)}_{\textnormal B} ( P_{X,Y} ) \wedge 2^{\rho \renent {\tirho} { X | Y }}. \label{eq:EveMomKeyConv}
\end{IEEEeqnarray}
\end{theorem}

Theorems~\ref{th:keyGuess} and \ref{th:key} are reminiscent of their counterparts for the scenario with a public and a secret hint, i.e., of Theorems~\ref{th:secMessGuess} and \ref{th:secMess}. The main difference is that in the current scenario $c$ and $|\setK|$, which relate to the ``unencrypted'' and the ``encrypted'' information stored on $M$, respectively, play the roles of $c$ and $|\setM_{\textnormal s}|$, which in the previous scenario relate to the information stored on the public and the secret hint, respectively. Like Theorems~\ref{th:secMessGuess} and \ref{th:secMess}, Theorems~\ref{th:keyGuess} and \ref{th:key} imply that in the current scenario Alice can describe $X$ deterministically by choosing a $\{ 0,1 \}$-valued conditional PMF \eqref{eq:aliceEncPMFKey}; there is no need for Alice to draw a one-time-pad like random variable, because she can use the secret key $K$ as a one-time-pad.

\section{Proofs}\label{sec:distStorProofs}

%This sections contains the proofs of the results in Section~\ref{sec:problemStatement}: Theorems~\ref{th:distStorGuess} and \ref{th:distStor} are proved in Section~\ref{sec:distStorProofsThmsFiniteBlockl}; Theorem~\ref{th:asympDistStor} in Section~\ref{sec:distStorProofsThmAsymp}; and Theorem~\ref{th:asympDistStorEB} in Section~\ref{sec:distStorProofsThmAsympEB}.

\subsection{A Proof of Theorems~\ref{th:distStorGuess} and \ref{th:distStor}} \label{sec:distStorProofsThmsFiniteBlockl}

We first establish the achievability results, i.e., \eqref{eq:BobMomDistStorDirGuess}--\eqref{eq:EveMomDistStorDirGuess} in the guessing version and \eqref{eq:BobMomDistStorDir}--\eqref{eq:EveMomDistStorDir} in the list version. To this end fix $( c_{\textnormal s},c_1,c_2 ) \in \naturals^3$ satisfying \eqref{bl:condCsC1C2Guess} in the guessing version and \eqref{bl:condCsC1C2} in the list version. For every $\nu \in \{ \textnormal s,1,2 \}$ let $V_\nu$ be a chance variable taking values in the set $\setV_\nu = \{ 0, \ldots, c_\nu - 1 \}$. Corollary~\ref{co:equivBunteResultGuessing} implies that there exists some $\{ 0,1 \}$-valued conditional PMF $\bigdistof { ( V_{\textnormal s}, V_1, V_2 ) = ( v_{\textnormal s}, v_1, v_2 ) \bigl| X = x, Y = y }$ for which
\begin{IEEEeqnarray}{l}
\min_{\guess {}{\cdot | Y, V_{\textnormal s}, V_1, V_2}} \bigEx {}{\guess {}{X|Y, V_{\textnormal s}, V_1, V_2}^\rho}  < 1 + 2^{\rho ( \renent {\tirho} {X | Y } - \log ( c_{\textnormal s} c_1 c_2 ) + 1 )}. \label{eq:BobMomDistStorDir1Guess}
\end{IEEEeqnarray}
Moreover, Theorem~\ref{th:optTaskEnc} implies that there exists some deterministic task-encoder $\enc {\cdot | Y } \colon \setX \rightarrow \setV_{\textnormal s} \times \setV_1 \times \setV_2$ for which
\begin{IEEEeqnarray}{l}
\BigEx {}{\bigl| \setL^Y_{V_s,V_1,V_2} \bigr|^\rho} < 1 + 2^{\rho ( \renent {\tirho} {X | Y} - \log ( c_s c_1 c_2 - \log \card \setX - 2 ) + 2 )}, \label{eq:BobMomDistStorDir1}
\end{IEEEeqnarray}
where $(V_{\textnormal s},V_1,V_2) = \enc{ X | Y }$. Both \eqref{bl:condCsC1C2Guess} and \eqref{bl:condCsC1C2} imply that $\card {\setM_1} \geq c_{\textnormal s} c_1$ and $\card {\setM_2} \geq c_{\textnormal s} c_2$. It thus suffices to prove \eqref{eq:BobMomDistStorDirGuess}--\eqref{eq:EveMomDistStorDirGuess} and \eqref{eq:BobMomDistStorDir}--\eqref{eq:EveMomDistStorDir} for a conditional PMF \eqref{eq:aliceEncPMF} that assigns positive probability only to $c_{\textnormal s} c_1$ elements of $\setM_1$ and $c_{\textnormal s} c_2$ elements of $\setM_2$. Therefore, we can assume w.l.g.\ that $\setM_1 = \setV_{\textnormal s} \times \setV_1$ and $\setM_2 = \setV_{\textnormal s} \times \setV_2$. That is, we can choose $M_1 = ( V_{\textnormal s} \oplus_{c_{\textnormal s}} \! U, V_1 )$ and $M_2 = ( U, V_2 )$, where $(V_{\textnormal s},V_1,V_2)$ is drawn according to one of the above conditional PMFs depending on the version, and where $U$ is independent of $( X, Y, V_{\textnormal s}, V_1, V_2 )$ and uniform over $\setV_{\textnormal s}$. Bob observes both hints and can thus recover $(V_{\textnormal s},V_1,V_2)$. Hence, in the guessing version \eqref{eq:BobMomDistStorDirGuess} follows from \eqref{eq:BobMomDistStorDir1Guess} and in the list version \eqref{eq:BobMomDistStorDir} follows from \eqref{eq:BobMomDistStorDir1}.

The proof of \eqref{eq:EveMomDistStorDirGuess} and \eqref{eq:EveMomDistStorDir} is more involved. It builds on the following two intermediate claims, which we prove next:
\begin{enumerate}
\item We can assume w.l.g.\ that Eve must guess not only $X$ but the pair $(X,U)$.
\item Given any pair of guessing functions $\guess 1 { \cdot,\cdot | Y, M_1 }$ and $\guess 2 { \cdot,\cdot | Y, M_2 }$ for $(X,U)$, there exist a chance variable $Z$ that takes values in a set of size at most $c_{\textnormal s} ( c_1 + c_2 )$ and a guessing function $\guess {}{\cdot,\cdot | Y,Z}$ for $(X,U)$ for which
\begin{IEEEeqnarray}{l}
\guess {}{X,U | Y,Z} = \guess 1 { X,U | Y, M_1 } \wedge \guess 2 { X,U | Y, M_2 }. \label{eq:distStorFact2Fact}
\end{IEEEeqnarray}
\end{enumerate}

We first prove the first intermediate claim. To this end note that in both versions (guessing and list) there exist some mappings $g_1 \colon \setX \times \setY \times \setM_1 \rightarrow \setV_s$ and $g_2 \colon \setX \times \setY \times \setM_2 \rightarrow \setV_s$ for which
\begin{equation} \label{eq:pfDistStorUFun}
U = g_1 ( X,Y,M_1 ) = g_2 ( X,Y,M_2 ).
\end{equation}
Given any guessing functions $\guess 1 {\cdot | Y,M_1 }$ and $\guess 2 {\cdot | Y,M_2 }$ for $X$, introduce some guessing functions $\guess 1 {\cdot, \cdot | Y,M_1 }$ and $\guess 2 {\cdot, \cdot | Y,M_2 }$ for $( X, U )$ satisfying, for every $(x,y) \in \setX \times \setY$, $m_1 \in \setM_1$, and $m_2 \in \setM_2$, that 
\begin{equation}
\bigguess k {x,g_k (x,y,m_k ) \bigl| y, m_k} = \guess k {x | y, m_k}, \,\, \forall \, k \in \{ 1,2 \}.
\end{equation}
From \eqref{eq:pfDistStorUFun} it follows that
\begin{equation}
\guess k { X, U | Y, M_k } = \guess k { X| Y, M_k }, \,\, \forall \, k \in \set{1,2} .\label{eq:distStorFact1}
\end{equation}
Consequently, Eve can guess $X$ and the pair $(X,U)$ with the same number of guesses. This proves the first intermediate claim.

We next prove the second intermediate claim. Given any pair of guessing functions $\guess 1 {\cdot, \cdot | Y,M_1 }$ and $\guess 2 {\cdot, \cdot | Y,M_2 }$ for $(X,U)$, define the triple of chance variables
\begin{IEEEeqnarray}{l} \label{eq:distStorTripIHatUHatV}
( I, \hat U, \hat V ) \triangleq \begin{cases} ( 1, V_{\textnormal s} \oplus_{c_{\textnormal s}} \! U, V_1 ) &\text{if} \, \guess 1 { X, U | Y, M_1 } \leq \guess 2 { X, U | Y, M_2 }, \\ ( 2,U,V_2 ) &\text{otherwise} \end{cases}
\end{IEEEeqnarray}
over the alphabet $\setI \times \setV_{\textnormal s} \times \hat \setV$, where $\setI = \{ 1,2 \}$ and $\hat \setV = \{0,1,\ldots, c_1  \vee c_2 - 1 \}$. Observing $(Y,I,\hat U, \hat V)$, Eve can guess $(X,U)$ using either $G_1$ or $G_2$ depending on the value of $I$. That is, Eve can guess $(X,U)$ using some guessing function $\guess {}{\cdot, \cdot | Y, I, \hat U, \hat V }$ satisfying, for every $y \in \setY$, $i \in \setI$, $\hat u \in \setV_s$, and $\hat v \in \{ 0,1,\ldots, c_i - 1 \}$, that 
\begin{IEEEeqnarray}{l}
\guess {}{\cdot,\cdot | y, i, \hat u, \hat v} = \bigguess i {\cdot,\cdot | y, ( \hat u, \hat v )}. \label{eq:distStorDefGuess}
\end{IEEEeqnarray}
By \eqref{eq:distStorTripIHatUHatV} the number of guesses that she needs to do so is given by
\begin{IEEEeqnarray}{l}
\guess {}{ X, U | Y, I, \hat U, \hat V } \nonumber \\
\quad = \bigguess { I }{ X,U | Y, ( \hat U, \hat V ) } \\
\quad = \guess I {X,U | Y, M_I} \\
\quad = \guess 1 {X,U | Y, M_1} \wedge \guess 2 {X,U | Y, M_2}. \label{eq:distStorFact2}
\end{IEEEeqnarray}
Consequently, \eqref{eq:distStorFact2Fact} holds when we set $Z = ( I, \hat U, \hat V )$. To conclude the proof of the second intermediate claim, note that the triple $( I, \hat U, \hat V )$ takes values in the set
\begin{IEEEeqnarray*}{l}
\bigl\{ ( 1, \hat u, \hat v ) \colon ( \hat u, \, \hat v ) \in \setV_{\textnormal s} \times \setV_1 \bigr\} \cup \bigl\{ ( 2, \hat u, \hat v ) \colon ( \hat u, \, \hat v ) \in \setV_{\textnormal s} \times \setV_2 \bigr\},
\end{IEEEeqnarray*}
whose cardinality is given by $$| \setV_{\textnormal s} \times \setV_1 | + | \setV_{\textnormal s} \times \setV_2 | = c_{\textnormal s} ( c_1 + c_2 ).$$

We are now ready to prove \eqref{eq:EveMomDistStorDirGuess} and \eqref{eq:EveMomDistStorDir}:
\begin{IEEEeqnarray}{l}
\bigEx {}{\guess 1 { X | Y, M_1 }^\rho \wedge \guess 2 { X | Y, M_2 }^\rho} \nonumber \\
\quad \stackrel{(a)}= \bigEx {}{\guess 1 { X, U | Y, M_1 }^\rho \wedge \guess 2 { X, U | Y, M_2 }^\rho} \\
\quad \stackrel{(b)}= \bigEx {}{\guess {}{X, U | Y, I, \hat U, \hat V}^\rho} \\
\quad \stackrel{(c)}\geq \bigl( 1 + \ln | \setX | \bigr)^{-\rho} 2^{\rho ( \renent {\tirho}{X,U|Y} - \log ( c_{\textnormal s} ( c_1 + c_2 ) ) )} \\
\quad \stackrel{(d)}= \bigl( 1 + \ln | \setX | \bigr)^{-\rho} 2^{\rho ( \renent {\tirho}{X|Y} - \log ( c_1 + c_2 ) )},
\end{IEEEeqnarray}
where $(a)$ holds by \eqref{eq:distStorFact1}; $(b)$ holds by \eqref{eq:distStorFact2}; $(c)$ follows from Corollary~\ref{co:equivBunteResultGuessing} and the fact that $( I, \hat U, \hat V )$ takes values in a set of size $c_{\textnormal s} ( c_1 + c_2 )$; and $(d)$ holds because
\begin{IEEEeqnarray}{l}
\renent {\tirho}{ X, U | Y } \nonumber \\
\quad = \frac{1}{\rho} \log \sum_{y \in \setY} \Biggl( \sum_{x \in \setX} \sum_{u \in \setV_{\textnormal s}} \bigr( P_{X,Y} ( x,y ) / | \setV_{\textnormal s} | \bigr)^{\tirho} \Biggr)^{\!\! 1+\rho} \label{eq:reEntrXUCondY} \\
\quad = \frac{1}{\rho} \log \! \left( \sum_{y \in \setY} \Biggl( \sum_{x \in \setX} P_{X,Y} ( x,y )^{\tirho} \Biggr)^{\!\! 1+\rho} | \setV_{\textnormal s} |^{\rho} \right) \\
\quad = \renent {\tirho} { X | Y } + \log c_{\textnormal s}.
\end{IEEEeqnarray}
The equality in~\eqref{eq:reEntrXUCondY} holds because $U$ is independent of $( X, Y )$ and uniform over the set $\setV_{\textnormal s}$ of size $| \setV_{\textnormal s} | = c_{\textnormal s}$. This concludes the proof of the achievability results.\\

It remains to establish the converse results, i.e., \eqref{eq:BobMomDistStorConvGuess}--\eqref{eq:EveMomDistStorConvGuessAlternative} in the guessing version and \eqref{eq:BobMomDistStorConv}--\eqref{eq:EveMomDistStorConvAlternative} in the list version. In the guessing version \eqref{eq:BobMomDistStorConvGuess} follows from Corollary~\ref{co:equivBunteResultGuessing}, and in the list version \eqref{eq:BobMomDistStorConv} follows from Theorem~\ref{th:optTaskEnc}. From \eqref{eq:eveAltAmb} we see that \eqref{eq:EveMomDistStorConvGuess} and \eqref{eq:EveMomDistStorConv} follow from \eqref{eq:EveMomDistStorConvGuessAlternative} and \eqref{eq:EveMomDistStorConvAlternative}, respectively, and hence it only remains to establish \eqref{eq:EveMomDistStorConvGuessAlternative} and \eqref{eq:EveMomDistStorConvAlternative}. By Corollary~\ref{co:impGuess}, it holds for every $k \in \{ 1,2 \}$ and $l \in \{ 1, 2 \} \setminus \{ k \}$ that
\begin{IEEEeqnarray}{l}
\min_{\guess {}{\cdot|Y, M_1, M_2}} \bigEx {}{\guess {}{X|Y,M_1,M_2}^\rho} \geq | \setM_l |^{-\rho} \min_{\guess k {\cdot | Y, M_k}} \bigEx {}{\guess k{X|Y,M_k}^\rho}. \label{eq:boundEveAltAmb}
\end{IEEEeqnarray}
Since $$\min_{\guess {}{\cdot|Y, M_1, M_2}} \bigEx {}{\guess {}{X|Y,M_1,M_2}^\rho} \leq \BigEx {}{ \bigl| \setL^Y_{M_1,M_2} \bigr|^\rho},$$ \eqref{eq:boundEveAltAmb} implies that in both versions the ambiguity $\tilde {\mathscr A_{\textnormal E}} ( P_{X,Y} )$ exceeds Bob's ambiguity by at most a factor of $\bigl( | \setM_1 | \wedge | \setM_2 | \bigr)^\rho$. That is, $\tilde {\mathscr A}_{\textnormal E} ( P_{X,Y} ) \leq \bigl( | \setM_1 | \wedge | \setM_2 | \bigr)^\rho \mathscr A^{(\textnormal g)}_{\textnormal B} ( P_{X,Y} )$ and $\tilde {\mathscr A}_{\textnormal E} ( P_{X,Y} ) \leq \bigl( | \setM_1 | \wedge | \setM_2 | \bigr)^\rho \mathscr A^{(\textnormal l)}_{\textnormal B} ( P_{X,Y} )$. Another upper bound on $\tilde {\mathscr A_{\textnormal E}} ( P_{X,Y} )$ is obtained by considering the case where Eve ignores the hint that she observes and guesses $X$ based on $Y$ alone. In this case it follows from Theorem~\ref{th:optGuessFun} that
\begin{equation}
\min_{\guess k {\cdot | Y, M_k}} \bigEx {}{\guess k{X|Y,M_k}^\rho} \leq 2^{\rho \renent {\tirho}{X|Y}}, \,\, \forall \, k \in \{ 1,2 \}. \label{eq:boundEveAltAmb2}
\end{equation}
From \eqref{eq:boundEveAltAmb2} we obtain that in both versions the ambiguity $\tilde {\mathscr A_{\textnormal E}} ( P_{X,Y} )$ cannot exceed $2^{\rho \renent {\tirho}{X|Y}}$, i.e., $\tilde {\mathscr A}_{\textnormal E} ( P_{X,Y} ) \leq 2^{\rho \renent {\tirho} { X | Y }}$. This concludes the proof of \eqref{eq:EveMomDistStorConvGuessAlternative} and \eqref{eq:EveMomDistStorConvAlternative} and consequently that of the converse results.

\subsection{A Proof of Theorem~\ref{th:asympDistStor}} \label{sec:distStorProofsThmAsymp}

If $R_1 + R_2 < \renent {\tirho}{ \rndvecX | \rndvecY }$, then \eqref{eq:BobMomDistStorConvGuess} in the guessing version and \eqref{eq:BobMomDistStorConv} in the list version imply that the privacy-exponent is negative infinity. We hence assume that $R_1 + R_2 > \renent {\tirho}{ \rndvecX | \rndvecY }$.

We first show that the privacy-exponent cannot exceed the RHS of \eqref{eq:secStrongConst}. To this end suppose that \eqref{eq:bobAmbTo1} holds and consequently
\begin{equation}
\limsup_{n \rightarrow \infty} \frac{\log \bigl( \mathscr A_{\textnormal B} ( P_{X^n,Y^n} ) \bigr)}{n} = 0.
\end{equation}
This, combined with \eqref{eq:EveMomDistStorConvGuess} in the guessing version and \eqref{eq:EveMomDistStorConv} in the list version, implies that
\begin{IEEEeqnarray}{l}
\limsup_{n \rightarrow \infty} \frac{\log \bigl( \mathscr A_{\textnormal E} ( P_{X^n,Y^n} ) \bigr)}{n} \leq \rho \bigl( R_1 \wedge R_2 \wedge \renent {\tirho}{ \rndvecX | \rndvecY } \bigr).
\end{IEEEeqnarray}
Hence, the privacy-exponent cannot exceed the RHS of \eqref{eq:secStrongConst}.

We next show that the privacy-exponent cannot be smaller than the RHS of \eqref{eq:secStrongConst}. By possibly relabeling the hints, we can assume w.l.g.\ that $R_2 = R_1 \wedge R_2$. Fix some $\epsilon > 0$ satisfying
%\begin{subequations}
\begin{equation} \label{eq:choiceEpsilon}
\epsilon \leq R_1 + R_2 - \renent {\tirho}{ \rndvecX | \rndvecY }.
\end{equation}
%and, if $0 < 2 R_2 - \renent {\tirho}{ \rndvecX | \rndvecY } \leq \renent {\tirho}{ \rndvecX | \rndvecY }$,
%\begin{equation}
%\epsilon < 2 R_2 - \renent {\tirho}{ \rndvecX | \rndvecY }.
%\end{equation}
%\end{subequations}
Choose a nonnegative rate-triple $( R_{\textnormal s}, \tilde R_1, \tilde R_2 ) \in ( \reals_0^+ )^3$ as follows:
\begin{enumerate}
\item If $R_2 \leq \renent {\tirho}{ \rndvecX | \rndvecY } / 2$, then choose
\begin{IEEEeqnarray}{l}
R_{\textnormal s} = 0, \quad \tilde R_1 = \renent {\tirho}{ \rndvecX | \rndvecY } - R_2 + \epsilon, \quad \tilde R_2 = R_2.
\end{IEEEeqnarray}
\item Else if $\renent {\tirho}{ \rndvecX | \rndvecY } / 2 < R_2 \leq \renent {\tirho}{ \rndvecX | \rndvecY }$, then choose
\begin{IEEEeqnarray}{l}
R_{\textnormal s} = 2 R_2 - \renent {\tirho}{ \rndvecX | \rndvecY } - \epsilon, \quad \tilde R_1 = \tilde R_2 = \renent {\tirho}{ \rndvecX | \rndvecY } - R_2 + \epsilon.
\end{IEEEeqnarray}
(To guarantee that $R_{\textnormal s} \geq 0$, we assume in this case that $\epsilon > 0$ is sufficiently small so that, in addition to \eqref{eq:choiceEpsilon}, also
\begin{equation}
\epsilon < 2 R_2 - \renent {\tirho}{ \rndvecX | \rndvecY }
\end{equation}
holds.)
\item Else if $\renent {\tirho}{ \rndvecX | \rndvecY } < R_2$, then choose
\begin{IEEEeqnarray}{l}
R_{\textnormal s} = R_2, \quad \tilde R_1 = \tilde R_2 = 0.
\end{IEEEeqnarray}
\end{enumerate}
Having chosen $(R_{\textnormal s}, \tilde R_1, \tilde R_2)$, choose the triple $(c_{\textnormal s},c_1,c_2) \in \naturals^3$ to be $( 2^{n R_{\textnormal s}}, 2^{n \tilde R_1}, 2^{n \tilde R_2} )$. For every sufficiently-large $n$, this choice implies \eqref{bl:condCsC1C2Guess} and \eqref{bl:condCsC1C2}, and by Theorems~\ref{th:distStorGuess} and Theorem~\ref{th:distStor} we can thus guarantee \eqref{eq:BobMomDistStorDirGuess}--\eqref{eq:EveMomDistStorDirGuess} in the guessing version and \eqref{eq:BobMomDistStorDir}--\eqref{eq:EveMomDistStorDir} in the list version. Note that
\begin{equation}
R_{\textnormal s} + \tilde R_1 + \tilde R_2 > \renent {\tirho}{ \rndvecX | \rndvecY }. \label{eq:privExpSumRateLargerRenEnt}
\end{equation}
Combining \eqref{eq:privExpSumRateLargerRenEnt} with \eqref{eq:BobMomDistStorDirGuess} in the guessing version and with \eqref{eq:BobMomDistStorDir} in the list version yields \eqref{eq:bobAmbTo1}. Moreover, combining \eqref{eq:privExpSumRateLargerRenEnt} with \eqref{eq:EveMomDistStorDirGuess} in the guessing version and with \eqref{eq:EveMomDistStorDir} in the list version implies that
\begin{IEEEeqnarray}{rCl}
\liminf_{n \rightarrow \infty} \frac{\log \bigl( \mathscr A_{\textnormal E} ( P_{X^n,Y^n} ) \bigr)}{n} & \geq & \rho \bigl( \renent {\tirho}{ \rndvecX | \rndvecY } - ( \tilde R_1 \vee \tilde R_2 ) \bigr) \\
& \geq & \rho \bigl( \left( R_1 \wedge R_2 - \epsilon \right) \wedge \renent {\tirho}{ \rndvecX | \rndvecY } \bigr).
\end{IEEEeqnarray}
Letting $\epsilon$ tend to zero proves that the privacy-exponent cannot be smaller than the RHS of \eqref{eq:secStrongConst}.

\subsection{A Proof of Theorem~\ref{th:asympDistStorEB}} \label{sec:distStorProofsThmAsympEB}

If $R_1 + R_2 < \renent {\tirho}{ \rndvecX | \rndvecY } - \rho^{-1} E_{\textnormal B}$, then \eqref{eq:BobMomDistStorConvGuess} in the guessing version and \eqref{eq:BobMomDistStorConv} in the list version imply that the modest privacy-exponent is negative infinity. We hence assume that $R_1 + R_2 \geq \renent {\tirho}{ \rndvecX | \rndvecY } - \rho^{-1} E_{\textnormal B}$.

We first show that the modest privacy-exponent cannot exceed the RHS of \eqref{eq:secModConst}. To this end suppose that \eqref{eq:bobAmbToEB} holds. This, combined with \eqref{eq:EveMomDistStorConvGuess} in the guessing version and \eqref{eq:EveMomDistStorConv} in the list version, implies that
\begin{IEEEeqnarray}{l}
\limsup_{n \rightarrow \infty} \frac{\log \bigl( \mathscr A_{\textnormal E} ( P_{X^n,Y^n} ) \bigr)}{n} \leq \bigl( \rho ( R_1 \wedge R_2 ) + E_{\textnormal B} \bigr) \wedge \rho \renent {\tirho}{ \rndvecX | \rndvecY }.
\end{IEEEeqnarray}
Hence, the modest privacy-exponent cannot exceed the RHS of \eqref{eq:secModConst}.

We next show that the modest privacy-exponent cannot be smaller than the RHS of \eqref{eq:secModConst}. By possibly relabeling the hints, we can assume w.l.g.\ that $R_2 = R_1 \wedge R_2$. Choose a nonnegative rate-triple $( R_{\textnormal s}, \tilde R_1, \tilde R_2 ) \in ( \reals_0^+ )^3$ as follows:
\begin{enumerate}
\item If $R_2 \leq \bigl( \renent {\tirho}{ \rndvecX | \rndvecY } - \rho^{-1} E_{\textnormal B} \bigr) / 2$, then choose
\begin{IEEEeqnarray}{l}
R_{\textnormal s} = 0, \quad \tilde R_1 = \renent {\tirho}{ \rndvecX | \rndvecY } - \rho^{-1} E_{\textnormal B} - R_2, \quad \tilde R_2 = R_2.
\end{IEEEeqnarray}
\item Else if $\bigl( \renent {\tirho}{ \rndvecX | \rndvecY } - \rho^{-1} E_{\textnormal B} \bigr) / 2 < R_2 \leq \renent {\tirho}{ \rndvecX | \rndvecY } - \rho^{-1} E_{\textnormal B}$, then choose
\begin{IEEEeqnarray}{l}
R_{\textnormal s} = 2 R_2 - \renent {\tirho}{ \rndvecX | \rndvecY } + \rho^{-1} E_{\textnormal B}, \quad \tilde R_1 = \tilde R_2 = \renent {\tirho}{ \rndvecX | \rndvecY } - \rho^{-1} E_{\text B} - R_2.
\end{IEEEeqnarray}
\item Else if $\renent {\tirho}{ \rndvecX | \rndvecY } - \rho^{-1} E_{\textnormal B} < R_2$, then choose
\begin{IEEEeqnarray}{l}
R_{\textnormal s} = R_2, \quad \tilde R_1 = \tilde R_2 = 0.
\end{IEEEeqnarray}
\end{enumerate}
Having chosen $(R_{\textnormal s}, \tilde R_1, \tilde R_2)$, choose the triple $(c_{\textnormal s},c_1,c_2) \in \naturals^3$ to be $( 2^{n R_{\textnormal s}}, 2^{n \tilde R_1}, 2^{n \tilde R_2} )$. For every sufficiently-large $n$, this choice implies \eqref{bl:condCsC1C2Guess} and \eqref{bl:condCsC1C2}, and by Theorems~\ref{th:distStorGuess} and Theorem~\ref{th:distStor} we can thus guarantee \eqref{eq:BobMomDistStorDirGuess}--\eqref{eq:EveMomDistStorDirGuess} in the guessing version and \eqref{eq:BobMomDistStorDir}--\eqref{eq:EveMomDistStorDir} in the list version. Note that
\begin{equation}
R_{\textnormal s} + \tilde R_1 + \tilde R_2 \geq \renent {\tirho}{ \rndvecX | \rndvecY } - \rho^{-1} E_{\textnormal B}. \label{eq:privExpSumRateLargerRenEntEB}
\end{equation}
Combining \eqref{eq:privExpSumRateLargerRenEntEB} with \eqref{eq:BobMomDistStorDirGuess} in the guessing version and with \eqref{eq:BobMomDistStorDir} in the list version yields \eqref{eq:bobAmbToEB}. Moreover, combining \eqref{eq:privExpSumRateLargerRenEntEB} with \eqref{eq:EveMomDistStorDirGuess} in the guessing version and with \eqref{eq:EveMomDistStorDir} in the list version implies that
\begin{IEEEeqnarray}{rCl}
\liminf_{n \rightarrow \infty} \frac{\log \bigl( \mathscr A_{\textnormal E} ( P_{X^n,Y^n} ) \bigr)}{n} & \geq & \rho \bigl( \renent {\tirho}{ \rndvecX | \rndvecY } - ( \tilde R_1 \vee \tilde R_2 ) \bigr) \\
& \geq & \bigl( \rho ( R_1 \wedge R_2 ) + E_{\textnormal B} \bigr) \wedge \rho \renent {\tirho}{ \rndvecX | \rndvecY }.
\end{IEEEeqnarray}
Consequently, the modest privacy-exponent cannot be smaller than the RHS of \eqref{eq:secModConst}, which concludes the proof.

\section{Resilience against Disk Failures} \label{sec:extDiskFailures}

In this section we generalize the model of Section~\ref{sec:problemStatement} to allow for Alice to produce $\delta$ hints (not necessarily two) and store them on different disks, for Bob to see $\nu \leq \delta$ (not necessarily 2) of those hints, and for Eve to see $\eta < \nu$ (not necessarily one) of the hints. We assume that, after observing $X$ and $Y$, an adversarial ``genie'' reveals to Bob the $\nu$ hints that maximize his ambiguity and to Eve the $\eta$ hints that minimize her ambiguity. The former guarantees that the system be robust against $\delta - \nu$ disk failures, no matter which disks fail; and the latter guarantees that Eve's ambiguity be ``large'' no matter which $\eta$ hints she sees. We allow the genie to observe $(X,Y)$, because, as we have seen, not allowing the genie to observe $(X,Y)$ would lead to a weaker form of secrecy (see Example~\ref{ex:notionOfSecrecy}).

The current network can be described as follows. As in Section~\ref{sec:problemStatement}, we consider two problems, the ``guessing version'' and the ``list version,'' which differ in the definition of Bob's ambiguity. Upon observing $( X,Y ) = ( x,y )$, Alice draws the $\delta$-tuple $\rndvecM = ( M_1, \ldots, M_\delta )$ from the finite set $\mathbb F_{2^s}^\delta$ according to some conditional PMF
\begin{IEEEeqnarray}{l}
\distof{ \rndvecM = \vecm | X = x, Y = y }, \quad \vecm \in \mathbb F_{2^s}^\delta. \label{eq:aliceEncPMFDiscFail}
\end{IEEEeqnarray}
We assume here that each hint comprises $s$ bits (i.e., that $\rndvecM$ takes values in $\mathbb F_{2^s}^\delta$); why this assumption is reasonable will be explained shorty (see Theorem~\ref{th:SiEqSBest} and Remark~\ref{re:hintSize} ahead). Bob gets to see a size-$\nu$ set $\setB \subseteq \{ 1, \ldots, \delta \}$, the components $\rndvecM_{\setB}$ of $\rndvecM$ indexed by $\setB$, and the side information $Y$. As already mentioned, the index set $\setB$ is chosen by an adversary of his. In the guessing version Bob guesses $X$ using an optimal guessing function $G_{\text{\setB}} (\cdot | Y,\rndvecM_{\setB} )$, which minimizes the $\rho$-th moment of the number of guesses that he needs. (As indicated by the subscript, the guessing function $\guess {\setB}{\cdot|Y, \rndvecM_{\setB}}$ can depend on $\setB$.) His min-max ambiguity about $X$ is thus given by
\begin{IEEEeqnarray}{l}
\mathscr A^{(\textnormal g)}_{\textnormal B} ( P_{X,Y} ) = \min_{\guess {\setB}{\cdot|Y, \rndvecM_{\setB}}} \BigEx {}{\max_{\setB} \guess {\setB}{X | Y,\rndvecM_{\setB} }^\rho}. \label{eq:bobAmbiguityGuessingDiscFail}
\end{IEEEeqnarray}
In the list version Bob's ambiguity about $X$ is
\begin{IEEEeqnarray}{l} 
\mathscr A^{(\textnormal l)}_{\textnormal B} ( P_{X,Y} ) = \BigEx {}{\max_{ \setB } \, \bigl| \setL^Y_{\rndvecM_{\setB}} \bigr|^\rho}, \label{eq:bobAmbiguityListDiscFail}
\end{IEEEeqnarray}
where for all $y \in \setY$ and $\vecm_{\setB} \in \mathbb F_{2^s}^\delta$
\begin{IEEEeqnarray}{l} 
\setL^y_{\vecm_{\setB}} =  \bigl\{ x \colon \distof{ X = x | Y = y, \rndvecM_{\setB} = \vecm_{\setB} } > 0 \bigr\}
\end{IEEEeqnarray}
is the list of all the realizations of $X$ of positive posterior probability
\begin{IEEEeqnarray}{l} 
\distof{ X = x | Y = y, \rndvecM_{\setB} = \vecm_{\setB} } \nonumber \\
\quad = \frac{P_{X,Y} ( x,y ) \, \distof{ \rndvecM_{\setB} = \vecm_{\setB} | X = x, Y = y }}{\sum_{\tilde x} P_{X,Y} ( \tilde x,y ) \, \distof{\rndvecM_{\setB} = \vecm_{\setB} | X = \tilde x, Y = y }}.
\end{IEEEeqnarray}
Note that for $\setB^c \triangleq \{ 1, \ldots, \delta \} \setminus \setB$ we have
$$\distof{ \rndvecM_{\setB} = \vecm_{\setB} | X = x, Y = y } = \sum_{\vecm_{\setB^c}} \distof{ \rndvecM = \vecm | X = x, Y = y }.$$ Eve observes a size-$\eta$ set $\setE \subseteq \{ 1, \ldots, \delta \}$, the components $\rndvecM_{\setE}$ of $\rndvecM$ indexed by $\setE$, and the side information $Y$. The index set $\setE$ is chosen by an accomplice of hers. Eve guesses $X$ using an optimal guessing function $\guess {\setE}{\cdot | X, \rndvecM_{\setE}}$, which minimizes the $\rho$-th moment of the number of guesses that she needs. (The guessing function $\guess {\setE}{\cdot | X, \rndvecM_{\setE}}$ can depend on $\setE$.) In both versions her ambiguity about $X$ is thus given by
\begin{IEEEeqnarray}{l} 
\mathscr A_{\textnormal E} ( P_{X,Y} ) = \min_{\guess {\setE}{\cdot | X, \rndvecM_{\setE}}} \BigEx {}{ \min_{ \setE } \guess {\setE} { X | Y, \rndvecM_{\setE} }^\rho }. \label{eq:distEncSecrecyMeasureDiscFail}
\end{IEEEeqnarray}
Optimizing over Alice's choice of the conditional PMF in \eqref{eq:aliceEncPMFDiscFail}, we wish to characterize the largest ambiguity that we can guarantee that Eve will have subject to a given upper bound on the ambiguity that Bob may have.\\

Of special interest to us is the asymptotic regime where $(X,Y)$ is an $n$-tuple (not necessarily drawn IID), and where each hint stores $$s = n R_s$$ bits, where $R_s$ is nonnegative and corresponds to the per-hint storage-rate. (We assume that $\delta$, $\nu$, and $\eta$ are fixed.) For both versions of the problem, we shall characterize the largest exponential growth that we can guarantee for Eve's ambiguity subject to the constraint that Bob's ambiguity tend to one, i.e., we shall characterize the privacy-exponent $\overbar{E_{\textnormal E}}$ defined in Definition~\ref{de:EvesAmbig}. In addition, we shall also characterize the largest exponential growth that we can guarantee for Eve's ambiguity in case Bob's ambiguity is allowed to grow exponentially with a given normalized (by $n$) exponent $E_{\textnormal B} \geq 0$, i.e., we shall characterize the modest privacy-exponent $\overbar{E_{\text E}^{\textnormal m} ( E_{\textnormal B} )}$ defined in Definition~\ref{de:EvesAmbigModest}. As for the model of Section~\ref{sec:problemStatement}, the privacy-exponent and the modest privacy-exponent turn out not to depend on the version of the problem, and in the asymptotic analysis $\mathscr A_{\textnormal B}$ can thus stand for either $\mathscr A_{\textnormal B}^{(\textnormal g)}$ or $\mathscr A_{\textnormal B}^{(\textnormal l)}$.\\

%We next present our results to the stated problems in the finite-blocklength regime (Section~\ref{sec:statDistStorFinBl}) and in the asymptotic regime (Section~\ref{sec:statDistStorAsymp}).

\subsection{Finite-Blocklength Results} \label{sec:statDistStorFinBl}

In the next two theorems $(\nu - \eta) r$ should be viewed as the number of information-bits that can be gleaned about $X$ from $\nu$ but not from $\eta$ hints. Moreover, for every $\gamma \in \{ \eta, \nu \}$, $\gamma p$ should be viewed as the number of information-bits that any $\gamma$ hints reveal about $X$. By adapting the proof of Theorems~\ref{th:discFailGuess} and \ref{th:discFail} to the case at hand (see Appendix~\ref{app:pfThsDiscFail}), we obtain the following results:

\begin{theorem}[Finite-Blocklength Guessing-Version] \label{th:discFailGuess}
For every pair $(p, r) \in \{ 0, \ldots, s \}^2$ satisfying
\begin{subequations} \label{bl:condGuessDiscFail}
\begin{IEEEeqnarray}{l}
p + r = s, \\
p, \, r \in \{ 0 \} \cup \bigl\{ \lceil \log \delta \rceil, \lceil \log \delta \rceil + 1, \ldots \bigr\},
\end{IEEEeqnarray}
\end{subequations}
there is a choice of the conditional PMF in \eqref{eq:aliceEncPMFDiscFail} for which Bob's ambiguity about $X$ is upper-bounded by
\begin{IEEEeqnarray}{l}
\mathscr A_{\textnormal B}^{(\textnormal g)} ( P_{X,Y} ) < 1 + 2^{\rho ( \renent {\tirho} {X | Y } - \nu s + \eta r + 1 )}, \label{eq:BobMomDiscFailDirGuess}
\end{IEEEeqnarray}
and Eve's ambiguity about $X$ is lower-bounded by
\begin{IEEEeqnarray}{l}
\mathscr A_{\textnormal E} ( P_{X,Y} ) \geq 2^{\rho ( \renent {\tirho} { X | Y } - \eta ( s - r ) - \eta \log \delta - \log ( 1 + \ln | \setX | ) )}. \label{eq:EveMomDiscFailDirGuess}
\end{IEEEeqnarray}
Conversely, for every conditional PMF, Bob's ambiguity is lower-bounded by
\begin{IEEEeqnarray}{l}
\mathscr A_{\textnormal B}^{(\textnormal g)} ( P_{X,Y} ) \geq 2^{\rho ( \renent {\tirho} {X | Y } - \nu s - \log (1 + \ln | \setX |) )} \vee 1, \label{eq:BobMomDiscFailConvGuess}
\end{IEEEeqnarray}
and Eve's ambiguity is upper-bounded by
\begin{IEEEeqnarray}{l}
\mathscr A_{\textnormal E} ( P_{X,Y} ) \leq 2^{\rho (\nu - \eta) s} \mathscr A^{(\textnormal g)}_{\textnormal B} ( P_{X,Y} ) \wedge 2^{\rho \renent {\tirho} { X | Y }}. \label{eq:EveMomDiscFailConvGuess}
\end{IEEEeqnarray}
\end{theorem}

\begin{proof}
See Appendix~\ref{sec:pfThsDiscFail}.
\end{proof}

\begin{theorem}[Finite-Blocklength List-Version] \label{th:discFail}
If $2^{\nu s} > \log | \setX | + 2$, then for every pair $( p, r ) \in \{ 0, \ldots, s \}$ satisfying
\begin{subequations} \label{bl:condListDiscFail}
\begin{IEEEeqnarray}{l}
p + r = s, \\
p, \, r \in \{ 0 \} \cup \bigl\{ \lceil \log \delta \rceil, \lceil \log \delta \rceil + 1, \ldots \bigr\}, \\
2^{\nu s - \eta r} > \log | \setX | + 2,
\end{IEEEeqnarray}
\end{subequations}
there is a choice of the conditional PMF in \eqref{eq:aliceEncPMFDiscFail} for which Bob's ambiguity about $X$ is upper-bounded by
\begin{IEEEeqnarray}{l}
\mathscr A_{\textnormal B}^{(\textnormal l)} ( P_{X,Y} ) < 1 +  2^{\rho ( \renent {\tirho} { X | Y } - \log ( 2^{\nu s - \eta r} - \log | \setX | - 2 ) + 2 )}, \label{eq:BobMomDiscFailDir}
\end{IEEEeqnarray}
and Eve's ambiguity about $X$ is lower-bounded by
\begin{IEEEeqnarray}{l}
\mathscr A_{\textnormal E} ( P_{X,Y} ) \geq 2^{\rho ( \renent {\tirho} { X | Y } - \eta ( s - r ) - \eta \log \delta - \log ( 1 + \ln | \setX |) )}. \label{eq:EveMomDiscFailDir}
\end{IEEEeqnarray}
Conversely, for every conditional PMF, Bob's ambiguity is lower-bounded by
\begin{IEEEeqnarray}{l}
\mathscr A_{\textnormal B}^{(\textnormal l)} ( P_{X,Y} ) \geq 2^{\rho ( \renent {\tirho} { X | Y } - \nu s )} \vee 1, \label{eq:BobMomDiscFailConv}
\end{IEEEeqnarray}
and Eve's ambiguity is upper-bounded by
\begin{IEEEeqnarray}{l}
\mathscr A_{\textnormal E} ( P_{X,Y} ) \leq 2^{\rho (\nu - \eta) s} \mathscr A^{(\textnormal l)}_{\textnormal B} ( P_{X,Y} ) \wedge 2^{\rho \renent {\tirho} { X | Y }}. \label{eq:EveMomDiscFailConv}
\end{IEEEeqnarray}
\end{theorem}

\begin{proof}
See Appendix~\ref{sec:pfThsDiscFail}.
\end{proof}

The bounds in Theorems~\ref{th:discFailGuess} and \ref{th:discFail} are tight in the sense that, with a judicious choice of $p$ and $r$, the achievability results (namely \eqref{eq:BobMomDiscFailDirGuess}--\eqref{eq:EveMomDiscFailDirGuess} in the guessing version and \eqref{eq:BobMomDiscFailDir}--\eqref{eq:EveMomDiscFailDir} in the list version) match the corresponding converse results (namely \eqref{eq:BobMomDiscFailConvGuess}--\eqref{eq:EveMomDiscFailConvGuess} in the guessing version and \eqref{eq:BobMomDiscFailConv}--\eqref{eq:EveMomDiscFailConv} in the list version) up to polynomial factors of $\delta^{\eta}$ and of $\ln | \setX |$. This can be seen from the following corollary to Theorems~\ref{th:discFailGuess} and \ref{th:discFail}, which states the achievability results in a simplified and more accessible form:

\begin{corollary}[Simplified Finite-Blocklength Achievability-Results] \label{co:discFailSimp}
In the guessing version, for any constant $\mathscr U_{\textnormal B}$ satisfying
\begin{IEEEeqnarray}{l}
\mathscr U_{\textnormal B} \geq 1 + 2^{\rho ( \renent {\tirho} {X | Y } - \nu s + 1 )}, \label{eq:discFailSimpUBGuess}
\end{IEEEeqnarray}
there is a choice of the conditional PMF in \eqref{eq:aliceEncPMFDiscFail} for which Bob's ambiguity about $X$ is upper-bounded by
\begin{IEEEeqnarray}{l}
\mathscr A_{\textnormal B}^{(\textnormal g)} ( P_{X,Y} ) < \mathscr U_{\textnormal B}, \label{eq:discFailSimpBobGuess}
\end{IEEEeqnarray}
and Eve's ambiguity about $X$ is lower-bounded by
\begin{IEEEeqnarray}{l}
\mathscr A_{\textnormal E} ( P_{X,Y} ) \geq \bigl( \delta^{\eta} ( 1 + \ln |\setX| ) \bigr)^{-\rho} \Bigl( \bigl( (2 \delta)^{-\rho \eta} 2^{\rho (\nu - \eta) s} (\mathscr U_{\textnormal B} - 1) \bigr) \wedge 2^{\rho \renent {\tirho} {X | Y }} \Bigr). \label{eq:discFailSimpEveGuess}
\end{IEEEeqnarray}

In the list version, for any constant $\mathscr U_{\textnormal B}$ satisfying
\begin{IEEEeqnarray}{l}
\mathscr U_{\textnormal B} \geq 1 + 2^{\rho ( \renent {\tirho} {X | Y } - \log ( 2^{\nu s} - \log | \setX | - 2 ) + 2 )}, \label{eq:discFailSimpUB}
\end{IEEEeqnarray}
there is a choice of the conditional PMF in \eqref{eq:aliceEncPMFDiscFail} for which Bob's ambiguity about $X$ is upper-bounded by
\ba 
\mathscr A_{\textnormal B}^{(\textnormal l)} ( P_{X,Y} ) < \mathscr U_{\textnormal B}, \label{eq:BobAmbSimpBList}
\ea
and Eve's ambiguity about $X$ is lower-bounded by
\begin{IEEEeqnarray}{rcl}
\!\!\!\!\!\!\!\! \mathscr A_{\textnormal E} ( P_{X,Y} ) \geq \bigl( \delta^{\eta} ( 1 + \ln |\setX| ) \bigr)^{-\rho} & \Biggl( & \bigl( 2^{-3 \rho} (2 \delta)^{-\rho \eta} 2^{\rho (\nu - \eta) s} (\mathscr U_{\textnormal B} - 1) \bigr) \nonumber \\
& & \wedge 2^{\rho \renent {\tirho} {X | Y }} \nonumber \\
& & \wedge \biggl( \Bigl( 2 (2 \delta)^{\eta} \bigl( 2 + \log |\setX| \bigr) \Bigr)^{-\rho} 2^{\rho ( (\nu - \eta) s + \renent {\tirho} { X | Y } ) } \biggr) \Biggr). \label{eq:EveAmbSimpBList}
\end{IEEEeqnarray}
\end{corollary}

\begin{proof}
The result is a corollary to Theorems~\ref{th:discFailGuess} and \ref{th:discFail}. See Appendix~\ref{app:pfCoDiscFail} for a detailed proof.
\end{proof}

We conclude this section by explaining why it is a good idea to store an equal number of bits on each disk. This can be seen from the next theorem:

\begin{theorem}[Converse Results: Disk~$\ell$ stores $s_\ell$ Bits] \label{th:SiEqSBest}
Suppose that for every $\ell \in \{ 1, \ldots, \delta \}$ Disk~$\ell$ stores $s_\ell$ bits, where $s_1 \leq \ldots \leq s_\delta$. For every conditional PMF in \eqref{eq:aliceEncPMFDiscFail}, Bob's ambiguity about $X$ is---depending on the version of the problem---lower-bounded by
\begin{subequations} \label{bl:BobMomDistStorConvSi}
\begin{IEEEeqnarray}{rCl}
\mathscr A_{\textnormal B}^{(\textnormal g)} ( P_{X,Y} ) & \geq & 2^{\rho ( \renent {\tirho} {X | Y } - \sum_{\ell = 1}^\nu s_\ell - \log (1 + \ln | \setX |) )} \vee 1, \label{eq:BobMomDistStorConvSiGuess} \\
\mathscr A_{\textnormal B}^{(\textnormal l)} ( P_{X,Y} ) & \geq & 2^{\rho ( \renent {\tirho} {X | Y } - \sum_{\ell = 1}^\nu s_\ell )} \vee 1, \label{eq:BobMomDistStorSiConv}
\end{IEEEeqnarray}
\end{subequations}
and Eve's ambiguity about $X$ is upper-bounded by
\begin{subequations} \label{bl:EveMomDistStorConvSi}
\begin{IEEEeqnarray}{rCl}
\mathscr A_{\textnormal E} ( P_{X,Y} ) & \leq & 2^{\rho \sum_{\ell = 1}^{\nu - \eta} s_\ell} \mathscr A^{(\textnormal g)}_{\textnormal B} ( P_{X,Y} ) \wedge 2^{\rho \renent {\tirho} {X | Y }}, \label{eq:EveMomDistStorConvSiGuess} \\
\mathscr A_{\textnormal E} ( P_{X,Y} ) & \leq & 2^{\rho \sum_{\ell = 1}^{\nu - \eta} s_\ell} \mathscr A^{(\textnormal l)}_{\textnormal B} ( P_{X,Y} ) \wedge 2^{\rho \renent {\tirho} {X | Y }}. \label{eq:EveMomDistStorSiConv}
\end{IEEEeqnarray}
\end{subequations}
\end{theorem}

\begin{proof}
See Appendix~\ref{app:pfThSiEqSBest}.
\end{proof} 

\begin{remark}[\textbf{Why Store $s$ Bits on Each Disk?}]\label{re:hintSize}
{\em Compare a scenario where for every $\ell \in \{ 1, \ldots, \delta \}$ Disk~$\ell$ stores $s_\ell$ bits, where $s_1 \leq \ldots \leq s_\delta$, with a scenario where each disk stores $\bigl\lfloor ( s_1 + \ldots + s_\delta ) / \delta \bigr\rfloor$ bits. Based on Theorem~\ref{th:SiEqSBest} and Corollary~\ref{co:discFailSimp}, neglecting polynomial factors of $\delta^\eta$ and of $\ln |\setX|$, every pair of ambiguities for Bob and Eve that is achievable in the former scenario is also achievable in the latter scenario.}
\end{remark}

\subsection{Asymptotic Results} \label{sec:statDistStorAsymp}

Suppose now that $(X,Y)$ is an $n$-tuple. We study the asymptotic regime where $n$ tends to infinity. Recall that in this regime we refer to both $\mathscr A_{\textnormal B}^{(\textnormal g)}$ and $\mathscr A_{\textnormal B}^{(\textnormal l)}$ by $\mathscr A_{\textnormal B}$, because the results are the same for both versions. As we prove in Appendix~\ref{app:pfThAsympDiscFail}, Theorems~\ref{th:discFailGuess} and \ref{th:discFail} and Corollary~\ref{co:discFailSimp} imply the following asymptotic result:

\begin{theorem}[Privacy-Exponent and Modest Privacy-Exponent] \label{th:asympDiscFail}
Let $\bigl\{ (X_i,Y_i) \bigr\}_{i \in \naturals}$ be a discrete-time stochastic process with finite alphabet $\setX \times \setY$, and suppose its conditional R\'enyi entropy-rate $\renent {\tirho}{ \rndvecX | \rndvecY }$ is well-defined. Given any nonnegative rate $R_s$, the privacy-exponent is
\begin{IEEEeqnarray}{l}
\overbar{ E_{\textnormal E} } = \begin{cases} \rho \bigl( R_s ( \nu - \eta ) \wedge \renent {\tirho}{ \rndvecX | \rndvecY } \bigr) & \nu R_s > \renent {\tirho}{ \rndvecX | \rndvecY }, \\ - \infty & \nu R_s < \renent {\tirho}{ \rndvecX | \rndvecY }, \end{cases} \label{eq:privExpDiscFail}
\end{IEEEeqnarray}
and the modest privacy-exponent for $E_{\textnormal B} \geq 0$ is
\begin{IEEEeqnarray}{l}
\overbar{E^{\textnormal m}_{\textnormal E} (E_{\textnormal B})} = \begin{cases} \bigl( \rho R_s ( \nu - \eta ) + E_{\textnormal B} \bigr) \wedge \rho \renent {\tirho}{ \rndvecX | \rndvecY } & \nu R_s \geq \renent {\tirho}{ \rndvecX | \rndvecY } - \rho^{-1} E_{\textnormal B}, \\ - \infty & \nu R_s < \renent {\tirho}{ \rndvecX | \rndvecY } - \rho^{-1} E_{\textnormal B}. \end{cases} \label{eq:privExpDiscFailEB}
\end{IEEEeqnarray}
\end{theorem}

By \eqref{eq:privExpDiscFail} we can achieve the maximum privacy-exponent $\rho \renent {\tirho}{ \rndvecX | \rndvecY }$ if the per-hint storage-rate satisfies $$R_s \geq \renent {\tirho}{ \rndvecX | \rndvecY } / ( \nu - \eta ),$$ where $\renent {\tirho}{ \rndvecX | \rndvecY }$ is the minimum rate that is necessary to describe the source for Bob. This agrees with the well-known result that the optimal share-size to share a $k$-bit secret so that any $\nu$ shares reveal $X$ and any $\eta$ shares provide no information about $X$ is $k /( \nu - \eta )$ (see, e.g., \cite{subramanianmclaughlin09}).

\section{Coding and Encryption under a Fidelity Criterion}\label{sec:rateDist}

In this section we study a rate-distortion version of the model of Section~\ref{sec:problemStatement}, where reconstructions are lossy but subject to a given fidelity criterion. We only treat the asymptotic regime where $(X,Y)$ is an $n$-tuple, and we shall assume that the $n$-tuple is drawn IID. Throughout this section, $\bigl\{ ( X_i,Y_i ) \bigr\}_{i \in \naturals}$ is thus a discrete-time stochastic process of IID pairs $( X_i,Y_i )$ that are drawn from the finite set $\setX \times \setY$ according to the PMF $P_{X,Y}$.

Consider some ``reconstruction alphabet'' $\hat \setX$ and some nonnegative ``distortion-function'' $d \colon \setX \times \hat \setX \rightarrow \reals^+_0$. We quantify the distortion between any pair of $n$-tuples $( \vecx, \hat \vecx ) \in \setX^n \times \hat \setX^n$ by their average distortion
\begin{equation}
d^{(n)} ( \vecx, \hat \vecx ) = \frac{1}{n} \sum^n_{i = 1} d ( x_i,\hat x_i ).
\end{equation}
The fidelity criterion we study is that any reconstruction $\hat \vecx \in \hat \setX^n$ of $X^n$ satisfy
\begin{equation}
d^{(n)} ( X^n, \hat \vecx ) \leq \Delta \label{eq:fidCrit}
\end{equation}
for some nonnegative ``distortion-level'' $\Delta \geq 0$. Following the convention of \cite{arikanmerhav98}, we assume that for every $x \in \setX$ there exists some $\hat x \in \hat \setX$ for which $d ( x, \hat x ) = 0$, i.e., that
\begin{equation}
\min_{\hat x \in \hat \setX} d ( x, \hat x ) = 0, \,\, \forall \, x \in \setX. \label{eq:minDistZero}
\end{equation}

To describe the results in this section, we denote by $\RDfun$ the classical rate-distortion function of $X$ given $Y$ under some fixed PMF $Q_{X,Y}$ on $\setX \times \setY$ \cite[Ch.~7]{csiszarkoerner11}
\begin{IEEEeqnarray}{l}
\RDfun = \min_{ \substack{ Q_{\hat X|X,Y} \colon \\ \smEx {}{ d (X, \hat X) } \leq \Delta } } I ( X, \hat X | Y ); \label{eq:defRDFun}
\end{IEEEeqnarray}
and we denote by $\relent {Q_{X,Y}}{P_{X,Y}}$ the Kullback-Leibler divergence between two PMFs $Q_{X,Y}$ and $P_{X,Y}$ on $\setX \times \setY$. By $\RDexp$ we refer to the functional
\begin{IEEEeqnarray}{l}
\RDexp = \sup_{Q_{X,Y}} \Bigl( \RDfun - \rho^{-1} \relent {Q_{X,Y}}{P_{X,Y}} \Bigr), \label{eq:defFunctional}
\end{IEEEeqnarray}
where the supremum is over all PMFs $Q_{X,Y}$ on $\setX \times \setY$.

The remainder of this section is structured as follows. Section~\ref{sec:optGuessTaskEncRD} summarizes some notions and results pertaining to the rate-distortion versions of the guessing and task-encoding problems. Section~\ref{sec:listsAndGuessesRD} extends the results on guessing and task-encoding of Section~\ref{sec:listsAndGuesses} to the case where the reconstruction is subject to the fidelity criterion~\eqref{eq:fidCrit}. Finally, Section~\ref{sec:distStorRD} studies a rate-distrotion version of the model of Section~\ref{sec:problemStatement}.

\subsection{Optimal Guessing Functions and Task-Encoders} \label{sec:optGuessTaskEncRD}

Suppose we want to guess a reconstruction $\hat \vecx \in \hat \setX^n$ of $X^n$ that satisfies the fidelity criterion \eqref{eq:fidCrit} with guesses of the form ``Is $d^{(n)} (X^n, \hat \vecx) \leq \Delta$?'' Similarly as in Section~\ref{sec:optGuessTaskEnc}, we call $\hatguessD {}{}{ \cdot | Y^n }$ a guessing function on $\hat \setX^n$ if for every $\vecy \in \setY^n$ the mapping $\hatguessD {}{}{ \cdot | \vecy } \colon \hat \setX^n \rightarrow \bigl\{ 1, \ldots, |\hat \setX|^n \bigr\}$ is one-to-one.\footnote{Unlike the guessing problem of Section~\ref{sec:optGuessTaskEnc}, where we guess over the source-sequence alphabet $\setX^n$, here we guess over the reconstruction-sequence alphabet $\hat \setX^n$.} The guessing function determines the guessing order: If we use $\hatguessD {}{}{ \cdot | Y^n }$ to guess a reconstruction of $X^n$ from the observation $Y^n$ and observe that $Y^n$ equals $\vecy$, then the question ``Is $d^{(n)} (X^n, \hat \vecx) \leq \Delta$?'' will be our $\hatguessD {}{}{ \hat \vecx | \vecy }$-th question.

Suppose we are given a guessing function $\hatguessD {}{}{ \cdot | Y^n }$. For every $\vecy \in \setY^n$ we define $$\guessD {\Delta}{}{ \cdot | \vecy } \colon \setX^n \rightarrow \bigl\{ 1, \ldots, |\hat \setX|^n \bigr\}$$ as the unique mapping satisfying that, if $(X^n,Y^n)$ equals $(\vecx,\vecy)$, then the first question that will be answered with ``Yes!'' will be our $\guessD {\Delta}{}{ \vecx | \vecy }$-th question.\footnote{By \eqref{eq:minDistZero} and because $\Delta \geq 0$, at least one question will be answered with ``Yes!''.} That is, for every pair $(\vecx,\vecy) \in \setX^n \times \setY^n$ we denote by $\guessD {\Delta}{}{ \vecx | \vecy }$ the smallest positive integer $j$ satisfying that $d^{(n)} ( \vecx, \hat \vecx ) \leq \Delta$ holds for the unique $n$-tuple $\hat \vecx \in \hat \setX^n$ for which $\hatguessD {}{}{ \hat \vecx | \vecy } = j$. The \emph{success function} corresponding to $\hatguessD {}{}{ \cdot | Y^n }$ is the collection $\bigl\{ \guessD {\Delta}{}{ \cdot | \vecy } \bigr\}_{\vecy \in \setY^n}$ and is denoted $\guessD {\Delta}{}{ \cdot | Y^n }$. For every $\vecy \in \setY^n$ we define $$\psifun{ \cdot | \vecy } \colon \setX^n \rightarrow \hat \setX^n$$ as the unique mapping satisfying that
\begin{IEEEeqnarray}{l} \label{eq:psiDefinition}
\Bigl( \psifun { \vecx | \vecy } = \hat \vecx \iff \guessD {\Delta}{} { \vecx | \vecy } = \hatguessD {}{} { \hat \vecx | \vecy } \Bigr), \,\, \forall \, (\vecx, \hat \vecx, \vecy) \in \setX^n \times \hat \setX^n \times \setY^n,
\end{IEEEeqnarray}
so if $(X^n,Y^n)$ equals $(\vecx,\vecy)$, then the question ``Is $d^{(n)} (X^n, \hat \vecx) \leq \Delta$?'' will be answered with ``Yes!'' for the first time when $\hat \vecx = \psifun { \vecx | \vecy }$. The \emph{ reconstruction function} corresponding to $\hatguessD {}{}{ \cdot | Y^n }$ is the collection $\bigl\{ \psifun { \cdot | \vecy } \bigr\}_{\vecy \in \setY^n}$ and is denoted $\psifun { \cdot | Y^n }$.

We assess the performance of a guessing function in terms of the $\rho$-th moment of the number of guesses that we need to guess a reconstruction $\hat \vecx$ that satisfies the fidelity criterion \eqref{eq:fidCrit}. That is, the performance of $\hatguessD {}{}{ \cdot | Y^n }$ is $\bigEx {}{ \guessD {\Delta}{}{ X^n|Y^n }^\rho }$, where $\guessD {\Delta}{}{ \cdot |Y^n }$ is the success function corresponding to $\hatguessD {}{}{ \cdot | Y^n }$. We say that a guessing function is optimal if its performance is optimal, i.e., $\hatguessD {}{}{ \cdot | Y^n }$ is optimal iff its corresponding success function minimizes $\bigEx {}{ \guessD {\Delta}{}{ X^n|Y^n }^\rho }$ among all success functions. We can use Arikan and Merhav's results in \cite{arikanmerhav98} to characterize the asymptotic performance of optimal guessing functions on $\hat \setX^n$:

\begin{theorem}[Asymptotic Performance of Optimal Guessing Functions on $\hat \setX^n$]\cite[Section~VI.~C.]{arikanmerhav98}\label{th:optGuessFunRD}
There exist guessing functions $\hatguessD {}{}{ \cdot | Y^n }$ whose corresponding success functions $\guessD {\Delta}{}{ \cdot |Y^n }$ satisfy
\begin{IEEEeqnarray}{l}
\lim_{n \rightarrow \infty} \frac{1}{n} \log \Bigl( \bigEx {}{\guessD \Delta {} { X^n | Y^n }^\rho} \Bigr) \leq \rho \RDexp.
\end{IEEEeqnarray}
Conversely, for every guessing functions $\hatguessD {}{}{ \cdot | Y^n }$ with corresponding success functions $\guessD {\Delta}{}{ \cdot |Y^n }$
\begin{IEEEeqnarray}{l}
\lim_{n \rightarrow \infty} \frac{1}{n} \log \Bigl( \bigEx {}{\guessD \Delta {} { X^n | Y^n }^\rho} \Bigr) \geq \rho \RDexp. \label{eq:lbOptGuessFunRD}
\end{IEEEeqnarray}
\end{theorem}

For task-encoders we adopt the terminology of~\cite[Section~7]{buntelapidoth14}. Given some finite set $\setZ$, a task-encoder $\enc{ \cdot | Y^n }$ for $X^n$ given side-information $Y^n$ is for every $\vecy \in \setY^n$ a mapping $\enc { \cdot | \vecy } \colon \setX^n \rightarrow \setZ$. A corresponding task-decoder $\decphi { \cdot | Y^n }$ is, for every $\vecy \in \setY^n$, a mapping $\decphi { \cdot | \vecy } \colon \setZ \rightarrow 2^{\hat \setX^n}$ for which
\begin{IEEEeqnarray}{l} \label{eq:decRD}
\forall \, \vecx \in \setX^n \textnormal{ s.t.\ } P_{X|Y}^n (\vecx|\vecy) > 0 \quad \exists \, \hat \vecx \in \bigdecphi{ \enc { \vecx | \vecy } \bigl| \vecy } \colon d^{(n)} ( \vecx, \hat \vecx ) \leq \Delta.
\end{IEEEeqnarray}
If, upon observing $Y^n$, the task-encoder describes $X^n$ by $Z = \enc {X^n|Y^n}$, then the corresponding decoder produces the list $\setL^{Y^n}_Z \triangleq \decphi{ Z | Y^n }$. By \eqref{eq:decRD} this list is guaranteed to contain a reconstruction $\hat \vecx \in \hat \setX^n$ of $X^n$ that satisfies the fidelity criterion \eqref{eq:fidCrit}.

As in Section~\ref{sec:optGuessTaskEnc}, a stochastic task-encoder associates with every realization $( \vecx, \vecy ) \in \setX^n \times \setY^n$ of the pair $( X^n, Y^n )$ a PMF on $\setZ$ and, upon observing the side information $\vecy$, describes $\vecx$ by drawing $Z$ from $\setZ$ according to the PMF associated with $(\vecx,\vecy)$, so conditonal on $(X,Y) = (\vecx,\vecy)$ the probability that $Z = z$ is
\begin{equation}
\distof { Z = z | X^n = \vecx, Y^n = \vecy }, \quad (\vecx,\vecy,z) \in \setX^n \times \setY^n \times \setZ. \label{eq:condPMFRelGuessEncRD}
\end{equation}
A corresponding task-decoder is a collection of lists $\{ \setL^{\vecy}_z \}$ for which
\begin{IEEEeqnarray}{l} \label{eq:decRDStoch}
\forall \, (\vecx,\vecy,z) \in \setX^n \times \setY^n \times \setZ \textnormal{ s.t.\ } P_{X,Y}^n (\vecx, \vecy) \, \distof{ Z = z | X^n = \vecx, Y^n = \vecy } > 0 \quad \exists \, \hat \vecx \in \setL^{\vecy}_z \colon \nonumber \\
\qquad d^{(n)} ( \vecx, \hat \vecx ) \leq \Delta.
\end{IEEEeqnarray}
If, upon observing $Y^n$, the task-encoder describes $X^n$ by $Z$, then the corresponding decoder produces the list $\setL^{Y^n}_Z \subseteq \hat \setX^n$. By \eqref{eq:decRDStoch} this list is guaranteed to contain a reconstruction $\hat \vecx \in \hat \setX^n$ of $X^n$ that satisfies the fidelity criterion \eqref{eq:fidCrit}.

We assess the performance of an encoder-decoder pair in terms of the $\rho$-th moment $\BigEx {}{ \bigl| \setL^{Y^n}_Z \bigr|^\rho }$ of the size of the list that the decoder produces. Bunte and Lapidoth characterized the asymptotic performance of optimal encoder-decoder pairs for the case where $Y^n$ is null and the task-encoder is deterministic \cite[Theorem~VII.1]{buntelapidoth14}. A generalization of the results in~\cite{buntelapidoth14} to the case at hand where $Y^n$ need not be null and the task-encoder may be stochastic is feasible but not carried out in this paper. Instead, we shall use the close connection between task-encoding and guessing to characterize the asymptotic performance of optimal encoder-decoder pairs. The performance guarantees for optimal encoder-decoder pairs are thus presented in Section~\ref{sec:listsAndGuessesRD} ahead (Corollary~\ref{co:optTaskEncRD} ahead).

\subsection{Lists and Guesses}\label{sec:listsAndGuessesRD}

This section extends the results of Section~\ref{sec:listsAndGuesses} to the case where the reconstruction $\hat \vecx \in \hat \setX^n$ of $X^n$ is subject to the fidelity criterion \eqref{eq:fidCrit}. We begin with the rate-distortion version of Lemma~\ref{le:ImproveGuess}, which quantifies how some additional informaiton $Z$ (e.g., some description produced by an encoder), can help guessing:

\begin{lemma}\label{le:ImproveGuessRD}  
Given a finite set $\setZ$, draw $Z$ from $\setZ$ according to some conditional PMF $P_{Z|X^n,Y^n}$, so $(X^n,Y^n,Z) \sim P^n_{X,Y} \times P_{Z|X^n,Y^n}$. For optimal guessing functions $\hatguessD {} \star { \cdot | Y^n, Z }$ and $\hatguessD {} \star { \cdot | Y^n }$ with corresponding success function $\guessD \Delta \star { \cdot | Y^n,Z }$ and $\guessD \Delta \star { \cdot | Y^n }$ (which minimize $\bigEx {}{\guessD \Delta {}{ X^n | Y^n, Z}^\rho}$ and $\bigEx {}{\guessD \Delta {}{ X^n | Y^n }^\rho}$, respectively)
\begin{IEEEeqnarray}{l}
\bigEx {}{\guessD \Delta \star { X^n | Y^n,Z }^\rho} \geq |\setZ|^{-\rho} \bigEx {}{\guessD \Delta {\ast}{ X^n | Y^n }^\rho}. \label{eq:impGuessMaxRD}
\end{IEEEeqnarray}
Conversely, if $\psifun { \cdot | Y^n }$ is the reconstruction function corresponding to $\hatguessD {} \star { \cdot | Y^n }$ (for which \eqref{eq:psiDefinition} holds when we substitute $\hatguessD {} \star { \hat \vecx | \vecy }$ for $\hatguessD {}{} { \hat \vecx | \vecy }$ and $\guessD \Delta \star { \vecx | \vecy }$ for $\guessD \Delta {} { \vecx | \vecy }$ in \eqref{eq:psiDefinition}) and $Z = f \bigl( \psifun {X^n|Y^n}, Y^n \bigr)$ for some mapping $f \colon \hat \setX^n \times \setY^n \rightarrow \setZ$ for which $f ( \hat \vecx, \vecy ) = f ( \hat \vecx^\prime, \vecy )$ implies either $\bigl\lceil \hatguessD {} \star { \hat \vecx | \vecy } / | \setZ | \bigr\rceil \neq \bigl\lceil \hatguessD {} \star { \hat \vecx^\prime | \vecy } / | \setZ | \bigr\rceil$ or $\hat \vecx = \hat \vecx^\prime$, then
\begin{IEEEeqnarray}{l}
\bigEx {}{\guessD \Delta \star {X^n|Y^n,Z}^\rho} \leq \BigEx {}{\bigl\lceil \guessD \Delta \star { X^n | Y^n } / | \setZ | \bigr\rceil^\rho}. \label{eq:impGuessMinRD}
\end{IEEEeqnarray}
Such a mapping $f$ always exists, because for all $l \in \naturals$ at most $|\setZ|$ different $\hat \vecx \in \hat \setX^n$ satisfy $\bigl\lceil \hatguessD {} \star { \hat \vecx | \vecy }  / | \setZ | \bigr\rceil = l$.
\end{lemma}

\begin{proof}
See Appendix~\ref{app:pfLeImproveGuessRD}.
\end{proof}

Lemma~\ref{le:ImproveGuessRD} and \eqref{eq:ceilApprox} imply the following rate-distortion version of Corollary~\ref{co:impGuess}:

\begin{corollary} \label{co:impGuessRD}
Given a finite set $\setZ$, there exists some mapping $f \colon \setX^n \times \setY^n \rightarrow \setZ$ such that $Z = f ( X^n, Y^n )$ satisfies
\begin{IEEEeqnarray}{l}
\min_{\hatguessD {}{} { \cdot | Y^n, Z }} \bigEx {}{\guessD \Delta {}{X^n|Y^n,Z}^\rho} < 1 + 2^{\rho} | \setZ |^{-\rho} \min_{\hatguessD {}{} { \cdot | Y^n }} \bigEx {}{\guessD \Delta {}{X^n | Y^n}^\rho}. \label{eq:lbGuessSIRD}
\end{IEEEeqnarray}
Conversely, for every chance variable $Z$ that takes values in $\setZ$
\begin{IEEEeqnarray}{l}
\min_{\hatguessD {}{} { \cdot | Y^n, Z }} \bigEx {}{\guessD \Delta {}{X^n|Y^n,Z}^\rho} \geq | \setZ |^{-\rho} \min_{\hatguessD {}{} { \cdot | Y^n }} \bigEx {}{\guessD \Delta {}{X^n | Y^n}^\rho} \vee 1. \label{eq:ubGuessSIRD}
\end{IEEEeqnarray}
(In \eqref{eq:lbGuessSIRD} and \eqref{eq:ubGuessSIRD} $\guessD \Delta {}{\cdot|Y^n,Z}$ and $\guessD \Delta {}{\cdot|Y^n}$ are the success functions corresponding to $\hatguessD {}{} { \cdot | Y^n, Z }$ and $\hatguessD {}{} { \cdot | Y^n }$, respectively.)
\end{corollary}

From Corollary~\ref{co:impGuessRD} and Theorem~\ref{th:optGuessFunRD}, which characterizes the asymptotic performance of optimal guessing functions $\hatguessD {} {} { \cdot | Y^n }$, we obtain the following asymptotic rate-distortion version of Corollary~\ref{co:equivBunteResultGuessing}:

\begin{corollary}\label{co:equivBunteResultGuessingRD}
Let $\hatguessD {}{} { \cdot | Y^n, Z }$ be guessing functions and let $\guessD \Delta {}{\cdot|Y^n,Z}$ be the corresponding success functions. Then, given a positive rate $R > 0$ and finite sets $\setZ_n$ satisfying
\begin{equation}
\lim_{n \rightarrow \infty} \frac{\log |\setZ_n|}{n} = R,
\end{equation}
there exist mappings $f_n \colon \setX^n \times \setY^n \rightarrow \setZ_n$ for which $Z_n = f_n ( X^n,Y^n )$ satisfy
\begin{IEEEeqnarray}{l}
\lim_{n \rightarrow \infty} \frac{1}{n} \log \biggl( \min_{\hatguessD {}{} { \cdot | Y^n, Z }} \bigEx {}{\guessD \Delta {}{X^n|Y^n,Z_n}^\rho} \biggr) \leq \rho \Bigl( \RDexp - R \Bigr) \vee 0. \label{eq:lbGuessSIExpRD}
\end{IEEEeqnarray}
Moreover, if $R > \RDexp$, then there exist mappings $f_n \colon \setX^n \times \setY^n \rightarrow \setZ_n$ for which $Z_n = f_n ( X^n,Y^n )$ satisfy
\begin{IEEEeqnarray}{l}
\lim_{n \rightarrow \infty} \min_{\hatguessD {}{} { \cdot | Y^n, Z }} \bigEx {}{\guessD \Delta {}{X^n|Y^n,Z_n}^\rho} = 1. \label{eq:lbGuessSIExpRD1}
\end{IEEEeqnarray}
Conversely, for all chance variables $Z_n$ taking values in $\setZ_n$
\begin{IEEEeqnarray}{l}
\lim_{n \rightarrow \infty} \frac{1}{n} \log \biggl( \min_{\hatguessD {}{} { \cdot | Y^n, Z }} \bigEx {}{\guessD \Delta {}{X^n|Y^n,Z_n}^\rho} \biggr) \geq \rho \Bigl( \RDexp - R \Bigr) \vee 0. \label{eq:ubGuessSIExpRD}
\end{IEEEeqnarray}
\end{corollary}

Our next result is a rate-distortion version of Theorem~\ref{th:relGuessEnc}:

\begin{theorem} \label{th:relGuessEncRD}
Let $\setZ$ be a finite set.
\begin{enumerate}
\item Given any stochastic task-encoder \eqref{eq:condPMFRelGuessEncRD}, every decoder with lists $\{ \setL^{\vecy}_z \}$ \eqref{eq:decRDStoch} induces a guessing function $\hatguessD {}{}{\cdot|Y^n}$ whose corresponding success function $\guessD \Delta {}{\cdot | Y^n }$ satisfies
\begin{IEEEeqnarray}{l}
\bigEx {}{\guessD \Delta {}{X^n|Y^n}^\rho} \leq \card \setZ^\rho \BigEx {}{ \bigl| \setL^{Y^n}_Z \bigr|^\rho}. \label{eq:listToGuessRD}
\end{IEEEeqnarray}
\item Every guessing function $\hatguessD {}{}{\cdot|Y^n}$ with corresponding success function $\guessD \Delta {}{\cdot | Y^n }$ and every positive integer $\omega \leq | \hat \setX |^n$ satisfying
\begin{IEEEeqnarray}{l}
| \setZ | \geq \omega \biggl( 1 + \Bigl\lfloor \log \bigl\lceil | \hat \setX |^n / \omega \bigr\rceil \Bigr\rfloor \biggr) \label{eq:relCardMandVRD}
\end{IEEEeqnarray}
induce a deterministic task-encoder, i.e., a stochastic task-encoder whose conditional PMF \eqref{eq:condPMFRelGuessEncRD} is $\{ 0, 1 \}$-valued, and a decoder whose lists $\{ \setL^{\vecy}_z \}$ \eqref{eq:decRDStoch} satisfy
\begin{IEEEeqnarray}{l} 
\BigEx {}{\bigl| \setL^{Y^n}_Z \bigr|^\rho} \leq \BigEx {}{ \bigl\lceil \guessD \Delta {} { X^n | Y^n } / \omega \bigr\rceil^\rho }. \label{eq:guessToListRD}
\end{IEEEeqnarray}
\end{enumerate}
\end{theorem}

\begin{proof}
See Appendix~\ref{app:pfThRelGuessEncRD}.
\end{proof}

The following rate-distortion version of Corollary~\ref{co:guessToBestList} results from Theorem~\ref{th:relGuessEncRD} and \eqref{eq:ceilApprox} by setting $$\omega = \biggl\lfloor | \setZ | / \Bigl( 1 + \bigl\lfloor \log | \hat \setX |^n \bigr\rfloor \Bigr) \biggr\rfloor$$ in Theorem~\ref{th:relGuessEncRD}.

\begin{corollary} \label{co:guessToBestListRD}
Given a set $\setZ$ of cardinality $|\setZ| \geq 1 + \bigl\lfloor \log |\hat \setX|^n \bigr\rfloor$, any guessing function $\hatguessD {}{}{\cdot|Y^n}$ with corresponding success function $\guessD \Delta {}{\cdot | Y^n }$ induces a deterministic task-encoder, i.e., a stochastic task-encoder whose conditional PMF \eqref{eq:condPMFRelGuessEncRD} is $\{ 0, 1 \}$-valued, and a decoder with lists $\{ \setL^{\vecy}_z \}$ \eqref{eq:decRDStoch} that satisfy
\begin{IEEEeqnarray}{l} 
\BigEx {}{\bigl| \setL^{Y^n}_Z \bigr|^\rho} \leq 1 + 2^{\rho} \bigEx {}{\guessD \Delta {}{ X^n | Y^n }^\rho} \biggl( \frac{ | \setZ | }{1 + \log | \hat \setX |^n} - 1 \biggr)^{-\rho}. \label{eq:guessToBestListRD}
\end{IEEEeqnarray}
\end{corollary}

We can combine \eqref{eq:listToGuessRD} and \eqref{eq:guessToBestListRD} with Theorem~\ref{th:optGuessFunRD}, which characterizes the asymptotic performance of an optimal guessing function $\hatguessD {} {} { \cdot | Y^n }$, to characterize the asymptotic performance of optimal encoder-decoder pairs:

\begin{corollary}[Asymptotic Performance of Optimal Encoder-Decoder Pairs] \label{co:optTaskEncRD}
Given a positive rate $R > 0$ and finite sets satisfying
\begin{equation}
\lim_{n \rightarrow \infty} \frac{ \log |\setZ_n| }{n} = R,
\end{equation}
there exist deterministic task-encoders, i.e., stochastic task-encoders whose conditional PMFs \eqref{eq:condPMFRelGuessEncRD} (where we substitute $\setZ_n$ for $\setZ$ in \eqref{eq:condPMFRelGuessEncRD}) are $\{ 0, 1 \}$-valued, and decoders whose lists $\bigl\{ \setL^{\vecy}_{z_n} \bigr\}$ satisfy \eqref{eq:decRDStoch} (when we substitute $\setZ_n$ for $\setZ$ in \eqref{eq:decRDStoch}) for which
\begin{IEEEeqnarray}{l}
\lim_{n \rightarrow \infty} \frac{1}{n} \log \BigEx {}{\bigl| \setL^{Y^n}_{Z_n} \bigr|^\rho} \leq \rho \Bigl( \RDexp - R \Bigr) \vee 0;
\end{IEEEeqnarray}
and if, moreover, $R > \RDexp$, then there exist encoder-decoder pairs for which
\begin{IEEEeqnarray}{l} 
\lim_{n \rightarrow \infty} \BigEx {}{\bigl| \setL^{Y^n}_{M_n} \bigr|^\rho} = 1.
\end{IEEEeqnarray}
Conversely, for any stochastic task-encoders \eqref{eq:condPMFRelGuessEncRD} (where we substitute $\setZ_n$ for $\setZ$ in \eqref{eq:condPMFRelGuessEncRD}) and decoders whose lists $\bigl\{ \setL^{\vecy}_{z_n} \bigr\}$ satisfy \eqref{eq:decRDStoch} (when we substitute $Z_n$ for $Z$, $z_n$ for $z$, and $\setZ_n$ for $\setZ$ in \eqref{eq:decRDStoch})
\begin{IEEEeqnarray}{l}
\lim_{n \rightarrow \infty} \frac{1}{n} \log \BigEx {}{\bigl| \setL^{Y^n}_{Z_n} \bigr|^\rho} \geq \rho \Bigl( \RDexp - R \Bigr) \vee 0. \label{eq:lbOptTaskEncRD}
\end{IEEEeqnarray}
\end{corollary}

Note that for the special case where $Y^n$ is null Corollary~\ref{co:optTaskEncRD} specializes to \cite[Theorem~VII.~1]{buntelapidoth14}.

Another interesting corollary to Theorem~\ref{th:relGuessEncRD}, that is to say a rate-distortion version of Corollary~\ref{co:guessToList}, results from the choice $\omega = 1$ in Theorem~\ref{th:relGuessEncRD}:

\begin{corollary}\label{co:guessToListRD}
Given a set $\setZ$ of cardinality $|\setZ| = 1 + \bigl\lfloor \log |\hat \setX|^n \bigr\rfloor$, any guessing function $\hatguessD {}{}{\cdot|Y^n}$ with corresponding success function $\guessD \Delta {}{\cdot | Y^n }$ induces a deterministic task-encoder, i.e., a stochastic task-encoder whose conditional PMF \eqref{eq:condPMFRelGuessEncRD} is $\{ 0, 1 \}$-valued, and a decoder with lists $\{ \setL^{\vecy}_z \}$ \eqref{eq:decRDStoch} that satisfy
\begin{equation}
\BigEx {}{\bigl| \setL^{Y^n}_Z \bigr|^\rho} \leq \bigEx {}{\guessD \Delta {} { X^n | Y^n }^\rho}. \label{eq:guessToListV1RD}
\end{equation}
\end{corollary}

\subsection{Distributed-Storage Systems} \label{sec:distStorRD}

We consider the following rate-distortion version of the model in Section~\ref{sec:problemStatement}. Upon observing $( X^n,Y^n ) = ( \vecx,\vecy )$, Alice draws the hints $M_1$ and $M_2$ from the finite set $\setM_1 \times \setM_2$ according to some conditional PMF
\begin{equation}
\distof{M_1 = m_1, M_1 = m_1 | X^n = \vecx, Y^n = \vecy }. \label{eq:aliceEncPMFRD}
\end{equation}
We assume here that $$\setM_1 = \{ 1, \ldots, 2^{n R_1} \}, \quad \setM_2 = \{ 1, \ldots, 2^{n R_2} \},$$ where $( R_1, R_2 )$ is a nonnegative pair corresponding to the rate. Bob sees both hints. In the guessing version he guesses a reconstruction of $X^n$ that satisfies \eqref{eq:fidCrit} based on the hints and the side information $Y^n$, and Bob's ambiguity about $X^n$ is thus
\begin{equation}
\mathscr A_{\textnormal B}^{(\textnormal g)} ( P^n_{X,Y}, \Delta ) = \min_{\hatguessD {}{}{\cdot|Y^n,M_1,M_2}} \bigEx {}{\guessD \Delta {}{X^n|Y^n,M_1,M_2}^\rho}, \label{eq:bobAmbiguityGuessingRD}
\end{equation}
where $\guessD \Delta {}{\cdot|Y^n,M_1,M_2}$ is the success function corresponding to the guessing function $\hatguessD {}{}{\cdot|Y^n,M_1,M_2}$. In the list version Bob's ambiguity about $X^n$ is
\begin{equation}
\mathscr A_{\textnormal B}^{(\textnormal l)} ( P^n_{X,Y}, \Delta ) = \BigEx {}{\bigl| \setL^{Y^n}_{M_1,M_1} \bigr|^\rho}, \label{eq:bobAmbiguityListRD}
\end{equation}
where $\bigl\{ \setL^{\vecy}_{m_1,m_1} \bigr\}$ are the lists of a decoder corresponding to the stochastic encoder \eqref{eq:aliceEncPMFRD} and thus satisfy \eqref{eq:decRDStoch} (when we substitute $(M_1,M_2)$ for $Z$, $(m_1,m_2)$ for $z$, and $\setM_1 \times \setM_2$ for $\setZ$ in \eqref{eq:decRDStoch}), so
\begin{IEEEeqnarray}{l}
P_{X,Y}^n ( \vecx, \vecy ) \, \distof{M_1 = m_1, M_1 = m_1 | X^n = \vecx, Y^n = \vecy } > 0 \nonumber \\
\quad \implies \exists \, \hat \vecx \in \setL^{\vecy}_{m_1,m_2} \colon d^{(n)} (\vecx, \hat \vecx) \leq \Delta.
\end{IEEEeqnarray}
Eve sees one of the hints and guesses a reconstruction of $X^n$ that satisfies \eqref{eq:fidCrit} based on this hint and the side information $Y$. We assume that an accomplice of hers chooses the hint so that her guessing efforts are minimum. In both versions Eve's ambiguity about $X$ is thus
\begin{IEEEeqnarray}{l}
\!\!\!\!\!\!\!\!\!\!\!\! \mathscr A_{\textnormal E} ( P^n_{X,Y}, \Delta ) = \min_{\hatguessD {}{(1)}{\cdot|Y^n,M_1}, \, \hatguessD {}{(2)}{\cdot|Y^n,M_2}} \BigEx {}{\guessD \Delta {(1)}{ X^n | Y^n, M_1 }^\rho \wedge \guessD \Delta {(2)}{ X^n | Y^n, M_2 }^\rho}, \label{eq:distEncSecrecyMeasureRD}
\end{IEEEeqnarray}
where $\guessD \Delta {(1)}{\cdot|Y^n,M_1}$ and $\guessD \Delta {(2)}{\cdot|Y^n,M_2}$ are the success functions corresponding to the guessing functions $\hatguessD {}{(1)}{\cdot|Y^n,M_1}$ and $\hatguessD {}{(2)}{\cdot|Y^n,M_2}$, respectively.

For both versions of the problem, we shall characterize the largest exponential growth that we can guarantee for Eve's ambiguity subject to the constraint that Bob's ambiguity tend to one, i.e., we shall characterize the privacy-exponent $\overbar{E_{\textnormal E}}$ defined in Definition~\ref{de:EvesAmbig}. In addition, we shall also characterize the largest exponential growth that we can guarantee for Eve's ambiguity in case Bob's ambiguity is allowed to grow exponentially with a given normalized (by $n$) exponent $E_{\textnormal B} \geq 0$, i.e., we shall characterize the modest privacy-exponent $\overbar{E_{\text E}^{\textnormal m} ( E_{\textnormal B} )}$ defined in Definition~\ref{de:EvesAmbigModest}. Like the model studied in Section~\ref{sec:problemStatement}, the privacy-exponent and the modest privacy-exponent turn out not to depend on the version of the problem, and $\mathscr A_{\textnormal B}$ can thus stand for either $\mathscr A_{\textnormal B}^{( \textnormal g )}$ or $\mathscr A_{\textnormal B}^{( \textnormal l )}$.

Our results are presented in the following theorem, which generalizes Theorems~\ref{th:asympDistStor} and \ref{th:asympDistStorEB}. To prove the theorem, we combine the proofs of Theorems~\ref{th:distStorGuess} and \ref{th:distStor} with the proofs of Theorems~\ref{th:asympDistStor} and \ref{th:asympDistStorEB}. Thereby, we replace the results of Section~\ref{sec:listsAndGuesses} with their rate-distortion versions, i.e., with the results of Section~\ref{sec:listsAndGuessesRD}. The main difficulty in adapting the proofs to the rate-distortion version of the problem is that Claim~1 in the proof of Theorems~\ref{th:distStorGuess} and \ref{th:distStor} need not hold, because Eve need not guess $X^n$ but only a reconstruction of it that satisfies \eqref{eq:fidCrit}.

\begin{theorem} \label{th:secrecyRD}
Given any nonnegative rate-pair $( R_1, R_2 )$ and distortion-level $\Delta \geq 0$, the privacy exponent is
\begin{IEEEeqnarray}{l}
\overbar{E_{\textnormal E}} = \begin{cases} \rho \Bigl( R_1 \wedge R_2 \wedge \RDexp \Bigr) & R_1 + R_2 > \RDexp, \\ - \infty & R_1 + R_2 < \RDexp; \end{cases} \label{eq:privExpBob1RD}
\end{IEEEeqnarray}
and the modest privacy exponent for $E_{\textnormal B} \geq 0$ is
\begin{IEEEeqnarray}{l}
\overbar{E_{\textnormal E}^{\textnormal m} ( E_{\textnormal B} )} = \begin{cases} \bigl( \rho ( R_1 \wedge R_2 ) + E_{\textnormal B} \bigr) \wedge \rho \RDexp &\\ &\hspace{-2cm} R_1 + R_2 \geq \RDexp - \rho^{-1} E_{\textnormal B}, \\ - \infty &\hspace{-2cm} R_1 + R_2 < \RDexp - \rho^{-1} E_{\textnormal B}. \end{cases} \label{eq:privExpBobEBRD}
\end{IEEEeqnarray}
\end{theorem}

\begin{proof}
See Appendix~\ref{app:pfSecrecyRD}.
\end{proof}

\section{Summary} \label{sec:conclusion}

This paper studies a distributed-storage system whose encoder, Alice,
observes some sensitive information $X$ (e.g., a password) that takes
values in a finite set $\setX$ and describes it using two hints, which
she stores in different locations. The legitimate receiver, Bob, sees
both hints, and---depending on the version of the problem---must
either guess $X$ (the guessing version) or must form a list that is
guaranteed to contain $X$ (the list version). The eavesdropper, Alice,
sees only one of the hints; an accomplice of hers controls
which. Based on her observation, Eve wishes to guess $X$. For an
arbitrary $\rho > 0$, Bob's and Eve's ambiguity about $X$ are
quantified as follows: In the guessing version we quantify Bob's
ambiguity by the $\rho$-th moment of the number of guesses that he
needs to guess $X$, and in the list version we quantify Bob's
ambiguity by the $\rho$-th moment of the size of the list that he must
form. In both versions we quantify Eve's ambiguity by the $\rho$-th
moment of the number of guesses that she needs to guess $X$. For each
version this paper characterizes---up to polylogarithmic factors of
$|\setX|$---the largest ambiguity that we can guarantee that Eve will
have subject to a given upper bound on the ambiguity that Bob may
have. Our results imply that, if the hint that is available to Bob but
not to Eve can assume $\sigma$ realizations, then---up to
polylogarithmic factors of $| \setX |$---the ambiguity that we can
guarantee that Eve will have either exceeds the ambiguity that Bob may
have by a factor of $\sigma^\rho$ or---in case the hint that Eve
observes reveals no information about $X$---is as large as it can
be. This holds even if we require that---up to polylogarithmic factors
of $| \setX |$---Bob's ambiguity be as small as it can be. The paper
also discusses extensions to a distributed-storage system that is
robust against disk failures and a rate-distortion version of the
problem.

The results for the guessing and the list version are remarkably similar: every pair of ambiguities for Bob and Eve that is achievable in the guessing version is---up to polylogarithmic factors of $\card \setX$---also achievable in the list version and vice versa. This can be explained by the close relation between Arikan's guessing problem \cite{arikan96} and Bunte and Lapidoth's task-encoding problem \cite{buntelapidoth14} that this paper reveals. The relation can be used to give alternative proofs of \cite[Theorems~I.2 and~VI.2]{buntelapidoth14} as well as the direct part of \cite[Theorem~I.1]{buntelapidoth14fb}. It holds also for the rate-distortion versions of the guessing and task-encoding problems, which were introduced in \cite{arikanmerhav98, buntelapidoth14}; and in this case it can be used to give an alternative proof of \cite[Theorem~VII.1]{buntelapidoth14}.

%The relation implies that in the asymptotic regime (i.e., when we wish to characterize encoding- or transmission-rates that are achievable when the blocklength tends to infinity) of any communication setting in which the encoder knows at least as much as the decoder (e.g., in a source-coding problem where no side information is available only at the decoder, or in a channel-coding problem with perfect feedback), a list-decoder (whose performance is measured in terms of the $\rho$-th moment of the size of the list that he must form) performs at least as well as a guessing decoder (whose performance is measured in terms of the $\rho$-th moment of the number of guesses that he needs).

\begin{appendix}

\section{A Proof of Corollary~\ref{prop:intGuessing}}\label{app:pfIntGuessing}

\begin{proof}
The converse results readily follow from the converse results of Theorem~\ref{th:distStorGuess}: \eqref{eq:BobMomDistStorConvGuess}  implies \eqref{eq:intUBLBConv} and \eqref{eq:EveMomDistStorConvGuess} implies \eqref{eq:intEveAmbUB}. The proof of the achievability results \eqref{eq:intBobAmbUB}--\eqref{eq:intEveAmbLB} is more involved. Suppose that \eqref{eq:intUBLB} holds. To show that there is a choice of the conditional PMF in \eqref{eq:aliceEncPMF} for which \eqref{eq:intBobAmbUB}--\eqref{eq:intEveAmbLB} hold, we will exhibit a judicious choice of the triple $( c_{\textnormal s}, c_1, c_2 ) \in \naturals^3$ for which \eqref{eq:intBobAmbUB} follows from \eqref{eq:BobMomDistStorDirGuess} and \eqref{eq:intEveAmbLB} from \eqref{eq:EveMomDistStorDirGuess}. By possibly relabeling the hints, we can assume w.l.g.\ that $| \setM_2 | = | \setM_1 | \wedge | \setM_2 |$. Our choice of $( c_{\textnormal s}, c_1, c_2 )$ depends on the constant $\mathscr U_{\textnormal B}$ and the cardinalities $| \setM_1 |$ and $| \setM_2 |$. Specifically, we distinguish between three different cases.

The first case is the case where
\begin{IEEEeqnarray}{l}
\mathscr U_{\text B} \geq 1 + 2^{\rho ( \renent {\tirho}{X|Y} - \log |\setM_2| + 1 )}. \label{eq:proofPrIntGuessingCase1}
\end{IEEEeqnarray}
In this case we choose
\begin{IEEEeqnarray}{l}
c_{\textnormal s} = | \setM_2 | \textrm{ and }  c_1 = c_2 = 1.
\end{IEEEeqnarray}
Note that this choice satisfies \eqref{bl:condCsC1C2Guess}. Consequently, \eqref{eq:BobMomDistStorDirGuess} implies that Bob's ambiguity satisfies \eqref{eq:intBobAmbUB}:
\begin{IEEEeqnarray}{rCl}
\mathscr A^{(\textnormal g)}_{\textnormal B} ( P_{X,Y} ) & < & 1 + 2^{\rho ( \renent {\tirho}{X|Y} - \log | \setM_2 | + 1 )} \\
& \leq & \mathscr U_{\textnormal B},
\end{IEEEeqnarray}
where the second inequality holds by \eqref{eq:proofPrIntGuessingCase1}. Moreover, it follows from \eqref{eq:EveMomDistStorDirGuess} that Eve's ambiguity satisfies \eqref{eq:intEveAmbLB}:
\begin{IEEEeqnarray}{rCl}
\mathscr A_{\textnormal E} ( P_{X,Y} ) & \geq & ( 1 + \ln | \setX | )^{-\rho} 2^{\rho ( \renent {\tirho}{X|Y} - \log 2 )} \\
& = & 2^{-\rho} ( 1 + \ln | \setX | )^{-\rho} 2^{\rho \renent {\tirho}{X|Y}}.
\end{IEEEeqnarray}

The second case is the case where
\begin{subequations} \label{bl:proofPrIntGuessingCase2}
\begin{IEEEeqnarray}{rCl}
\mathscr U_{\textnormal B} & \geq & 1 + \bigl\lfloor | \setM_1 | / | \setM_2 | \bigr\rfloor^{-\rho} 2^{\rho (\renent {\tirho}{X|Y} - \log |\setM_2| + 1)} \label{eq:proofPrIntGuessingCase2LB}
\end{IEEEeqnarray}
and
\begin{IEEEeqnarray}{rCl}
\mathscr U_{\textnormal B} & < & 1 + 2^{\rho ( \renent {\tirho}{X|Y} - \log |\setM_2| + 1 ) }.  \label{eq:proofPrIntGuessingCase2UB}
\end{IEEEeqnarray}
\end{subequations}
In this case we choose
\begin{IEEEeqnarray}{l}
c_{\textnormal s} = | \setM_2 |, \quad c_1 = \Bigl\lceil 2^{\renent {\tirho}{X|Y} - \log | \setM_2 | + 1 - \rho^{-1} \log ( \mathscr U_{\text B} - 1 )} \Bigr\rceil, \quad c_2 = 1. \label{eq:proofPrIntGuessingCase2DefTriple}
\end{IEEEeqnarray}
By \eqref{eq:proofPrIntGuessingCase2LB}, this choice satisfies \eqref{bl:condCsC1C2Guess}. Moreover, note that
\begin{IEEEeqnarray}{rCl}
c_{\textnormal s} c_1 c_2 & \geq & | \setM_2 | \, 2^{\renent {\tirho}{X|Y} - \log | \setM_2 | + 1 - \rho^{-1} \log ( \mathscr U_{\text B} - 1 )} \\
& = & 2^{\renent {\tirho}{X|Y} + 1 - \rho^{-1} \log ( \mathscr U_{\text B} - 1 )}. \label{eq:proofPrIntGuessingCase2ProdTriple}
\end{IEEEeqnarray}
Consequently, it follows from \eqref{eq:BobMomDistStorDirGuess} that Bob's ambiguity satisfies \eqref{eq:intBobAmbUB}:
\begin{IEEEeqnarray}{rCl}
\mathscr A^{(\textnormal g)}_{\textnormal B} ( P_{X,Y} ) & < & 1 + 2^{\rho ( \renent {\tirho}{X|Y} - ( \renent {\tirho}{X|Y} + 1 - \rho^{-1} \log ( \mathscr U_{\text B} - 1 ) ) + 1 )} \\
& = & \mathscr U_{\textnormal B}.
\end{IEEEeqnarray}
From \eqref{eq:proofPrIntGuessingCase2UB} it follows that
\begin{IEEEeqnarray}{l} \label{eq:boundCase2M2}
1 < 2^{\renent {\tirho}{X|Y} - \log | \setM_2 | + 1 - \rho^{-1} \log ( \mathscr U_{\textnormal B} - 1 )}.
\end{IEEEeqnarray}
Note that, for every $\xi > 1$, it holds that $\lceil \xi \rceil < 2 \xi$. Consequently, \eqref{eq:proofPrIntGuessingCase2DefTriple} and \eqref{eq:boundCase2M2} imply that %This, \eqref{eq:proofPrIntGuessingCase2DefTriple}, and $$\lceil \xi \rceil < 2 \xi, \quad \xi > 1$$ imply that
\begin{IEEEeqnarray}{rCl}
c_1 + c_2 & = & c_1 + 1 \\
& < & 2 c_1 \\ \label{eq:bounfCase2C1}
& < & 2^{\renent {\tirho}{X|Y} - \log | \setM_2 | + 3 - \rho^{-1} \log ( \mathscr U_{\textnormal B} - 1 ) }.
\end{IEEEeqnarray}	 
Eve's ambiguity satisfies \eqref{eq:intEveAmbLB}, because from~\eqref{eq:EveMomDistStorDirGuess} and~\eqref{eq:bounfCase2C1} it follows that:
\begin{IEEEeqnarray}{rCl}
\mathscr A_{\textnormal E} ( P_{X,Y} ) & > & \bigl( 1 + \ln | \setX | \bigr)^{-\rho} 2^{\rho ( \renent {\tirho}{X|Y} - ( \renent {\tirho}{X|Y} - \log | \setM_2 | + 3 - \rho^{-1} \log ( \mathscr U_{\textnormal B} - 1 ) ) )} \\
& = & 2^{-3 \rho} \bigl( 1 + \ln | \setX | \bigr)^{-\rho} |\setM_2|^{\rho} ( \mathscr U_{\textnormal B} - 1 ) \\
& = & 2^{-3 \rho} \bigl( 1 + \ln | \setX | \bigr)^{-\rho} \bigl( |\setM_1| \wedge | \setM_2 | \bigr)^{\rho} ( \mathscr U_{\textnormal B} - 1 ),
\end{IEEEeqnarray}
where the last equality holds by the assumption that $|\setM_2| = |\setM_1| \wedge | \setM_2 |$.

The third and last case is the case where 
\begin{IEEEeqnarray}{l}
\mathscr U_{\text B} < 1 + \bigl\lfloor |\setM_1| / |\setM_2| \bigr\rfloor^{-\rho} 2^{\rho ( \renent {\tirho}{X|Y} - \log |\setM_2| + 1 ) }. \label{eq:proofPrIntGuessingCase3}
\end{IEEEeqnarray}
In this case we let $k^\star \in \naturals$ be the largest positive integer $k$ for which
\begin{IEEEeqnarray}{l}
1 + 2^\rho k^{-\rho} \bigl\lfloor |\setM_1| / k \bigr\rfloor^{-\rho} \bigl\lfloor |\setM_2| / k \bigr\rfloor^{-\rho} 2^{\rho \renent {\tirho}{X|Y} } \leq \mathscr U_{\textnormal B}, \label{eq:intkForcs}
\end{IEEEeqnarray}
and we choose
\begin{IEEEeqnarray}{l}
c_{\textnormal s} = k^\star, \quad c_1 = \bigl\lfloor |\setM_1| / k^\star \bigr\rfloor, \quad c_2 = \bigl\lfloor |\setM_2| / k^\star \bigr\rfloor. \label{eq:proofPrIntGuessingCase3DefTriple}
\end{IEEEeqnarray}
The existence of such a $k^\star$ follows from \eqref{eq:intUBLB}, which implies that \eqref{eq:intkForcs} holds when we substitute $1$ for $k$. The choice in~\eqref{eq:proofPrIntGuessingCase3DefTriple} satisfies \eqref{bl:condCsC1C2Guess}. Consequently, \eqref{eq:BobMomDistStorDirGuess} implies that Bob's ambiguity satisfies \eqref{eq:intBobAmbUB}:
\begin{IEEEeqnarray}{rCl}
\mathscr A^{(\textnormal g)}_{\textnormal B} ( P_{X,Y} ) & < & 1 + 2^{\rho ( \renent {\tirho}{X|Y} - \log ( c_{\textnormal s} \lfloor |\setM_1| / c_{\textnormal s} \rfloor \lfloor |\setM_2| / c_{\textnormal s} \rfloor ) + 1 )} \\
& \leq & \mathscr U_{\textnormal B},
\end{IEEEeqnarray}
where in the second inequality we used that \eqref{eq:intkForcs} holds when we substitute $c_{\textnormal s}$ for $k$. By the choice of $c_{\textnormal s}$ in \eqref{eq:proofPrIntGuessingCase3DefTriple} we also have
\begin{IEEEeqnarray}{rCl}
2^{-\rho \renent {\tirho}{X|Y}} ( \mathscr U_{\textnormal B} - 1 ) & \stackrel{(a)}< & 2^{\rho} ( c_{\textnormal s} + 1 )^{-\rho} \biggl\lfloor \frac{|\setM_1|}{c_{\textnormal s} + 1} \biggr\rfloor^{-\rho} \biggl\lfloor \frac{| \setM_2 |}{c_{\textnormal s} + 1} \biggr\rfloor^{-\rho} \\
& \stackrel{(b)}< & 2^{3 \rho} \biggl( \frac{c_{\textnormal s} + 1}{|\setM_1| \, |\setM_2|} \biggr)^{\!\! \rho}\\
& \stackrel{(c)}\leq & 2^{4 \rho} \biggl( \frac{c_{\textnormal s}}{|\setM_1| \, |\setM_2|} \biggr)^{\!\! \rho}, \label{eq:boundForC1}
\end{IEEEeqnarray}
where $(a)$ holds because $c_{\textnormal s}$ is the largest positive integer $k$ for which \eqref{eq:intkForcs} holds and consequently $$\mathscr U_{\textnormal B} < 1 + 2^\rho (c_{\textnormal s} + 1)^{-\rho} \biggl\lfloor \frac{|\setM_1|}{c_{\textnormal s} + 1} \biggr\rfloor^{-\rho} \biggl\lfloor \frac{| \setM_2 |}{c_{\textnormal s} + 1} \biggr\rfloor^{-\rho} 2^{\rho \renent {\tirho}{X|Y} };$$ $(b)$ holds because \eqref{eq:proofPrIntGuessingCase3} and the fact that \eqref{eq:intkForcs} holds for every positive integer $k < c_{\textnormal s} + 1$ imply that $|\setM_2| \geq c_{\textnormal s} + 1$ and consequently that $|\setM_1| \wedge |\setM_2| \geq c_{\textnormal s} + 1$, and because $$\xi / 2 < \lfloor \xi \rfloor, \quad \xi \geq 1;$$ and $(c)$ holds because $c_{\textnormal s} \geq 1$ and consequently $c_{\textnormal s} + 1 \leq 2 c_{\textnormal s}$. From \eqref{eq:boundForC1} we obtain that
\begin{IEEEeqnarray}{rCl}
( c_1 + c_2 )^{-\rho} & \stackrel{(a)}= & \Bigl( \bigl\lfloor |\setM_1| / c_{\textnormal s} \bigr\rfloor + \bigl\lfloor |\setM_2| / c_{\textnormal s} \bigr\rfloor \Bigr)^{\! -\rho} \\
& \stackrel{(b)}\geq & 2^{-\rho} \biggl( \frac{c_{\textnormal s}}{|\setM_1|} \biggr)^{\!\! \rho} \\
& \stackrel{(c)}> & 2^{-5 \rho} |\setM_2|^\rho \, 2^{-\rho \renent {\tirho}{X|Y}} ( \mathscr U_{\textnormal B} - 1 ), \label{eq:boundForC1PlusC2}
\end{IEEEeqnarray}
where $(a)$ holds by \eqref{eq:proofPrIntGuessingCase3DefTriple}; $(b)$ holds by the assumption that $|\setM_2| \leq |\setM_1|$; and $(c)$ holds by \eqref{eq:boundForC1}. From \eqref{eq:boundForC1PlusC2} and \eqref{eq:EveMomDistStorDirGuess} we obtain that Eve's ambiguity satisfies \eqref{eq:intEveAmbLB}:
\begin{IEEEeqnarray}{rCl}
\mathscr A_{\textnormal E} ( P_{X,Y} ) & > & 2^{-5 \rho} \bigl( 1 + \ln | \setX | \bigr)^{-\rho} | \setM_2 |^{\rho} ( \mathscr U_{\textnormal B} - 1 ) \\
& = & 2^{-5 \rho} \bigl( 1 + \ln | \setX | \bigr)^{-\rho} \bigl( | \setM_1 | \wedge | \setM_2 | \bigr)^{\rho} ( \mathscr U_{\textnormal B} - 1 ),
\end{IEEEeqnarray}
where the last equality holds by the assumption that $| \setM_2 | = | \setM_1 | \wedge | \setM_2 |$.
\end{proof}

\section{A Proof of Corollary~\ref{prop:intList}}\label{app:pfIntList}

\begin{proof}
The converse results readily follow from the converse results of Theorem~\ref{th:distStor}: \eqref{eq:BobMomDistStorConv} implies \eqref{eq:intUBLBConvList}, and \eqref{eq:EveMomDistStorConv} implies \eqref{eq:intEveAmbUBList}. The proof of the achievability results \eqref{eq:intBobAmbUBList}--\eqref{eq:intEveAmbLBList} is more involved. Suppose that $|\setM_1| \, |\setM_2| > \log | \setX | + 2$ and that \eqref{eq:intUBLBList} holds. To show that there is a choice of the conditional PMF in \eqref{eq:aliceEncPMF} for which \eqref{eq:intBobAmbUBList}--\eqref{eq:intEveAmbLBList} hold, we will exhibit a judicious choice of the triple $( c_{\textnormal s}, c_1, c_2 ) \in \naturals^3$ for which \eqref{eq:intBobAmbUBList} follows from \eqref{eq:BobMomDistStorDir} and \eqref{eq:intEveAmbLBList} from \eqref{eq:EveMomDistStorDir}. By possibly relabeling the hints, we can assume w.l.g.\ that $| \setM_2 | = | \setM_1 | \wedge | \setM_2 |$. Our choice of $( c_{\textnormal s}, c_1, c_2 )$ depends on $\mathscr U_{\textnormal B}$, $| \setM_1 |$, and $| \setM_2 |$; specifically, we distinguish three different cases.

The first case is the case where 
\begin{IEEEeqnarray}{l}
\mathscr U_{\textnormal B} \geq 1 + 2^{\rho ( \renent {\tirho}{X|Y} - \log ( | \setM_2 | - \log | \setX | - 2 ) + 2 ) }. \label{eq:proofPrIntListCase1}
\end{IEEEeqnarray}
In this case we choose
\begin{IEEEeqnarray}{l}
c_{\textnormal s} = | \setM_2 |, \quad c_1 = c_2 = 1.
\end{IEEEeqnarray}
Note that this choice satisfies \eqref{bl:condCsC1C2}. Consequently, \eqref{eq:BobMomDistStorDir} implies that Bob's ambiguity satisfies \eqref{eq:intBobAmbUBList}, because
\begin{IEEEeqnarray}{rCl}
\mathscr A^{(\textnormal l)}_{\textnormal B} ( P_{X,Y} ) & < & 1 + 2^{\rho ( \renent {\tirho}{X|Y} - \log ( | \setM_2 | - \log | \setX | - 2 ) + 2 )} \\
& \leq & \mathscr U_{\textnormal B},
\end{IEEEeqnarray}
where the second inequality holds by \eqref{eq:proofPrIntListCase1}. Moreover, from \eqref{eq:EveMomDistStorDir} it follows that Eve's ambiguity satisfies \eqref{eq:intEveAmbLBList}:
\begin{IEEEeqnarray}{rCl}
\mathscr A_{\textnormal E} ( P_{X,Y} ) & \geq & \bigl( 1 + \ln | \setX | \bigr)^{-\rho} 2^{\rho ( \renent {\tirho}{X|Y} - \log 2 )} \\
& = & 2^{-\rho} \bigl( 1 + \ln | \setX | \bigr)^{-\rho} 2^{\rho \renent {\tirho}{X|Y}}.
\end{IEEEeqnarray}

The second case is the case where 
\begin{subequations} \label{bl:proofPrIntListCase2}
\begin{IEEEeqnarray}{rCl}
\mathscr U_{\textnormal B} & \geq & 1 + 2^{\rho ( \renent {\tirho}{X|Y} - \log ( | \setM_2 | \, \lfloor | \setM_1 | / | \setM_2 | \rfloor - \log | \setX | - 2 ) + 2 ) } \label{eq:proofPrIntListCase2LB}
\end{IEEEeqnarray}
and
\begin{IEEEeqnarray}{rCl}
\mathscr U_{\textnormal B} & < & 1 + 2^{\rho ( \renent {\tirho}{X|Y} - \log ( | \setM_2 | - \log | \setX | - 2 ) + 2 ) }.  \label{eq:proofPrIntListCase2UB}
\end{IEEEeqnarray}
\end{subequations}
In this case we choose
\begin{IEEEeqnarray}{l}
c_{\textnormal s} = | \setM_2 |, \quad c_1 = \Bigl\lceil \bigl( 2^{\renent {\tirho}{X|Y} + 2 - \rho^{-1} \log ( \mathscr U_{\text B} - 1 )} + \log | \setX | + 2 \bigr) / | \setM_2 | \Bigr\rceil, \quad c_2 = 1. \label{eq:proofPrIntListCase2DefTriple}
\end{IEEEeqnarray}
By \eqref{eq:proofPrIntListCase2LB}, this choice satisfies \eqref{bl:condCsC1C2}. Moreover, note that
\begin{IEEEeqnarray}{rCl}
c_{\textnormal s} c_1 c_2 & \geq & 2^{\renent {\tirho}{X|Y} + 2 - \rho^{-1} \log ( \mathscr U_{\text B} - 1 )} + \log | \setX | + 2. \label{eq:proofPrIntListCase2ProdTriple}
\end{IEEEeqnarray}
Consequently, \eqref{eq:BobMomDistStorDir} implies that Bob's ambiguity satisfies \eqref{eq:intBobAmbUBList}, because 
\begin{IEEEeqnarray}{rCl}
\mathscr A^{(\textnormal l)}_{\textnormal B} ( P_{X,Y} ) & < & 1 + 2^{\rho \bigl( \renent {\tirho}{X|Y} - \log \bigl( 2^{\renent {\tirho}{X|Y} + 2 - \rho^{-1} \log ( \mathscr U_{\text B} - 1 ) } \bigr) + 2 \bigr)} \\
& = & \mathscr U_{\textnormal B}.
\end{IEEEeqnarray}
From \eqref{eq:proofPrIntListCase2UB} it follows that
\begin{IEEEeqnarray}{l} \label{eq:secondCaseBobsAmbigIntermdBound}
1 < \Bigl( 2^{\renent {\tirho}{X|Y} + 2 - \rho^{-1} \log ( \mathscr U_{\text B} - 1 )} + \log | \setX | + 2 \Bigr) / | \setM_2 |.
\end{IEEEeqnarray}
Note that, for every $\xi > 1$, it holds that $\lceil \xi \rceil < 2 \xi$. Consequently,~\eqref{eq:proofPrIntListCase2DefTriple} and~\eqref{eq:secondCaseBobsAmbigIntermdBound} imply that
\begin{IEEEeqnarray}{rCl}
c_1 + c_2 & = & c_1 + 1 \\
& < & 2 c_1 \\ \label{eq:secondCaseBobsAmibgC1Bound}
& < & 4 \Bigl( 2^{\renent {\tirho}{X|Y} + 2 - \rho^{-1} \log ( \mathscr U_{\text B} - 1 )} + \log | \setX | + 2 \Bigr) / | \setM_2 |.
\end{IEEEeqnarray}
From \eqref{eq:EveMomDistStorDir} and~\eqref{eq:secondCaseBobsAmibgC1Bound} it follows that Eve's ambiguity satisfies \eqref{eq:intEveAmbLBList}:
\begin{IEEEeqnarray}{rCl}
\mathscr A_{\textnormal E} ( P_{X,Y} ) & > & 2^{-2 \rho} \bigl( 1 + \ln | \setX | \bigr)^{-\rho} |\setM_2|^{\rho} \nonumber \\
& & \times 2^{\rho \bigl(\renent {\tirho}{X|Y} - \log \bigl( 2^{\renent {\tirho}{X|Y} + 2 - \rho^{-1} \log ( \mathscr U_{\text B} - 1 )} + \log | \setX | + 2 \bigr) \bigr)} \\
& = & 2^{-2 \rho} \bigl( 1 + \ln | \setX | \bigr)^{-\rho} |\setM_2|^{\rho} \, 2^{\rho \renent {\tirho}{X|Y}} \nonumber \\
& & \times \bigl( 2^{\renent {\tirho}{X|Y} + 2 - \rho^{-1} \log ( \mathscr U_{\textnormal B} - 1 )} + \log | \setX | + 2 \bigr)^{-\rho} \\
& \stackrel{(a)}\geq & 2^{-5 \rho} \bigl( 1 + \ln | \setX | \bigr)^{-\rho} | \setM_2 |^{\rho} ( \mathscr U_{\textnormal B} - 1 ) \nonumber \\
& & \wedge 2^{-3 \rho} \bigl( 1 + \ln | \setX | \bigr)^{-\rho} ( 2 + \log | \setX | )^{-\rho} | \setM_2 |^\rho \, 2^{\rho \renent {\tirho}{X|Y}} \\
& \stackrel{(b)}= & 2^{-5 \rho} \bigl( 1 + \ln | \setX | \bigr)^{-\rho} \bigl( |\setM_1| \wedge | \setM_2 | \bigr)^{\rho} ( \mathscr U_{\textnormal B} - 1 ) \nonumber \\
& & \wedge 2^{-3 \rho} \bigl( 1 + \ln | \setX | \bigr)^{-\rho} ( 2 + \log | \setX | )^{-\rho} \bigl( |\setM_1| \wedge | \setM_2 | \bigr)^\rho 2^{\rho \renent {\tirho}{X|Y}},
\end{IEEEeqnarray}
where $(a)$ holds because $$\frac{1}{a + b} \geq \frac{1}{2a} \wedge \frac{1}{2b}, \quad a, \, b > 0;$$ and $(b)$ holds by the assumption that $|\setM_2| = |\setM_1| \wedge | \setM_2 |$.

The third and last case is the case where
\begin{IEEEeqnarray}{l}
\mathscr U_{\text B} < 1 + 2^{\rho ( \renent {\tirho}{X|Y} - \log ( | \setM_2 | \, \lfloor | \setM_1 | / | \setM_2 | \rfloor - \log | \setX | - 2 ) + 2 ) }. \label{eq:proofPrIntListCase3}
\end{IEEEeqnarray}
In this case we let $k^\star \in \naturals$ be the largest positive integer $k$ for which
\begin{IEEEeqnarray}{l}
1 + 2^{\rho ( \renent {\tirho}{X|Y} - \log ( k \lfloor | \setM_1 | / k \rfloor \lfloor | \setM_2 | / k \rfloor - \log | \setX | - 2 ) + 2 ) } \leq \mathscr U_{\textnormal B}, \label{eq:intkForcsList}
\end{IEEEeqnarray}
and we choose
\begin{IEEEeqnarray}{l}
c_{\textnormal s} = k^\star, \quad c_1 = \bigl\lfloor |\setM_1| / k^\star \bigr\rfloor, \quad c_2 = \bigl\lfloor |\setM_2| / k^\star \bigr\rfloor. \label{eq:proofPrIntListCase3DefTriple}
\end{IEEEeqnarray}
The existence of such a $k^\star$ follows from \eqref{eq:intUBLBList}, which implies that \eqref{eq:intkForcsList} holds when we substitute $1$ for $k$. Note that the choice in~\eqref{eq:proofPrIntListCase3DefTriple} satisfies \eqref{bl:condCsC1C2}. Consequently, \eqref{eq:BobMomDistStorDir} implies that Bob's ambiguity satisfies \eqref{eq:intBobAmbUBList}, because
\begin{IEEEeqnarray}{rCl}
\mathscr A^{(\textnormal l)}_{\textnormal B} ( P_{X,Y} ) & < & 1 + 2^{\rho ( \renent {\tirho}{X|Y} - \log ( c_{\textnormal s} \lfloor |\setM_1| / c_{\textnormal s} \rfloor \lfloor |\setM_2| / c_{\textnormal s} \rfloor  - \log | \setX | - 2 ) + 2 )} \\
& \leq & \mathscr U_{\textnormal B},
\end{IEEEeqnarray}
where in the second inequality we used that \eqref{eq:intkForcsList} holds when we substitute $c_{\textnormal s}$ for $k$. By the choice of $c_{\textnormal s}$ in \eqref{eq:proofPrIntListCase3DefTriple} we also have
\begin{IEEEeqnarray}{rCl}
2^{-\rho ( \renent {\tirho}{X|Y} + 2 )} ( \mathscr U_{\textnormal B} - 1 ) & \stackrel{(a)}< & \Biggl( ( c_{\textnormal s} + 1 ) \biggl\lfloor \frac{|\setM_1|}{c_{\textnormal s} + 1} \biggr\rfloor \biggl\lfloor \frac{| \setM_2 |}{c_{\textnormal s} + 1} \biggr\rfloor - \log |\setX| - 2 \Biggr)^{\!\! -\rho} \\
& \stackrel{(b)}< & \biggl( \frac{| \setM_1 | \, | \setM_2 |}{4 ( c_{\textnormal s} + 1 )} - \log | \setX | - 2 \biggr)^{\!\! -\rho} \\
& \stackrel{(c)}\leq & \biggl( \frac{| \setM_1 | \, | \setM_2 |}{8 c_{\textnormal s}} - \log | \setX | - 2 \biggr)^{\!\! -\rho}, \label{eq:boundForC1List}
\end{IEEEeqnarray}
where $(a)$ holds because $c_{\textnormal s}$ is the largest positive integer $k$ for which \eqref{eq:intkForcsList} holds and consequently
\begin{IEEEeqnarray*}{l}
\mathscr U_{\textnormal B} < 1 + 2^{\rho ( \renent {\tirho}{X|Y} - \log ( (c_{\textnormal s} + 1) \lfloor | \setM_1 | / (c_{\textnormal s} + 1) \rfloor \lfloor | \setM_2 | / (c_{\textnormal s} + 1) \rfloor - \log | \setX | - 2 ) + 2 ) };
\end{IEEEeqnarray*}
$(b)$ holds because \eqref{eq:proofPrIntListCase3} and the fact that \eqref{eq:intkForcsList} holds for every positive integer $k < c_{\textnormal s} + 1$ imply that $|\setM_2| \geq c_{\textnormal s} + 1$ and consequently that $|\setM_1| \wedge |\setM_2| \geq c_{\textnormal s} + 1$, and because $$\xi / 2 < \lfloor \xi \rfloor, \quad \xi \geq 1;$$ and $(c)$ holds because $c_{\textnormal s} \geq 1$ and consequently $c_{\textnormal s} + 1 \leq 2 c_{\textnormal s}$. From \eqref{eq:boundForC1List} we obtain that
\begin{IEEEeqnarray}{l}
\biggl( \frac{c_{\textnormal s}}{| \setM_1 |} \biggr)^{\!\! \rho} > 2^{-3 \rho} | \setM_2 |^\rho \Bigl( ( \mathscr U_{\textnormal B} - 1 )^{-1 / \rho} 2^{\renent {\tirho}{X|Y} + 2} + \log | \setX | + 2 \Bigr)^{\! -\rho}, \label{eq:boundForC1CsList}
\end{IEEEeqnarray}
and consequently that
\begin{IEEEeqnarray}{rCl}
( c_1 + c_2 )^{-\rho} & \stackrel{(a)}= & \Bigl( \bigl\lfloor |\setM_1| / c_{\textnormal s} \bigr\rfloor + \bigl\lfloor |\setM_2| / c_{\textnormal s} \bigr\rfloor \Bigr)^{\! -\rho} \\
& \stackrel{(b)}\geq & 2^{-\rho} \biggl( \frac{c_{\textnormal s}}{|\setM_1|} \biggr)^{\!\! \rho} \\
& \stackrel{(c)}> & 2^{-4 \rho} |\setM_2|^\rho \Bigl( ( \mathscr U_{\textnormal B} - 1 )^{-1 / \rho} 2^{\renent {\tirho}{X|Y} + 2} + \log | \setX | + 2 \Bigr)^{\! -\rho} \\
& \stackrel{(d)}\geq & 2^{-7 \rho} | \setM_2 |^\rho ( \mathscr U_{\textnormal B} - 1 ) 2^{-\rho \renent {\tirho}{X|Y}} \wedge 2^{-5 \rho} ( 2 + \log | \setX | )^{-\rho} | \setM_2 |^\rho, \label{eq:boundForC1PlusC2List}
\end{IEEEeqnarray}
where $(a)$ holds by \eqref{eq:proofPrIntListCase3DefTriple}; $(b)$ holds by the assumption that $|\setM_2| \leq |\setM_1|$; $(c)$ holds by \eqref{eq:boundForC1CsList}; and $(d)$ holds because $$\frac{1}{a + b} \geq \frac{1}{2a} \wedge \frac{1}{2b}, \quad a, \, b > 0.$$ From \eqref{eq:boundForC1PlusC2List} and \eqref{eq:EveMomDistStorDir} we obtain that Eve's ambiguity satisfies \eqref{eq:intEveAmbLBList}:
\begin{IEEEeqnarray}{rCl}
\mathscr A_{\textnormal E} ( P_{X,Y} ) & > & 2^{-5 \rho} \bigl( 1 + \ln | \setX | \bigr)^{-\rho} | \setM_2 |^{\rho} \Bigl( 2^{-2 \rho} ( \mathscr U_{\textnormal B} - 1 ) \nonumber \\
& & \wedge \bigl( 2 + \log | \setX | \bigr)^{-\rho} 2^{\rho \renent {\tirho}{X|Y} } \Bigr) \\
& = & 2^{-5 \rho} \bigl( 1 + \ln | \setX | \bigr)^{-\rho} \bigl( |\setM_1| \wedge | \setM_2 | \bigr)^{\rho} \Bigl( 2^{-2 \rho} ( \mathscr U_{\textnormal B} - 1 ) \nonumber \\
& & \wedge \bigl( 2 + \log | \setX | \bigr)^{-\rho} 2^{\rho \renent {\tirho}{X|Y} } \Bigr),
\end{IEEEeqnarray}
where the last equality holds by the assumption that $| \setM_2 | = | \setM_1 | \wedge | \setM_2 |$.
\end{proof}

\section{A Proof of Theorem~\ref{th:fairOpponent}}\label{app:pfThFairOpponent}

\begin{proof}
We first establish the achievability results, i.e., \eqref{eq:BobMomEveListsDir}--\eqref{eq:EveMomEveListsDir}. To this end suppose that $|\setM_1| \wedge |\setM_2| \geq 1 + \bigl\lfloor \log |\setX| \bigr\rfloor$. Let
\begin{IEEEeqnarray}{l}
c_{\textnormal s} = 1 + \bigl\lfloor \log |\setX| \bigr\rfloor, \quad c_1 = \biggl\lfloor \frac{|\setM_1|}{c_{\textnormal s}} \biggr\rfloor, \quad c_2 = \biggl\lfloor \frac{|\setM_2|}{c_{\textnormal s}} \biggr\rfloor, \label{eq:defCsC1C2EveLists}
\end{IEEEeqnarray}
and for each $\nu \in \{ c_{\textnormal s}, c_1, c_2 \}$ let $V_\nu$ be a chance variable taking values in the set $\setV_\nu = \{ 0, \ldots, c_\nu - 1 \}$. Corollary~\ref{co:equivBunteResultGuessing} implies that there exists some $\{ 0, 1 \}$-valued conditional PMF $\bigdistof { ( V_1, V_2 ) = ( v_1, v_2 ) \bigl| X = x, Y = y}$ for which
\begin{IEEEeqnarray}{l}
\min_{\guess {}{\cdot | Y, V_1, V_2}} \bigEx {}{\guess {}{X|Y, V_1, V_2}^\rho} < 1 + 2^{\rho ( \renent {\tirho} { X | Y } - \log ( c_1 c_2 ) + 1 )}. \label{eq:BobMomEveListsDir1Guess}
\end{IEEEeqnarray}
Draw $(V_1,V_2)$ from $\setV_1 \times \setV_2$ according to the above conditional PMF. Fix $\epsilon > 0$ and draw $(V_1^\prime,V_2^\prime)$ from $\setV_1 \times \setV_2$ according to the conditional PMF
\begin{IEEEeqnarray}{l}
\bigdistof { ( V_1^\prime, V_2^\prime ) = ( v_1^\prime, v_2^\prime ) \bigl| ( V_1, V_2 ) = (v_1,v_2)} \nonumber \\
\quad = \biggl( 1 - 2^{-\epsilon} - \frac{2^{-\epsilon}}{|\setV_1| \, |\setV_2|} \biggr) \ind {\{ ( v_1^\prime, v_2^\prime ) = (v_1,v_2) \}} + \frac{2^{-\epsilon}}{|\setV_1| \, |\setV_2|}.
\end{IEEEeqnarray}
Note that, irrespective of the realization $(v_1,v_2)$ of $( V_1^\prime, V_2^\prime )$, the probability that $(V_1^\prime,V_2^\prime)$ equals $(v_1, v_2)$ is $1 - 2^{-\epsilon}$. Let $\guess \star { \cdot | Y, V_1, V_2}$ be an optimal guessing function, which minimizes $\bigEx {}{\guess {}{X|Y, V_1, V_2}^\rho}$. Define the guessing function $\guess {}{ \cdot | Y, V_1^\prime, V_2^\prime }$ by
\begin{IEEEeqnarray}{l}
\guess {}{ x | y, v_1^\prime, v_2^\prime } = \guess \star { x | y, v_1^\prime, v_2^\prime }, \,\, \forall \, (x,y,v_1^\prime,v_2^\prime) \in \setX \times \setY \times \setV_1 \times \setV_2.
\end{IEEEeqnarray}
Using the trivial bound $$\guess {}{x|y, v_1^\prime, v_2^\prime} \leq |\setX|, \,\, \forall \, (x,y,v_1^\prime,v_2^\prime) \in \setX \times \setY \times \setV_1 \times \setV_2,$$ we obtain that
\begin{equation}
\bigEx {}{\guess {}{X|Y, V_1^\prime, V_2^\prime}^\rho} \leq (1 - 2^{-\epsilon}) \bigEx {}{\guess \star {X|Y, V_1, V_2}^\rho} + 2^{-\epsilon} |\setX|^\rho.
\end{equation}
Consequently,
\begin{IEEEeqnarray}{l}
\min_{\guess {}{ \cdot | Y, V_1^\prime, V_2^\prime }} \bigEx {}{\guess {}{X|Y, V_1^\prime, V_2^\prime}^\rho} \nonumber \\
\quad \leq ( 1 - 2^{-\epsilon} ) \min_{\guess {}{\cdot | Y, V_1, V_2}} \bigEx {}{\guess {}{X|Y, V_1, V_2}^\rho} + 2^{-\epsilon} |\setX|^\rho \\
\quad < 1 + 2^{-(\epsilon - \rho \log \! |\setX|)} + 2^{\rho ( \renent {\tirho} {X | Y } - \log ( c_1 c_2 ) + 1 )}, \label{eq:BobMomEveListsDir2Guess}
\end{IEEEeqnarray}
where~\eqref{eq:BobMomEveListsDir2Guess} follows from \eqref{eq:BobMomEveListsDir1Guess}. Corollary~\ref{co:guessToList} and \eqref{eq:defCsC1C2EveLists} imply that there exists some $\{ 0, 1 \}$-valued conditional PMF $$\distof { V_{\textnormal s} = v_{\textnormal s} | X = x, Y = y, V_1^\prime = v_1, V_2^\prime = v_2 }$$ for which
\begin{IEEEeqnarray}{rCl}
\BigEx {}{ \bigl| \setL^Y_{V_{\textnormal s}, V_1^\prime, V_2^\prime} \bigr|^\rho } & \leq & \min_{\guess {}{ \cdot | Y, V_1^\prime, V_2^\prime }} \bigEx {}{ \guess {}{ X | Y, V_1^\prime, V_2^\prime }^\rho } \\
& < & 1 + 2^{-(\epsilon - \rho \log \! |\setX|)} + 2^{\rho ( \renent {\tirho} {X | Y } - \log ( c_1 c_2 ) + 1 )}. \label{eq:BobMomEveListsDir1}
\end{IEEEeqnarray}
Draw $V_{\textnormal s}$ from $\setV_{\textnormal s}$ according to the above conditional PMF. Using the assumption that $|\setM_1| \wedge |\setM_2| \geq 1 + \bigl\lfloor \log |\setX| \bigr\rfloor$ and \eqref{eq:defCsC1C2EveLists}, we obtain that
\begin{equation} \label{eq:lowerBoundOnCk}
c_k > \frac{ |\setM_k| }{2 \bigl( 1 + \bigl\lfloor \log |\setX| \bigr\rfloor \bigr)}, \quad k \in \{ 1, 2 \}.
\end{equation}
From~\eqref{eq:BobMomEveListsDir1} and~\eqref{eq:lowerBoundOnCk} it follows that
\begin{IEEEeqnarray}{rCl}
\!\!\!\! \BigEx {}{ \bigl| \setL^Y_{V_{\textnormal s}, V_1^\prime, V_2^\prime} \bigr|^\rho } & < & 1 + 2^{-(\epsilon - \rho \log |\setX|)} + 2^{\rho ( \renent {\tirho} {X | Y} - \log ( |\setM_1| \, |\setM_2| ) + 2 \log (1 + \lfloor \log |\setX| \rfloor) + 3 )}. \label{eq:BobMomEveListsDir2}
\end{IEEEeqnarray}
By \eqref{eq:defCsC1C2EveLists} $| \setM_1 | \geq c_{\textnormal s} c_1$ and $| \setM_2 | \geq c_{\textnormal s} c_2$, and hence it suffices to prove \eqref{eq:BobMomEveListsDir}--\eqref{eq:EveMomEveListsDir} for a conditional PMF \eqref{eq:aliceEncPMF} that assigns positive probability only to $c_{\textnormal s} c_1$ elements of $\setM_1$ and $c_{\textnormal s} c_2$ elements of $\setM_2$, and we thus assume w.l.g.\ that $\setM_1 = \setV_{\textnormal s} \times \setV_1$ and $\setM_2 = \setV_{\textnormal s} \times \setV_2$. That is, we can choose $M_1 = ( V_{\textnormal s} \oplus_{c_{\textnormal s}} \! U, V_1^\prime )$ and $M_2 = ( U, V_2^\prime )$, where $U$ is independent of $( X, Y, V_{\textnormal s}, V_1^\prime, V_2^\prime )$ and uniform over $\setV_{\textnormal s}$. For this choice it follows from \eqref{eq:BobMomEveListsDir2} that
\begin{IEEEeqnarray}{l}
\mathscr A_{\textnormal B}^{(\textnormal l)} ( P_{X,Y} ) < 1 + 2^{-(\epsilon - \rho \log |\setX|)} + 2^{\rho ( \renent {\tirho} {X | Y} - \log ( |\setM_1| \, |\setM_2| ) + 2 \log (1 + \lfloor \log |\setX| \rfloor) + 3 )}.
\end{IEEEeqnarray}
This proves that \eqref{eq:BobMomEveListsDir} holds for every sufficiently-large $\epsilon$. As to \eqref{eq:EveMomDistStorDir}, note that for every $\epsilon > 0$
\begin{equation}
\setL_{M_1}^Y = \setL_{M_2}^Y = \setL^Y,
\end{equation}
because
\begin{IEEEeqnarray}{l}
\distof {M_1 = m_1, M_2 = m_2 | X = x, Y = y} > 0, \,\, \forall \, (x,y,m_1,m_2) \in \setX \times \setY \times \setM_1 \times \setM_2.
\end{IEEEeqnarray}

We next conclude by establishing the converse results \eqref{eq:BobMomEveListsConv}--\eqref{eq:EveMomEveListsConv}. Theorem~\ref{th:optTaskEnc} implies \eqref{eq:BobMomEveListsConv}; and \eqref{eq:EveMomEveListsConv} trivially holds, because the list that Eve forms based on $Y$ and the hint that she observes cannot be larger than the list that she would have to form if she were to observe only $Y$.
\end{proof}

\section{A Proof of Theorems~\ref{th:secMessGuess} and \ref{th:secMess}}\label{app:pfThsSecMess}

\begin{proof}
We first establish the achievability results, i.e., \eqref{eq:BobMomSecMsgDirGuess}--\eqref{eq:EveMomSecMsgDirGuess} in the guessing version and \eqref{eq:BobMomSecMsgDir}--\eqref{eq:EveMomSecMsgDir} in the list version. To this end, fix $c \in \naturals$ satisfying \eqref{eq:condCGuess} in the guessing version and \eqref{eq:condCList} in the list version. Both \eqref{eq:condCGuess} and \eqref{eq:condCList} imply that $c \leq | \setM_{\textnormal p} |$. Hence it suffices to prove \eqref{eq:BobMomSecMsgDirGuess}--\eqref{eq:EveMomSecMsgDirGuess} and \eqref{eq:BobMomSecMsgDir}--\eqref{eq:EveMomSecMsgDir} for a $\{ 0,1 \}$-valued conditional PMF as in~\eqref{eq:aliceEncPMFSecMsg} that assigns positive probability only to $c$ elements of $\setM_{\textnormal p}$. We can thus assume w.l.g.\ that $| \setM_{\textnormal p} | = c$. Corollary~\ref{co:equivBunteResultGuessing} implies that there exists some $\{ 0, 1 \}$-valued conditional PMF $$\distof {M_{\textnormal p} = m_{\textnormal p}, M_{\textnormal s} = m_{\textnormal s} | X = x, Y = y}$$ for which
\begin{IEEEeqnarray}{rCl}
\min_{\guess {}{ \cdot | Y, M_{\textnormal p}, M_{\textnormal s} }} \bigEx {}{ \guess {}{X|Y,M_{\textnormal p},M_s}^\rho } & < & 1 + 2^{\rho ( \renent {\tirho} { X | Y } - \log ( | \setM_{\textnormal p} | \, | \setM_{\textnormal s} | ) + 1 )} \\
& = & 1 + 2^{\rho ( \renent {\tirho} {X | Y } - \log ( c \, | \setM_{\textnormal s} | ) + 1 )}. \label{eq:BobMomSecMsgDirGuessPf}
\end{IEEEeqnarray}
In addition, Theorem~\ref{th:optTaskEnc} implies that there exists some deterministic task-encoder $\enc { \cdot | Y } \colon \setX \rightarrow \setM_{\textnormal p} \times \setM_{\textnormal s}$ for which
\begin{IEEEeqnarray}{rCl}
\BigEx {}{\bigl|\setL^Y_{M_{\textnormal p},M_{\textnormal s}}\bigr|^\rho} & < & 1 + 2^{\rho ( \renent {\tirho} { X | Y } - \log ( |\setM_{\textnormal p} \times \setM_{\textnormal s}| - \log | \setX | - 2 ) + 2 )} \\
& = & 1 + 2^{\rho ( \renent {\tirho} { X | Y } - \log ( c \, | \setM_{\textnormal s} | - \log | \setX | - 2) + 2 )}, \label{eq:BobMomSecMsgDirPf}
\end{IEEEeqnarray}
where $(M_{\textnormal p}, M_{\textnormal s}) = \enc { X | Y }$. Accordingly, in the guessing version \eqref{eq:BobMomSecMsgDirGuess} follows from \eqref{eq:BobMomSecMsgDirGuessPf} and in the list version \eqref{eq:BobMomSecMsgDir} follows from \eqref{eq:BobMomSecMsgDirPf}. Moreover, Corollary~\ref{co:equivBunteResultGuessing} implies \eqref{eq:EveMomSecMsgDirGuess} in the guessing version and \eqref{eq:EveMomSecMsgDir} in the list version:
\begin{IEEEeqnarray}{rCl}
\min_{\guess {\cdot}{X | Y, M_{\textnormal p} }} \Ex {}{\guess {}{X | Y, M_{\textnormal p} }^\rho} & \geq & \bigl( 1 + \ln | \setX | \bigr)^{-\rho} 2^{\rho ( \renent{\tirho}{X|Y} - \log | \setM_{\textnormal p} | )} \\
& = & \bigl( 1 + \ln | \setX | \bigr)^{-\rho} 2^{\rho ( \renent{\tirho}{ X | Y } - \log c )}.
\end{IEEEeqnarray}

It remains to establish the converse results, i.e., \eqref{eq:BobMomSecMsgConvGuess}--\eqref{eq:EveMomSecMsgConvGuess} in the guessing version and \eqref{eq:BobMomSecMsgConv}--\eqref{eq:EveMomSecMsgConv} in the list version. In the guessing version \eqref{eq:BobMomSecMsgConvGuess} follows from Corollary~\ref{co:equivBunteResultGuessing}, and in the list version \eqref{eq:BobMomSecMsgConv} follows from Theorem~\ref{th:optTaskEnc}. To prove \eqref{eq:EveMomSecMsgConvGuess} and \eqref{eq:EveMomSecMsgConv}, we first note from Corollary~\ref{co:impGuess} that
\begin{IEEEeqnarray}{l}
\min_{\guess {}{\cdot|Y,M_{\textnormal p},M_{\textnormal s}}} \bigEx {}{\guess {}{X|Y,M_{\textnormal p},M_{\textnormal s}}^\rho} \geq | \setM_{\textnormal s} |^{-\rho} \min_{\guess {}{\cdot|Y,M_{\textnormal p}}} \bigEx {}{\guess {}{X|Y,M_{\textnormal p}}^\rho}. \label{eq:pfSecMsgConvEve}
\end{IEEEeqnarray}
Moreover, we also note that
\begin{IEEEeqnarray}{l}
\min_{\guess {}{\cdot|Y,M_{\textnormal p},M_{\textnormal s}}} \bigEx {}{\guess {}{X|Y,M_{\textnormal p},M_{\textnormal s}}^\rho} \leq \BigEx {}{\bigl| \setL^Y_{M_{\textnormal p},M_{\textnormal s}} \bigr|^\rho}. \label{eq:pfSecMsgConvEve2ndPart}
\end{IEEEeqnarray}
From \eqref{eq:pfSecMsgConvEve} and \eqref{eq:pfSecMsgConvEve2ndPart} it follows that in both versions Eve's ambiguity exceeds Bob's by at most a factor of $| \setM_{\textnormal s} |^\rho$, i.e., $\mathscr A_{\textnormal E} ( P_{X,Y} ) \leq | \setM_{\textnormal s} |^\rho \mathscr A^{(\textnormal g)}_{\textnormal B} ( P_{X,Y} )$ and $\mathscr A_{\textnormal E} ( P_{X,Y} ) \leq | \setM_{\textnormal s} |^\rho \mathscr A^{(\textnormal l)}_{\textnormal B} ( P_{X,Y} )$. Since Eve can ignore $M_{\textnormal p}$ and guess $X$ based on $Y$ alone, we obtain from Theorem~\ref{th:optGuessFun} that in both versions Eve's ambiguity cannot exceed $2^{\rho \renent {\tirho} {X | Y }}$. That is,
\begin{equation}
\mathscr A_{\textnormal E} ( P_{X,Y} ) = \min_{\guess {}{ \cdot | Y, M_{\textnormal p}} } \bigEx {}{\guess {}{X | Y, M_{\textnormal p} }^\rho} \leq 2^{\rho \renent {\tirho} {X | Y }}.
\end{equation}
This concludes the proof of \eqref{eq:EveMomSecMsgConvGuess} and \eqref{eq:EveMomSecMsgConv} and consequently that of the converse results.
\end{proof}

\section{A Proof of Theorems~\ref{th:keyGuess} and \ref{th:key}}\label{app:pfThsKey}

\begin{proof}
We first establish the achievability results, i.e., \eqref{eq:BobMomKeyDirGuess}--\eqref{eq:EveMomKeyDirGuess} in the guessing version and \eqref{eq:BobMomKeyDir}--\eqref{eq:EveMomKeyDir} in the list version. To this end fix $c \in \naturals$ satisfying \eqref{eq:condCGuessKey} in the guessing version and \eqref{eq:condCListKey} in the list version. Let $M_{\textnormal p}$ be a chance variable that takes values in the set $\setM_{\textnormal p}$, and let $M_{\textnormal s}$ be a chance variable that takes values in the set $\setK$. Corollary~\ref{co:equivBunteResultGuessing} implies that there exists some $\{ 0, 1 \}$-valued conditional PMF $\distof {M_{\textnormal p} = m_{\textnormal p}, M_{\textnormal s} = m_{\textnormal s} | X = x, Y = y}$ for which
\begin{IEEEeqnarray}{rCl}
\min_{\guess {}{\cdot|Y,M_{\textnormal p}, M_{\textnormal s}}} \bigEx {}{\guess {}{X|Y,M_{\textnormal p},M_{\textnormal s}}^\rho } & < & 1 + 2^{\rho ( \renent {\tirho} {X | Y } - \log ( | \setM_{\textnormal p} | \, | \setM_{\textnormal s} | ) - 1 )} \\
& = & 1 + 2^{\rho ( \renent {\tirho} {X | Y } - \log ( c \, | \setK | ) - 1 )}. \label{eq:keyMpMsGuess}
\end{IEEEeqnarray}
Theorem~\ref{th:optTaskEnc} implies that there exists some deterministic task-encoder $\enc {\cdot | Y } \colon \setX \rightarrow \setM_{\textnormal p} \times \setM_{\textnormal s}$ for which
\begin{IEEEeqnarray}{rCl}
\BigEx {}{ \bigl| \setL^Y_{M_{\textnormal p}, M_{\textnormal s}} \bigr|^\rho } & < & 1 + 2^{\rho ( \renent {\tirho} { X | Y } - \log ( |\setM_{\textnormal p} | \, | \setM_{\textnormal s} | - \log | \setX | - 2 ) + 2 )} \\
& = & 1 + 2^{\rho ( \renent {\tirho} { X | Y } - \log ( c \, | \setK | - \log | \setX | - 2 ) + 2 )}, \label{eq:keyMpMsList}
\end{IEEEeqnarray}
where $(M_{\textnormal p}, M_{\textnormal s}) = \enc { X | Y }$. Both \eqref{eq:condCGuessKey} and \eqref{eq:condCListKey} imply that $c \, | \setK | \leq | \setM |$. Hence it suffices to prove \eqref{eq:BobMomKeyDirGuess}--\eqref{eq:EveMomKeyDirGuess} and \eqref{eq:BobMomKeyDir}--\eqref{eq:EveMomKeyDir} for a $\{ 0, 1 \}$-valued conditional PMF as in~\eqref{eq:aliceEncPMFKey} that assigns positive probability only to $c \, | \setK |$ elements of $\setM$. We can thus assume w.l.g.\ that $\setM = \setK \times \setM_{\textnormal p}$, where $\setM_{\textnormal p}$ is a set of cardinality $c$, and $\setK = \bigl\{ 0, \ldots, | \setK | - 1 \bigr\}$. That is, we can choose $M = ( M_{\textnormal s} \oplus_{|\setK|} \! K, M_{\textnormal p} )$, where $(M_{\textnormal s}, M_{\textnormal p})$ is drawn according to one of the above conditional PMFs depending on the version. Bob observes the hint $M$ and the secret key $K$ and can thus recover the pair $(M_{\textnormal s}, M_{\textnormal p})$. Hence, in the guessing version \eqref{eq:BobMomKeyDirGuess} follows from \eqref{eq:keyMpMsGuess}, and in the list version \eqref{eq:BobMomKeyDir} follows from \eqref{eq:keyMpMsList}.

The proof of \eqref{eq:EveMomKeyDirGuess} and \eqref{eq:EveMomKeyDir} is more involved. Note that in both versions (guessing and list) there exists some mapping $g \colon \setX \times \setY \times \setM \rightarrow \setK$ for which
\begin{equation}
K = g ( X, Y, M ). \label{eq:pfKeyKeyFuncXYM}
\end{equation}
Given any guessing function $\guess {}{ \cdot | Y, M }$ for $X$, introduce some guessing function $\guess {}{\cdot, \cdot | Y, M }$ for $( X, K )$ satisfying that
\begin{equation}
\bigguess {}{x, g (x,y,m) \bigl| y, m} = \guess {}{x | y, m}, \,\, \forall \, (x,y,m) \in \setX \times \setY \times \setM.
\end{equation}
From \eqref{eq:pfKeyKeyFuncXYM} it then follows that
\begin{equation}
\guess {}{ X, K | Y, M } = \guess {}{ X | Y, M },
\end{equation}
and consequently that Eve can guess $X$ and the pair $(X,K)$ with the same number of guesses. In particular,
\begin{equation}
\bigEx {}{ \guess {}{ X | Y, M }^\rho } = \bigEx {}{ \guess {}{ X, K | Y, M }^\rho }. \label{eq:EveMomKeyDir1}
\end{equation}
Corollary~\ref{co:equivBunteResultGuessing} implies that
\begin{IEEEeqnarray}{rCl}
\min_{\guess {}{ \cdot, \cdot | Y,M }} \bigEx {}{ \guess {}{ X, K | Y, M }^\rho } & \geq & \bigl( 1 + \ln | \setX | \bigr)^{-\rho} 2^{\rho ( \renent {\tirho} { X, K | Y } - \log | \setM | )} \\
& = & \bigl( 1 + \ln | \setX | \bigr)^{-\rho} 2^{\rho ( \renent{\tirho}{ X, K | Y } - \log ( c \, | \setK | ) )}. \label{eq:EveMomKeyDir2}
\end{IEEEeqnarray}

Note, that 
\begin{IEEEeqnarray}{rCl}
\renent {\tirho}{ X, K | Y } & = & \frac{1}{\rho} \log \sum_{y \in \setY} \Biggl( \sum_{x \in \setX} \sum_{k \in \setK} \biggl( \frac{P_{X,Y} ( x,y )}{| \setK |} \biggr)^{\!\! \tirho} \Biggr)^{\!\! 1+\rho} \\
& = & \frac{1}{\rho} \log \! \left( \sum_{y \in \setY} \Biggl( \sum_{x \in \setX} P_{X,Y} ( x,y )^{\tirho} \Biggr)^{\!\! 1+\rho} | \setK |^{\rho} \right) \\
& = & \renent {\tirho} { X | Y } + \log | \setK |,
\end{IEEEeqnarray}
where the first equality holds because $K$ is independent of $(X,Y)$ and uniform over the set $\setK$. Consequently, \eqref{eq:EveMomKeyDir1} and \eqref{eq:EveMomKeyDir2} imply \eqref{eq:EveMomKeyDirGuess} in the guessing version and \eqref{eq:EveMomKeyDir} in the list version.
 
It remains to establish the converse results, i.e., \eqref{eq:BobMomKeyConvGuess}--\eqref{eq:EveMomKeyConvGuess} in the guessing version and \eqref{eq:BobMomKeyConv}--\eqref{eq:EveMomKeyConv} in the list version. To this end we first note that
\ba
\renent {\tirho}{ X | Y, K } &= \frac{\alpha}{1-\alpha} \log \sum_{y \in \setY} \sum_{k \in \setK} \Biggl( \sum_{x \in \setX} \biggl( \frac{P_{X,Y} ( x,y )}{| \setK |} \biggl)^{\!\! \alpha} \Biggr)^{\!\! \frac{1}{\alpha}} \\
&= \frac{\alpha}{1-\alpha} \log \sum_{y \in \setY} \biggl( \sum_{x \in \setX} P_{X,Y} ( x,y )^\alpha \biggr)^{\!\! \frac{1}{\alpha}} \\
&= \renent {\tirho} { X | Y }, \label{eq:pfKeyRenentXGivYK}
\ea
where the first equality holds because $K$ is independent of $(X,Y)$ and uniform over the set $\setK$. In the guessing version \eqref{eq:BobMomKeyConvGuess} follows from Corollary~\ref{co:equivBunteResultGuessing} and \eqref{eq:pfKeyRenentXGivYK}, and in the list version \eqref{eq:BobMomKeyConv} follows from Theorem~\ref{th:optTaskEnc} and \eqref{eq:pfKeyRenentXGivYK}. To prove \eqref{eq:EveMomKeyConvGuess} and \eqref{eq:EveMomKeyConv}, we first note that by Corollary~\ref{co:impGuess}
\begin{IEEEeqnarray}{l}
\min_{\guess {}{ \cdot | Y, K, M }} \bigEx {}{\guess {}{ X | Y, K, M }^\rho} \geq | \setK |^{-\rho} \min_{\guess {}{ \cdot | Y, M }} \bigEx {}{\guess {}{ X | Y, M }^\rho}. \label{eq:pfKeyConvEve}
\end{IEEEeqnarray}
Because $$\min_{\guess {}{ \cdot | Y, K, M }} \bigEx {}{ \guess {}{ X | Y, K, M }^\rho } \leq \BigEx {}{ \bigl| \setL^{Y,K}_M \bigr|^\rho },$$ \eqref{eq:pfKeyConvEve} implies that in both versions Eve's ambiguity exceeds Bob's by at most a factor of $| \setK |^\rho$, i.e., $\mathscr A_{\textnormal E} ( P_{X,Y} ) \leq | \setK |^\rho \mathscr A^{(\textnormal g)}_{\textnormal B} ( P_{X,Y} )$ and $\mathscr A_{\textnormal E} ( P_{X,Y} ) \leq | \setK |^\rho \mathscr A^{(\textnormal l)}_{\textnormal B} ( P_{X,Y} )$. Since Eve can ignore $M$ and guess $X$ based on $Y$ alone, we obtain from Theorem~\ref{th:optGuessFun} that in both versions Eve's ambiguity cannot exceed $2^{\rho \renent {\tirho} {X | Y }}$:
\begin{IEEEeqnarray}{l}
A_{\textnormal E} ( P_{X,Y} ) = \min_{\guess {}{ \cdot | Y, M }} \bigEx {}{\guess {}{ X | Y, M }^\rho} \leq 2^{\rho \renent {\tirho} {X | Y }}.
\end{IEEEeqnarray}
This concludes the proof of \eqref{eq:EveMomKeyConvGuess} and \eqref{eq:EveMomKeyConv} and consequently that of the converse results.
\end{proof}

\section{A Proof of Theorems~\ref{th:discFailGuess} and \ref{th:discFail}} \label{app:pfThsDiscFail}

In Section~\ref{sec:MDS} we summarize the results on maximum-distance separable (MDS) codes that we shall use in the proof of Theorems~\ref{th:discFailGuess} and \ref{th:discFail}. Theorems~\ref{th:discFailGuess} and \ref{th:discFail} are proved in Section~\ref{sec:pfThsDiscFail}.

\subsection{Properties of MDS Codes} \label{sec:MDS}

The following results on maximum-distance separable (MDS) codes can be found, e.g., in \cite{roth06}. An $(n,k)$ linear code $\mathscr C$ over a finite field $\mathbb F_q$ is a $k$-dimensional linear subspace of the vector space $\mathbb F_q^n$ of all $n$-tuples over $\mathbb F_q$. An $(n,k,d)$ linear code is an $(n,k)$ linear code satisfying that the minimum Hamming distance between any two codewords (or, equivalently, the minimum Hamming weight of any nonzero codeword) is $d$. By the Singleton bound $k \leq n - d + 1$, where equality is achieved iff the following holds for every size-$k$ set $\setC \subseteq [1:n]$, where $k = n - d + 1$: if we reduce all $q^k$ codewords to the components indexed by $\setC$, then we obtain all $q^k$ $k$-tuples over $\mathbb F_q$. An MDS code is a linear code that satisfies the Singleton bound with equality.

In this paper we are interested in the case where $q = 2^{\ell}, \,\, \ell \in \naturals$, and we denote by $\alpha$ a primitive element of $\mathbb F_q$. If $n = q$, then for every $k \in \{ 1, \ldots, n \}$
\begin{IEEEeqnarray}{l}
G_{k,q} =
\begin{pmatrix}
1 & 1 & 1 & \ldots & 1 \\ 0 & 1 & \alpha & \ldots & \alpha^{-1} \\ \vdots & \vdots & \vdots & & \vdots \\ 0 & 1 & \alpha^{k-1} & \ldots & \alpha^{-(k-1)}
\end{pmatrix} \in \mathbb F_q^{k \times q} \label{eq:genMatrixMDS}
\end{IEEEeqnarray}
is a generator matrix of a $(q,k)$ MDS code. (More precisely, $G_{k,q}$ is a generator matrix of a Reed-Solomon (RS) code.) To see this, note that
\begin{IEEEeqnarray}{l}
\vecu G_{k,q} = \bigl( u (0), u (\alpha^0), u (\alpha), \ldots, u (\alpha^{-1}) \bigr), \quad \vecu \in \mathbb F_q^k,
\end{IEEEeqnarray}
where $u (\beta) = \sum_{j = 0}^{k-1} u_j \beta^j, \,\, \beta \in \mathbb F_q$ is computed in the field $\mathbb F_{2^\ell}$. Hence, the first component of $\vecu G_{k,q}$ is zero iff zero is a root of $u (z)$, and for every $i \in \{ 2, \ldots, q \}$ the $i$-th component of $\vecu G_{k,q}$ is zero iff $\alpha^{i-2}$ is a root of $u (z)$. Since $\alpha$ is a primitive element of $\mathbb F_q$, we know that $0, \, 1, \, \alpha, \ldots, \, \alpha^{q-1}$ are distinct elements of $\mathbb F_q$. Moreover, the polynomial $u (z)$ has degree at most $k - 1$, and hence the fundamental theorem of algebra asserts that if $u (z) \neq 0$, then $u (z)$ can have at most $k-1$ roots in $\mathbb F_q$. Consequently, at most $k-1$ components of any nonzero codeword can be zero, and hence every nonzero codeword has Hamming weight at least $n - k + 1$. This and the Singleton bound imply that $d = n - k + 1$ and consequently that the code with generator matrix \eqref{eq:genMatrixMDS} is a $(q,k)$ MDS code.

If $k \leq n \leq q$, then the matrix $G_{k,n} \in \mathbb F_q^{k \times n}$ that we obtain by taking the first $n$ columns of $G_{k,q}$ is a generator matrix of an $(n,k)$ MDS code. To see this, note that reducing $G_{k,q}$ to its first $n$ columns is tantamount to reducing each codeword to its first $n$ components. This implies that the Hamming weight of any codeword or, equivalently, the Hamming distance between any two codewords can decrease by at most $q - n$, and consequently that the minimum Hamming distance between any two codewords can decrease by at most $q - n$. Consequently, the new code is an $(n,k,d)$ linear code with $d \geq q - k + 1 - (q - n) = n - k + 1$. This and the Singleton bound imply that $d = n - k + 1$ and consequently that the new code is an MDS code.

We also note here that, for any generator matrix $G_{k,n}$ of an $(n,k)$ MDS code over $\mathbb F_q$, where $k \leq n \leq q$, and any $k^\prime < k$, the matrix $G_{k^\prime,n}$ that we obtain by taking the first $k^\prime$ rows of $G_{k,n}$ is a generator matrix of an $(n,k^\prime)$ MDS code.

\subsection{A Proof of Theorems~\ref{th:discFailGuess} and \ref{th:discFail}} \label{sec:pfThsDiscFail}

\begin{proof}
We first establish the achievability results, i.e., \eqref{eq:BobMomDiscFailDirGuess}--\eqref{eq:EveMomDiscFailDirGuess} in the guessing version and \eqref{eq:BobMomDiscFailDir}--\eqref{eq:EveMomDiscFailDir} in the list version. We begin with an outline of the proof ideas. We shall use the following coding scheme. Upon observing $( X,Y )$, Alice describes $X$ deterministically by a tuple $( V,W )$, where $V$ takes values in the finite field $\mathbb F^\nu_{2^p}$ and $W$ in $\mathbb F^{\nu - \eta}_{2^r}$. Depending on the version, she chooses the description $( V,W )$ so that, if Bob's observation were $( V,W )$, then his ambiguity about $X$ would satisfy \eqref{eq:BobMomDiscFailDirGuess} in the guessing version and \eqref{eq:BobMomDiscFailDir} in the list version. Then, she maps $V$ to a length-$\delta$ codeword of a $( \delta,\nu,\delta-\nu+1 )$ MDS code over $\mathbb F_{2^p}$ and stores each codeword symbol on a different disc. Since the code is MDS, any $\gamma \leq \nu$ hints reveal $\gamma p$ bits of $V$. Independently of $( X,Y )$, Alice draws a random variable $U$ uniformly over the field $\mathbb F^{\eta}_{2^r}$, maps $( W,U )$ to a length-$\delta$ codeword of a $( \delta,\nu,\delta-\nu+1 )$ MDS code over the field $\mathbb F_{2^r}$, and stores each codeword symbol on a different disc. She chooses the mapping so that any $\eta$ codeword symbols are independent of $W$ or, equivalently, that given $W$ it is possible to reconstruct $U$ from any $\eta$ codeword symbols. (As in \cite{subramanianmclaughlin09}, this is accomplished using nested MDS codes.) As a consequence, $W$ can be recovered from any $\nu$ hints, while any $\eta$ hints reveal no information about $W$.

Summing up, the outlined coding scheme guarantees that, upon observing $\nu$ hints, Bob can reconstruct the tuple $( V,W )$. Hence, his ambiguity about $X$ satisfies \eqref{eq:BobMomDiscFailDirGuess} in the guessing version and \eqref{eq:BobMomDiscFailDir} in the list version. Observing $\eta$ hints enables Eve to recover $\eta p$ bits of $V$, but it does not enable her to recover any information about $W$. Using the results of Section~\ref{sec:listsAndGuesses}, we can thus show that observing $\eta$ hints can decrease Eve's guessing efforts by at most a factor of $2^{- \rho \nu p}$.\footnote{The coding scheme is reminiscent of the coding scheme in the proof of Theorem~\ref{th:distStorGuess} and \ref{th:distStor}, where after describing $X$ Alice stores part of the description (insecurely) on the first hint, another part (insecurely) on the second hint, and the remaining portion (securely) so that it can only be computed from both hints.} Since we quantify Eve's ambiguity by \eqref{eq:distEncSecrecyMeasureDiscFail}, we assume that---upon observing all the hints and $( X,Y )$---an adversarial genie reveals to Eve the $\eta$ hints that minimize her ambiguity. In doing so, the genie can decrease Eve's ambiguity by an additional factor of at most $\delta^{-\rho \eta}$ (this is due to Corollary~\ref{co:impGuess} and the fact that there are ${\delta \choose \eta} \leq \delta^\eta$ size-$\eta$ subsets of $\{ 1, \ldots, \delta \}$).

The described MDS codes exist if each nonnegative integer $p$ and $r$ is either zero or at least $\log \delta$ (see Appendix~\ref{sec:MDS}). Recalling that each disc stores up to $s$ bits, we can thus construct the MSD codes whenever $p$ and $r$ satisfy \eqref{bl:condGuessDiscFail}. In the list version the stronger requirement~\eqref{bl:condListDiscFail}---in addition to guaranteeing the existence of the described MDS codes---allows us to use Theorem~\ref{th:optTaskEnc} in order to guarantee that Bob's ambiguity satisfy \eqref{eq:BobMomDiscFailDir}.

We are now ready to give a formal proof of the achievability results, i.e., \eqref{eq:BobMomDiscFailDirGuess}--\eqref{eq:EveMomDiscFailDirGuess} in the guessing version and \eqref{eq:BobMomDiscFailDir}--\eqref{eq:EveMomDiscFailDir} in the list version. To this end fix $p, \, r \in \{ 1, \ldots, s \}$ satisfying \eqref{bl:condGuessDiscFail} in the guessing version and \eqref{bl:condListDiscFail} in the list version, and let $V$ and $W$ be chance variables taking values in $\setV = \mathbb {F}_{2^{p}}^{ \nu }$ and $\setW = \mathbb {F}_{2^{r}}^{ \nu - \eta }$, respectively. Corollary~\ref{co:equivBunteResultGuessing} implies that there exists some $\{ 0, 1 \}$-valued conditional PMF $\bigdistof { (V, W) = (v, w) \bigl| X = x, Y = y}$ for which
\begin{IEEEeqnarray}{l}
\min_{\guess {}{\cdot|Y,V,W}} \bigEx {}{\guess {}{X|Y,V,W}^\rho} <  1 + 2^{\rho ( \renent {\tirho} { X | Y } - \nu s + \eta r + 1 )}. \label{eq:BobMomDiscFailDir1Guess}
\end{IEEEeqnarray}
Theorem~\ref{th:optTaskEnc} implies that there exists some deterministic task-encoder $\enc { \cdot | Y } \colon \setX \rightarrow \setV \times \setW$ for which
\begin{IEEEeqnarray}{l}
\BigEx {}{\bigl| \setL^Y_{V,W} \bigr|^\rho} < 1 + 2^{\rho ( \renent {\tirho} {X | Y} - \log ( 2^{ \nu s - \eta r } - \log | \setX | - 2 ) + 2 )}, \label{eq:BobMomDiscFailDir1}
\end{IEEEeqnarray}
where $( V,W ) = \enc {X|Y}$. Draw $U$ independently of $(X,Y)$ and uniformly over $\mathbb {F}_{2^r}^\eta$. Choose $G_{\setV} \in \mathbb {F}_{2^{p}}^{ \nu \times \delta }$, $G_{\setW} \in \mathbb {F}_{2^{r}}^{ ( \nu - \eta ) \times \delta }$, and $G_{\setU} \in \mathbb {F}_{2^{r}}^{ \eta \times \delta }$ so that $$G_{\setV}, \quad \begin{pmatrix} G_{\setU} \\ G_{\setW} \end{pmatrix}, \quad G_{\setU}$$ are generator matrices of MDS codes. (This is possible, because both \eqref{bl:condGuessDiscFail} and \eqref{bl:condListDiscFail} imply that
\begin{subequations}
\begin{IEEEeqnarray}{l}
p > 0 \implies 2^p \geq \delta, \\
r > 0 \implies 2^r \geq \delta;
\end{IEEEeqnarray}
\end{subequations}
if $p = 0$, then $V$ can assume but one value, and hence we do not need $G_{\setV}$; and if $r = 0$, then $( W,U )$ can assume but one value, and hence we do not need $G_{\setW}$ and $G_{\setU}$.) Define the chance variables
\begin{subequations} \label{bl:discFailDefMpMr}
\begin{IEEEeqnarray}{rCl}
M_p & = & V \, G_{\setV}, \\
M_r & = & U \, G_{\setU} \oplus W \, G_{\setW} = \begin{pmatrix} U \,\, W \end{pmatrix} \! \begin{pmatrix} G_{\setU} \\ G_{\setW} \end{pmatrix}, \label{eq:discFailDefMr}
\end{IEEEeqnarray}
\end{subequations}
where $M_p$ is computed in the field $\mathbb {F}_{2^{p}}$ and $M_r$ in $\mathbb {F}_{2^{r}}$. Note that $M_p \in \mathbb {F}_{2^{p}}^\delta$ and $M_r \in \mathbb {F}_{2^{r}}^\delta$. Since both \eqref{bl:condGuessDiscFail} in the guessing version and \eqref{bl:condListDiscFail} in the list version imply that $s = p + r$, Alice can choose the $\ell$-th hint to comprise the $\ell$-th components of $M_p$ and $M_r$, so
\begin{equation}
M_{\ell} = \bigl( [ M_p ]_{\ell}, [ M_r ]_{\ell} \bigr), \quad \ell \in \{ 1, \ldots, \delta \}.
\end{equation}
For this choice of the hints Bob can recover $( V,W,U )$ no matter which $\nu$ hints he observes, because $$G_{\setV}, \quad \begin{pmatrix} G_{\setU} \\ G_{\setW} \end{pmatrix}$$ are generator matrices of MDS codes. Hence, in the guessing version \eqref{eq:BobMomDiscFailDirGuess} follows from \eqref{eq:BobMomDiscFailDir1Guess}, and in the list version \eqref{eq:BobMomDiscFailDir} follows from \eqref{eq:BobMomDiscFailDir1}.

The proof of \eqref{eq:EveMomDiscFailDirGuess} and \eqref{eq:EveMomDiscFailDir} is more involved. Recall that Eve observes a size-$\eta$ set $\setE \subset \{ 1, \ldots, \delta \}$ and the components $\rndvecM_\setE$ of $\rndvecM$ indexed by $\setE$. Index the possible sets that $\setE$ could denote by the elements of some size-${\delta \choose \eta}$ set $\setK$, and denote by $\setE (k)$ the set that is indexed by~$k$. The proof of \eqref{eq:EveMomDiscFailDirGuess} and \eqref{eq:EveMomDiscFailDir} builds on the following two intermediate claims, which we prove next:
\begin{enumerate}
\item Eve's ambiguity can be alternatively expressed as
\begin{IEEEeqnarray}{l}
\mathscr A_{\textnormal E} (P_{X,Y}) = \min_{ K, \, \guess {} { \cdot | Y, \rndvecM_{\setE (K)}, K }} \bigEx {}{ \guess {} { X | Y, \rndvecM_{\setE (K)}, K }^\rho },
\end{IEEEeqnarray}
where $K$ is a chance variable of support $\setK$, and where the minimization is over all conditional PMFs of $K$ given $( X, Y, \rndvecM )$ and all guessing functions $\guess {} { \cdot | Y, \rndvecM_{\setE (K)}, K }$.
\item We can assume w.l.g.\ that Eve must guess not only $X$ but the pair $(X,U)$.
\end{enumerate}

We first prove Claim~1, i.e., that
\begin{IEEEeqnarray}{l}
\min_{\guess {\setE} { \cdot | Y, \rndvecM_{\setE} }} \BigEx {}{ \min_{ \setE } \guess {\setE} { X | Y, \rndvecM_{\setE} }^\rho } \nonumber \\
\quad = \min_{ K, \, \guess {} { \cdot | Y, \rndvecM_{\setE (K)}, K }} \bigEx {}{ \guess {} { X | Y, \rndvecM_{\setE (K)}, K }^\rho }. \label{eq:equivExpEveAmbDiscFail}
\end{IEEEeqnarray}
Note that
\begin{IEEEeqnarray}{l}
\min_{ \setE } \guess {\setE} { X | Y, \rndvecM_{\setE} } = \min_{ k } \guess {\setE (k)} { X | Y, \rndvecM_{\setE (k)} };
\end{IEEEeqnarray}
and for any given $\guess {\setE (k)} { \cdot | Y, \rndvecM_{\setE (k)} }, \,\, k \in \setK$, define
\begin{IEEEeqnarray}{l}
K = \argmin_k \guess {\setE (k)} { X | Y, \rndvecM_{\setE (k)} }, \label{eq:defKDiscFail}
\end{IEEEeqnarray}
and introduce the guessing function $\guess {} { \cdot | Y, \rndvecM_{\setE (K)}, K }$ satisfying that, for every $(x,y) \in \setX \times \setY$, $\vecm_{\setE (k)} \in \mathbb F_{2^s}^\eta$, and $k \in \setK$, 
\begin{IEEEeqnarray}{l}
\guess {} { x | y, \vecm_{\setE (k)}, k } = \guess {\setE (k)} { x | y, \vecm_{\setE (k)} }.
\end{IEEEeqnarray}
We then obtain that
\begin{IEEEeqnarray}{l} 
\bigEx {}{ \guess {} { X | Y, \rndvecM_{\setE (K)}, K }^\rho } = \BigEx {}{ \min_{ \setE } \guess {\setE} { X | Y, \rndvecM_{\setE} }^\rho },
\end{IEEEeqnarray}
and consequently that
\begin{IEEEeqnarray}{l} 
\min_{\guess {\setE} { \cdot | Y, \rndvecM_{\setE} }} \BigEx {}{ \min_{ \setE } \guess {\setE} { X | Y, \rndvecM_{\setE} }^\rho } \nonumber \\
\quad \geq \min_{ K, \, \guess {} { \cdot | Y, \rndvecM_{\setE}, K }} \bigEx {}{ \guess {} { X | Y, \rndvecM_{\setE (K)}, K }^\rho }. \label{eq:equivExpEveAmb1DiscFail}
\end{IEEEeqnarray}
To see that equality holds, note that, irrespective of $K$ and $\guess {} { \cdot | Y, \rndvecM_{\setE (K)}, K }$,
\begin{IEEEeqnarray}{l} 
\bigEx {}{ \guess {} { X | Y, \rndvecM_{\setE (K)}, K }^\rho } \geq \BigEx {}{ \min_k \guess {} { X | Y, \rndvecM_{\setE (k)}, k }^\rho }. \label{eq:equivExpEveAmbFor2DiscFail}
\end{IEEEeqnarray}
For any given $\guess {} { \cdot | Y, \rndvecM_{\setE (K)}, K }$ introduce the collection of guessing functions $\guess {\setE (k)} { \cdot | Y, \rndvecM_{\setE (K)} }, \,\, k \in \setK$ that, for every $(x,y) \in \setX \times \setY$ and $\vecm_{\setE (k)} \in \mathbb F_{2^s}^\eta$, satisfy
\begin{IEEEeqnarray}{l} 
\guess {\setE (k)} { x | y, \vecm_{\setE (k)} } = \guess {} { x | y, \vecm_{\setE (k)}, k }.
\end{IEEEeqnarray}
We then obtain from \eqref{eq:equivExpEveAmbFor2DiscFail} that
\begin{IEEEeqnarray}{l} 
\bigEx {}{ \guess {} { X | Y, \rndvecM_{\setE (K)}, K }^\rho } \geq \BigEx {}{ \min_{k} \guess {\setE (k)} { X | Y, \rndvecM_{\setE (k)} }^\rho },
\end{IEEEeqnarray}
and consequently that
\begin{IEEEeqnarray}{l}
\min_{\guess {\setE} { \cdot | Y, \rndvecM_{\setE} }} \BigEx {}{ \min_{ \setE } \guess {\setE} { X | Y, \rndvecM_{\setE} }^\rho } \nonumber \\
\quad \leq \min_{ K, \, \guess {} { \cdot | Y, \rndvecM_{\setE}, K }} \bigEx {}{ \guess {} { X | Y, \rndvecM_{\setE (K)}, K }^\rho }. \label{eq:equivExpEveAmb2DiscFail}
\end{IEEEeqnarray}
From \eqref{eq:equivExpEveAmb1DiscFail} and \eqref{eq:equivExpEveAmb2DiscFail} we conclude that \eqref{eq:equivExpEveAmbDiscFail} holds.

We next prove Claim~2. To this end we shall use Claim~1. Let $K$ be any chance variable of finite support $\setK$, and note that $W$ is deterministic given $(X,Y)$. By \eqref{eq:discFailDefMr}
\begin{IEEEeqnarray}{l}
[ U \, G_{\setU} ]_{\setE (K)} = [M_r]_{\setE (K)} \ominus [ W \, G_{\setV} ]_{\setE (K)},
\end{IEEEeqnarray}
where the computation is in the field $\mathbb {F}_{2^{r}}$. Consequently, $[ U \, G_{\setU} ]_{\setE (K)}$ is deterministic given $( X,Y,\rndvecM_{\setE (K)},K )$. Because $G_{\setU}$ is a generator matrix of an MDS code, and because $| \setE (K) | = \eta$, it follows that $U$ is deterministic given $(X,Y,\rndvecM_{\setE (K)}, K)$, i.e., that there exists some mapping $$g \colon \setX \times \setY \times \mathbb {F}_{2^r}^\eta \times \setK \rightarrow \setU$$ for which
\begin{equation}
U = g ( X,Y,\rndvecM_{\setE (K)},K ). \label{eq:uFunEveObsDiscFail}
\end{equation}
Given any guessing function $\guess {} { \cdot | Y, \rndvecM_{\setE (K)}, K }$ for $X$, introduce some guessing function $\guess {}{\cdot,\cdot|Y,\rndvecM_{\setE (K)},K}$ for $(X,U)$ satisfying that
\begin{IEEEeqnarray}{l}
\bigguess {}{ X, g (X,Y,\rndvecM_{\setE (K)},K) \bigl| Y, \rndvecM_{\setE (K)}, K } = \guess {}{ X | Y, \rndvecM_{\setE (K)}, K },
\end{IEEEeqnarray}
and note that
\begin{IEEEeqnarray}{l}
\guess {}{X,U|Y,\rndvecM_{\setE (K)},K} = \guess {}{X|Y,\rndvecM_{\setE (K)},K}. \label{eq:guessXUDiscFail}
\end{IEEEeqnarray}
This proves Claim~2.

Having established Claims~1 and 2, we are now ready to prove \eqref{eq:EveMomDiscFailDirGuess} and \eqref{eq:EveMomDiscFailDir}:
\begin{IEEEeqnarray}{l}
\min_{\guess {\setE} { \cdot | Y, \rndvecM_{\setE} }} \BigEx {}{ \min_{ \setE } \guess {\setE} { X \left| Y, \rndvecM_{\setE} \right. }^\rho } \nonumber \\
\quad \stackrel{(a)}= \min_{ K, \, \guess {} { \cdot | Y, \rndvecM_{\setE}, K }} \bigEx {}{ \guess {} { X | Y, \rndvecM_{\setE (K)}, K }^\rho } \\
\quad \stackrel{(b)}= \min_{ K, \, \guess {\setE} { \cdot | Y, \rndvecM_{\setE}, K }} \bigEx {}{ \guess {} { X, U | Y, \rndvecM_{\setE (K)}, K }^\rho } \\
\quad \stackrel{(c)}\geq 2^{\rho ( \renent{\tirho}{X,U|Y} - \eta s - \log {\delta \choose \eta} - \log ( 1 + \ln |\setX| ) ) } \\
\quad \stackrel{(d)}\geq 2^{\rho ( \renent{\tirho}{X|Y} - \eta ( s - r ) - \eta \log \delta - \log ( 1 + \ln |\setX| ) ) },
\end{IEEEeqnarray}
where $(a)$ holds by \eqref{eq:equivExpEveAmbDiscFail}; $(b)$ holds by \eqref{eq:guessXUDiscFail}; $(c)$ follows from Corollary~\ref{co:equivBunteResultGuessing} and the fact that $(\rndvecM_{\setE (K)}, K)$ takes values in a set of size $2^{\eta s} \, { \delta \choose \eta }$; and $(d)$ holds because ${ \delta \choose \eta } \leq \delta^\eta$ and
\begin{IEEEeqnarray}{l}
\renent {\tirho}{X, U | Y } \nonumber \\
\quad \stackrel{(e)}= \frac{1}{\rho} \log \sum_{y \in \setY} \Biggl( \sum_{x \in \setX} \sum_{u \in \mathbb {F}_{2^r}^\eta} \bigl( P_{X,Y} ( x,y ) / 2^{\eta r} \bigr)^{\tirho} \Biggr)^{\!\! 1+\rho} \nonumber \\
\quad = \frac{1}{\rho} \log \! \left( \sum_{y \in \setY} \Biggl( \sum_{x \in \setX} P_{X,Y} ( x,y )^{\tirho} \Biggr)^{\!\! 1+\rho} 2^{\rho \eta r} \right) \nonumber \\
\quad = \renent {\tirho} { X | Y } + \eta r, \label{eq:reEntrXUCondYDiscFail}
\end{IEEEeqnarray}
where $(e)$ holds because $U$ is independent of $( X,Y )$ and uniform over the set $\mathbb {F}_{2^r}^\eta$ of size $2^{\eta r}$. This concludes the proof of the achievability results.\\

It remains to establish the converse results, i.e., \eqref{eq:BobMomDiscFailConvGuess}--\eqref{eq:EveMomDiscFailConvGuess} in the guessing version and \eqref{eq:BobMomDiscFailConv}--\eqref{eq:EveMomDiscFailConv} in the list version. To this end we first note that
\begin{subequations} \label{bl:BobAltAmbDiscFail}
\begin{IEEEeqnarray}{rCl}
\mathscr A^{(g)}_{\textnormal B} (P_{X,Y}) & = & \min_{\guess {\setB} { \cdot | Y, \rndvecM_{\setB} }} \BigEx {}{\max_{\setB} \guess {\setB}{ X | Y,\rndvecM_{\setB} }^\rho} \nonumber \\
& \geq & \min_{\guess {\setB} { \cdot | Y, \rndvecM_{\setB} }} \max_{\setB} \bigEx {}{ \guess {\setB}{X | Y,\rndvecM_{\setB} }^\rho}, \label{eq:BobAltAmbDiscFailGuess} \\
\mathscr A^{(l)}_{\textnormal B} (P_{X,Y}) & = & \BigEx {}{\max_{ \setB } \bigl| \setL^Y_{\rndvecM_{\setB}} \bigr|^\rho} \nonumber \\
& \geq & \max_{ \setB } \BigEx {}{\bigl| \setL^Y_{\rndvecM_{\setB}} \bigr|^\rho}. \label{eq:BobAltAmbDiscFail}
\end{IEEEeqnarray}
\end{subequations}
Because $\setB \subseteq  \{ 1, \ldots, \delta \}$ is a size-$\nu$ set, in the guessing version \eqref{eq:BobMomDiscFailConvGuess} follows from \eqref{eq:BobAltAmbDiscFailGuess} and Corollary~\ref{co:equivBunteResultGuessing}, and in the list version \eqref{eq:BobMomDiscFailConv} follows from \eqref{eq:BobAltAmbDiscFail} and Theorem~\ref{th:optTaskEnc}. To prove \eqref{eq:EveMomDiscFailConvGuess} and \eqref{eq:EveMomDiscFailConv}, we first note that
\begin{IEEEeqnarray}{rCl}
A_{\textnormal E} (P_{X,Y}) & = & \min_{\guess {\setE}{ \cdot | Y,\rndvecM_{\setE} }} \BigEx {}{\min_{\setE} \guess {\setE}{ X | Y,\rndvecM_{\setE} }^\rho} \\
& \leq & \min_{\setE, \, \guess {\setE}{ \cdot | Y,\rndvecM_{\setE} }} \bigEx {}{ \guess {\setE}{ X | Y, \rndvecM_{\setE} }^\rho}. \label{eq:EveAltAmbDiscFail}
\end{IEEEeqnarray}
Corollary~\ref{co:impGuess} implies that, for every size-$\nu$ set $\setB \subseteq \{ 1, \ldots, \delta \}$ and every size-$\eta$ set $\setE \subset \setB$,
\begin{IEEEeqnarray}{l}
\min_{\guess {\setB} { \cdot | Y, \rndvecM_{\setB} }} \bigEx {}{\guess {\setB}{ X | Y,\rndvecM_{\setB} }^\rho} \geq 2^{ - \rho (\nu-\eta) s} \min_{\guess {\setE}{ \cdot | Y,\rndvecM_{\setE} }} \bigEx {}{\guess {\setE}{X | Y,\rndvecM_{\setE} }^\rho}; \label{eq:EveAltAmbDiscFailBoundBobAmb}
\end{IEEEeqnarray}
and, because $$\min_{\guess {\setB} { \cdot | Y, \rndvecM_{\setB} }} \bigEx {}{\guess {\setB}{ X | Y,\rndvecM_{\setB} }^\rho} \leq \BigEx {}{\bigl| \setL^Y_{\rndvecM_{\setB}} \bigr|},$$ \eqref{eq:EveAltAmbDiscFail} and \eqref{eq:EveAltAmbDiscFailBoundBobAmb} imply that in both versions Eve's ambiguity exceeds Bob's by at most a factor of $2^{\rho (\nu-\eta) s}$, i.e., $\mathscr A_{\textnormal E} (P_{X,Y}) \leq 2^{\rho (\nu-\eta) s} \mathscr A^{(g)}_{\textnormal B} (P_{X,Y})$ and $\mathscr A_{\textnormal E} (P_{X,Y}) \leq 2^{\rho (\nu-\eta) s} \mathscr A^{(l)}_{\textnormal B} (P_{X,Y})$. Since Eve can ignore the hints that she observes and guess $X$ based on $Y$ alone, we obtain from Theorem~\ref{th:optGuessFun} that, for every size-$\eta$ set $\setE \subset \{ 1, \ldots, \delta \}$,
\begin{IEEEeqnarray}{l}
\min_{\guess {\setE}{ \cdot | Y,\rndvecM_{\setE} }} \bigEx {}{\guess {\setE}{X | Y,\rndvecM_{\setE} }^\rho} \leq 2^{\rho \renent {\tirho}{X|Y}}; \label{eq:EveAltAmbDiscFailBoundOnlyY}
\end{IEEEeqnarray}
and \eqref{eq:EveAltAmbDiscFail} and \eqref{eq:EveAltAmbDiscFailBoundOnlyY} imply that in both versions Eve's ambiguity cannot exceed $2^{\rho \renent {\tirho}{X|Y}}$, i.e., $A_{\textnormal E} (P_{X,Y}) \leq 2^{\rho \renent {\tirho}{X|Y}}$. This concludes the proof of \eqref{eq:EveMomDiscFailConvGuess} and \eqref{eq:EveMomDiscFailConv} and consequently that of the converse results.
\end{proof}

\section{A Proof of Corollary~\ref{co:discFailSimp}} \label{app:pfCoDiscFail}

\begin{proof}
For the guessing version, the results in~\eqref{eq:discFailSimpBobGuess}--\eqref{eq:discFailSimpEveGuess} follow from Theorem~\ref{th:discFailGuess} if we let
\begin{IEEEeqnarray}{rCl}
\tilde r & = & \frac{\nu s + \rho^{-1} \log (\mathscr U_{\textnormal B} - 1) - \renent {\tirho} { X | Y } - 1}{\eta}, \\
r & = & \begin{cases} 0 &\lfloor \tilde r \rfloor \in ( -\infty, \log \delta ), \\ \lfloor \tilde r \rfloor &\lfloor \tilde r \rfloor \in [ \log \delta, s - \log \delta ), \\ s - \lceil \log \delta \rceil &\lfloor \tilde r \rfloor \in [ s - \log \delta, s), \\ s &\lfloor \tilde r \rfloor \in [s, \infty), \end{cases} \label{eq:choiceRSimpVDiscFail} \\
p & = & s - r,
\end{IEEEeqnarray}
and note that $$r \neq s \implies \tilde r - r < \log \delta + 1.$$

To obtain the results in~\eqref{eq:BobAmbSimpBList}--\eqref{eq:EveAmbSimpBList} for the list version, let
\begin{IEEEeqnarray}{l}
\tilde r = \frac{\nu s - \log \Bigl( 2^{ \renent {\tirho} { X | Y } - \frac{1}{\rho} \log ( \mathscr U_{\textnormal B} - 1 ) + 2 } + \log | \setX | + 2 \Bigr) }{\eta},
\end{IEEEeqnarray}
and choose $r$ as in \eqref{eq:choiceRSimpVDiscFail}. Then, \eqref{eq:BobMomDiscFailDir} implies that Bob's ambiguity satisfies \eqref{eq:BobAmbSimpBList}. Since $$r \neq s \implies \tilde r - r < \log \delta + 1,$$ we obtain from \eqref{eq:EveMomDiscFailDir} that, if $r \neq s$, then
\begin{IEEEeqnarray}{rCl}
\mathscr A_{\text{E}} ( P_{X,Y} ) & > & 2^{\rho ( \renent {\tirho} { X | Y } + (\nu - \eta) s - 2 \eta \log \delta - \eta - \log ( 1 + \ln |\setX| ) ) } \nonumber \\
& & \times \Bigl( 2^{ \renent {\tirho} { X | Y } - \frac{1}{\rho} \log ( \mathscr U_{\textnormal B} - 1 ) + 2 } + \log | \setX | + 2 \Bigr)^{-\rho}. \label{eq:EbeAmbSimpDiscFail1}
\end{IEEEeqnarray}
Because $$\frac{1}{a + b} \geq \frac{1}{2 a} \wedge \frac{1}{2 b}, \quad a, \, b > 0, $$ the second factor satisfies the lower bound
\begin{IEEEeqnarray}{l}
\Bigl( 2^{ \renent {\tirho} { X | Y } - \frac{1}{\rho} \log ( \mathscr U_{\textnormal B} - 1 ) + 2 } + \log | \setX | + 2 \Bigr)^{-\rho} \nonumber \\
\quad \geq 2^{ -\rho ( \renent {\tirho} { X | Y } - \frac{1}{\rho} \log ( \mathscr U_{\textnormal B} - 1 ) + 3 )} \wedge \bigl( 2 ( \log | \setX | + 2 ) \bigr)^{-\rho}. \label{eq:EbeAmbSimpDiscFail2}
\end{IEEEeqnarray}
We are now ready to conclude the proof of \eqref{eq:EveAmbSimpBList}: if $r \neq s$, then \eqref{eq:EveAmbSimpBList} follows from \eqref{eq:EbeAmbSimpDiscFail1} and \eqref{eq:EbeAmbSimpDiscFail2}; and if $r = s$, then \eqref{eq:EveMomDiscFailDir} implies that
\begin{IEEEeqnarray}{l}
\mathscr A_{\textnormal E} ( P_{X,Y} ) \geq 2^{\rho ( \renent {\tirho} { X | Y } - \eta \log \delta - \log ( 1 + \ln | \setX | ) )}
\end{IEEEeqnarray}
and consequently that \eqref{eq:EveAmbSimpBList} holds.
\end{proof}

\section{A Proof of Theorem~\ref{th:SiEqSBest}} \label{app:pfThSiEqSBest}

\begin{proof}
If we choose $\setB =  \{ 1, \ldots, \nu \}$, then in the guessing version \eqref{eq:BobMomDistStorConvSiGuess} follows from \eqref{eq:BobAltAmbDiscFailGuess} and Corollary~\ref{co:equivBunteResultGuessing}, and in the list version \eqref{eq:BobMomDistStorSiConv} follows from \eqref{eq:BobAltAmbDiscFail} and Theorem~\ref{th:optTaskEnc}. For $\setB = \{ 1, \ldots, \nu \}$ and $\setE = \{ \nu - \eta + 1, \ldots, \nu \}$, Corollary~\ref{co:impGuess} implies that, 
\begin{IEEEeqnarray}{l}
\min_{\guess {\setB} { \cdot | Y, \rndvecM_{\setB} }} \bigEx {}{\guess {\setB}{ X | Y,\rndvecM_{\setB} }^\rho} \geq 2^{ - \rho \sum^{\eta - \nu}_{\ell = 1} s_\ell} \min_{\guess {\setE}{ \cdot | Y,\rndvecM_{\setE} }} \bigEx {}{\guess {\setE}{X | Y,\rndvecM_{\setE} }^\rho}. \label{eq:EveAltAmbDiscFailBoundBobAmbSi}
\end{IEEEeqnarray}
Since $$\min_{\guess {\setB} { \cdot | Y, \rndvecM_{\setB} }} \bigEx {}{\guess {\setB}{ X | Y,\rndvecM_{\setB} }^\rho} \leq \BigEx {}{\bigl| \setL^Y_{\rndvecM_{\setB}} \bigr|},$$ \eqref{eq:EveAltAmbDiscFail} and \eqref{eq:EveAltAmbDiscFailBoundBobAmbSi} imply that in both versions Eve's ambiguity exceeds Bob's by at most a factor of $2^{ \rho \sum^{\eta - \nu}_{\ell = 1} s_\ell}$. That is, $$\mathscr A_{\textnormal E} (P_{X,Y}) \leq 2^{ \rho \sum^{\eta - \nu}_{\ell = 1} s_\ell} \mathscr A^{(g)}_{\textnormal B} (P_{X,Y})$$ and $$\mathscr A_{\textnormal E} (P_{X,Y}) \leq 2^{ \rho \sum^{\eta - \nu}_{\ell = 1} s_\ell} \mathscr A^{(l)}_{\textnormal B} (P_{X,Y}).$$ Moreover, \eqref{eq:EveAltAmbDiscFail} and \eqref{eq:EveAltAmbDiscFailBoundOnlyY} imply that in both versions Eve's ambiguity cannot exceed $2^{\rho \renent {\tirho}{X|Y}}$. That is, $$A_{\textnormal E} (P_{X,Y}) \leq 2^{\rho \renent {\tirho}{X|Y}},$$ which concludes the proof of \eqref{bl:EveMomDistStorConvSi}.
\end{proof}

\section{A Proof of Theorem~\ref{th:asympDiscFail}} \label{app:pfThAsympDiscFail}

\begin{proof}
We first prove \eqref{eq:privExpDiscFail}. If $\nu R_s < \renent {\tirho}{ \rndvecX | \rndvecY }$, then \eqref{eq:BobMomDiscFailConvGuess} in the guessing version and \eqref{eq:BobMomDiscFailConv} in the list version imply that the privacy-exponent is negative infinity. We hence assume that $\nu R_s > \renent {\tirho}{ \rndvecX | \rndvecY }$.

We start by showing that the privacy-exponent cannot exceed the RHS of \eqref{eq:privExpDiscFail}. To this end, suppose that \eqref{eq:bobAmbTo1} holds and consequently
\begin{IEEEeqnarray}{l} \label{eq:amibExponBob}
\limsup_{n \rightarrow \infty} \frac{\log \bigl( \mathscr A_{\textnormal B} (P_{X^n,Y^n}) \bigr)}{n} = 0.
\end{IEEEeqnarray}
Combining~\eqref{eq:EveMomDiscFailConvGuess} with~\eqref{eq:amibExponBob} in the guessing version and \eqref{eq:EveMomDiscFailConv} in the list version implies that
\begin{IEEEeqnarray}{l}
\limsup_{n \rightarrow \infty} \frac{\log \bigl( \mathscr A_{\textnormal E} ( P_{X^n,Y^n} ) \bigr) }{n} \leq \rho \bigl( R_s ( \nu - \eta ) \wedge \renent {\tirho}{ \rndvecX | \rndvecY } \bigr).
\end{IEEEeqnarray}
Hence, the privacy-exponent cannot exceed the RHS of \eqref{eq:privExpDiscFail}.

We next show that the privacy-exponent cannot be smaller than the RHS of \eqref{eq:privExpDiscFail}. To this end fix $0 < \epsilon < \nu R_s - \renent {\tirho}{ \rndvecX | \rndvecY }$ and let
\begin{equation}
\mathscr U_{\textnormal B} (n) = 1 + 2^{-n \epsilon}.
\end{equation}
Note that $\mathscr U_{\text B} (n)$ converges to one as $n$ tends to infinity. By Corollary~\ref{co:discFailSimp} we can guarantee that Bob's ambiguity not exceed $\mathscr U_{\textnormal B} (n)$ whenever $n$ is sufficiently large and that
\begin{IEEEeqnarray}{l}
\liminf_{n \rightarrow \infty} \frac{\log \bigl( A_{\textnormal E} ( P_{X^n,Y^n} ) \bigr)}{n} \geq \rho \Bigl( \bigl( R_s ( \nu - \eta ) - \epsilon \bigr) \wedge \renent {\tirho}{ \rndvecX | \rndvecY } \Bigr).
\end{IEEEeqnarray}
By letting $\epsilon$ tend to zero we thus find that the privacy-exponent cannot be smaller than the RHS of \eqref{eq:privExpDiscFail}.\\

To prove \eqref{eq:privExpDiscFailEB}, we first note that if $\nu R_s < \renent {\tirho}{ \rndvecX | \rndvecY } - \rho^{-1} E_{\textnormal B}$, then \eqref{eq:BobMomDiscFailConvGuess} in the guessing version and \eqref{eq:BobMomDiscFailConv} in the list version imply that the modest privacy-exponent is negative infinity. We hence assume that $\nu R_s \geq - \rho^{-1} E_{\textnormal B}$.

We start by showing that the modest privacy-exponent cannot exceed the RHS of \eqref{eq:privExpDiscFailEB}. To this end, suppose that \eqref{eq:bobAmbToEB} holds. Due to~\eqref{eq:EveMomDiscFailConvGuess} in the guessing version and~\eqref{eq:EveMomDiscFailConv} in the list version, it follows that
\begin{IEEEeqnarray}{l}
\limsup_{n \rightarrow \infty} \frac{\log \bigl( \mathscr A_{\textnormal E} ( P_{X^n,Y^n} ) \bigr)}{n} \leq \bigl( \rho R_s ( \nu - \eta ) + E_{\textnormal B} \bigr) \wedge \rho \renent {\tirho}{ \rndvecX | \rndvecY }.
\end{IEEEeqnarray}
Hence, the privacy-exponent cannot exceed the RHS of \eqref{eq:privExpDiscFailEB}.

We next show that the privacy-exponent cannot be smaller than the RHS of \eqref{eq:privExpDiscFailEB}. To this end let
\begin{equation}
\mathscr U_{\textnormal B} (n) = 2^{\rho n E_{\textnormal B}}.
\end{equation}
By Corollary~\ref{co:discFailSimp} we can guarantee that Bob's ambiguity not exceed $\mathscr U_{\textnormal B} (n)$ whenever $n$ is sufficiently large and that
\begin{IEEEeqnarray}{l}
\liminf_{n \rightarrow \infty} \frac{\log \bigl( A_{\textnormal E} ( P_{X^n,Y^n} ) \bigr)}{n} \geq \bigl( \rho R_s ( \nu - \eta ) + E_{\textnormal B} \bigr) \wedge \rho \renent {\tirho}{ \rndvecX | \rndvecY }.
\end{IEEEeqnarray}
This proves that the modest privacy-exponent cannot be smaller than the RHS of \eqref{eq:privExpDiscFailEB}.
\end{proof}

\section{A Proof of Lemma~\ref{le:ImproveGuessRD}} \label{app:pfLeImproveGuessRD}

\begin{proof}
To prove \eqref{eq:impGuessMaxRD}, fix some optimal guessing function $\hatguessD {} \star { \cdot | Y^n, Z}$ with corresponding success function $\guessD \Delta \star { \cdot | Y^n, Z}$. The success function $\guessD \Delta \star { \cdot | Y^n, Z}$ minimizes $\bigEx {}{\guessD \Delta \star { X | Y^n, Z}^\rho}$. Let $\psifun { \cdot | Y^n, Z }$ be the corresponding reconstruction function, i.e., the unique mapping satisfying that
\begin{equation}
 \psifun { \vecx | \vecy, z } = \hat \vecx \iff \guessD \Delta \star { \vecx | \vecy, z } = \hatguessD {} \star { \hat \vecx | \vecy, z }, \,\, \forall \, ( \vecx, \hat \vecx, \vecy, z ) \in \setX^n \times \hat \setX^n \times \setY^n \times \setZ. \label{eq:psiFunStarGivYZ}
\end{equation}
For every $\vecy \in \setY^n$ consider a guessing order on $\hat \setX^n$ where we first guess the elements of the set $$\Bigl\{ \hat \vecx \in \hat \setX^n \colon \min_{z \in \setZ} \hatguessD {} \star { \hat \vecx | \vecy, z } = 1 \Bigr\}$$ in some arbitrary order followed by the elements of the set $$\Bigl\{ \hat \vecx \in \hat \setX^n \colon \min_{z \in \setZ} \hatguessD {} \star { \hat \vecx | \vecy, z } = 2 \Bigr\},$$ and where we continue until concluding by guessing the elements of $\hat \setX^n$ for which $\min_{z \in \setZ} \hatguessD {} \star { \hat \vecx | \vecy, z }$ is maximum. Let $\hatguessD {} {} { \cdot | Y^n }$ be the corresponding guessing function. For every $\hat \vecx, \, \hat \vecx^\prime \in \hat \setX^n$ and $\vecy \in \setY^n$ a necessary condition for $\hatguessD {} {}{ \hat \vecx^\prime \bigl| \vecy } \leq \hatguessD {} {}{ \hat \vecx \bigl| \vecy }$ is that $$\min_{z \in \setZ} \hatguessD {} \star { \hat \vecx^\prime \bigl| \vecy, z } \leq \min_{z \in \setZ} \hatguessD {} \star { \hat \vecx \bigl| \vecy, z }.$$ In addition, for every $z^\prime \in \setZ$ the mapping $\hatguessD {} \star { \cdot \bigl| \vecy, z^\prime } \colon \hat \setX^n \rightarrow \bigl[ 1 : |\hat \setX|^n \bigr]$ is one-to-one, and consequently the number of $\hat \vecx^\prime \in \hat \setX^n$ satisfying $$\hatguessD {} \star { \hat \vecx^\prime \bigl| \vecy, z^\prime } \leq \min_{z \in \setZ} \hatguessD {} \star { \hat \vecx \bigl| \vecy, z }$$ is $\min_{z \in \setZ} \hatguessD {} \star { \hat \vecx \bigl| \vecy, z }$. Consequently, 
\begin{IEEEeqnarray}{rCl}
\bighatguessD {}{}{ \psifun{ X^n | Y^n, Z } \bigl| Y^n } & \leq & | \setZ | \min_{z \in \setZ} \bighatguessD {} \star { \psifun{ X^n | Y^n, Z } \bigl| Y^n, z } \\
& \leq & | \setZ | \, \bighatguessD {} \star { \psifun{ X^n | Y^n, Z } \bigl| Y^n, Z }. \label{eq:impHatGuessMaxRD}
\end{IEEEeqnarray}
From \eqref{eq:impHatGuessMaxRD} it follows that the success function $\guessD \Delta {}{ \cdot | Y^n }$ corresponding to $\hatguessD {} {} { \cdot | Y^n }$ satisfies
\begin{IEEEeqnarray}{rCl}
\guessD \Delta {}{ X^n | Y^n } & \stackrel{(a)}\leq & \bighatguessD {}{}{ \psifun{ X^n | Y^n, Z } \bigl| Y^n } \\
& \stackrel{(b)}\leq & | \setZ | \, \bighatguessD {} \star { \psifun{ X^n | Y^n, Z } \bigl| Y^n, Z } \\
& \stackrel{(c)}= & | \setZ | \, \guessD \Delta \star { X^n | Y^n, Z },
\end{IEEEeqnarray}
where $(a)$ holds because $d^{(n)} \bigl( X^n, \psifun{X^n|Y^n,Z} \bigr) \leq \Delta$; $(b)$ holds by \eqref{eq:impHatGuessMaxRD}; and $(c)$ holds because $\psifun{ \cdot | Y^n, Z }$ satisfies \eqref{eq:psiFunStarGivYZ}. Since $\hatguessD {} \star { \cdot | Y^n, Z}$ is an optimal guessing function, this concludes the proof of \eqref{eq:impGuessMaxRD}.

To prove \eqref{eq:impGuessMinRD}, fix some optimal guessing function $\hatguessD {} \star { \cdot | Y^n }$ with a corresponding success function $\guessD \Delta \star { \cdot | Y^n }$. The success function $\guessD \Delta \star { \cdot | Y^n }$ minimizes $\bigEx {}{\guessD \Delta \star { X | Y^n }^\rho}$. Let $\psifun { \cdot | Y^n }$ be the corresponding reconstruction function for which \eqref{eq:psiDefinition} holds when we substitute $\hatguessD {} \star { \hat \vecx | \vecy }$ for $\hatguessD {}{}{ \hat \vecx | \vecy }$ and $\guessD \Delta \star { \vecx | \vecy }$ for $\guessD \Delta {} { \vecx | \vecy }$ in \eqref{eq:psiDefinition}. Let $f \colon \hat \setX^n \times \setY^n \rightarrow \setZ$ be some mapping for which $f ( \hat \vecx, \vecy ) = f ( \hat \vecx^\prime, \vecy )$ implies either $\bigl\lceil \hatguessD \Delta \star { \hat \vecx | \vecy } / | \setZ | \bigr\rceil \neq \bigl\lceil \hatguessD \Delta \star { \hat \vecx^\prime | \vecy } / | \setZ | \bigr\rceil$ or $\hat \vecx = \hat \vecx^\prime$. The mapping $f$ could be any mapping for which, for every $(\hat \vecx, \vecy) \in \hat \setX^n \times \setY^n$, $f ( \hat \vecx, \vecy )$ is---up to relabeling the elements of $\setZ$---the remainder of the Euclidean division of $\hatguessD {} \star {\hat \vecx | \vecy} - 1$ by $| \setZ |$. Define the chance variable $\hat X^n = \psifun { X^n | Y^n }$, which takes values in $\hat \setX^n$. Lemma~\ref{le:ImproveGuess} implies that for $Z = f ( \hat X^n, Y^n )$ there exists some guessing function $\hatguessD {}{}{ \cdot | Y^n, Z }$ for $\hat X^n$ for which
\begin{IEEEeqnarray}{l}
\bigEx {}{\hatguessD {}{}{ \hat X^n | Y^n, Z }^\rho} = \BigEx {}{ \bigl\lceil \hatguessD {}{}{ \hat X^n | Y^n } / |\setZ| \bigr\rceil^\rho }.
\end{IEEEeqnarray}
In fact, in the proof of Lemma~\ref{le:ImproveGuess} it is shown that there exists some guessing function $\hatguessD {}{}{ \cdot | Y^n, Z }$ for $\hat X^n$ for which
\begin{IEEEeqnarray}{l}
\hatguessD {}{}{ \hat X^n | Y^n, Z } =  \bigl\lceil \hatguessD {}{}{ \hat X^n | Y^n } / |\setZ| \bigr\rceil. \label{eq:hatFunRelWWZRD}
\end{IEEEeqnarray}
Let $\hatguessD {}{}{ \cdot | Y^n, Z }$ be a guessing function as in \eqref{eq:hatFunRelWWZRD} with corresponding success function $\guessD \Delta {} { \cdot | Y^n, Z }$. Note that
\begin{IEEEeqnarray}{rCl}
\guessD \Delta {} { X^n | Y^n, Z } & \stackrel{(a)}\leq & \bighatguessD {} {} { \psifun { X^n | Y^n } \bigl| Y^n, Z } \\
& \stackrel{(b)}= & \Bigl\lceil \bighatguessD {} \star { \psifun { X^n | Y^n } \bigl| Y^n } / | \setZ | \Bigr\rceil \\
& \stackrel{(c)}= & \bigl\lceil \guessD \Delta \star { X^n | Y^n } / | \setZ | \bigr\rceil,
\end{IEEEeqnarray}
where $(a)$ holds because $d^{(n)} \bigl( X^n, \psifun { X^n | Y^n } \bigr) \leq \Delta$; $(b)$ holds because $\hat X^n = \psifun { X^n | Y^n }$ and by \eqref{eq:hatFunRelWWZRD}; and $(c)$ holds because $\psifun { \cdot | Y^n }$ satisfies \eqref{eq:psiDefinition} when we substitute $\hatguessD {} \star { \hat \vecx | \vecy }$ for $\hatguessD {}{}{ \hat \vecx | \vecy }$ and $\guessD \Delta \star { \vecx | \vecy }$ for $\guessD \Delta {} { \vecx | \vecy }$ in \eqref{eq:psiDefinition}. Since $\hatguessD {} \star { \cdot | Y^n }$ is an optimal guessing function, this concludes the proof of \eqref{eq:impGuessMinRD}.
\end{proof}

\section{A Proof of Theorem~\ref{th:relGuessEncRD}} \label{app:pfThRelGuessEncRD}

\begin{proof}
As to the first part, suppose we are given a stochastic task-encoder \eqref{eq:condPMFRelGuessEncRD} and a decoder with lists $\{ \setL^\vecy_z \}$ satisfying \eqref{eq:decRDStoch}. For every $\vecy \in \setY^n$ order the lists $\set{\setL^\vecy_z}_{z \in \setZ}$ in increasing order of their cardinalities, and order the elements in each list in some arbitrary way. Now consider the guessing order where we first guess the elements of the first (and smallest) list in their respective order followed by those elements in the second list that have not yet been guessed (i.e., that are not contained in the first list). We continue until concluding by guessing those elements of the last (and longest) list that have not been previously guessed. Let $\hatguessD {}{}{ \cdot | Y^n }$ be the corresponding guessing function, let $\guessD \Delta {}{ \cdot | Y^n }$ be its success function, and let $\psifun{ \cdot | Y^n }$ be its reconstruction function (which satisfies \eqref{eq:psiDefinition}). Observe that
\begin{IEEEeqnarray}{rCl} 
\bigEx {}{ \guessD \Delta {}{ X^n | Y^n }^\rho } & \stackrel{(a)}= & \BigEx {}{ \bighatguessD {} {} {\psifun { X^n | Y^n } \bigl| Y^n }^\rho } \\
& = & \sum_{\vecx,\vecy} P^n_{X,Y} ( \vecx,\vecy ) \, \Bigl| \Bigl\{ \hat \vecx \colon \hatguessD {} {}{ \hat \vecx | \vecy } \leq \bighatguessD {} {}{ \psifun{ \vecx | \vecy } \bigl| \vecy } \Bigr\} \Bigr|^\rho \\
& \stackrel{(b)}\leq & \sum_{\vecx,\vecy} P^n_{X,Y} ( \vecx,\vecy ) \, | \setZ |^\rho \min_{z \colon \psifun{ \vecx | \vecy } \in \setL^\vecy_z } | \setL^\vecy_z |^\rho \\
&\stackrel{(c)}\leq & | \setZ |^\rho \BigEx {}{\bigl| \setL^{Y^n}_Z \bigr|^\rho},
\end{IEEEeqnarray}
where $(a)$ holds because $\psifun { \cdot | Y^n }$ satisfies \eqref{eq:psiDefinition}; $(b)$ holds because for every $\vecx \in \setX^n$, $\hat \vecx \in \hat \setX^n$, and $\vecy \in \setY^n$, a necessary condition for $\hatguessD {} {}{ \hat \vecx | \vecy } \leq \bighatguessD {} {}{ \psifun{ \vecx | \vecy }\bigr| \vecy }$ is that $\hat \vecx \in \setL^\vecy_{\tilde z}$ for some $\tilde z \in \setZ$ satisfying $| \setL^\vecy_{\tilde z} | \leq \min_{z \colon \psifun{ \vecx | \vecy } \in \setL^\vecy_z} | \setL^\vecy_z |$, and because the number of lists whose size does not exceed $\min_{z \colon \psifun{ \vecx | \vecy } \in \setL^\vecy_z} | \setL^\vecy_z |$ is at most $| \setZ |$; and $(c)$ is true because by \eqref{eq:decRDStoch} the list $\setL^{Y^n}_Z$ contains a reconstruction $\hat \vecx \in \hat \setX^n$ of $X^n$ that satisfies the fidelity criterion \eqref{eq:fidCrit}, and because \eqref{eq:psiDefinition} implies that
\begin{equation}
\bighatguessD {} {} {\psifun { \vecx | \vecy } \bigl| \vecy } \leq \hatguessD {} {} { \hat \vecx | \vecy }, \,\, \forall \, \hat \vecx \textnormal{ s.t.\ } d^{(n)} ( X^n, \hat \vecx ) \leq \Delta,
\end{equation}
and consequently that
\begin{equation}
\min_{z \colon \psifun{ \vecx | \vecy } \in \setL^\vecy_z } | \setL^\vecy_z | \leq \min_{z \colon \hat \vecx \in \setL^\vecy_z } | \setL^\vecy_z |, \,\, \forall \, \hat \vecx \textnormal{ s.t.\ } d^{(n)} ( X^n, \hat \vecx ) \leq \Delta.
\end{equation}
This concludes the proof of \eqref{eq:listToGuessRD}.

As to the second part, suppose we are given a positive integer $\omega \leq | \hat \setX |^n$ satisfying \eqref{eq:relCardMandVRD} and a guessing function $\hatguessD {} {} { \cdot | Y^n }$ with corresponding success function $\guessD \Delta {} { \cdot | Y^n }$ and reconstruction function $\psifun { \cdot | Y^n }$ satisfying \eqref{eq:psiDefinition}. Define the chance variable $\hat X^n = \psifun { X^n | Y^n }$, which takes values in $\hat \setX^n$. Theorem~\ref{th:relGuessEnc} implies that $\hatguessD {} {} { \cdot | Y^n }$ and $\omega$ induce a $\{ 0, 1 \}$-valued conditional PMF
\begin{IEEEeqnarray}{l}
\distof { Z = z | \hat X^n = \hat \vecx, Y^n = \vecy }, \,\, \forall \, ( \hat \vecx, \vecy, z ) \in \hat \setX^n \times \setY^n \times \setZ,
\end{IEEEeqnarray}
whose associated decoding lists
 \begin{IEEEeqnarray}{l}
\hat \setL^\vecy_z = \bigl\{ \hat \vecx \in \hat \setX^n \colon \distof { \hat X^n = \hat \vecx | Y^n = \vecy, Z = z } > 0 \bigr\}, \,\, \forall \, (\vecy,z) \in \setY^n \times \setZ \label{eq:pfThRelGuessEncRDDefListsHatVecX}
\end{IEEEeqnarray}
satisfy
\begin{IEEEeqnarray}{l}
\BigEx {}{\bigl| \hat \setL^{Y^n}_Z \bigr|^\rho} \leq \BigEx {}{ \bigl\lceil \hatguessD {} {} { \hat X^n | Y^n } / \omega \bigr\rceil^\rho }. \label{eq:pfThRelGuessEncRDHatListsSatisfy}
\end{IEEEeqnarray}
Define the $\{ 0,1 \}$-valued conditional PMF
\begin{IEEEeqnarray}{l}
\distof {Z = z |X^n = \vecx, Y^n = \vecy} \nonumber \\
\quad = \bigdistof { Z = z \bigl| \hat X^n = \psifun { \vecx | \vecy }, Y^n = \vecy }, \,\, \forall \, (\vecx,\vecy,z) \in \setX^n \times \setY^n \times \setZ, \label{eq:pfThRelGuessEncRDDefCondPMF}
\end{IEEEeqnarray}
and the lists
\begin{IEEEeqnarray}{l}
\setL^\vecy_z = \hat \setL^\vecy_z, \,\, \forall \, (\vecy,z) \in \setY^n \times \setZ. \label{eq:pfThRelGuessEncRDDefLists}
\end{IEEEeqnarray}
Because $\hat X^n = \psifun { X^n | Y^n }$, \eqref{eq:pfThRelGuessEncRDDefListsHatVecX}, \eqref{eq:pfThRelGuessEncRDDefCondPMF}, and \eqref{eq:pfThRelGuessEncRDDefLists} imply that
\begin{equation}
\psifun { X^n | Y^n } \in \setL^{Y^n}_Z. \label{eq:pfThRelGuessEncRDPsifunInList}
\end{equation}
Since $$d^{(n)} \bigl( \vecx, \psifun { \vecx | \vecy } \bigr) \leq \Delta, \,\, \forall \, (\vecx, \vecy) \in \setX^n \times \setY^n,$$ this implies that the decoding lists $\{ \setL^\vecy_z \}$ satisfy \eqref{eq:decRDStoch}. Hence, \eqref{eq:pfThRelGuessEncRDDefCondPMF} is a deterministic task-encoder  (whose conditional PMF \eqref{eq:condPMFRelGuessEncRD} is $\{ 0,1 \}$-valued) for which the decoder with lists \eqref{eq:pfThRelGuessEncRDDefLists} satisfies \eqref{eq:decRDStoch}. We are now ready to conclude the proof of \eqref{eq:guessToListRD}:
\begin{IEEEeqnarray}{rCl}
\BigEx {}{\bigl| \setL^{Y^n}_Z \bigr|^\rho} & \stackrel{(a)}= & \BigEx {}{\bigl| \hat \setL^{Y^n}_Z \bigr|^\rho} \\
& \stackrel{(b)}\leq & \BigEx {}{\bigl\lceil \hatguessD {} {} { \hat X^n | Y^n } / \omega \bigr\rceil^\rho} \\
& \stackrel{(c)}= & \BigEx {}{\bigl\lceil \guessD \Delta {} { X^n | Y^n } / \omega \bigr\rceil^\rho},
\end{IEEEeqnarray}
where $(a)$ holds by \eqref{eq:pfThRelGuessEncRDDefLists}; $(b)$ holds by \eqref{eq:pfThRelGuessEncRDHatListsSatisfy}; and $(c)$ holds because $\hat X^n = \psifun { X^n | Y^n }$, where $\psifun { \cdot | Y^n }$ satisfies \eqref{eq:psiDefinition}.
\end{proof}

\section{A Proof of Theorem~\ref{th:secrecyRD}}\label{app:pfSecrecyRD}

\begin{proof}
We first prove \eqref{eq:privExpBob1RD}. If $R_1 + R_2 < \RDexp$, then Corollary~\ref{co:equivBunteResultGuessingRD} in the guessing version and Corollary~\ref{co:optTaskEncRD} in the list version imply that the privacy-exponent is negative infinity. We hence assume that $R_1 + R_2 > \RDexp$. In this case Corollary~\ref{co:equivBunteResultGuessingRD} in the guessing version and Corollary~\ref{co:optTaskEncRD} in the list version imply that the constraint
\begin{equation}
\lim_{n \rightarrow \infty} \mathscr A_{\textnormal B} ( P^n_{X,Y}, \Delta ) = 1 \label{eq:bobAmbTo1RD}
\end{equation}
can be met.

We first show that the privacy-exponent cannot exceed the RHS of \eqref{eq:privExpBob1RD}. To this end we note that it holds for every $n \in \naturals$ that
\begin{IEEEeqnarray}{rCl}
\!\!\!\!\!\!\!\! \mathscr A_{\textnormal E} ( P^n_{X,Y} ) & = & \min_{ \hatguessD {}{(1)} { \cdot | Y^n, M_1 }, \, \hatguessD {}{(2)} { \cdot | Y^n, M_2 } } \BigEx {}{\guessD \Delta {(1)}{X^n|Y^n,M_1}^\rho \wedge \guessD \Delta {(2)}{X^n|Y^n,M_2}^\rho} \\
& \leq & \min_{k \in \{1,2\}} \biggl( \min_{ \hatguessD {}{(k)} { \cdot | Y^n, M_k } } \BigEx {}{\guessD \Delta {(k)}{X^n|Y^n,M_k}^\rho} \biggr). \label{eq:eveAltAmbRD}
\end{IEEEeqnarray}
By Corollary~\ref{co:impGuessRD} it holds for every $k \in \{1,2\}$ and $l \in \{ 1, 2 \} \setminus \{ k \}$ that
\begin{IEEEeqnarray}{l}
\min_{ \hatguessD {}{} { \cdot | Y^n, M_1, M_2 }} \bigEx {}{\guessD \Delta {}{X^n|Y^n,M_1,M_2}^\rho} \nonumber \\
\quad \geq | \setM_l |^{-\rho} \min_{\hatguessD {}{(k)} { \cdot | Y^n, M_k }} \BigEx {}{\guessD \Delta {(k)}{X^n|Y^n,M_k}^\rho}. \label{eq:eveAltAmbRDBoundBob}
\end{IEEEeqnarray}
Because $$\min_{ \hatguessD {}{} { \cdot | Y^n, M_1, M_2 } } \bigEx {}{\guessD \Delta{}{X^n|Y^n,M_1,M_2}^\rho} \leq \BigEx {}{ \bigl| \setL^{Y^n}_{M_1,M_2} \bigr|^\rho},$$ \eqref{eq:eveAltAmbRD} and \eqref{eq:eveAltAmbRDBoundBob} imply that in both versions Eve's ambiguity exceeds Bob's by at most a factor of $| \setM_1 |^\rho \wedge | \setM_2 |^\rho$. That is,
\begin{IEEEeqnarray}{l}
\mathscr A_{\textnormal E} ( P^n_{X,Y}, \Delta ) \leq \bigl( |\setM_1| \wedge |\setM_2| \bigr)^\rho \mathscr A_{\textnormal B} ( P^n_{X,Y}, \Delta ). \label{eq:eveAltAmbRDBoundBob1}
\end{IEEEeqnarray}
Suppose that \eqref{eq:bobAmbTo1RD} holds and consequently
\begin{IEEEeqnarray}{l}
\limsup_{n \rightarrow \infty} \frac{ \log \bigl( \mathscr A_{\textnormal B} ( P^n_{X,Y}, \Delta ) \bigr) }{n} = 0. \label{eq:bobAmbTo1EBIs0}
\end{IEEEeqnarray}
From~\eqref{eq:eveAltAmbRDBoundBob1} and~\eqref{eq:bobAmbTo1EBIs0} it follows that
\begin{IEEEeqnarray}{l}
\limsup_{n \rightarrow \infty} \frac{\log \bigl( \mathscr A_{\textnormal E} ( P^n_{X,Y}, \Delta ) \bigr)}{n} \leq \rho (R_1 \wedge R_2). \label{eq:eveAltAmbRDBound1}
\end{IEEEeqnarray}
Eve can ignore the hint that she observes and guess a reconstruction for $X^n$ based on $Y^n$ alone. Hence, we obtain from Theorem~\ref{th:optGuessFunRD} that
\begin{IEEEeqnarray}{l}
\limsup_{n \rightarrow \infty} \frac{\log \bigl( \mathscr A_{\textnormal E} ( P^n_{X,Y}, \Delta ) \bigr)}{n} \leq \rho \RDexp. \label{eq:eveAltAmbRDBound2}
\end{IEEEeqnarray}
From \eqref{eq:eveAltAmbRDBound1} and \eqref{eq:eveAltAmbRDBound2} we conclude that the privacy-exponent cannot exceed the RHS of \eqref{eq:privExpBob1RD}:
\begin{IEEEeqnarray}{l}
\limsup_{n \rightarrow \infty} \frac{\log \bigl( \mathscr A_{\textnormal E} ( P^n_{X,Y}, \Delta ) \bigr)}{n} \leq \rho \Bigl( R_1 \wedge R_2 \wedge \RDexp \Bigr). \label{eq:eveAltAmbRDBound}
\end{IEEEeqnarray}

We next show that the privacy-exponent cannot be smaller than the RHS of \eqref{eq:privExpBob1RD}. By possibly relabeling the hints, we can assume w.l.g.\ that $R_2 = R_1 \wedge R_2$. Fix some $\epsilon > 0$ satisfying
\begin{equation}
\epsilon \leq R_1 + R_2 - \RDexp. \label{eq:epsChoiceRDPf}
\end{equation}
Choose a nonnegative rate-triple $(R_{\textnormal s},\tilde R_1,\tilde R_2)$ as follows:
\begin{enumerate}
\item If $R_2 \leq \RDexp / 2$, then choose
\begin{IEEEeqnarray}{l}
R_{\textnormal s} = 0, \quad \tilde R_1 = \RDexp - R_2 + \epsilon, \quad \tilde R_2 = R_2.
\end{IEEEeqnarray}
\item Else if $\RDexp / 2 < R_2 \leq \RDexp$, then choose
\begin{IEEEeqnarray}{l}
R_{\textnormal s} = 2 R_2 - \RDexp - \epsilon, \quad \tilde R_1 = \tilde R_2 = \RDexp - R_2 + \epsilon.
\end{IEEEeqnarray}
(To guarantee that $R_{\textnormal s} \geq 0$, we assume in this case that $\epsilon > 0$ is sufficiently small so that, in addition to \eqref{eq:epsChoiceRDPf}, also
\begin{equation}
\epsilon < 2 R_2 - \RDexp
\end{equation}
holds.)
\item Else if $\RDexp < R_2$, then choose
\begin{IEEEeqnarray}{l}
R_{\textnormal s} = R_2, \quad \tilde R_1 = \tilde R_2 = 0.
\end{IEEEeqnarray}
\end{enumerate}
%\ba
%\nonumber
%& (R_s,\tilde R_1,\tilde R_2) \\ & \quad = \begin{cases} (0, \RDexp - R_2 + \epsilon, R_2), & \hspace{-5cm} \textrm{ if } \, R_2 \leq \frac{1}{2} \RDexp\\ ( 2 R_2 - \RDexp - \epsilon,\RDexp - R_2 + \epsilon,\RDexp - R_2 + \epsilon ), &\\ \quad \textrm{ if } \, \frac{1}{2} \RDexp < R_2 \leq \RDexp \\  (R_2, 0, 0), &\hspace{-5cm} \textrm{ if } \, \RDexp < R_2 \end{cases}
%\ea
Having chosen $(R_{\textnormal s},\tilde R_1,\tilde R_2)$, choose the triple $(c_{\textnormal s}, c_1, c_2) \in \naturals^3$ to be
\begin{IEEEeqnarray}{l}
(c_{\textnormal s}, c_1, c_2) = (2^{n R_{\textnormal s}}, 2^{n \tilde R_1}, 2^{n \tilde R_2}).
\end{IEEEeqnarray}
For each $\nu \in \{ \textnormal s, 1, 2 \}$, let $V_\nu$ be a chance variable taking values in the set $\setV_\nu = \{ 0, \ldots, c_\nu - 1 \}$. Because our choice of $(R_{\textnormal s},\tilde R_1,\tilde R_2)$ satisfies
\begin{equation}
R_{\textnormal s} + \tilde R_1 + \tilde R_2 >  \RDexp,
\end{equation}
Corollary~\ref{co:optTaskEncRD} implies that there exist $\{ 0, 1 \}$-valued conditional PMFs $$\bigdistof { ( V_{\textnormal s}, V_1, V_2 ) = ( v_{\textnormal s}, v_1, v_2 ) \bigl| X^n = \vecx, Y^n = \vecy}$$ and decoders, whose lists $$\bigl\{ \setL^{\vecy}_{v_{\textnormal s}, v_1, v_2} \bigr\}_{ ( \vecy,v_{\textnormal s}, v_1, v_2 ) \in \setY^n \times \setV_{\textnormal s} \times \setV_1 \times \setV_2 }$$ satsify
\begin{IEEEeqnarray}{l}
\exists \, \hat \vecx \in \setL^{Y^n}_{V_{\textnormal s}, V_1, V_2} \textnormal { s.t.\ } d^{(n)} (X^n, \hat \vecx) \leq \Delta,
\end{IEEEeqnarray}
%\begin{IEEEeqnarray}{l}
%\biggl( \forall \, \bigl( \vecx, \vecy, v^{(n)}_{\textnormal s}, v^{(n)}_1, v^{(n)}_2 \bigr) \textnormal { s.t.\ } \nonumber \\
%\qquad P_{X,Y}^n (\vecx,\vecy) \, \Bigdistof { \bigl( V^{(n)}_{\textnormal s}, V^{(n)}_1, V^{(n)}_2 \bigr) = \bigl( v^{(n)}_{\textnormal s}, v^{(n)}_1, v^{(n)}_2 \bigr) \Bigl| X^n = \vecx, Y^n = \vecy} > 0 \biggr) \nonumber \\
%\quad \exists \, \hat \vecx \in \setL^{\vecy}_{v^{(n)}_{\textnormal s}, v^{(n)}_1, v^{(n)}_2} \textnormal { s.t.\ } d^{(n)} (\vecx, \hat \vecx) \leq \Delta.
%\end{IEEEeqnarray}
for which
\begin{equation}
\lim_{n \rightarrow \infty} \BigEx {}{\bigl| \setL^{Y^n}_{V_{\textnormal s}, V_1, V_2} \bigr|^\rho} = 1. \label{eq:BobMomDistStorDir1RD}
\end{equation}
Because $$\min_{\hatguessD {}{} {\cdot | Y^n,V_{\textnormal s}, V_1, V_2 }} \bigEx {}{\guessD \Delta {} {X^n | Y^n,V_{\textnormal s}, V_1, V_2 }^\rho} \leq \BigEx {}{\bigl| \setL^{Y^n}_{V_{\textnormal s}, V_1, V_2} \bigr|^\rho},$$ \eqref{eq:BobMomDistStorDir1RD} implies that
\begin{equation}
\lim_{n \rightarrow \infty} \min_{\hatguessD {}{} {\cdot | Y^n,V_{\textnormal s}, V_1, V_2 }} \bigEx {}{\guessD \Delta {} {X^n \bigl| Y^n,V_{\textnormal s}, V_1, V_2 }^\rho} = 1. \label{eq:BobMomDistStorDir1GuessRD}
\end{equation}
Our choice of $(R_{\textnormal s},\tilde R_1,\tilde R_2)$ satisfies
\begin{equation}
R_1 \geq R_{\textnormal s} + \tilde R_1, \quad R_2 \geq R_{\textnormal s} + \tilde R_2,
\end{equation}
and hence we can for every blocklength~$n$ choose some conditional PMF  \eqref{eq:aliceEncPMFRD} that assigns positive probability only to $c_{\textnormal s} c_1$ elements of $\setM_1$ and $c_{\textnormal s} c_2$ elements of $\setM_2$. Therefore, we can assume w.l.g.\ that $\setM_1 = \setV_{\textnormal s} \times \setV_1$ and $\setM_2 = \setV_{\textnormal s} \times \setV_2$ and choose $M_1 = ( V_{\textnormal s} \oplus_{c_{\textnormal s}} \! U, V_1 )$ and $M_2 = ( U, V_2 )$, where $(V_{\textnormal s},V_1,V_2)$ is drawn according to the above conditional PMF, and where $U$ is independent of $( X^n, Y^n, V_{\textnormal s}, V_1, V_2 )$ and uniform over $\setV_{\textnormal s}$. For this choice \eqref{eq:bobAmbTo1RD} follows from \eqref{eq:BobMomDistStorDir1RD} in the list version and from \eqref{eq:BobMomDistStorDir1GuessRD} in the guessing version.

It remains to show that for the above choice of the conditional PMFs \eqref{eq:aliceEncPMFRD}
\begin{IEEEeqnarray}{l}
\liminf_{n \rightarrow \infty} \frac{\log \bigl( \mathscr A_{\textnormal E} ( P^n_{X,Y}, \Delta ) \bigr)}{n} \geq \rho \Bigl( R_1 \wedge R_2 \wedge \RDexp \Bigr). \label{eq:eveAmbRDLB}
\end{IEEEeqnarray}
Define the triple of chance variables
\begin{IEEEeqnarray}{l}
( I, \hat U, \hat V ) \triangleq \begin{cases} ( 1, V_s \oplus_{c_{\textnormal s}} \! U, V_1 ) &\text{if } \guessD \Delta {(1)} { X^n | Y^n, M_1 } \leq \guessD \Delta {(2)} { X^n | Y^n, M_2 }, \\ ( 2,U,V_2 ) &\text{otherwise} \end{cases} \label{eq:pfEveAmbRDTripRVs}
\end{IEEEeqnarray}
with alphabet $\setI \times \setV_{\textnormal s} \times \hat \setV$, where $\setI = \{ 1,2 \}$ and $\hat \setV = \{ 0, \ldots, c_1 \vee c_2 - 1 \}$. From $(Y^n, I, U, \hat V)$ Eve can guess a reconstruction $\hat \vecx \in \hat \setX^n$ of $X^n$ using either $\hatguessD {}{(1)}{ \cdot | Y^n, M_1 }$ or $\hatguessD {}{(2)}{ \cdot | Y^n, M_2 }$ depending on the value of $I$. That is, Eve can use some guessing function $\hatguessD {}{}{ \cdot | Y^n, I, \hat U, \hat V }$ satisfying that, for every $\vecy \in \setY^n$, $i \in \setI$, $\hat u \in \setV_{\textnormal s}$, and $\hat v \in \{ 0, \ldots, c_i - 1 \}$,
\begin{equation}
\hatguessD {}{}{ \hat \vecx | \vecy, i, \hat u, \hat v } = \bighatguessD {}{(i)}{ \hat \vecx \bigl| \vecy, (\hat u, \hat v) },
\end{equation}
where by \eqref{eq:pfEveAmbRDTripRVs} the success function $\guessD \Delta {}{ \cdot | Y^n, I, \hat U, \hat V }$ corresponding to $\hatguessD {}{}{ \cdot | Y^n, I, \hat U, \hat V }$ satisfies
\begin{IEEEeqnarray}{l}
\guessD \Delta {}{ X^n | Y^n, I, \hat U, \hat V } \nonumber \\
\quad = \bigguessD \Delta {(I)}{ X^n \bigl| Y^n, (\hat U, \hat V) } \\
\quad = \guessD \Delta {(I)}{ X^n | Y^n, M_I } \\
\quad = \guessD \Delta {(1)}{ X^n | Y^n, M_1 }^\rho \wedge \guessD \Delta {(2)}{ X^n | Y^n, M_2 }.
\end{IEEEeqnarray}
Let $\psifun { \cdot | Y^n, I, \hat U, \hat V }$ be the reconstruction function corresponding to $\hatguessD {}{}{ \cdot | Y^n, I, \hat U, \hat V }$, i.e., the unique mapping satisfying that
\begin{IEEEeqnarray}{l}
\Bigl( \psifun { \vecx | \vecy, i, \hat u, \hat v } = \hat \vecx \iff \guessD \Delta {}{ \vecx | \vecy, i, \hat u, \hat v } = \hatguessD {}{}{ \hat \vecx | \vecy, i, \hat u, \hat v } \Bigr), \nonumber \\
\quad \,\, \forall \, (\vecx, \hat \vecx, \vecy, i, \hat u, \hat v) \in \setX^n \times \hat \setX^n \times \setY^n \times \setI \times \setV_{\textnormal s} \times \hat \setV,
\end{IEEEeqnarray}
and define the chance variable $\hat X^n = \psifun { X^n | Y^n, I, \hat U, \hat V }$. Note that
\begin{IEEEeqnarray}{l}
\bigEx {}{\hatguessD {} {} { \hat X^n | Y^n, I, \hat U, \hat V }^\rho} \geq \min_{\guess {} { \cdot, \cdot, \cdot | Y^n, I, \hat U, \hat V } } \bigEx {}{\guess {} { \hat X^n, I, \hat U | Y^n, I, \hat U, \hat V }^\rho}.
\end{IEEEeqnarray}
This implies that
\begin{IEEEeqnarray}{rCl}
\mathscr A_{\textnormal E} ( P^n_{X,Y}, \Delta ) & \geq & \min_{\guess {} { \cdot, \cdot, \cdot | Y^n, I, \hat U, \hat V } } \bigEx {}{\guess {} { \hat X^n, I, \hat U | Y^n, I, \hat U, \hat V }^\rho} \\
& \stackrel{(a)}\geq & \bigl( |\setI| \, | \setV_{\textnormal s} | \, |\hat \setV| \bigr)^{-\rho} \min_{\guess {} { \cdot, \cdot, \cdot | Y^n } } \bigEx {}{\guess {} { \hat X^n, I, \hat U | Y^n }^\rho} \\
& \geq & 2^{- \rho - n \rho ( R_{\textnormal s} + \tilde R_1 \vee \tilde R_2 )} \min_{ \guess {} { \cdot, \cdot, \cdot | Y^n } } \bigEx {}{\guess {} { \hat X^n, I, \hat U | Y^n }^\rho}, \label{eq:pfEveAmbRDFirstExpr}
\end{IEEEeqnarray}
where $(a)$ follows from Corollary~\ref{co:impGuess} and the fact that $( I, \hat U, \hat V )$ takes values in the set
\begin{IEEEeqnarray*}{l}
\bigl\{ ( 1, \hat u, \hat v ) \colon ( \hat u, \hat v ) \in \setV_{\textnormal s} \times \setV_1 \bigr\} \cup \bigl\{ ( 2, \hat u, \hat v ) \colon ( \hat u, \hat v ) \in \setV_{\textnormal s} \times \setV_2 \bigr\},
\end{IEEEeqnarray*}
which is of size
\begin{IEEEeqnarray*}{l}
| \setV_{\textnormal s} \times \setV_1 | + | \setV_{\textnormal s} \times \setV_2 | = c_{\textnormal s} ( c_1 + c_2 ).
\end{IEEEeqnarray*}
From \eqref{eq:pfEveAmbRDFirstExpr} it follows that
\begin{IEEEeqnarray}{l}
\liminf_{n \rightarrow \infty} \frac{\log \bigl( \mathscr A_{\textnormal E} ( P^n_{X,Y}, \Delta ) \bigr)}{n} \nonumber \\
\quad \geq \liminf_{n \rightarrow \infty} \min_{ \guess {} { \cdot, \cdot, \cdot | Y^n } } \frac{\log \Bigl( \bigEx {}{ \guess {}{\hat X^n, I, \hat U \bigl| Y^n }^\rho } \Bigr)}{n} - \rho ( R_{\textnormal s} + \tilde R_1 \vee \tilde R_2 ). \label{eq:secrecyRD1}
\end{IEEEeqnarray}
Therefore, if we can show that
\begin{IEEEeqnarray}{l}
\liminf_{n \rightarrow \infty} \min_{ \guess {} { \cdot, \cdot, \cdot | Y^n } } \frac{\log \Bigl( \bigEx {}{ \guess {}{\hat X^n, I, \hat U | Y^n }^\rho } \Bigr)}{n} \geq \rho \Bigl( \RDexp + R_{\textnormal s} \Bigr), \label{eq:secrecyRD2}
\end{IEEEeqnarray}
then we can let $\epsilon$ tend to zero to conclude from \eqref{eq:secrecyRD1} that \eqref{eq:eveAmbRDLB} holds:
\begin{IEEEeqnarray}{rCl}
\liminf_{n \rightarrow \infty} \frac{\log \bigl( \mathscr A_{\textnormal E} ( P^n_{X,Y}, \Delta ) \bigr)}{n} & \geq & \rho \Bigl( R_2 \wedge \RDexp \Bigr) \\
& \geq & \rho \Bigl( R_1 \wedge R_2 \wedge \RDexp \Bigr).
\end{IEEEeqnarray}

We next conclude the proof of \eqref{eq:eveAmbRDLB} by establishing \eqref{eq:secrecyRD2}. By Theorem~\ref{th:optGuessFun}
\begin{IEEEeqnarray}{l}
\liminf_{n \rightarrow \infty} \min_{ \guess {} { \cdot, \cdot, \cdot | Y^n } } \frac{\log \Bigl( \bigEx {}{ \guess {}{\hat X^n, I, \hat U | Y^n }^\rho } \Bigr)}{n} \geq \rho \renent {\tirho}{\hat X^n, I, \hat U | Y^n }.  \label{eq:secrecyRD3}
\end{IEEEeqnarray}
In \cite[Appendix~B]{buntelapidoth14} it is shown that for every pair of chance variables $( A, B )$ taking values in some finite set $\setA \times \setB$ according to som PMF $P_{A,B}$
\begin{IEEEeqnarray}{l}
\renent {\tirho}{ A | B } = \max_{\substack{Q \in \mathscr P ( \setB ), \\ V \in \mathscr P ( \setA | \setB )}} \ent { V | Q } - \rho^{-1} \relent { Q \times V }{P_{A,B}}, \label{eq:condRenEntAltExp}
\end{IEEEeqnarray}
where $\mathscr P ( \setB )$ denotes the set of PMFs on $\setB$, and $\mathscr P ( \setA | \setB )$ denotes the set of transition laws from $\setB$ to $\setA$. We shall use \eqref{eq:condRenEntAltExp} to lower-bound the RHS of \eqref{eq:secrecyRD3}, where we will substitute $( \hat X^n, I, \hat U )$ for $A$ and $Y^n$ for $B$ in \eqref{eq:condRenEntAltExp}. To that end denote by $V_n$ the conditional PMF of $( \hat X^n, I, \hat U )$ given $( X^n, Y^n, U )$, and denote by $\tilde V_n$ the conditional PMF of $( Y^n, \hat X^n, I, \hat U )$ given $( X^n, Y^n, U )$. Note that $V_n$ and $\tilde V_n$ are both $\{ 0,1 \}$-valued. Fix any PMF $Q_{X,Y}$ on $\setX \times \setY$, let $P_U$ denote the uniform distribution on $\setV_{\textnormal s}$, and define the PMF on $\setX^n \times \setY^n \times \setU \times \hat \setX^n \times \setI \times \setU$ $$Q_{X^n,Y^n,U,\hat X^n,I,\hat U} = \bigl( Q_{X,Y}^n \times P_U \bigr) \times V_n.$$ As to $\relent { Q \times V }{P_{A,B}}$, we then find that
\begin{IEEEeqnarray}{l}
\Bigrelent {Q_Y^n \times \bigl( ( Q_{X | Y }^n \times P_U ) \, V_n \bigr)}{P_{Y}^n \times \bigl( ( P_{ X | Y }^n \times P_U ) \, V_n \bigr) }\\
\quad = \bigrelent{ ( Q_{X,Y}^n \times P_U ) \, \tilde V_n } { ( P_{X,Y}^n \times P_U ) \, \tilde V_n } \\
\quad \stackrel{(a)}\leq \relent { Q_{X,Y}^n \times P_U }{ P_{X,Y}^n \times P_U } \\
\quad = \relent { Q_{X,Y}^n }{ P_{X,Y}^n } \\
\quad = n \relent { Q_{X,Y} }{ P_{X,Y} }, \label{eq:divBoundRDDistStor}
\end{IEEEeqnarray}
where $(a)$ follows from the Data-Processing inequality \cite[Lemma~3.11]{csiszarkoerner11}. As to $\ent {V | Q }$, we find that
\begin{IEEEeqnarray}{l}
\bigcondent { ( Q_{ X | Y }^n \times P_U ) \, V_n } { Q_Y^n } \\
\quad \stackrel{(a)}\geq \bigcondmuti { Q_{X | Y }^n \times P_U }{ V_n }{ Q_Y^n } \\
\quad \stackrel{(b)}= \condmuti { Q_{X | Y }^n }{ P_U \, V_n }{ Q_Y^n } + \bigcondmuti { P_U }{ V_n }{ Q_{X,Y}^n } \\
\quad \stackrel{(c)}= \condmuti { Q_{ X | Y }^n }{ P_U \, V_n }{ Q_Y^n } + \log | \setV_{\textnormal s} | \\
\quad \stackrel{(d)}\geq \condmuti { Q_{ X | Y }^n}{ Q_{ \hat X^n | X^n,Y^n } }{ Q_Y^n } + \log | \setV_{\textnormal s} | \\
\quad \stackrel{(e)}\geq n R_{ X | Y } ( Q_{X,Y}, \Delta ) + \log |\setV_{\textnormal s} |, \label{eq:entBoundRDDistStor}
\end{IEEEeqnarray}
where $(a)$ holds because entropy is nonnegative; $(b)$ follows from chain rule; $(c)$ holds because $U$ is independent of $( X^n, Y^n )$ and uniform over its support $\setV_{\textnormal s}$, and because $U$ is deterministic given $\bigl( X^n, Y^n, \hat X^n, I, \hat U \bigr)$ (which holds by \eqref{eq:pfEveAmbRDTripRVs} and because $V_{\textnormal s}$ is deterministic given $( X^n, Y^n )$); $(d)$ holds for the conditional PMF
\begin{IEEEeqnarray}{l}
Q_{\hat X^n|X^n,Y^n} ( \hat \vecx | \vecx, \vecy ) = \sum_{u, i, \hat u} P_U (u) \, V_n ( \hat \vecx, i, \hat u | \vecx, \vecy, u ), \, \forall \, ( \vecx, \hat \vecx, \vecy ) \in \setX^n \times \hat \setX^n \times \setY^n, \label{eq:condPMFHatXnGivXnYn}
\end{IEEEeqnarray}
because conditioning cannot increase entropy; and $(e)$ follows from the conditional Rate-Distortion theorem \cite{leinergray74} and
\begin{IEEEeqnarray}{l}
\Bigdistsubof { Q_{X,Y}^n \times Q_{\hat X^n|X^n,Y^n} }{ d^{(n)} \bigl( X^n, \hat X^n \bigr) \leq \Delta } = 1, \label{eq:distSmallerDWP1DistStor}
\end{IEEEeqnarray}
which holds by \eqref{eq:condPMFHatXnGivXnYn} and because
\begin{IEEEeqnarray}{l}
\Bigl( V_n ( \hat \vecx, i, \hat u | \vecx, \vecy, u ) > 0 \implies \exists \, \hat v \in \hat \setV \colon \hat \vecx = \psifun { \vecx | \vecy, i, \hat u, \hat v } \Bigr), \nonumber \\
\quad \,\, \forall \, (\vecx, \hat \vecx, \vecy, i, \hat u) \in \setX^n \times \hat \setX^n \times \setY^n \times \setI \times \setU.
\end{IEEEeqnarray}

More precisely, $(e)$ can be established as follows. Draw $\bigl( \underline {X^n}, \underline {Y^n}, \underline {\hat X^n} \bigr)$ from $\setX^n \times \setY^n \times \hat \setX^n$ according to the PMF $Q_{X,Y}^n \times Q_{\hat X^n|X^n,Y^n}$. By \eqref{eq:distSmallerDWP1DistStor}
\begin{IEEEeqnarray}{l}
\BigEx{}{d^{(n)} ( \underline {X^n}, \underline{\hat X^n} )} = \BigEx{Q_{X,Y}^n \times Q_{\hat X^n|X^n,Y^n}}{ d^{(n)}( X^n, \hat X^n ) } \leq \Delta. \label{eq:distSmallerDWP1DistStorUnderline}
\end{IEEEeqnarray}
Consequently, we find that
\begin{IEEEeqnarray}{l}
\condmuti {Q_{X|Y}^n}{Q_{\hat X^n|X^n,Y^n}}{Q_Y^n} \nonumber \\
\quad = \condmutivars {\underline {X^n}}{\underline {\hat X^n}}{\underline {Y^n}} \\
\quad \stackrel{(f)}= \sum^n_{i = 1} \bigcondmutivars {\underline {X_i}}{\underline {\hat X^n}}{\underline {Y^n}, \underline {X^{i-1}}} \\
\quad \stackrel{(g)}\geq \sum^n_{i = 1} \bigcondmutivars {\underline {X_i}}{\underline {\hat X_i}}{\underline {Y_i}} \\
\quad \stackrel{(h)}= n \Biggl( \frac{1}{n} \sum^n_{i = 1} \bigcondmuti {Q_{X|Y}}{Q_{\hat X_i| X_i, Y_i}}{Q_Y} \Biggr) \\
\quad \stackrel{(i)}\geq n \Biggcondmuti {Q_{X|Y}}{\frac{1}{n} \sum^n_{i = 1} Q_{\hat X_i|X_i,Y_i}}{Q_Y} \\
\quad \stackrel{(j)}\geq n \min_{\substack{ Q_{\hat X| X, Y} \colon \\ \Ex {}{d ( X, \hat X ) \leq \Delta}}} \condmuti {Q_{X|Y}}{Q_{\hat X | X,Y}}{Q_Y} \\
\quad \stackrel{(k)}= \RDfun,
\end{IEEEeqnarray}
where $(f)$ follows from the chain rule; $(g)$ holds because $\underline {X_i}$ and $\bigl( \underline {Y^{i-1}}, \underline {Y_{i+1}^n}, \underline {X^{i-1}} \bigr)$ are independent, and because conditioning cannot increase entropy; $(h)$ holds for the conditional PMFs $Q_{\hat X_i| X_i, Y_i}, \,\, i \in [1:n]$ that satisfy
\begin{IEEEeqnarray*}{l}
Q_{\hat X_i| X_i, Y_i} (\hat x_i | x_i, y_i) = \sum_{\substack{x^{i-1}, \hat x^{i-1}, y^{i-1}, \\ x^n_{i + 1}, \hat x^n_{i + 1}, y^n_{i + 1} }} Q_{X,Y}^{i-1} (x^{i-1}, y^{i-1}) \, Q_{X,Y}^{n - i} (x^n_{i+1}, y^n_{i+1}) \, Q_{\hat X^n|X^n,Y^n} (\hat x^n | x^n, y^n), \\
\quad \,\, \forall \, ( x_i, \hat x_i, y_i ) \in \setX \times \hat \setX \times \setY;
\end{IEEEeqnarray*}
$(i)$ holds because mutual information is convex in the transition law (here $Q_{\hat X_i|X_i,Y_i}$); $(j)$ holds because \eqref{eq:distSmallerDWP1DistStorUnderline} implies that
\begin{IEEEeqnarray}{rCl}
\Delta & \geq & \BigEx{}{d^{(n)} ( \underline {X^n}, \underline{\hat X^n} )} \\
& = & \BiggEx{}{\frac{1}{n} \sum^n_{i = 1} d \bigl( \underline {X_i}, \underline {\hat X_i} \bigr)} \\
& = & \bigEx{Q_{X,Y} \times \bigl( \frac{1}{n} \sum^n_{i = 1} Q_{\hat X_i|X_i, Y_i} \bigr)}{d ( X, \hat X )};
\end{IEEEeqnarray}
and $(k)$ holds by the definition of the rate-distortion function under the PMF $Q_{X,Y}$ \eqref{eq:defRDFun}. This concludes the proof of \eqref{eq:entBoundRDDistStor}.

Having established \eqref{eq:entBoundRDDistStor}, we are now ready to conclude the proof of \eqref{eq:secrecyRD2}. By substituting $( \hat X^n, I, \hat U )$ for $A$ and $Y^n$ for $B$ in \eqref{eq:condRenEntAltExp}, we obtain from \eqref{eq:condRenEntAltExp}, \eqref{eq:divBoundRDDistStor}, and \eqref{eq:entBoundRDDistStor} that
\begin{IEEEeqnarray}{l}
\renent {\tirho}{\hat X^n, I, \hat U \bigl| Y^n } \nonumber \\
\quad \geq \bigcondent { ( Q_{X | Y }^n \times P_U ) \, V_n } { Q_Y^n } \nonumber \\
\qquad - \rho^{-1} \Bigrelent {Q_Y^n \times \bigl( ( Q_{X | Y }^n \times P_U ) \, V_n \bigr)}{P_{Y}^n \times \bigl( ( P_{X|Y}^n \times P_U ) \, V_n \bigr)} \\
\quad \geq n \Bigl( \RDfun - \rho^{-1} \relent { Q_{X,Y} }{ P_{X,Y} } \Bigr) + \log |\setV_{\textnormal s}|.
\end{IEEEeqnarray}
Because this holds for every PMF $Q_{X,Y}$ on $\setX \times \setY$, and by the definition of $\RDexp$ \eqref{eq:defFunctional},
\begin{IEEEeqnarray}{l}
\renent {\tirho}{\hat X^n, I, \hat U | Y^n } \nonumber \\
\quad \geq n \sup_{Q_{X,Y}} \Bigl( \RDfun - \rho^{-1} \relent { Q_{X,Y} }{ P_{X,Y} } \Bigr) + \log |\setV_{\textnormal s}| \\
\quad = n \RDexp + \log |\setV_{\textnormal s}|.
\end{IEEEeqnarray}
This, $|\setV_{\textnormal s}| = 2^{n R_{\textnormal s}}$, and \eqref{eq:secrecyRD3} imply \eqref{eq:secrecyRD2}. This concludes the proof of \eqref{eq:privExpBob1RD}.\\

We next prove \eqref{eq:privExpBobEBRD}. If $R_1 + R_2 < \RDexp - \rho^{-1} E_{\textnormal B}$, then Corollary~\ref{co:equivBunteResultGuessingRD} in the guessing version and Corollary~\ref{co:optTaskEncRD} in the list version imply that the modest privacy-exponent is negative infinity. We hence assume that $R_1 + R_2 > \RDexp - \rho^{-1} E_{\textnormal B}$. We can now use the same line of argument as in the proof of \eqref{eq:privExpBob1RD} but with \eqref{eq:bobAmbTo1EBIs0} replaced by
\begin{equation}
\limsup_{n \rightarrow \infty} \frac{\log \bigl( \mathscr A_{\textnormal B} ( P^n_{X,Y}, \Delta ) \bigr)}{n} \leq E_{\textnormal B} \label{eq:bobAmbEBRD}
\end{equation}
to show that the modest privacy-exponent cannot exceed the RHS of \eqref{eq:privExpBobEBRD}. To show that the modest privacy-exponent is lower-bounded by the RHS of \eqref{eq:privExpBobEBRD}, we argue as for the privacy-exponent, except that here we choose the nonnegative triple $( R_{\textnormal s},\tilde R_1,\tilde R_2 )$ as follows:
\begin{enumerate}
\item If $R_2 \leq \bigl( \RDexp - \rho^{-1} E_{\textnormal B} \bigr) / 2$, then choose
\begin{IEEEeqnarray}{l}
R_{\textnormal s} = 0, \quad \tilde R_1 = \RDexp - \rho^{-1} E_{\textnormal B} - R_2, \quad \tilde R_2 = R_2.
\end{IEEEeqnarray}
\item Else if $\bigl( \RDexp - \rho^{-1} E_{\textnormal B} \bigr) / 2 < R_2 \leq \RDexp - \rho^{-1} E_{\textnormal B}$, then choose
\begin{IEEEeqnarray}{l}
R_{\textnormal s} = 2 R_2 - \RDexp + \rho^{-1} E_{\textnormal B}, \nonumber \\
\tilde R_1 = \tilde R_2 = \RDexp - \rho^{-1} E_{\textnormal B} - R_2.
\end{IEEEeqnarray}
\item Else if $\RDexp - \rho^{-1} E_{\textnormal B} < R_2$, then choose
\begin{IEEEeqnarray}{l}
R_{\textnormal s} = R_2, \quad \tilde R_1 = \tilde R_2 = 0.
\end{IEEEeqnarray}
\end{enumerate}
\end{proof}

%%%%%%%%%%%%%%%%%%%%%%%%%%%%%%%%%%%%%%%%%%%%%%%%%%%%%%%%%%%%%%%%%%
%% BIBLIOGRAPHY
%%%%%%%%%%%%%%%%%%%%%%%%%%%%%%%%%%%%%%%%%%%%%%%%%%%%%%%%%%%%%%%%%%
%%
%%%%%%%%%%%%%%%%%%%%%%%%%%%%%%%%%%%%%%%%%%%%%%%%%%%%%%%%%%%%%%%%%%%%%%%%%% 
% \bibliographystyle{IEEEtran}
% \bibliography{defshort1,biblio1}


%%%%%%%%%%%%%%%%%%%%%%%%%%%%%%%%%%%%%%%%%%%%%%%%%%%%%%%%%%%%%%%%%%
%% END OF APPENDIX
%%%%%%%%%%%%%%%%%%%%%%%%%%%%%%%%%%%%%%%%%%%%%%%%%%%%%%%%%%%%%%%%%%
%
\end{appendix}
\end{document}

%% file: math-macrosRD.tex
\usepackage{amsmath}
\usepackage{amssymb}
\usepackage{amsfonts}
\usepackage{mathrsfs}
\usepackage{xspace}
\usepackage{bm}
\usepackage{fancyref}
\usepackage{textcomp}
\usepackage[Symbol]{upgreek}

\newcommand{\safemath}[2]{\newcommand{#1}{\ensuremath{#2}\xspace}}
\renewcommand{\safemath}[2]{\newcommand{#1}{\ensuremath{#2}\xspace}}

\newcommand{\ssa}{\mathsf{a}}
\newcommand{\ssb}{\mathsf{b}}
\newcommand{\ssc}{\mathsf{c}}
\newcommand{\ssd}{\mathsf{d}}
\newcommand{\sse}{\mathsf{e}}
\newcommand{\ssf}{\mathsf{f}}
\newcommand{\ssg}{\mathsf{g}}
\newcommand{\ssh}{\mathsf{h}}
\newcommand{\ssi}{\mathsf{i}}
\newcommand{\ssj}{\mathsf{j}}
\newcommand{\ssk}{\mathsf{k}}
\newcommand{\ssl}{\mathsf{l}}
\newcommand{\ssm}{\mathsf{m}}
\newcommand{\ssn}{\mathsf{n}}
\newcommand{\sso}{\mathsf{o}}
\newcommand{\ssp}{\mathsf{p}}
\newcommand{\ssq}{\mathsf{q}}
\newcommand{\ssr}{\mathsf{r}}
\newcommand{\sss}{\mathsf{s}}
\newcommand{\sst}{\mathsf{t}}
\newcommand{\ssu}{\mathsf{u}}
\newcommand{\ssv}{\mathsf{v}}
\newcommand{\ssw}{\mathsf{w}}
\newcommand{\ssx}{\mathsf{x}}
\newcommand{\ssy}{\mathsf{y}}
\newcommand{\ssz}{\mathsf{z}}

\safemath{\bmsa}{\bm{\ssa}}
\safemath{\bmsb}{\bm{\ssb}}
\safemath{\bmsc}{\bm{\ssc}}
\safemath{\bmsd}{\bm{\ssd}}
\safemath{\bmse}{\bm{\sse}}
\safemath{\bmsf}{\bm{\ssf}}
\safemath{\bmsg}{\bm{\ssg}}
\safemath{\bmsh}{\bm{\ssh}}
\safemath{\bmsi}{\bm{\ssi}}
\safemath{\bmsj}{\bm{\ssj}}
\safemath{\bmsk}{\bm{\ssk}}
\safemath{\bmsl}{\bm{\ssl}}
\safemath{\bmsm}{\bm{\ssm}}
\safemath{\bmsn}{\bm{\ssn}}
\safemath{\bmso}{\bm{\sso}}
\safemath{\bmsp}{\bm{\ssp}}
\safemath{\bmsq}{\bm{\ssq}}
\safemath{\bmsr}{\bm{\ssr}}
\safemath{\bmss}{\bm{\sss}}
\safemath{\bmst}{\bm{\sst}}
\safemath{\bmsu}{\bm{\ssu}}
\safemath{\bmsv}{\bm{\ssv}}
\safemath{\bmsw}{\bm{\ssw}}
\safemath{\bmsx}{\bm{\ssx}}
\safemath{\bmsy}{\bm{\ssy}}
\safemath{\bmsz}{\bm{\ssz}}

\bmdefine{\bmualphad}{\upalpha}
\bmdefine{\bmubetad}{\upbeta}
\bmdefine{\bmuchid}{\upchi}
\bmdefine{\bmudeltad}{\updelta}
\bmdefine{\bmuepsilond}{\upepsilon}
\bmdefine{\bmuvarepsilond}{\upvarepsilon}
\bmdefine{\bmuetad}{\upeta}
\bmdefine{\bmugammad}{\upgamma}
\bmdefine{\bmuiotad}{\upiota}
\bmdefine{\bmukappad}{\upkappa}
\bmdefine{\bmulambdad}{\uplambda}
\bmdefine{\bmumud}{\upmu}
\bmdefine{\bmunud}{\upnu}
\bmdefine{\bmuomegad}{\upomega}
\bmdefine{\bmuphid}{\upphi}
\bmdefine{\bmuvarphid}{\upvarphi}
\bmdefine{\bmupid}{\uppi}
\bmdefine{\bmuvarpid}{\upvarpi}
\bmdefine{\bmupsid}{\uppsi}
\bmdefine{\bmurhod}{\uprho}
\bmdefine{\bmuvarrhod}{\upvarrho}
\bmdefine{\bmusigmad}{\upsigma}
\bmdefine{\bmuvarsigmad}{\upvarsigma}
\bmdefine{\bmutaud}{\uptau}
\bmdefine{\bmuthetad}{\uptheta}
\bmdefine{\bmuvarthetad}{\upvartheta}
\bmdefine{\bmuupsilond}{\upupsilon}
\bmdefine{\bmuxid}{\upxi}
\bmdefine{\bmuzetad}{\upzeta}

\safemath{\bmua}{\mathbf{a}}
\safemath{\bmub}{\mathbf{b}}
\safemath{\bmuc}{\mathbf{c}}
\safemath{\bmud}{\mathbf{d}}
\safemath{\bmue}{\mathbf{e}}
\safemath{\bmuf}{\mathbf{f}}
\safemath{\bmug}{\mathbf{g}}
\safemath{\bmuh}{\mathbf{h}}
\safemath{\bmui}{\mathbf{i}}
\safemath{\bmuj}{\mathbf{j}}
\safemath{\bmuk}{\mathbf{k}}
\safemath{\bmul}{\mathbf{l}}
\safemath{\bmum}{\mathbf{m}}
\safemath{\bmun}{\mathbf{n}}
\safemath{\bmuo}{\mathbf{o}}
\safemath{\bmup}{\mathbf{p}}
\safemath{\bmuq}{\mathbf{q}}
\safemath{\bmur}{\mathbf{r}}
\safemath{\bmus}{\mathbf{s}}
\safemath{\bmut}{\mathbf{t}}
\safemath{\bmuu}{\mathbf{u}}
\safemath{\bmuv}{\mathbf{v}}
\safemath{\bmuw}{\mathbf{w}}
\safemath{\bmux}{\mathbf{x}}
\safemath{\bmuy}{\mathbf{y}}
\safemath{\bmuz}{\mathbf{z}}

\safemath{\bmualpha}{\bmualphad}
\safemath{\bmubeta}{\bmubetad}
\safemath{\bmuchi}{\bumchid}
\safemath{\bmudelta}{\bmudeltad}
\safemath{\bmuepsilon}{\bmuepsilond}
\safemath{\bmuvarepsilon}{\bmuvarepsilond}
\safemath{\bmueta}{\bmuetad}
\safemath{\bmugamma}{\bmugammad}
\safemath{\bmuiota}{\bmuiotad}
\safemath{\bmukappa}{\bmukappad}
\safemath{\bmulambda}{\bmulambdad}
\safemath{\bmumu}{\bmumud}
\safemath{\bmunu}{\bmunud}
\safemath{\bmuomega}{\bmuomegad}
\safemath{\bmuphi}{\bmuphid}
\safemath{\bmuvarphi}{\bmuvarphid}
\safemath{\bmupi}{\bmupid}
\safemath{\bmuvarpi}{\bmuvarpid}
\safemath{\bmupsi}{\bmupsid}
\safemath{\bmurho}{\bmurhod}
\safemath{\bmuvarrho}{\bmuvarrhod}
\safemath{\bmusigma}{\bmusigmad}
\safemath{\bmuvarsigma}{\bmuvarsigmad}
\safemath{\bmutau}{\bmutaud}
\safemath{\bmutheta}{\bmuthetad}
\safemath{\bmuvartheta}{\bmuvarthetad}
\safemath{\bmuupsilon}{\bmuupsilond}
\safemath{\bmuxi}{\bmuxid}
\safemath{\bmuzeta}{\bmuzetad}

\bmdefine{\bmiad}{a}
\bmdefine{\bmibd}{b}
\bmdefine{\bmicd}{c}
\bmdefine{\bmidd}{d}
\bmdefine{\bmied}{e}
\bmdefine{\bmifd}{f}
\bmdefine{\bmigd}{g}
\bmdefine{\bmihd}{h}
\bmdefine{\bmiid}{i}
\bmdefine{\bmijd}{j}
\bmdefine{\bmikd}{k}
\bmdefine{\bmild}{l}
\bmdefine{\bmimd}{m}
\bmdefine{\bmind}{n}
\bmdefine{\bmiod}{o}
\bmdefine{\bmipd}{p}
\bmdefine{\bmiqd}{q}
\bmdefine{\bmird}{r}
\bmdefine{\bmisd}{s}
\bmdefine{\bmitd}{t}
\bmdefine{\bmiud}{u}
\bmdefine{\bmivd}{v}
\bmdefine{\bmiwd}{w}
\bmdefine{\bmixd}{x}
\bmdefine{\bmiyd}{y}
\bmdefine{\bmizd}{z}

\bmdefine{\bmialphad}{\alpha}
\bmdefine{\bmibetad}{\beta}
\bmdefine{\bmichid}{\chi}
\bmdefine{\bmideltad}{\delta}
\bmdefine{\bmiepsilond}{\epsilon}
\bmdefine{\bmivarepsilond}{\varepsilon}
\bmdefine{\bmietad}{\eta}
\bmdefine{\bmigammad}{\gamma}
\bmdefine{\bmiiotad}{\iota}
\bmdefine{\bmikappad}{\kappa}
\bmdefine{\bmivarkappad}{\varkappa}
\bmdefine{\bmilambdad}{\lambda}
\bmdefine{\bmimud}{\mu}
\bmdefine{\bminud}{\nu}
\bmdefine{\bmiomegad}{\omega}
\bmdefine{\bmiphid}{\phi}
\bmdefine{\bmivarphid}{\varphi}
\bmdefine{\bmipid}{\pi}
\bmdefine{\bmivarpid}{\varpi}
\bmdefine{\bmipsid}{\psi}
\bmdefine{\bmirhod}{\rho}
\bmdefine{\bmivarrhod}{\varrho}
\bmdefine{\bmisigmad}{\sigma}
\bmdefine{\bmivarsigmad}{\varsigma}
\bmdefine{\bmitaud}{\tau}
\bmdefine{\bmithetad}{\theta}
\bmdefine{\bmivarthetad}{\vartheta}
\bmdefine{\bmiupsilond}{\upsilon}
\bmdefine{\bmixid}{\xi}
\bmdefine{\bmizetad}{\zeta}

\safemath{\bmia}{\bmiad}
\safemath{\bmib}{\bmibd}
\safemath{\bmic}{\bmicd}
\safemath{\bmid}{\bmidd}
\safemath{\bmie}{\bmied}
\safemath{\bmif}{\bmifd}
\safemath{\bmig}{\bmigd}
\safemath{\bmih}{\bmihd}
\safemath{\bmii}{\bmiid}
\safemath{\bmij}{\bmijd}
\safemath{\bmik}{\bmikd}
\safemath{\bmil}{\bmild}
\safemath{\bmim}{\bmimd}
\safemath{\bmin}{\bmind}
\safemath{\bmio}{\bmiod}
\safemath{\bmip}{\bmipd}
\safemath{\bmiq}{\bmiqd}
\safemath{\bmir}{\bmird}
\safemath{\bmis}{\bmisd}
\safemath{\bmit}{\bmitd}
\safemath{\bmiu}{\bmiud}
\safemath{\bmiv}{\bmivd}
\safemath{\bmiw}{\bmiwd}
\safemath{\bmix}{\bmixd}
\safemath{\bmiy}{\bmiyd}
\safemath{\bmiz}{\bmizd}

\safemath{\bmialpha}{\bmialphad}
\safemath{\bmibeta}{\bmibetad}
\safemath{\bmichi}{\bmichid}
\safemath{\bmidelta}{\bmideltad}
\safemath{\bmiepsilon}{\bmiepsilond}
\safemath{\bmivarepsilon}{\bmivarepsilond}
\safemath{\bmieta}{\bmietad}
\safemath{\bmigamma}{\bmigammad}
\safemath{\bmiiota}{\bmiiotad}
\safemath{\bmikappa}{\bmikappad}
\safemath{\bmivarkappa}{\bmivarkappad}
\safemath{\bmilambda}{\bmilambdad}
\safemath{\bmimu}{\bmimud}
\safemath{\bminu}{\bminud}
\safemath{\bmiomega}{\bmiomegad}
\safemath{\bmiphi}{\bmiphid}
\safemath{\bmivarphi}{\bmivarphid}
\safemath{\bmipi}{\bmipid}
\safemath{\bmivarpi}{\bmivarpid}
\safemath{\bmipsi}{\bmipsid}
\safemath{\bmirho}{\bmirhod}
\safemath{\bmivarrho}{\bmivarrhod}
\safemath{\bmisigma}{\bmisigmad}
\safemath{\bmivarsigma}{\bmivarsigmad}
\safemath{\bmitau}{\bmitaud}
\safemath{\bmitheta}{\bmithetad}
\safemath{\bmivartheta}{\bmivarthetad}
\safemath{\bmiupsilon}{\bmiupsilond}
\safemath{\bmixi}{\bmixid}
\safemath{\bmizeta}{\bmizetad}

\bmdefine{\bmuDeltad}{\Updelta}
\bmdefine{\bmuGammad}{\Upgamma}
\bmdefine{\bmuLambdad}{\Uplambda}
\bmdefine{\bmuOmegad}{\Upomega}
\bmdefine{\bmuPhid}{\Upphi}
\bmdefine{\bmuPid}{\Uppi}
\bmdefine{\bmuPsid}{\Uppsi}
\bmdefine{\bmuSigmad}{\Upsigma}
\bmdefine{\bmuThetad}{\Uptheta}
\bmdefine{\bmuUpsilond}{\Upupsilon}
\bmdefine{\bmuXid}{\Upxi}

\safemath{\bmuA}{\mathbf{A}}
\safemath{\bmuB}{\mathbf{B}}
\safemath{\bmuC}{\mathbf{C}}
\safemath{\bmuD}{\mathbf{D}}
\safemath{\bmuE}{\mathbf{E}}
\safemath{\bmuF}{\mathbf{F}}
\safemath{\bmuG}{\mathbf{G}}
\safemath{\bmuH}{\mathbf{H}}
\safemath{\bmuI}{\mathbf{I}}
\safemath{\bmuJ}{\mathbf{J}}
\safemath{\bmuK}{\mathbf{K}}
\safemath{\bmuL}{\mathbf{L}}
\safemath{\bmuM}{\mathbf{M}}
\safemath{\bmuN}{\mathbf{N}}
\safemath{\bmuO}{\mathbf{O}}
\safemath{\bmuP}{\mathbf{P}}
\safemath{\bmuQ}{\mathbf{Q}}
\safemath{\bmuR}{\mathbf{R}}
\safemath{\bmuS}{\mathbf{S}}
\safemath{\bmuT}{\mathbf{T}}
\safemath{\bmuU}{\mathbf{U}}
\safemath{\bmuV}{\mathbf{V}}
\safemath{\bmuW}{\mathbf{W}}
\safemath{\bmuX}{\mathbf{X}}
\safemath{\bmuY}{\mathbf{Y}}
\safemath{\bmuZ}{\mathbf{Z}}

\safemath{\bmuZero}{\mathbf{0}}
\safemath{\bmuOne}{\mathbf{1}}

\safemath{\bmuDelta}{\bmuDeltad}
\safemath{\bmuGamma}{\bmuGammad}
\safemath{\bmuLambda}{\bmuLambdad}
\safemath{\bmuOmega}{\bmuOmegad}
\safemath{\bmuPhi}{\bmuPhid}
\safemath{\bmuPi}{\bmuPid}
\safemath{\bmuPsi}{\bmuPsid}
\safemath{\bmuSigma}{\bmuSigmad}
\safemath{\bmuTheta}{\bmuThetad}
\safemath{\bmuUpsilon}{\bmuUpsilond}
\safemath{\bmuXi}{\bmuXid}

\bmdefine{\bmiAd}{A}
\bmdefine{\bmiBd}{B}
\bmdefine{\bmiCd}{C}
\bmdefine{\bmiDd}{D}
\bmdefine{\bmiEd}{E}
\bmdefine{\bmiFd}{F}
\bmdefine{\bmiGd}{G}
\bmdefine{\bmiHd}{H}
\bmdefine{\bmiId}{I}
\bmdefine{\bmiJd}{J}
\bmdefine{\bmiKd}{K}
\bmdefine{\bmiLd}{L}
\bmdefine{\bmiMd}{M}
\bmdefine{\bmiOd}{N}
\bmdefine{\bmiPd}{O}
\bmdefine{\bmiQd}{P}
\bmdefine{\bmiRd}{R}
\bmdefine{\bmiSd}{S}
\bmdefine{\bmiTd}{T}
\bmdefine{\bmiUd}{U}
\bmdefine{\bmiVd}{V}
\bmdefine{\bmiWd}{W}
\bmdefine{\bmiXd}{X}
\bmdefine{\bmiYd}{Y}
\bmdefine{\bmiZd}{Z}

\bmdefine{\bmiDeltad}{\Delta}
\bmdefine{\bmiGammad}{\Gamma}
\bmdefine{\bmiLambdad}{\Lambda}
\bmdefine{\bmiOmegad}{\Omega}
\bmdefine{\bmiPhid}{\Phi}
\bmdefine{\bmiPid}{\Pi}
\bmdefine{\bmiPsid}{\Psi}
\bmdefine{\bmiSigmad}{\Sigma}
\bmdefine{\bmiThetad}{\Theta}
\bmdefine{\bmiUpsilond}{\Upsilon}
\bmdefine{\bmiXid}{\Xi}

\safemath{\bmiA}{\bmiAd}
\safemath{\bmiB}{\bmiBd}
\safemath{\bmiC}{\bmiCd}
\safemath{\bmiD}{\bmiDd}
\safemath{\bmiE}{\bmiEd}
\safemath{\bmiF}{\bmiFd}
\safemath{\bmiG}{\bmiGd}
\safemath{\bmiH}{\bmiHd}
\safemath{\bmiI}{\bmiId}
\safemath{\bmiJ}{\bmiJd}
\safemath{\bmiK}{\bmiKd}
\safemath{\bmiL}{\bmiLd}
\safemath{\bmiM}{\bmiMd}
\safemath{\bmiN}{\bmiNd}
\safemath{\bmiO}{\bmiOd}
\safemath{\bmiP}{\bmiPd}
\safemath{\bmiQ}{\bmiQd}
\safemath{\bmiR}{\bmiRd}
\safemath{\bmiS}{\bmiSd}
\safemath{\bmiT}{\bmiTd}
\safemath{\bmiU}{\bmiUd}
\safemath{\bmiV}{\bmiVd}
\safemath{\bmiW}{\bmiWd}
\safemath{\bmiX}{\bmiXd}
\safemath{\bmiY}{\bmiYd}
\safemath{\bmiZ}{\bmiZd}

\safemath{\bmiDelta}{\bmiDeltad}
\safemath{\bmiGamma}{\bmiGammad}
\safemath{\bmiLambda}{\bmiLambdad}
\safemath{\bmiOmega}{\bmiOmegad}
\safemath{\bmiPhi}{\bmiPhid}
\safemath{\bmiPi}{\bmiPid}
\safemath{\bmiPsi}{\bmiPsid}
\safemath{\bmiSigma}{\bmiSigmad}
\safemath{\bmiTheta}{\bmiThetad}
\safemath{\bmiUpsilon}{\bmiUpsilond}
\safemath{\bmiXi}{\bmiXid}

\safemath{\setA}{\mathcal{A}}
\safemath{\setB}{\mathcal{B}}
\safemath{\setC}{\mathcal{C}}
\safemath{\setD}{\mathcal{D}}
\safemath{\setE}{\mathcal{E}}
\safemath{\setF}{\mathcal{F}}
\safemath{\setG}{\mathcal{G}}
\safemath{\setH}{\mathcal{H}}
\safemath{\setI}{\mathcal{I}}
\safemath{\setJ}{\mathcal{J}}
\safemath{\setK}{\mathcal{K}}
\safemath{\setL}{\mathcal{L}}
\safemath{\setM}{\mathcal{M}}
\safemath{\setN}{\mathcal{N}}
\safemath{\setO}{\mathcal{O}}
\safemath{\setP}{\mathcal{P}}
\safemath{\setQ}{\mathcal{Q}}
\safemath{\setR}{\mathcal{R}}
\safemath{\setS}{\mathcal{S}}
\safemath{\setT}{\mathcal{T}}
\safemath{\setU}{\mathcal{U}}
\safemath{\setV}{\mathcal{V}}
\safemath{\setW}{\mathcal{W}}
\safemath{\setX}{\mathcal{X}}
\safemath{\setY}{\mathcal{Y}}
\safemath{\setZ}{\mathcal{Z}}
\safemath{\emptySet}{\varnothing}

\safemath{\colA}{\mathscr{A}}
\safemath{\colB}{\mathscr{B}}
\safemath{\colC}{\mathscr{C}}
\safemath{\colD}{\mathscr{D}}
\safemath{\colE}{\mathscr{E}}
\safemath{\colF}{\mathscr{F}}
\safemath{\colG}{\mathscr{G}}
\safemath{\colH}{\mathscr{H}}
\safemath{\colI}{\mathscr{I}}
\safemath{\colJ}{\mathscr{J}}
\safemath{\colK}{\mathscr{K}}
\safemath{\colL}{\mathscr{L}}
\safemath{\colM}{\mathscr{M}}
\safemath{\colN}{\mathscr{N}}
\safemath{\colO}{\mathscr{O}}
\safemath{\colP}{\mathscr{P}}
\safemath{\colQ}{\mathscr{Q}}
\safemath{\colR}{\mathscr{R}}
\safemath{\colS}{\mathscr{S}}
\safemath{\colT}{\mathscr{T}}
\safemath{\colU}{\mathscr{U}}
\safemath{\colV}{\mathscr{V}}
\safemath{\colW}{\mathscr{W}}
\safemath{\colX}{\mathscr{X}}
\safemath{\colY}{\mathscr{Y}}
\safemath{\colZ}{\mathscr{Z}}

\safemath{\opA}{\mathbb{A}}
\safemath{\opB}{\mathbb{B}}
\safemath{\opC}{\mathbb{C}}
\safemath{\opD}{\mathbb{D}}
\safemath{\opE}{\mathbb{E}}
\safemath{\opF}{\mathbb{F}}
\safemath{\opG}{\mathbb{G}}
\safemath{\opH}{\mathbb{H}}
\safemath{\opI}{\mathbb{I}}
\safemath{\opJ}{\mathbb{J}}
\safemath{\opK}{\mathbb{K}}
\safemath{\opL}{\mathbb{L}}
\safemath{\opM}{\mathbb{M}}
\safemath{\opN}{\mathbb{N}}
\safemath{\opO}{\mathbb{O}}
\safemath{\opP}{\mathbb{P}}
\safemath{\opQ}{\mathbb{Q}}
\safemath{\opR}{\mathbb{R}}
\safemath{\opS}{\mathbb{S}}
\safemath{\opT}{\mathbb{T}}
\safemath{\opU}{\mathbb{U}}
\safemath{\opV}{\mathbb{V}}
\safemath{\opW}{\mathbb{W}}
\safemath{\opX}{\mathbb{X}}
\safemath{\opY}{\mathbb{Y}}
\safemath{\opZ}{\mathbb{Z}}
\safemath{\opZero}{\mathbb{O}}
\safemath{\identityop}{\opI}

\safemath{\sca}{a}
\safemath{\scb}{b}
\safemath{\scc}{c}
\safemath{\scd}{d}
\safemath{\sce}{e}
\safemath{\scf}{f}
\safemath{\scg}{g}
\safemath{\sch}{h}
\safemath{\sci}{i}
\safemath{\scj}{j}
\safemath{\sck}{k}
\safemath{\scl}{l}
\safemath{\scm}{m}
\safemath{\scn}{n}
\safemath{\sco}{o}
\safemath{\scp}{p}
\safemath{\scq}{q}
\safemath{\scr}{r}
\safemath{\scs}{s}
\safemath{\sct}{t}
\safemath{\scu}{u}
\safemath{\scv}{v}
\safemath{\scw}{w}
\safemath{\scx}{x}
\safemath{\scy}{y}
\safemath{\scz}{z}

\safemath{\scA}{A}
\safemath{\scB}{B}
\safemath{\scC}{C}
\safemath{\scD}{D}
\safemath{\scE}{E}
\safemath{\scF}{F}
\safemath{\scG}{G}
\safemath{\scH}{H}
\safemath{\scI}{I}
\safemath{\scJ}{J}
\safemath{\scK}{K}
\safemath{\scL}{L}
\safemath{\scM}{M}
\safemath{\scN}{N}
\safemath{\scO}{O}
\safemath{\scP}{P}
\safemath{\scQ}{Q}
\safemath{\scR}{R}
\safemath{\scS}{S}
\safemath{\scT}{T}
\safemath{\scU}{U}
\safemath{\scV}{V}
\safemath{\scW}{W}
\safemath{\scX}{X}
\safemath{\scY}{Y}
\safemath{\scZ}{Z}

\safemath{\scalpha}{\alpha}
\safemath{\scbeta}{\beta}
\safemath{\scchi}{\chi}
\safemath{\scdelta}{\delta}
\safemath{\scepsilon}{\epsilon}
\safemath{\scvarepsilon}{\varepsilon}
\safemath{\sceta}{\eta}
\safemath{\scgamma}{\gamma}
\safemath{\sciota}{\iota}
\safemath{\sckappa}{\kappa}
\safemath{\scvarkappa}{\varkappa}
\safemath{\sclambda}{\lambda}
\safemath{\scmu}{\mu}
\safemath{\scnu}{\nu}
\safemath{\scomega}{\omega}
\safemath{\scphi}{\phi}
\safemath{\scvarphi}{\varphi}
\safemath{\scpi}{\pi}
\safemath{\scvarpi}{\varpi}
\safemath{\scpsi}{\psi}
\safemath{\scrho}{\rho}
\safemath{\scvarrho}{\varrho}
\safemath{\scsigma}{\sigma}
\safemath{\scvarsigma}{\varsigma}
\safemath{\sctau}{\tau}
\safemath{\sctheta}{\theta}
\safemath{\scvartheta}{\vartheta}
\safemath{\scupsilon}{\upsilon}
\safemath{\scxi}{\xi}
\safemath{\sczeta}{\zeta}

\safemath{\veca}{\mathbf{a}}
\safemath{\vecb}{\mathbf{b}}
\safemath{\vecc}{\mathbf{c}}
\safemath{\vecd}{\mathbf{d}}
\safemath{\vece}{\mathbf{e}}
\safemath{\vecf}{\mathbf{f}}
\safemath{\vecg}{\mathbf{g}}
\safemath{\vech}{\mathbf{h}}
\safemath{\veci}{\mathbf{i}}
\safemath{\vecj}{\mathbf{j}}
\safemath{\veck}{\mathbf{k}}
\safemath{\vecl}{\mathbf{l}}
\safemath{\vecm}{\mathbf{m}}
\safemath{\vecn}{\mathbf{n}}
\safemath{\veco}{\mathbf{o}}
\safemath{\vecp}{\mathbf{p}}
\safemath{\vecq}{\mathbf{q}}
\safemath{\vecr}{\mathbf{r}}
\safemath{\vecs}{\mathbf{s}}
\safemath{\vect}{\mathbf{t}}
\safemath{\vecu}{\mathbf{u}}
\safemath{\vectU}{\mathbf{U}}
\safemath{\vecv}{\mathbf{v}}
\safemath{\vecw}{\mathbf{w}}
\safemath{\vecx}{\mathbf{x}}
\safemath{\vectX}{\mathbf{X}}
\safemath{\vecy}{\mathbf{y}}
\safemath{\vecz}{\mathbf{z}}
\safemath{\veczero}{\mathbf{0}}
\safemath{\vecone}{\mathbf{1}}

\safemath{\vecalpha}{\upalpha}
\safemath{\vecbeta}{\upbeta}
\safemath{\vecchi}{\upchi}
\safemath{\vecdelta}{\updelta}
\safemath{\vecepsilon}{\upepsilon}
\safemath{\vecvarepsilon}{\upvarepsilon}
\safemath{\veceta}{\upeta}
\safemath{\vecgamma}{\upgamma}
\safemath{\veciota}{\upiota}
\safemath{\veckappa}{\upkappa}
\safemath{\veclambda}{\uplambda}
\safemath{\vecmu}{\text{\textmu}}
\safemath{\vecnu}{\upnu}
\safemath{\vecomega}{\upomega}
\safemath{\vecphi}{\upphi}
\safemath{\vecvarphi}{\upvarphi}
\safemath{\vecpi}{\uppi}
\safemath{\vecvarpi}{\upvarpi}
\safemath{\vecpsi}{\uppsi}
\safemath{\vecrho}{\uprho}
\safemath{\vecvarrho}{\upvarrho}
\safemath{\vecsigma}{\upsigma}
\safemath{\vecvarsigma}{\upvarsigma}
\safemath{\vectau}{\uptau}
\safemath{\vectheta}{\uptheta}
\safemath{\vecvartheta}{\upvartheta}
\safemath{\vecupsilon}{\upupsilon}
\safemath{\vecxi}{\upxi}
\safemath{\veczeta}{\upzeta}

\safemath{\matA}{\mathrm{A}}
\safemath{\matB}{\mathrm{B}}
\safemath{\matC}{\mathrm{C}}
\safemath{\matD}{\mathrm{D}}
\safemath{\matE}{\mathrm{E}}
\safemath{\matF}{\mathrm{F}}
\safemath{\matG}{\mathrm{G}}
\safemath{\matH}{\mathrm{H}}
\safemath{\matI}{\mathrm{I}}
\safemath{\matJ}{\mathrm{J}}
\safemath{\matK}{\mathrm{K}}
\safemath{\matL}{\mathrm{L}}
\safemath{\matM}{\mathrm{M}}
\safemath{\matN}{\mathrm{N}}
\safemath{\matO}{\mathrm{O}}
\safemath{\matP}{\mathrm{P}}
\safemath{\matQ}{\mathrm{Q}}
\safemath{\matR}{\mathrm{R}}
\safemath{\matS}{\mathrm{S}}
\safemath{\matT}{\mathrm{T}}
\safemath{\matU}{\mathrm{U}}
\safemath{\matV}{\mathrm{V}}
\safemath{\matW}{\mathrm{W}}
\safemath{\matX}{\mathrm{X}}
\safemath{\matY}{\mathrm{Y}}
\safemath{\matZ}{\mathrm{Z}}
\safemath{\matzero}{\mathrm{0}}

\safemath{\matDelta}{\Updelta}
\safemath{\matGamma}{\Upgammma}
\safemath{\matLambda}{\Uplambda}
\safemath{\matOmega}{\Upomega}
\safemath{\matPhi}{\Upphi}
\safemath{\matPi}{\Uppi}
\safemath{\matPsi}{\Uppsi}
\safemath{\matSigma}{\Upsigma}
\safemath{\matTheta}{\Uptheta}
\safemath{\matUpsilon}{\Upupsilon}
\safemath{\matXi}{\Upxi}

\safemath{\matidentity}{\matI}
\safemath{\matone}{\matO}

\safemath{\rnda}{\bmia}
\safemath{\rndb}{\bmib}
\safemath{\rndc}{\bmic}
\safemath{\rndd}{\bmid}
\safemath{\rnde}{\bmie}
\safemath{\rndf}{\bmif}
\safemath{\rndg}{\bmig}
\safemath{\rndh}{\bmih}
\safemath{\rndi}{\bmii}
\safemath{\rndj}{\bmij}
\safemath{\rndk}{\bmik}
\safemath{\rndl}{\bmil}
\safemath{\rndm}{\bmim}
\safemath{\rndn}{\bmin}
\safemath{\rndo}{\bmio}
\safemath{\rndp}{\bmip}
\safemath{\rndq}{\bmiq}
\safemath{\rndr}{\bmir}
\safemath{\rnds}{\bmis}
\safemath{\rndt}{\bmit}
\safemath{\rndu}{\bmiu}
\safemath{\rndv}{\bmiv}
\safemath{\rndw}{\bmiw}
\safemath{\rndx}{\bmix}
\safemath{\rndy}{\bmiy}
\safemath{\rndz}{\bmiz}

\safemath{\rndA}{\scA}
\safemath{\rndB}{\scB}
\safemath{\rndC}{\scC}
\safemath{\rndD}{\scD}
\safemath{\rndE}{\scE}
\safemath{\rndF}{\scF}
\safemath{\rndG}{\scG}
\safemath{\rndH}{\scH}
\safemath{\rndI}{\scI}
\safemath{\rndJ}{\scJ}
\safemath{\rndK}{\scK}
\safemath{\rndL}{\scL}
\safemath{\rndM}{\scM}
\safemath{\rndN}{\scN}
\safemath{\rndO}{\scO}
\safemath{\rndP}{\scP}
\safemath{\rndQ}{\scQ}
\safemath{\rndR}{\scR}
\safemath{\rndS}{\scS}
\safemath{\rndT}{\scT}
\safemath{\rndU}{\scU}
\safemath{\rndV}{\scV}
\safemath{\rndW}{\scW}
\safemath{\rndX}{\scX}
\safemath{\rndY}{\scY}
\safemath{\rndZ}{\scZ}

\safemath{\rndalpha}{\bmialpha}
\safemath{\rndbeta}{\bmibeta}
\safemath{\rndchi}{\bmichi}
\safemath{\rnddelta}{\bmidelta}
\safemath{\rndepsilon}{\bmiepsilon}
\safemath{\rndvarepsilon}{\bmivarepsilon}
\safemath{\rndeta}{\bmieta}
\safemath{\rndgamma}{\bmigamma}
\safemath{\rndiota}{\bmiiota}
\safemath{\rndkappa}{\bmikappa}
\safemath{\rndlambda}{\bmilambda}
\safemath{\rndmu}{\bmimu}
\safemath{\rndnu}{\bminu}
\safemath{\rndomega}{\bmiomega}
\safemath{\rndphi}{\bmiphi}
\safemath{\rndvarphi}{\bmivarphi}
\safemath{\rndpi}{\bmipi}
\safemath{\rndvarpi}{\bmivarpi}
\safemath{\rndpsi}{\bmipsi}
\safemath{\rndrho}{\bmirho}
\safemath{\rndvarrho}{\bmivarrho}
\safemath{\rndsigma}{\bmisigma}
\safemath{\rndvarsigma}{\bmivarsigma}
\safemath{\rndtau}{\bmitau}
\safemath{\rndtheta}{\bmitheta}
\safemath{\rndvartheta}{\bmivartheta}
\safemath{\rndupsilon}{\bmiupsilon}
\safemath{\rndxi}{\bmixi}
\safemath{\rndzeta}{\bmizeta}

\safemath{\rveca}{\bmua}
\safemath{\rvecb}{\bmub}
\safemath{\rvecc}{\bmuc}
\safemath{\rvecd}{\bmud}
\safemath{\rvece}{\bmue}
\safemath{\rvecf}{\bmuf}
\safemath{\rvecg}{\bmug}
\safemath{\rvech}{\bmuh}
\safemath{\rveci}{\bmui}
\safemath{\rvecj}{\bmuj}
\safemath{\rveck}{\bmuk}
\safemath{\rvecl}{\bmul}
\safemath{\rvecm}{\bmum}
\safemath{\rvecn}{\bmun}
\safemath{\rveco}{\bmuo}
\safemath{\rvecp}{\bmup}
\safemath{\rvecq}{\bmuq}
\safemath{\rvecr}{\bmur}
\safemath{\rvecs}{\bmus}
\safemath{\rvect}{\bmut}
\safemath{\rvecu}{\bmuu}
\safemath{\rvecv}{\bmuv}
\safemath{\rvecw}{\bmuw}
\safemath{\rvecx}{\bmux}
\safemath{\rvecy}{\bmuy}
\safemath{\rvecz}{\bmuz}

\safemath{\rvecalpha}{\bmualpha}
\safemath{\rvecbeta}{\bmubeta}
\safemath{\rvecchi}{\bmuchi}
\safemath{\rvecdelta}{\bmudelta}
\safemath{\rvecepsilon}{\bmuepsilon}
\safemath{\rvecvarepsilon}{\bmuvarepsilon}
\safemath{\rveceta}{\bmueta}
\safemath{\rvecgamma}{\bmugamma}
\safemath{\rveciota}{\bmuiota}
\safemath{\rveckappa}{\bmukappa}
\safemath{\rveclambda}{\bmulambda}
\safemath{\rvecmu}{\bmumu}
\safemath{\rvecnu}{\bmunu}
\safemath{\rvecomega}{\bmuomega}
\safemath{\rvecphi}{\bmuphi}
\safemath{\rvecvarphi}{\bmuvarphi}
\safemath{\rvecpi}{\bmupi}
\safemath{\rvecvarpi}{\bmuvarpi}
\safemath{\rvecpsi}{\bmupsi}
\safemath{\rvecrho}{\bmurho}
\safemath{\rvecvarrho}{\bmuvarrho}
\safemath{\rvecsigma}{\bmusigma}
\safemath{\rvecvarsigma}{\bmuvarsigma}
\safemath{\rvectau}{\bmutau}
\safemath{\rvectheta}{\bmutheta}
\safemath{\rvecvartheta}{\bmuvartheta}
\safemath{\rvecupsilon}{\bmuupsilon}
\safemath{\rvecxi}{\bmuxi}
\safemath{\rveczeta}{\bmuzeta}

\safemath{\rmatA}{\bmuA}
\safemath{\rmatB}{\bmuB}
\safemath{\rmatC}{\bmuC}
\safemath{\rmatD}{\bmuD}
\safemath{\rmatE}{\bmuE}
\safemath{\rmatF}{\bmuF}
\safemath{\rmatG}{\bmuG}
\safemath{\rmatH}{\bmuH}
\safemath{\rmatI}{\bmuI}
\safemath{\rmatJ}{\bmuJ}
\safemath{\rmatK}{\bmuK}
\safemath{\rmatL}{\bmuL}
\safemath{\rmatM}{\bmuM}
\safemath{\rmatN}{\bmuN}
\safemath{\rmatO}{\bmuO}
\safemath{\rmatP}{\bmuP}
\safemath{\rmatQ}{\bmuQ}
\safemath{\rmatR}{\bmuR}
\safemath{\rmatS}{\bmuS}
\safemath{\rmatT}{\bmuT}
\safemath{\rmatU}{\bmuU}
\safemath{\rmatV}{\bmuV}
\safemath{\rmatW}{\bmuW}
\safemath{\rmatX}{\bmuX}
\safemath{\rmatY}{\bmuY}
\safemath{\rmatZ}{\bmuZ}

\safemath{\rmatDelta}{\bmuDelta}
\safemath{\rmatGamma}{\bmuGamma}
\safemath{\rmatLambda}{\bmuLambda}
\safemath{\rmatOmega}{\bmuOmega}
\safemath{\rmatPhi}{\bmuPhi}
\safemath{\rmatPi}{\bmuPi}
\safemath{\rmatPsi}{\bmuPsi}
\safemath{\rmatSigma}{\bmuSigma}
\safemath{\rmatTheta}{\bmuTheta}
\safemath{\rmatUpsilon}{\bmuUpsilon}
\safemath{\rmatXi}{\bmuXi}

\safemath{\rndvecU}{\bmiU} % random vector U
\safemath{\rndvecW}{\bmiW} % random vector W
\safemath{\rndvecM}{\bmiM} % random vector M
\safemath{\rndvecV}{\bmiV} % random vector V
\safemath{\rndvecX}{\bmiX} % random vector X
\safemath{\rndvecY}{\bmiY} % random vector Y
\safemath{\rndvecC}{\bmiC} % random vector C

\newenvironment{textbmatrix}{	\setlength{\arraycolsep}{2.5pt}%
								\big[\begin{matrix}}{\end{matrix}\big]%
								\raisebox{0.08ex}{\vphantom{M}}}

\def\be{\begin{equation}}
\def\ee{\end{equation}}
\def\een{\nonumber \end{equation}}
\def\mat{\begin{bmatrix}}
\def\emat{\end{bmatrix}}
\def\btm{\begin{textbmatrix}}
\def\etm{\end{textbmatrix}}

\def\ba#1\ea{\begin{align}#1\end{align}}
\def\bas#1\eas{\begin{align*}#1\end{align*}}
\def\bs#1\es{\begin{split}#1\end{split}} 
\def\bg#1\eg{\begin{gather}#1\end{gather}}
\def\bml#1\eml{\begin{multline}#1\end{multline}}
\def\bi#1\ei{\begin{itemize}#1\end{itemize}} 
\def\bipi#1\eipi{\begin{inparaitem}#1\end{inparaitem}}

\DeclareMathOperator{\adj}{adj}				% adjunct matrix
\DeclareMathOperator*{\argmin}{arg\;min}		% arg min
\DeclareMathOperator{\Exop}{\opE}			% expectation operator
\safemath{\fun}{\scf}						% generic scalar function
\safemath{\altfun}{\scg}
\safemath{\aaltfun}{\sch}
\safemath{\bel}{\sce}					% basis element
\safemath{\altbel}{\sce}					% alternative basis element
\safemath{\frel}{g}					% frame element
\safemath{\altfrel}{g}					% alternative frame element
\safemath{\dfrel}{\tilde{g}}					% dual frame element
\safemath{\altdfrel}{\tilde{g}}					% alternative dual frame element

\newcommand{\nullspace}{\setN}	 			% nullspace
\newcommand{\Ex}[2]{\ensuremath{\Exop_{#1} [#2]}} 	% expectation
\newcommand{\smEx}[2]{\ensuremath{\Exop_{#1} [#2 ]}} 	% expectation
\newcommand{\bigEx}[2]{\ensuremath{\Exop_{#1} \bigl[#2\bigr]}} 	% expectation
\newcommand{\BigEx}[2]{\ensuremath{\Exop_{#1} \Bigl[#2\Bigr]}} 	% expectation
\newcommand{\BiggEx}[2]{\ensuremath{\Exop_{#1} \Biggl[#2\Biggr]}} 	% expectation
\newcommand{\card}[1]{| #1 |}			% cardinality of a set
\newcommand{\ind}[1]{\mathbbm{1}_{#1}}				% indicator function
\newcommand{\conj}[1]{\ensuremath{\overline{#1}}} 	% conjugate 		
\newcommand{\inv}[1]{\ensuremath{#1^{-1}}} 	% inverse
\safemath{\dirac}{\delta}					% Dirac delta
\safemath{\diracp}{\dirac(\time)}			% 	''	parametrized
\safemath{\krond}{\dirac}					% Kronecker delta
\newcommand{\set}[1]{\{#1\}}		% set notation
\safemath{\upto}{\uparrow}
\safemath{\downto}{\downarrow}
\safemath{\iu}{i}							% imaginary unit
\safemath{\maj}{\succ}
\newcommand{\dftmat}[1]{\matF_{#1}}			% DFT matrix
\safemath{\mdft}{\dftmat{}}					% 	''
\safemath{\runity}{\beta}					% root of unity
\safemath{\eval}{\lambda}					% eigenvalue
\safemath{\veval}{\veclambda}				% eigenvalue vector
\safemath{\littleo}{\sco}					% Landau\s little o

\let\im\undefined
\safemath{\re}{\mathfrak{Re}}				% real part
\safemath{\im}{\mathfrak{Im}}				% imaginary part
\safemath{\euclidspace}{\complexset}			% Euclidean space
\safemath{\confunspace}{\setC}				% space of continuous functions
\newcommand{\banachseqspace}[1]{l^{#1}}		% Banach sequence space
\safemath{\hilseqspace}{\banachseqspace{2}}	% Hilbert sequence space
\newcommand{\banachfunspace}[1]{\setL^{#1}}	% Banach function space
\safemath{\hilfunspace}{\banachfunspace{2}}	% Hilbert function space
\safemath{\schwarzspace}{\setS}				% Schwarz space

\newcommand{\hadj}[1]{#1^{\star}}			% Hilbert adjoint operator

\safemath{\SNR}{\text{\sc snr}} 				% signal to noise ratio
\safemath{\No}{N_0}							% noise spectral density
\safemath{\Es}{E_s}							% energy per symbol
\safemath{\Eb}{E_b}							% energy per bit
\safemath{\EbNo}{\frac{\Eb}{\No}}
\safemath{\EsNo}{\frac{\Es}{\No}}

\let\time\undefined
\safemath{\time}{\sct}						% continuous time
\safemath{\dtime}{\sck}						% discrete time
\safemath{\delay}{\sctau}					% continuous delay
\safemath{\ddelay}{\scl}						% discrete delay
\safemath{\doppler}{\scnu}					% continuous doppler
\safemath{\ddoppler}{\scm}					% discrete doppler
\safemath{\freq}{\scf}						% frequency
\safemath{\dfreq}{\scn}						% discrete frequency
\safemath{\Dtime}{\Delta\time}
\safemath{\Dfreq}{\Delta\freq}
\safemath{\Ddtime}{\Delta\dtime}
\safemath{\Ddfreq}{\Delta\dfreq}
\safemath{\bandwidth}{\scB}
\safemath{\maxdoppler}{\doppler_{0}}			% maximum Doppler shift
\safemath{\maxdelay}{\delay_{0}}				% maximum delay
\safemath{\spread}{\Delta_{\CHop}}			% total channel spread
\DeclareMathOperator{\CHop}{\ensuremath{\opH}} % channel operator
\safemath{\kernel}{\rndk_{\CHop}}			% operator kernel
\safemath{\kernelp}{\kernel(\time,\time')}	% 	''	parametrized
\safemath{\tvir}{\rndh_{\CHop}}				% time-varying impulse response
\safemath{\tvirp}{\tvir(\time,\delay)}		%	''	parametrized
\safemath{\tvirc}{\conj{\rndh}_{\CHop}}		% 	''	parametrized
\safemath{\tvtf}{\rndl_{\CHop}}				% time-varying transfer function
\safemath{\tvtfp}{\tvtf(\time,\freq)}			%	''	parametrized
\safemath{\tvtfc}{\conj{\rndl}_{\CHop}}		%	''	parametrized
\safemath{\spf}{\rnds_{\CHop}}				% spreading function
\safemath{\spfp}{\spf(\doppler,\delay)}		%	''	parametrized
\safemath{\spfc}{\conj{\rnds}_{\CHop}}		%	''	parametrized
\safemath{\bff}{\rndb_{\CHop}}				% bi-freuqency function
\safemath{\bffp}{\bff(\doppler,\freq)}		%	''	parametrized

\safemath{\irc}{\scr_{\rndh}}				% impulse response correlation fn.
\safemath{\tfc}{\scr_{\rndl}}				% time-frequency correlation fn.
\safemath{\spc}{\scr_{\rnds}}				% spreading fn. correlation fn.
\safemath{\bfc}{\scr_{\rndb}}				% bi-frequency correlation fn.

\safemath{\scaf}{\scc_{\rnds}}				% scattering function
\safemath{\scafp}{\scaf(\doppler,\delay)}		% 	''	parametrized
\safemath{\ccf}{\scc_{\rndl}}				% WSSUS tvtf correlation
\safemath{\ccfp}{\ccf(\Dtime,\Dfreq)}			% 	''	parametrized
\safemath{\cic}{\scc_{\rndh}}				% WSSUS tvir correlation
\safemath{\cicp}{\cic(\Dtime,\delay)}			% 	''	parametrized

\safemath{\mi}{\scI}							% mutual information
\safemath{\capacity}{\scC}					% capacity

\DeclareMathOperator{\Prob}{\opP}		% probability of an event
\safemath{\normal}{\mathcal{N}}			% normal distribution
\safemath{\jpg}{\mathcal{CN}}			% jointly proper Gaussian
\safemath{\mchain}{\leftrightarrow}		% Markov chain
\safemath{\dB}{\,\mathrm{dB}}
\safemath{\dBm}{\,\mathrm{dBm}}
\safemath{\Hz}{\,\mathrm{Hz}}
\safemath{\kHz}{\,\mathrm{kHz}}
\safemath{\MHz}{\,\mathrm{MHz}}
\safemath{\GHz}{\,\mathrm{GHz}}
\safemath{\s}{\,\mathrm{s}}
\safemath{\ms}{\,\mathrm{ms}}
\safemath{\mus}{\,\mathrm{\text{\textmu}s}}
\safemath{\ns}{\,\mathrm{ns}}
\safemath{\ps}{\,\mathrm{ps}}
\safemath{\meter}{\,\mathrm{m}}
\safemath{\mm}{\,\mathrm{mm}}
\safemath{\cm}{\,\mathrm{cm}}
\safemath{\m}{\,\mathrm{m}}
\safemath{\W}{\,\mathrm{W}}
\safemath{\mW}{\, \mathrm{mW}}
\safemath{\J}{\,\mathrm{J}}
\safemath{\K}{\,\mathrm{K}}
\safemath{\bit}{\,\mathrm{bit}}
\safemath{\nat}{\,\mathrm{nat}}

\safemath{\define}{\triangleq}					% definition

\safemath{\equivalent}{\sim}
\safemath{\distas}{\sim}					% distributed according to
\newcommand{\OR}{\,\mathbf{or}\,}		% logical OR
\safemath{\sdiff}{\Delta}				% symmetric set difference

\safemath{\reals}{\mathbb R}
\safemath{\positivereals}{\reals_{+}}
\safemath{\integers}{\mathbb Z}
\safemath{\posint}{\integers_{+}}
\safemath{\naturals}{\mathbb N}
\safemath{\posnaturals}{\naturals_{+}}
\safemath{\complexset}{\mathbb C}
\safemath{\rationals}{\mathbb Q}

\newcommand*{\fancyrefparlabelprefix}{par}		% Part
\newcommand*{\fancyrefchalabelprefix}{cha}		% Chapter
\newcommand*{\fancyrefapplabelprefix}{app}		% Appendix
\newcommand*{\fancyrefthmlabelprefix}{thm}		% Theorem
\newcommand*{\fancyreflemlabelprefix}{lem}		% Lemma
\newcommand*{\fancyrefcorlabelprefix}{cor}		% Corolary
\newcommand*{\fancyrefdeflabelprefix}{def}		% Definition

\frefformat{vario}{\fancyrefparlabelprefix}{Part~#1}
\frefformat{vario}{\fancyrefchalabelprefix}{Chapter~#1}
\frefformat{vario}{\fancyrefseclabelprefix}{Section~#1}
\frefformat{vario}{\fancyrefthmlabelprefix}{Theorem~#1}
\frefformat{vario}{\fancyreflemlabelprefix}{Lemma~#1}
\frefformat{vario}{\fancyrefcorlabelprefix}{Corollary~#1}
\frefformat{vario}{\fancyrefdeflabelprefix}{Definition~#1}
\frefformat{vario}{\fancyreffiglabelprefix}{Figure~#1}
\frefformat{vario}{\fancyrefapplabelprefix}{Appendix~#1}
\frefformat{vario}{\fancyrefeqlabelprefix}{(#1)}

\safemath{\iSet}{\setI}
\safemath{\rel}{\bowtie}					% relation
\safemath{\eqrel}{\sim}					% equivalence relation
\safemath{\rlord}{\leq}					% reflexive linear ordering
\safemath{\slord}{<}						% strict linear ordering
\safemath{\rpord}{\preceq}				% reflexive partial ordering
\safemath{\rrpord}{\succeq}				% reversed reflexive partial ordering
\safemath{\spord}{\prec}					% strict partial ordering
\safemath{\sig}{\sigma}					% sigma-{algebra, ring,...}
\safemath{\metric}{d}					% metric
\safemath{\setfun}{\Phi}					% set function
\safemath{\measure}{\mu}					% measure
\newcommand{\outerm}[1]{#1^{\star}}		% marks outer measures.
\newcommand{\innerm}[1]{#1_{\star}}		% marks inner measures.
\safemath{\omeasure}{\outerm{\measure}}		% outer measure
\safemath{\imeasure}{\innerm{\measure}}		% inner measure	
\safemath{\aecol}{\colS^{\star}_{\measure}} % collection of almost equal sets
\safemath{\emeasure}{\bar{\measure}_{0}}	% measure extension
\safemath{\rmeasure}{\tilde{\measure}}	% restricted measure
\safemath{\bmeasure}{\measure_{0}}		% basic measure on a semiring
\safemath{\glength}{\measure_{\altfun}}	% generalized length
\safemath{\lebmea}{\lambda}				% Lebesgue length and measure
\safemath{\blebmea}{\lebmea_{0}}			% pre-lebesgue-measure
\safemath{\sfun}{s}						% simple function
\safemath{\absintspace}{\colL^{1}}		% space of abs. integrable functions
\safemath{\sqintspace}{\colL^{2}}		% space of square integrable functions
\safemath{\abssumspace}{l^{1}}		% space of abs. summable sequences
\safemath{\sqsumspace}{l^{2}}		% space of square summable sequences
\safemath{\boundspace}{l^{\infty}}	% space of bounded sequences

\safemath{\sfield}{\setF}				% scalar field
\safemath{\vectors}{\setV}				% set of vectors
\safemath{\vecspace}{(\vectors,\sfield)}	% vector space
\safemath{\linspace}{\setV}				% linear space
\safemath{\clinspace}{(\linspace,\sfield)} % linear space
\safemath{\nspace}{\setU}				% normed space
\safemath{\metspace}{\setM}				% metric space
\safemath{\bspace}{\setB}				% Banach space
\safemath{\ipspace}{\setG}				% inner product space
\safemath{\hilspace}{\setH}				% Hilbert space
\safemath{\blospace}{\setG}				% set of bounded linear oprators
\safemath{\lop}{\opT}					% linear operator
\safemath{\altlop}{\opS}					% alternative linear operator
\safemath{\nullsp}{\nullspace(\lop)}		% null space of the linear operator
\safemath{\lfun}{l}						% linear functional
\safemath{\altlfun}{g}					% alternative linear functional
\newcommand{\dual}[1]{#1^{'}}			% dual space
\safemath{\dsum}{\oplus}					% direct sum
\safemath{\funspace}{\colL}				% function space
\renewcommand{\adj}[1]{#1^{\times}}		% adjoint operator generator
\safemath{\adjlop}{\adj{\lop}}			% adjoint operator
\safemath{\hadjlop}{\hadj{\lop}}			% Hilbert adjoint operator
\safemath{\tow}{\xrightarrow{w}}			% weak convergence
\safemath{\tows}{\xrightarrow{w^{*}}}		% weak* convergence
\safemath{\cparam}{\lambda}
\safemath{\lopl}{\lop_{\cparam}}		
\safemath{\iop}{\opI}					% identity operator
\safemath{\resolop}{\opR}				% resolvent operator
\safemath{\resolvent}{\resolop_{\cparam}(\lop)}	% resolvent operator
\safemath{\altresolvent}{\resolop_{\cparam}(\altlop)} % resolvent operator (alternative linear operator)
\safemath{\reset}{\setQ}
\safemath{\spectrum}{\setS}
\safemath{\resolset}{\reset(\lop)}		% resolvent set
\safemath{\lopspec}{\spectrum(\lop)}		% spectrum of a linear operator
\safemath{\altlopspec}{\spectrum(\altlop)} %spectrum of alternative linear operator	
\safemath{\pspec}{\spectrum_{p}(\lop)}	% point spectrum
\safemath{\cspec}{\spectrum_{c}(\lop)}	% continuous spectrum
\safemath{\rspec}{\spectrum_{r}(\lop)}	% residual spectrum
\safemath{\ev}{\cparam}					% eigenvalue
\newcommand{\specrad}[1]{r_{#1}}			% spectral radius
\safemath{\lopsrad}{\specrad{\lop}}		% spectral radius
\safemath{\pop}{\opP}					% projection operator

\safemath{\specfam}{\colE}				% spectral family
\safemath{\specop}{\opE_{\cparam}}		% spectral projection operator
\safemath{\altspecop}{\opE_{\mu}}		% alternat spectral projection operator
\safemath{\vmulti}{\vecone}				% vector multiplicative identity
\safemath{\unitaryop}{\opU}				% unitary operator
\safemath{\sval}{\sigma}					% singular value
\safemath{\corrcoef}{\rho}				% canonical correlation coefficient
\safemath{\sangle}{\theta}				% angle between subspaces

\let\time\undefined
\safemath{\iset}{\setI}				% index set
\safemath{\shift}{\nu}
\safemath{\scale}{\alpha}
\safemath{\time}{t}
\safemath{\specfreq}{\theta}	
\newcommand{\transopgen}[1]{\opT_{#1}} % translation operator
\safemath{\transop}{\transopgen{\delay}}
\newcommand{\modopgen}[1]{\opM_{#1}}	% modulation operator
\safemath{\modop}{\modopgen{\shift}}
\newcommand{\dilopgen}[1]{\opD_{#1}}	% dilation operator
\safemath{\dilop}{\dilopgen{\scale}}
\safemath{\fram}{\setF}				% frame
\safemath{\dfram}{\dual{\fram}}		% dual frame
\safemath{\ufb}{\scB}					% upper frame bound
\safemath{\lfb}{\scA}					% lower frame bound
\safemath{\sop}{\hadj{\aop}}				% frame synthesis operator
\safemath{\aop}{\opT}			% frame analysis operator 		// modifide by christoph bunte
\safemath{\fop}{\opS}				% frame operator 			// modified by christoph bunte 
\safemath{\daop}{\tilde\opT}			% dual frame analysis operator 		// added by christoph bunte
\safemath{\dsop}{\hadj{\tilde\opT}}				% dual frame synthesis operator 			// added by christoph bunte 

\safemath{\ifop}{\inv{\fop}}			% inverse frame operator
\safemath{\rifop}{\fop^{-1/2}}			% square root of inverse frame operator
\newcommand{\ft}[1]{\widehat{#1}}	% Fourier transform
\safemath{\transeq}{\setT}			% sequence of translates
\safemath{\nfun}{\Phi}				% ``Nyquist series''
\safemath{\funvec}{\vecf}			%  function as a vector
\safemath{\altfunvec}{\vecg}

\safemath{\samplespace}{\Omega}
\safemath{\probspace}{(\samplespace,\sfield,\Prob)}	% probability space
\safemath{\ccoef}{\rho}			% correlation coefficient
\safemath{\infstate}{\vecpi}				% steady state vector
\safemath{\typset}{\setA_{\epsilon}^{(N)}}	% typical set
\safemath{\expequal}{\doteq}				% equal to first order in the exponent
\safemath{\code}{C}						% code
\safemath{\dstringset}{\setD^{\star}}		% set of finite length D-ary strings
\safemath{\cwlength}{l}					% codeword length
\safemath{\elength}{L}					% expected codeword length
\safemath{\extension}{C^{\star}}			% code extension
\safemath{\approaches}{\rightarrow}		% i.e., x_{n} -> x
\safemath{\evnt}{\setA}					% event A
\safemath{\altevnt}{\setB}					% event B
\safemath{\rv}{\rndx}					% random variable X
\safemath{\altrv}{\rndy}					% random variable Y
\safemath{\complexrv}{\rndu}					% complex random variable U
\safemath{\altcrv}{\rndv}				% complex random variable V
\safemath{\rvec}{\rvecx}					% random vector X
\safemath{\altrvec}{\rvecy}				% random vector Y
\safemath{\crvec}{\rvecu}				% complex random vector U
\safemath{\altcrvec}{\rvecv}				% complex random vector V
\safemath{\variance}{\sigma^{2}}			% variance
\safemath{\map}{T}						% mapping
\safemath{\jacobian}{J}					% jacobian
\safemath{\wvec}{\rvecw}					% white random vector
\safemath{\wrv}{\rndw}					% white noise process
\safemath{\orthmat}{\matQ}				% orthogonal matrix
\safemath{\evmat}{\matLambda}			% eigenvalue matrix (diagonal)
\safemath{\identity}{\matidentity}		% identity matrix
\safemath{\innovec}{\vecv}				% innovations vector
\safemath{\convas}{\xrightarrow{\text{a.s.}}}	% almost sure convergence
\safemath{\convr}{\xrightarrow{\text{r}}}	% convergence in r-th mean
\safemath{\convp}{\xrightarrow{\text{P}}}	% convergence in probability
\safemath{\convd}{\xrightarrow{\text{D}}}	% convergence in distribution
\safemath{\ltis}{\opL}				% LTI system
\safemath{\ir}{h}					% impulse response
\safemath{\tf}{\MakeUppercase{\ir}}	% transfer function
{\begin{list}{}%
         {\setlength{\leftmargin}{#1}}%
         \item[]%
}
{\end{list}}

\safemath{\signal}{\scx} % signal in Hilbert space
\safemath{\signalFt}{\widehat{\signal}}	% FT of signal
\safemath{\sigTime}{\signal \! \left( \time \right)}	% time domain signal
\safemath{\orFt}{\widehat{\signal} \! \left(\freq\right)}	% FT of original signal
\safemath{\sampFt}{\widehat{\signal}_{\scd} \! \left(\freq\right)}	% FT of sampled signal
\safemath{\fZero}{\freq_{0}}	% highest frequency
\safemath{\fSamp}{\freq_{\scs}}	% sampling frequency
\safemath{\specOc}{\setI} % spectral Occupancy
\safemath{\sampSet}{\setP} % sampling Set
\safemath{\sampVal}{\signal \! \left(\time_\scn\right)}	% sampled signal points
\safemath{\beuDen}{\setD^{-} \! \left(\sampSet\right)}	% lower Beurling density
\safemath{\samPer}{\scT_\scs}	% sampling Period
\safemath{\hSp}{\setH}	% Hilbert Space
\safemath{\multSig}{\setB \! \left( \specOc \right)}	% Multi-Band signals supported on a given interval
\safemath{\Ltwo}{\sqintspace \! \left(\reals\right)}	% Hilbert space L^2
\safemath{\cellNo}{\scL} % number of cells for the multicoset sampling
\safemath{\noIn}{s} % No. of intervals in which the support set is partitioned
\safemath{\cosetNo}{\scK} % No. of cosets for multicoset sampling
\safemath{\sampSig}{\scX \! \left(\freq\right)} % Signal after application of coset sampling operator
\newcommand{\vandEnt}[1]{\scx_{#1}} % Entries of Vandermonde matrix
\safemath{\vandM}{\matV \! \left(\vandEnt{0}, \vandEnt{1}, \ldots, \vandEnt{\cosetNo - 1} \right)} % Vandermonde matrix
\safemath{\sampMat}{\matA} % The matrix according to which the measurements of the spectrum are taken
\safemath{\suppXhat}{\gamma} % Support of the vector corresponding to the spectrum
\safemath{\orFtSupp}{\widehat{\signal}_{\suppXhat} \! \left(\freq\right)} % nonzero entries of spectrum vector
\safemath{\maxCard}{\scC} % maximal cardinality of spectral occupancy
\safemath{\dict}{\matD} % dictionary
\safemath{\measVec}{\vecy} % measurement vector
\safemath{\sigVec}{\vecx} % signal
\safemath{\meas}{\scm} % number of measurements
\safemath{\sigDim}{\scn} % dimension of the signal vector
\safemath{\spars}{\scs} % sparsity level of the signal vector
\safemath{\FM}{\matF}	% Fourier matrix
\safemath{\IM}{\matI}	% Identity matrix
\safemath{\fONB}{\matA}	% first ONB when the dictionary is a concatenation
\safemath{\sONB}{\matB}	% second ONB when the dictionary is a concatenation
\safemath{\dictCol}{d} % dictionary column
\safemath{\recSigVec}{\hat{\sigVec}} % signal vector for recovery problems
\safemath{\PZ}{\left( \text{P0} \right)} % recovery problem (P0)
\safemath{\BP}{\left( \text{BP} \right)} % recovery problem (BP)
\safemath{\PO}{\left( \text{P1} \right)} % recovery problem (P1), same as (BP)
\safemath{\coher}{\mu} % dictionary coherence
\safemath{\supp}{\setS}	% support of the signal vector
\safemath{\corMeas}{\vecz} % corrupted measurement vector
\safemath{\erVec}{\vece} % error vector
\safemath{\concSig}{\check{\sigVec}} % vector obtained by concatenation of signal and error vector
\safemath{\fVec}{\vecp} % vector for representation in first frame
\safemath{\sVec}{\vecq} % vector for representation in second frame
\safemath{\conVec}{\vecv} % concatenation of the above vectors
\safemath{\sgnl}{\vecs} % signal Ap = Bq = s
\safemath{\fONBCol}{\veca} % column of the first frame
\safemath{\sONBCol}{\vecb} % column of the second frame
\safemath{\fSupp}{\setP} % support of the representation in the first frame
\safemath{\sSupp}{\setQ} % support of the representation in the second frame
\safemath{\fDim}{\sigDim_{\fONBCol}} % number of columns in the first frame
\safemath{\sDim}{\sigDim_{\sONBCol}} % number of columns in the second frame
\safemath{\fCoher}{\coher_{\fONBCol}} % coherence of the first frame
\safemath{\sCoher}{\coher_{\sONBCol}} % coherence of the second frame
\safemath{\sigSupp}{\setX} % support of the signal vector
\safemath{\erSupp}{\setE} % support of the error vector
\safemath{\PZErSup}{\left( \text{P0, } \erSupp \right)} % recovery problem (P0) with known error support
\safemath{\falseSigVec}{\tilde{\sigVec}} % alternative signal vector found by recovery problem
\safemath{\falseErVec}{\tilde{\erVec}} % alternative error vector found by recovery problem
\safemath{\sigSpars}{\sigDim_{\sigVec}} % signal sparsity
\safemath{\erSpars}{\sigDim_{\erVec}} % error sparsity
\safemath{\BPErSup}{\left( \text{BP, } \erSupp \right)} % recovery problem (BP) with known error support
\safemath{\gram}{\matG} % gram matrix D^H D

\safemath{\WHT}{\scT}	%time parameter
\safemath{\WHF}{\scF}	%frequence parameter
\safemath{\prot}{\scg}	%prototype function
\safemath{\vprot}{\vecg}	%vector prototype function
\safemath{\proti}{\prot_{\scm, \scn}}	%prototype function with indices
\safemath{\prott}{\prot \! \left( \time \right) }	%prototype function with parameter t
\safemath{\protit}{\proti \! \left(\time\right)}	%prototype function with indices and time
\safemath{\vproti}{\vprot_{\scm, \scn}}	%vector prototype function with indices
\safemath{\dprot}{\tilde{\scg}}	%dual prototype function
\safemath{\dvprot}{\tilde{\vecg}}	%dual vector prototype function
\safemath{\dproti}{\dprot_{\scm, \scn}}	%dual prototype function with indices
\safemath{\dprott}{\dprot \! \left( \time \right) }	%dual prototype function with parameter t
\safemath{\dprotit}{\dproti \! \left(\time\right)}	%dual prototype function with indices and time
\safemath{\dvproti}{\dvprot_{\scm, \scn}}	%dual vector prototype function with indices
\safemath{\WO}{\scW}		%Weyl operator
\safemath{\WOpam}{\WO^{\left( \WHT, \WHF \right)}_{\scm,\scn}}	%Weyl operator with parameters
\newcommand{\WOind}[2]{\WO_{#1, #2}} %Weyl operator with indices
\safemath{\WOsind}{\WOind \scm \scn }%Weyl operator with standard indices
\safemath{\zak}{\opZ} %Zak transform
\safemath{\zakpar}{\zak^{\WHT,\WHF}} %Zak transform with parameters
\safemath{\zaksig}{\zak_{\signal} \! \left( \time, \freq \right)} %Zak transform with standard arguments
\newcommand{\zakprot}[2]{\zak_{\prot} \! \left( #1, #2 \right)} %Zak transform of prototype function
\safemath{\zakprots}{\zakprot{\time}{\freq}}
\safemath{\zakprotis}{\zak_{\prot_{\scm,\scn}} \! \left( \time, \freq \right)} %Zak transform of shifted prototype function with standard arguments
\safemath{\tfr}{\scR} %R = T F
\safemath{\funmin}{\scm \! \left( \prott; \WHT \right)}	%minimum of  function
\safemath{\funmax}{\scM \! \left( \prott; \WHT \right)}	%maximum of  function

\safemath{\sclf}{\phi}	%scaling function
\safemath{\vscf}{\vecphi}	%vector scaling function
\newcommand{\scfa}[1]{\sclf \! \left( #1 \right)}	%scaling function with arguments
\safemath{\scft}{\scfa \time}	%scaling function with time argument
\newcommand{\vscfi}[2]{\vscf_{#1,#2}}	%vector scaling function with indices
\safemath{\vscfs}{\vscfi \scj \scn}
\safemath{\fscf}{\ft \sclf}	%Fourier transform of scaling function
\safemath{\vfscf}{\ft \vscf}	%Fourier transform of vector scaling function
\newcommand{\fscfa}[1]{\fscf \! \left( #1 \right)}	%Fourier transform of scaling function with argument
\safemath{\fscff}{\fscfa \freq}	%Fourier transform of scaling function with frequency
\safemath{\wav}{\psi}	%wavelet function
\safemath{\vwav}{\vecpsi}	%vector wavelet function
\newcommand{\wava}[1]{\wav \! \left( #1 \right)}	%wavelet function with arguments
\safemath{\wavt}{\wava \time}	%waveket function with time argument
\newcommand{\vwavi}[2]{\vwav_{#1,#2}}	%vector wavelet function with indices
\safemath{\vwavs}{\vwavi \scj \scn} %vector wavelet function with standard indices
\safemath{\fwav}{\ft \wav}	%Fourier transform of wavelet function
\safemath{\vfwav}{\ft \vwav}	%Fourier transform of vector wavelet function
\newcommand{\fwava}[1]{\fwav \! \left( #1 \right)}	%Fourier transform of wavelet function with argument
\safemath{\fwavf}{\fwava \freq}	%Fourier transform of wavelet function with frequency

\newcommand{\condmuti}[3]{\scI ( #1 , #2 | #3 )} %conditional mutual information
\newcommand{\bigcondmuti}[3]{\scI \bigl( #1 , #2 \bigl| #3 \bigr)} %conditional mutual information
\newcommand{\Biggcondmuti}[3]{\scI \Biggl( #1 , #2 \Biggl| #3 \Biggr)} %conditional mutual information
\newcommand{\condmutivars}[3]{\scI ( #1 ; #2 | #3 )} %conditional mutual information
\newcommand{\bigcondmutivars}[3]{\scI \bigl( #1 ; #2 \bigl| #3 \bigr)} %conditional mutual information
\newcommand{\ent}[1]{\scH ( #1 )} %entropy
\newcommand{\bigcondent}[2]{\scH \bigl( #1 \bigl| #2 \bigr)} %conditional entropy
\newcommand{\renent}[2]{\scH_{#1} ( #2 )} %Renyi entropy
\safemath{\oprobf}{\hat{p}}	% optimal function
\safemath{\probf}{p} % pmf function

\safemath{\binind}{\setI}			%set of bin indices
\safemath{\pool}{\bm{ \mathcal P}}			%pool
\safemath{\poolre}{\setP}			%realization bin
\safemath{\eptyp}{\setT^{(n)}_\epsilon}	%set of epsylon typical sequences
\safemath{\dist}{\mathbb{P}}			%distribution
\safemath{\tdist}{\tilde \dist}			%distribution
\newcommand{\distof}[1]{\dist [ #1 ]}	%distribution with arguments
\newcommand{\bigdistof}[1]{\dist \bigl[ #1 \bigr]}	%distribution with arguments
\newcommand{\Bigdistsubof}[2]{\dist_{#1} \! \Bigl[ #2 \Bigr]}	%distribution with arguments and subscripts
\safemath{\ry}{\setY}			%Receiver Y
\safemath{\rz}{\setZ}			%Receiver Z
\newcommand{\enc}[1]{f ( #1 )}	    %encoding function
\newcommand{\dec}[1]{f^{-1} ( #1 )}	%decoding function
\newcommand{\bigdec}[1]{f^{-1} \bigl( #1 \bigr)}	%decoding function
\newcommand{\decphi}[1]{\phi ( #1 )}	%decoding function
\newcommand{\bigdecphi}[1]{\phi \bigl( #1 \bigr)}	%decoding function
\newcommand{\guess}[2]{\scG_{#1} ( #2 )}	%guessing function
\newcommand{\bigguess}[2]{\scG_{#1} \bigl( #2 \bigr)}	%guessing function
\newcommand{\guessast}[2]{\scG^\star_{#1} ( #2 )}	%optimal guessing function
\newcommand{\bigguessast}[2]{\scG^\star_{#1} \bigl( #2 \bigr)}	%optimal guessing function
\newcommand{\guessD}[3]{\scG_{#1}^{#2} ( #3 )}	%guessing function
\newcommand{\bigguessD}[3]{\scG_{#1}^{#2} \bigl( #3 \bigr)}	%guessing function
\newcommand{\hatguessD}[3]{\hat \scG_{#1}^{#2} ( #3 )}	%guessing function
\newcommand{\bighatguessD}[3]{\hat \scG_{#1}^{#2} \bigl( #3 \bigr)}	%guessing function
\safemath{\trho}{\frac{1}{1+\rho}}
\safemath{\tirho}{\tilde \rho}
\newcommand{\psifun}[1]{\psi ( #1 )}	    %encoding function
\newcommand{\RDexp}{E^{(\rho)}_{X|Y} (P_{X,Y}, \Delta)}	    %rate-distortion guessing exponent
\newcommand{\RDfun}{R_{X|Y} (Q_{X,Y}, \Delta)}			%rate-distortion function
\newcommand{\relent}[2]{D (#1 || #2 )}			%relative entropy
\newcommand{\bigrelent}[2]{D \bigl(#1 \bigl| \bigl| #2 \bigr)}			%relative entropy
\newcommand{\Bigrelent}[2]{D \Bigl(#1 \Bigl| \Bigl| #2 \Bigr)}			%relative entropy